\documentclass[oneside,12pt]{IIScthesisPSnPDF}
\usepackage{etoolbox}
\usepackage{graphicx}
\usepackage{epsfig}
\usepackage{amsfonts}
\usepackage{amssymb}
\usepackage{mathrsfs}
\usepackage{hyperref}
\usepackage{amsmath}
\usepackage{url}
\usepackage{amsthm}
\usepackage{rotating}
\usepackage{tikz}
\usepackage{filecontents}
\usepackage{import}
\usepackage{verbatim}
\usepackage{wrapfig}
\usepackage{slashbox}
\usepackage{array}
\usepackage{multirow}
\usepackage{calc}
\usepackage{blindtext}
\usepackage{pifont}
\usepackage{textcomp}
\usepackage[ruled,vlined]{algorithm2e}
\usepackage{rotating}
\usepackage{lipsum}
\usepackage{enumitem}
\usepackage{latexsym}
\usepackage{multicol}
\usepackage{booktabs}
\usepackage{fancyhdr}
\usepackage[numbers]{natbib}
\usepackage{cleveref}
\usepackage{cancel}
\usepackage{mathrsfs}
\usepackage[compact]{titlesec}
\usepackage{mdwlist}
\usepackage{tikz}
\usepackage{varwidth}
\usepackage[vertfit]{breakurl}
\usepackage{datetime}
\usepackage{pdfpages}
\usepackage{appendix}
\usepackage{chngcntr}
\usepackage{lipsum}

\newtoggle{clr}
\toggletrue{clr}

\AtBeginEnvironment{subappendices}{%
}

\newtheorem{definition}{Definition}[chapter]
\newtheorem{proposition}{Proposition}[chapter]
\newtheorem{theorem}{Theorem}[chapter]
\newtheorem{lemma}{Lemma}[chapter]
\newtheorem{corollary}{Corollary}[chapter]
\newtheorem{example}{Example}[chapter]
\newtheorem{observation}{Observation}[chapter]

\newtheorem{property}{Property}[chapter]
\newtheorem{remark}{Remark}[chapter]

\newcount\Comments  
\newcommand{\kibitz}[2]{\ifnum\Comments=1\textcolor{#1}{#2}\fi}
\newcommand{\sg}[1]  {\ifnum\Comments=1      {[SG: #1]}\fi}
\newcommand{\yn}[1]  {\ifnum\Comments=1      {[YN: #1]}\fi}
\newcommand{\ads}[1]  {\ifnum\Comments=1      {[ADS: #1]}\fi}
\newcommand{\dcp}[1]  {\ifnum\Comments=1      {[DCP: #1]}\fi}
\newcommand{\jz}[1]  {\ifnum\Comments=1      {[JZ: #1]}\fi}

\newcommand\blfootnote[1]{%
  \begingroup
  \renewcommand\thefootnote{}\footnote{#1}%
  \addtocounter{footnote}{-1}%
  \endgroup
}

\newcommand{\squishlist}{\begin{itemize}}
\newcommand{\squishend}{\end{itemize}}

\newcommand\T{\rule{0pt}{2.4ex}}
\newcommand\B{\rule[-1.2ex]{0pt}{0pt}}

\DeclareMathOperator*{\argmax}{arg\,max}
\DeclareMathOperator*{\argmin}{arg\,min}
\DeclareMathOperator*{\avg}{avg\,}

\newtheorem{innercustomthm}{Theorem}
\newenvironment{customthm}[1]
  {\renewcommand\theinnercustomthm{#1}\innercustomthm}
  {\endinnercustomthm}
\newtheorem{innercustomprop}{Proposition}
\newenvironment{customprop}[1]
  {\renewcommand\theinnercustomprop{#1}\innercustomprop}
  {\endinnercustomprop}
\newtheorem{innercustomlem}{Lemma}
\newenvironment{customlem}[1]
  {\renewcommand\theinnercustomlem{#1}\innercustomlem}
  {\endinnercustomlem}

\newdateformat{monthyeardate}{ \monthname[\THEMONTH], \THEYEAR}

\newcommand{\blankpage}{
\newpage
\thispagestyle{empty}
\mbox{}
\newpage
}

\newcommand{\blankpagewithnumber}{
\newpage
\mbox{}
\newpage
}

\crefname{observation}{observation}{observations}
\crefname{algorithm}{algorithm}{algorithms}
\crefname{align}{equation}{equations}
\crefname{eqnarray}{equation}{equations}

\begin{document}

\title{New Models and Methods for Formation and Analysis of\\Social Networks} 

\submitdate{September, 2016} 
\phd
\dept{Computer Science and Automation}
\faculty{Faculty of Engineering}
\author{Swapnil Dhamal}


\maketitle
\blankpage
\pagenumbering{gobble}
%
%
%
%
%
%
%
%
%
%
%
%

\vspace*{\fill}
\begin{center}
\large\bf \textcopyright \ Swapnil Dhamal\\
\large\bf September, 2016\\
\large\bf All rights reserved
\end{center}
\vspace*{\fill}
\thispagestyle{empty}

\blankpage

\vspace*{\fill}
\begin{center}
DEDICATED TO \\[2em]
\Large\it My parents, sister, and grandparents for their love and support,\\[1em]
\Large\it Narahari sir for making my Ph.D. an enjoyable experience,\\[1em]
\Large\it And all animals and animal lovers
\end{center}
\vspace*{\fill}
\thispagestyle{empty}
\blankpage



\setcounter{secnumdepth}{3}
\setcounter{tocdepth}{3}

\frontmatter 
\pagenumbering{roman}

\prefacesection{Acknowledgments}

First and foremost, I thank Prof. Y. Narahari for being an ideal guide, by allowing me to explore and think freely, suggesting me problems which would be of interest to the community, continuously keeping me motivated to come up with better solutions and approaches, inspiring me to have a firm goal in life, and supporting and encouraging me during times of paper rejections. I am really grateful to him for giving me a chance to pursue Ph.D. under his guidance at the time when I felt the need of proving myself to myself. I also thank him for creating a great atmosphere in the department under his chairmanship. He has guided many students towards achieving their goals and I wish that he continues to guide many more for many more years.

I thank Prabuchandran K. J., Satyanath Bhat, and Jay Thakkar for being alongside me for most part of my IISc life and also for readily participating in crazy games and activities invented by me.
I am greatly grateful to Shweta Jain and Manish Agrawal for all the awesome time at trips, meals, and games;
Rohith D. Vallam for always being there like an elder brother;
Vibhuti Shali for being my ultimate partner in craziness and senselessness;
Monika Dhok for letting me freely express my anger and frustration at times;
Divya Padmanabhan for giving me moral and emotional support in times of need;
Ganesh Ghalme for making me laugh in any situation;
Surabhi Akotiya for being the awesomest company to share my stories and experiences with;
Sujit Gujar for showing that age is not a criterion for being young; and
Amleshwar Kumar for inspiring me to try and work towards being phenomenal.

I thank my colleagues with whom I got a chance to work with on a number of problems, projects, and papers. Specifically, I would like to thank Rohith D. Vallam, Prabuchandran K. J., Akanksha Meghlan, Surabhi Akotiya, Shaifali Gupta,  Satyanath Bhat, Anoop K. R., and Varun Embar for lengthy yet interesting technical discussions and tireless joint efforts.
Being a mentor for {\em Game Theory} course projects, I thank the group members of the projects allotted to me, namely,
Sravya Vankayala, Aakar Deora, Pankhuri Sai, Arti Bhat, Nilam Tathawadekar, Cressida Hamlet, Marilyn George, Aiswarya S, Chandana Dasari,
Bhupendra Singh Solanki, Hemanth Kumar N., Rakesh S., Sumit Neelam,
Mrinal Ekka, Lomesh Meshram, Arunkumar K., Narendra Sabale,
Avinash Mohan, Samrat Mukhopadhyay, Avishek Ghosh,
Debmalya Mandal, Durga Datta Kandel, Gaurav Chaudhory, and Pramod M. J.
Being a teaching assistant for the undergraduate course of {\em Algorithms and Programming} enabled me to interact with enthusiatic youngsters, while further developing my own algorithmic and programming skills.

I am extremely thankful to Dinesh Garg and Ramasuri Narayanam for some of the most useful feedback on my research.
I thank Shiv Kumar Saini, Balaji Vasan Srinivasan, Anandhavelu N., Arava Sai Kumar, and Harsh Jhamtani, from Adobe research labs for useful discussions and feedback during meetings on various topics related to marketing using social networks.
I also thank the many lab visitors and anonymous paper reviewers for their invaluable feedback.

There were lots of hands behind the Facebook app that we developed. In particular, I thank Rohith D. Vallam for doing most of the coding and Akanksha Meghlan for helping kick-start the user interface.
I also thank Mani Doraisamy for technical help regarding Google App Engine;
Tharun Niranjan and Srinidhi Karthik B. S. for helping with the code; and
Nilam Tathawadekar, Cressida Hamlet, Marilyn George, Aiswarya Sreekantan, and Chandana Dasari for helping in designing the incentive scheme for the app as well as in enhancing its user interface.
I would like to acknowledge 
Google India, Bangalore (in particular Mr. Ashwani Sharma) for providing us free credits for hosting our app on Google App Engine;
and
Pixabay.com which provides images free of copyright and other issues.
I am thankful to several of my CSA colleagues for testing and lots of suggestions.
I also thank Prof. Ashish Goel, Prof. Sarit Kraus, Prof. V. S. Subrahmanian, and Prof. Madhav Marathe for their useful feedback on the Facebook app.

{The images in the introduction chapter of this thesis are created using raw images obtained from sites such as {deviantart.com}, youtube.com, flickr.com, google.com, and blogdoalex.com.}

I am thankful to the members of Game Theory lab, past and present, for creating a great environment for doing research while having fun. In particular, I thank Rohith D. Vallam for being, as what Narahari sir says, the pillar of the lab in many ways;
 Satyanath Bhat, Sourav Medya, Ashutosh Verma, and Ratul Ray for being some of the best batchmates having joined the lab with me;
Shweta Jain, Priyanka Bhatt, and Akanksha Meghlan for lighting up the lab atmosphere with their unbounded energy and enthusiasm;
Surabhi Akotiya, Vibhuti Shali, Ganesh Ghalme, and Shaifali Gupta for creating a fun-loving environment;
Debmalya Mandal for being the center of discussions on a variety of topics;
Swaprava Nath for being a practical guide for a variety of issues;
Pankaj Dayama and Praphul Chandra for attending lab talks to give feedback despite their busy job schedules; 
Amleshwar Kumar for adding a new flavor to the lab by being one of the most knowledgeable persons despite being in his undergraduate stage;
and
Sujit Gujar and Chetan Yadati for adding a lot of value to the lab as research associates.
I also thank other lab members, including but not limited to,
Divya Padmanabhan,
Palash Dey,
Arpita Biswas, 
Arupratan Ray,
Aritra Chatterjee,
Sneha Mondal,
Chidananda Sridhar, 
Santosh Srinivas,
Dilpreet Kaur,
Thirumulanathan Dhayaparan,
Manohar Maddineni,
Moon Chetry, and
Shourya Roy.

I am utmost thankful to my pre-IISc friends, namely, Abhishek Dhumane, Mihir Mehta, Amol Walunj, Tejas Pradhan, Siddharth Wagh, Nilesh Kulkarni, Anupam Mitra, and Praneet Mhatre for taking me outside my Ph.D. world every once in a while.
I thank Priyanka Agrawal for being my first friend in IISc;
Chandramouli Kamanchi, Sindhu Padakandla, and Ranganath Biligere for being a great company every once in a while; and Shruthi Puranik for being a great animal-loving friend with whom I could freely talk about animal welfare.
I also thank Anirudh Santhiar, Narayan Hegde, Joy Monteiro, Aravind Acharya, Prachi Goyal, Pallavi Maiya, Thejas C. R., Shalini Kaleeswaran, Pranavadatta Devaki, Lakshmi Sundaresh, Ninad Rajgopal, K. Vasanta Lakshmi, Roshan Dathathri, Chandan G., Prasanna Pandit, Malavika Samak, Vinayaka Bandishti, Suthambhara N., and others, with whom I shared dinner table at times while listening to their out-of-the-world stories. For infrequent yet fun interactions, I thank Chandrashekar Lakshminarayanan, Brijendra Kumar, Akhilesh Prabhu, Vishesh Garg, Goutham Tholpadi, Divyanshu Joshi, Ankur Miglani, Pramod Mane, Surabhi Punjabi, Disha Makhija, Lawqueen Kanesh, Varun, Nikhil Jain, Saurabh Prasad, Ankur Agrawal, the family of Karmveer, Bharti, and Manan Sharma, among others.

I thank IISc for providing great facilities and environment for high-quality research. I would also like to thank its administration, security, mess and canteen workers, cleaners, among others, for making my stay in IISc a very pleasant one.
I thank CSA office staff and security for enabling smooth functioning of our department. In particular, I thank Sudesh Baidya for creating a friendly atmosphere whenever I entered and exited the department while doing his job with utmost sincerity; and also Lalitha madam, Suguna madam, Meenakshi madam, and Kushal madam for being a highly efficient and friendly office staff.
I would also like to thank IBM Research Labs for awarding me Ph.D. fellowship for the years 2013-2015.

Ph.D. being a fairly long journey during one's learning phase, one is bound to learn some life lessons on the way. I specifically thank
Vibhuti Shali for teaching me that one should stick to what one feels right without worrying about what others think;
Monika Dhok for making me understand that one should either solve or forget a problem and not complicate it;
Shweta Jain for constantly guiding me on how to and how not to behave in social situations;
Prabuchandran K. J. for making me understand that one should not have expectations from others beyond a certain limit;
Satyanath Bhat and Divya Padmanabhan for teaching me the importance of having lots of contacts since one cannot wholly rely on a selected few;
and Sireesha and Shravya Yakkali for teaching me the importance of being positive and successful in life.

I would like to appreciate the Indian cricket team for time and again showing me and everyone else, that nothing is impossible and that we can always spring back from any situation. Specifically, I am grateful to Sir Ravindra Jadeja, Mahendra Singh Dhoni, Virat Kohli, and Suresh Raina for showing the importance of working with dedication, keeping calm in the worst of situations, focusing on what is at hand, and being a good team player.

Finally and most importantly, I express my gratitude to my family, especially my parents, Sudarshana and Vilas Dhamal, and my sister, Swati Dhamal, for unconditional love and constant support throughout my life, and making it possible for me to enjoy my Ph.D. journey in the best possible way.

\prefacesection{Abstract}
Social networks are an inseparable part of human lives, and play a major role in a wide range of activities in our day-to-day as well as long-term lives. 
The rapid growth of online social networks has enabled people to reach each other, while bridging the gaps of geographical locations, age groups, socioeconomic classes, etc. 
It is natural for government agencies, political parties, product companies, etc. to harness social networks for planning the well-being of society, launching effective campaigns, making profits for themselves, etc.
Social networks can be effectively used to spread a certain information so as to increase the sales of a product using word-of-mouth marketing, to create awareness about something, to influence people about a viewpoint, etc.
Social networks can also be used to know the viewpoints of a large population by knowing the viewpoints of only a few selected people; this could help in predicting outcomes of elections, obtaining suggestions for improving a product, etc.
The study on social network formation helps us know how one forms social and political contacts, how terrorist networks are formed, and how one's position in a social network makes one an influential person or enables one to achieve a particular level of socioeconomic status. 

This doctoral work focuses on three main problems related to social networks:
\begin{itemize}
\item 
{\em Orchestrating Network Formation}: 
We consider the problem of orchestrating formation of a social network having a certain given topology that may be desirable for the intended usecases. Assuming the social network nodes to be strategic in forming relationships, we derive conditions under which a given topology can be uniquely obtained. We also study the efficiency and robustness of the derived conditions.

\item 
{\em Multi-phase Influence Maximization}: 
We propose that information diffusion be carried out in multiple phases rather than in a single instalment. With the objective of achieving better diffusion, we discover optimal ways of splitting the available budget among the phases, determining the time delay between consecutive phases, and also finding the individuals to be targeted for initiating the diffusion process.

\item 
{\em Scalable Preference Aggregation}: 
It is extremely useful to determine a small number of  representatives of a social network such that the individual preferences of these nodes, when aggregated, reflect the aggregate preference of the entire network. Using real-world data collected from Facebook with human subjects, we discover a model that faithfully captures the spread of preferences in a social network. We hence propose fast and reliable ways of computing a truly representative aggregate preference of the entire network.
\end{itemize}
%
%

\noindent
In particular, we develop models and methods for solving the above problems, which primarily deal with formation and analysis of social networks. 

\prefacesection{Publications based on this Thesis}

\begin{itemize}
\item
\textbf{Forming networks of strategic agents with desired topologies.} \\
Swapnil Dhamal and Y. Narahari. \\
{\em Internet and Network Economics (WINE)},
  Lecture Notes in Computer Science, pages 504--511. Springer Berlin
  Heidelberg, 2012.
  
  \item
  \textbf{Scalable preference aggregation in social networks.} \\
  Swapnil Dhamal and Y. Narahari. \\
{\em First AAAI Conference on Human Computation and
  Crowdsourcing (HCOMP)}, pages 42--50. AAAI, 2013.

\item
\textbf{Multi-phase information diffusion in social networks.} \\
Swapnil Dhamal, Prabuchandran~K.~J., and Y. Narahari. \\
{\em 14th International Conference on Autonomous Agents \& Multiagent Systems (AAMAS)},
pages 1787--1788. IFAAMAS, 2015.
\end{itemize}

\begin{itemize}

\item
\textbf{Formation of stable strategic networks with desired topologies.} \\
  Swapnil Dhamal and Y. Narahari. \\
  {\em Studies in Microeconomics}, 3(2): pages 158--213. Sage Publications, 2015.
\end{itemize}

\begin{itemize}
\item
\textbf{Information diffusion in social networks in two phases.} \\
Swapnil Dhamal, Prabuchandran~K.~J., and Y. Narahari. \\
{\em Transactions on Network Science and Engineering}, 3(4): pages 197--210. IEEE, 2016.

  \item
  \textbf{Modeling spread of preferences in social
                  networks for sampling-based preference
                  aggregation.} \\
  Swapnil Dhamal, Rohith D. Vallam, and Y. Narahari. \\
{\em Journal Submission}, 2016.
\end{itemize}

\blankpagewithnumber
\tableofcontents
\prefacesection{Acronyms and Notation}
\label{acronyms and notation}

%
%
%

\noindent
{\bf Chapter~\ref{chap:nfsc} \\ \nameref{chap:nfsc}}\\
\vspace{-0.3cm}
\hrule
\begin{tabbing}

xxxxxxxxxxxx \= \kill
$g$ \> A social graph or network \\
$N$ \> Set of nodes present in a given social network \\
$n$ \> Number of nodes present in a given social network \\
$e$ \> Typical edge in a social network \\
$u_j(g)$ \> The net utility of node $j$ in social network $g$ \\
$l(j,w)$ \> The shortest path distance between nodes $j$ and $w$ \\
$b_i$ \> The benefit obtained from a node at distance $i$ in absence of rents \\
$c$ \> The cost for maintaining link with an immediate neighbor \\
$c_0$ \> The network entry factor \\
$d_j$ \> The degree of node $j$ \\
$E(x,y)$ \> Set of essential nodes connecting nodes $x$ and $y$ \\
$e(x,y)$ \> Number of essential nodes connecting nodes $x$ and $y$ \\
$\gamma$ \> The fraction of indirect benefit paid to corresponding set of essential nodes \\  
$\text{T}(j)$ \> The node to which node $j$ connects to enter the network \\
$\textbf{I}_{\{j=\text{NE}\}}$ \> 1 when $j$ is a newly entering node about to create its first  link,  else 0 \\ 
GED \> Graph edit distance \\
$\Delta$ \> highest degree in the graph \\

\end{tabbing}

\noindent
{\bf Chapter~\ref{chap:mpid} \\ \nameref{chap:mpid}}\\
\vspace{-0.3cm}
\hrule
\begin{tabbing}

xxxxxxxxxxxx  \= \kill
$G$ \> A social graph or network \\
$N$ \> Set of nodes in a given social network \\
$n$ \> Number of nodes in a given social network \\
$E$ \> Set of weighted and directed edges in a given social network \\
$m$ \> Number of weighted and directed edges in a given social network \\
$(u,v)$ \> A typical directed edge \\
$\mathcal{P}$ \> Set of probabilities associated with the edges \\
$p_{uv}$ \> The weight of edge $(u,v)$ \\
IC \> Independent Cascade \\
LT \> Linear Threshold \\
WC \> Weighted Cascade \\
TV \> Trivalency \\
$k$ \> The total budget \\
$\mathbb{A}$ \> A typical influence maximizing algorithm \\
$k_1$ \> Number of seed nodes to be selected for the first phase \\
$k_2$ \> Number of seed nodes to be selected for the second phase \\
$d$ \> Delay between the first phase and the second phase \\
$D$ \> The length of the longest path in a graph \\
$X$ \> A typical live graph \\
$p(X)$ \> The probability of occurrence of a live graph $X$ \\
$S,T$ \> Typical sets of nodes \\
$\sigma^X(S)$ \> The number of nodes reachable from set $S$ in live graph $X$ \\
$\sigma(S)$ \> The expected number of influenced nodes in single phase diffusion with $S$ as \\ \> the seed set \\
$Y$ \> Partial observation of diffusion at beginning of the second phase \\
$\mathcal{A}^Y$ \> Set of already activated nodes as per observation $Y$ \\
$\mathcal{R}^Y$ \> Set of recently activated nodes as per observation $Y$ \\
$f(\cdot)$ \> The two-phase objective function \\
$\mathcal{MC}$ \> Monte-Carlo simulations \\
$\mathcal{M}$ \> Number of Monte-Carlo simulations \\
$\Delta$ \> The maximum degree in a graph \\
GDD \> Generalized Degree Discount \\
SD \> Single Discount \\
WD \> Weighted Discount \\
PMIA \> Prefix excluding Maximum Influence Arborescence \\
FACE \> Fully Adaptive Cross Entropy \\
SPIN \> ShaPley value based discovery of Influential Nodes \\
SPIC \> ShaPley value based discovery of influential nodes in IC model \\
$t$ \> A typical time step \\
$\sigma_{(t)}(S)$ \> The expected number of newly activated nodes at time step $t$ \\
$t_j^{X,S}$ \> The minimum number of time steps in which node $j$ can be reached from set $S$ \\ \> in live graph $X$ \\
$\Gamma(t)$ \> The value obtained for influencing a node a time step $t$ \\
$\delta$ \> Decay factor \\
$\nu(\cdot)$ \> The two-phase objective function with temporal constraints \\
$b_{u,v}$ \> The degree of influence that node $v$ has on node $u$ (specific to LT model) \\
$\chi_u$ \> The influence threshold of node $u$ (specific to LT model) \\

\end{tabbing}

\noindent
{\bf Chapter~\ref{chap:pasn} \\ \nameref{chap:pasn} \\}
\vspace{-0.3cm}
\hrule
\begin{tabbing}

xxxxxxxxxxxx  \= \kill
$G$ \> A social graph or network \\
$N$ \> Set of nodes in a given social network \\
$n$ \> Number of nodes in a given social network \\
$E$ \> Set of weighted and undirected edges in a given social network \\
$m$ \> Number of weighted and undirected edges in a given social network \\
$r$ \> Number of alternatives \\
$P$ \> Preference profile of set $N$ \\
$f$ \> Preference aggregation rule \\
$M$ \> Set of representatives who report their preferences \\
$k$ \> $|M|$, cardinality of $M$ \\
$\mathcal{D}$ \> Discretized Truncated Gaussian Distribution \\
$\mu_{ij}$ \> $\mu$ parameter of the Gaussian distribution from which distribution $\mathcal{D}$ between \\
\> preferences of nodes $i$ and $j$ is derived \\
$\sigma_{ij}$ \> $\sigma$ parameter of the Gaussian distribution from which distribution $\mathcal{D}$ between \\
\> preferences of nodes $i$ and $j$ is derived \\
$d(i,j)$ \> Expected distance between preferences of nodes $i$ and $j$ \\
$c(i,j)$ \> $1-d(i,j)$ \\
$\tilde{d}(x,y)$ \> Distance between preferences $x$ and $y$ \\
$\mathbb{T}$ \> Number of generated topics \\
$A_j$ \> Set of assigned neighbors of node $j$ \\
$deg(i)$ \> Degree of node $i$ \\
RPM \> Random Preferences Model \\
$\Phi(S,i)$ \> Representative of node $i$ in set $S$ \\
$R$ \> A generic preference profile of set $M$ \\
$Q$ \> Profile containing unweighted preferences of $M$ \\ 
$Q'$ \> Profile containing weighted preferences of $M$ \\ 
$\Delta$ \> Error operator between aggregate preferences \\ 
TU \> Transferable Utility \\
$\nu(\cdot)$ \> Characteristic function of a TU game \\

\end{tabbing}

\listoffigures
\listoftables
\blankpagewithnumber

\mainmatter 
\setcounter{page}{1}

\chapter[Introduction]{Introduction}
\label{chap:intro}

\begin{quote}
In this chapter, we provide the background and motivation for this thesis work. We bring out the context for this work  and describe the main contributions of the thesis. Finally we provide an outline and organization of the rest of the thesis.
\end{quote}


\noindent
We encounter networks everywhere in our lives, some we can see, some we can perceive, while some remain hidden or even unknown to us. They appear in a wide variety of domains, both natural and artificial,
ranging
from
biological networks which capture interactions among species in an ecosystem,
to
organizational networks which capture business and other relations among organizations,
from
transportation networks which capture the feasibilities of transportation between any two physical locations,
to
telecommunications networks which capture the feasibilities of communication between any two agents,
from
biological neural networks which capture a series of interconnections among neurons in a living being,
to 
artificial neural networks which are used to estimate functions that can depend on a large number of inputs,
from 
computer networks which capture the exchange of data among computing devices,
to 
social networks which capture connections and interactions among us.


Social networks have been an integral part of human lives for thousands of years, hence the term `social animal'. They are used by people for a variety of purposes, ranging from basic ones such as help in times of need, to extravagant one such as bonus points in online games. The development, the socioeconomic status, and the quality of life in general, of an individual heavily depends on the individual's social circle. So people think and act rationally while deciding who to have in their social circles and who not to have. 
An individual's social circle also determines his or her views and habits owing to the regular interactions involved. 
In addition, social networks act as a very effective means of spreading information since people readily listen to their friends and trust the information provided by them; this trust is either limited or absent when it comes to other means such as mass media.

\begin{figure}[t]
\centering
\includegraphics[scale=.4]{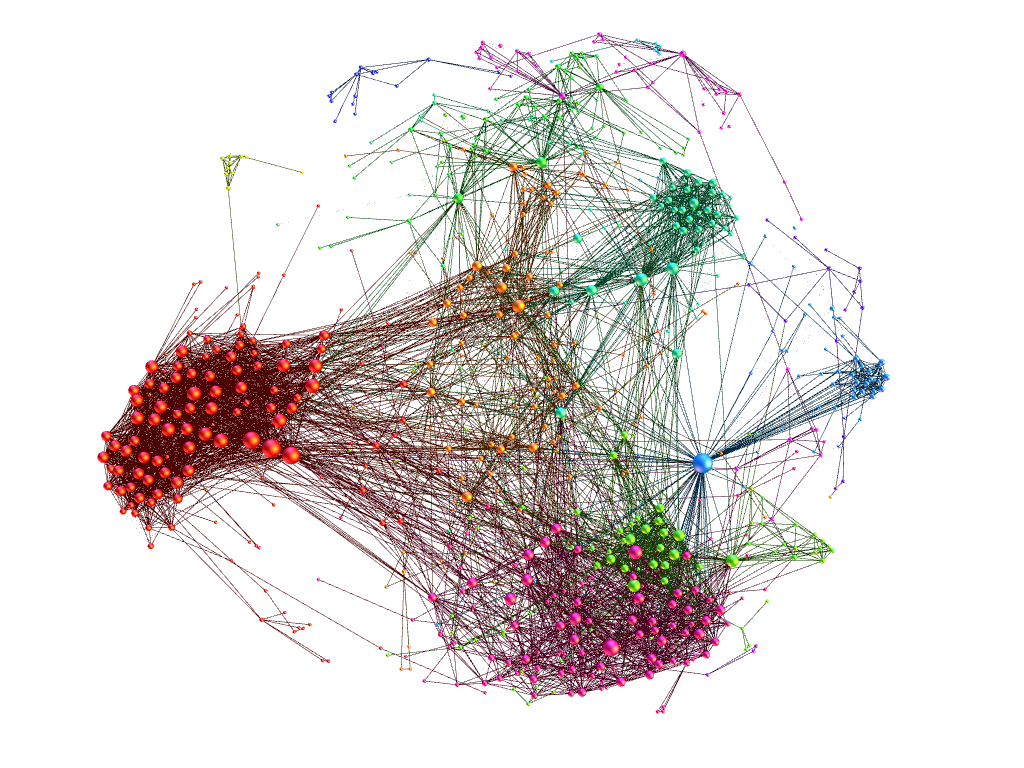}
\vspace{-5mm}
\caption{An example of a social network constructed from Facebook data}
\label{fig:intro_fig}
\end{figure}

In recent times, online social networks (OSNs) and social media have become very common and popular. Online social networks can be viewed as a digital version of real-world social networks. With over 1.2 billion users on Facebook as of 2016, the resourcefulness and effects of OSNs are something which cannot be ignored.
Figure~\ref{fig:intro_fig} shows a network that was created from Facebook data that was collected by us using a Facebook app developed as part of this thesis work.
The impact of OSNs on day-to-day routines of people is also growing day-by-day. Most people access their OSN accounts on a daily basis and it won't be an exaggeration to say that a large portion of them would become restless if they don't get to access them for a long time. Accessing their OSN accounts is one of the first activities after getting access to the Internet. One cannot definitively say if OSNs are a boon or a curse to human existence; this coin also has two sides like any other invention. But it is difficult to deny that one's presence on OSNs has started to become more of a necessity than a leisurely activity. 

\begin{center}
\begin{tabular}{p{0.9\textwidth}}
\textit{``We don't have a choice on whether we do social media, the question is how well we do it.''}
\\
\textit{ \hfill - Erik Qualman, author of Socialnomics}
\end{tabular}
\end{center}

\begin{center}
\begin{tabular}{p{0.9\textwidth}}
\textit{``Social media is like a snowball rolling down the hill. It's picking up speed. Five years from now, it's going to be the standard.''}
\\
\textit{ \hfill - Jeff Antaya (five years ago as of 2016), chief marketing officer at Plante Moran}
\end{tabular}
\end{center}
%
%
%
%
There have been several endorsements of OSNs since 
they help reach people easier and quicker especially in times of emergencies,
they enable people to come together and fight for a cause,
they allow people to express their opinions immediately to a wide audience, 
etc.
 \begin{center}
 \begin{tabular}{p{0.9\textwidth}}
 \textit{``Right now, with social networks and other tools on the Internet, all of these 500 million people (at the time) have a way to say what they're thinking and have their voice be heard.''}
 \\
 \textit{ \hfill - Mark Zuckerberg, co-founder of Facebook}
 \end{tabular}
 \end{center}
 \begin{center}
 \begin{tabular}{p{0.9\textwidth}}
 \textit{``The power of social media is that it forces necessary change.''}
 \\
 \textit{ \hfill - Erik Qualman, author of Socialnomics}
 \end{tabular}
 \end{center}
 On the other hand, there have been several criticisms since 
 they allow people to spread information immediately without checking its validity or usefulness, 
they have increased the exposure of personal information,
 they have lured people into spending their social time and effort for a large number of Internet friendships rather than a few good real-world friendships,
 etc.
  \begin{center}
  \begin{tabular}{p{0.9\textwidth}}
  \textit{``Twitter provides us with a wonderful platform to discuss and confront societal problems. We trend Justin Bieber instead.''}
  \\
  \textit{ \hfill - Lauren Leto, co-founder of TFLN and Bnter}
  \end{tabular}
  \end{center}
In the next section, we give a more elaborate introduction to social networks and their important properties.
 
\section{Characteristics of Social Networks}
\label{sec:basicprops}

Social networks exhibit a variety of properties, which have been consistently validated using empirical observations over a large number and variety of experiments since 1960s, across local and global networks, online and offline networks, friendship and collaboration networks, etc. 
To add further interest to these properties, these are not only followed by social networks, but several other networks such as network among Web pages, citation networks, etc.
We now describe some of these properties in brief.

\subsection{Birds of a Feather Flock Together}
\label{sec:homophily}

There is a natural tendency of individuals to form friendships with others who are similar to them.
The similarities could vary from belonging to the same race or geographical location, to sharing similar educational background or behavior.
On the other hand, people also change their mutable characteristics to align themselves more closely with the characteristics of their friends.
Examples of such mutable characteristics are behavior, health, attitude, etc.
In addition to these factors, there is a high likelihood of dissolution of friendships between dissimilar individuals.
So in a typical social network, there is a bias in friendships between individuals with similar characteristics. 
This phenomenon is termed {\em homophily}. 
%

\subsection{Friend of a Friend Becomes a Friend}
\label{sec:triadic}

It is typically observed that, ``If two people in a social network have a friend in common, then there is an increased likelihood that they will become friends themselves at some point in the future''.
This phenomenon is termed {\em triadic closure}.
One of the primary reasons for this effect is that, owing to having a common friend, the two people have a higher chance of meeting each other as well as a good basis for trusting each other. The common friend may also find it beneficial to bring the two friends together so as to avoid resources such as time and energy to be separately expended in the two friendships. If the closure is not formed, it is likely that one of the existing friendships would weaken or break, owing to resource sharing or stress between the two friendships.

\subsection{Weak Ties are Strong}
\label{sec:weakties}

It has been deduced by interviewing people, with great regularity that, ``their best job leads came from acquaintances rather than close friends''.
The reason is that people in a tightly-knit group (those connected to each other with strong ties) tend to have access to similar sources of information. On the other hand, weak ties or acquaintance connections enable people to reach into a different part of the network, and offer them access to useful information they otherwise may not have access to.
Hence in scenarios such as finding good job opportunities, {\em the strength of weak ties} comes into play.

\subsection{Managers Hold Powerful Positions}
\label{sec:strholes}

Empirical studies have correlated an individual's success within an organization to the individual's access to weak ties in the organizational network. Individuals with good access to weak ties typically hold managerial positions in organizations, allowing them to play the role of what are termed as {\em structural holes} in the network. They act as connections between two groups that do not otherwise interact closely.
The advantages of being in such a role are that, they have early access to information originating in non-interacting parts of the network, they have opportunities to develop novel ideas by combining these disparate sources of information in new ways, and they can regulate and control the information flow from one group to another, thus making their positions in the organization, powerful.

\subsection{Rich Get Richer}
\label{sec:degdist}

The {\em rich-get-richer} phenomenon states that an individual's value grows at a rate proportional to its current value. 
The value could be in terms of popularity, number of friends, socioeconomic status, etc. 
Social network structures also follow this rule where degree of an individual grows at a rate proportional to its current degree, where degree is defined as the number of direct connections that an individual has.
The reasoning behind this phenomenon is the natural tendency of people to connect to individuals who already have many connections, since it gives some indication of trust, usefulness, etc.
Furthermore, the degree distribution is observed to have a {\em long tail}, that is, though most individuals have low degrees, there exist individuals having extremely high degrees.

\subsection{It's a Small World}
\label{sec:smallworld}

Existence of short paths has been empirically observed on a consistent basis in global social networks. This phenomenon is termed as the {\em small-world} phenomenon since it takes a small number of friendship hops to connect almost any two people in this world. It is also popularly referred to as the {\em six degrees of separation} owing to experimental observations that one can reach anyone in this world within a certain number of hops, the median being six.
In fact, not only do short paths exist, but they are in abundance and people are observed to be effective at collectively finding these short paths.

\subsection{There is Core and There is Periphery}
\label{sec:coreper}

Large social networks tend to have a {\em core-periphery structure}, where high-status people are linked in a densely-connected core, while the low-status people are atomized around the periphery of the network. 
A primary reason is that high-status people have the resources to travel widely and establish links in the network that span geographic and social boundaries, while low-status people tend to form links that are more local. 
So two low-status people who are geographically or socially far apart, are typically connected through some high-status people in the core.
This property suggests that network structure is intertwined with status and the varied positions that different groups occupy in society.

\subsection{And It Cascades}
\label{sec:cascading}

Cascading is common in social networks wherein, some effect or information tends to propagate from one part of the network to another. Whether or not a successful cascade takes place, depends on certain conditions which are observed to follow a threshold.
It is often a result of individuals imitating others, even if it is inconsistent with their own information. For instance, one would prefer using a social networking site if most of his or her friends use it, despite knowing about an alternative site that offers better features. 
Financial crisis may result owing to the cascade of financial failures from one individual or organization to another. 
Information diffusion in another form of cascade (its effect can be judged based on the truthfulness and intensity of the information being propagated).
Viral marketing is yet another form of cascade used by companies, wherein individuals recommend their friends to buy a product, who in turn recommend their friends to buy the product, and so on.
\\

For detailed descriptions of the properties discussed above, the reader is referred to \cite{networkscrowdsmarkets,jacksonbook}.


\section{Specific Topics in Social Networks Investigated in this Thesis}

A typical social network goes through several events, the most important ones being the changes in its structure, the information diffusing over it, and the development of preferences of its individuals. These events are correlated, and in fact, one event often leads to the other. For instance, the structure of the network plays a leading role in determining how the information spreads among the individuals and how an individual's connections are likely to change his or her preferences over time. 
The information spreading over the network often changes the preferences of individuals and may determine how new links form and disappear over time leading to change in its structure.
Similarly, the preferences of individuals may lead to formation and deletion of connections and also change their influencing strengths thus altering the way information diffuses. 
This thesis investigates novel problems in the aforementioned three broad topics, which we now introduce in brief.

\subsection{Network Formation}

It is a part of human psychology to naturally want to be a part of the society, by forming connections with other people. The number and strengths of connections may vary to a great extent across people depending on a variety of criteria. Apart from the basic human psychology, connecting with other people goes a long way in helping an individual to develop intellectually, financially, emotionally, etc. In short, networking plays a key role for an individual to have a good quality of life. 
   \begin{center}
   \begin{tabular}{p{0.9\textwidth}}
   \textit{``If you want to go fast, go alone. If you want to go far, go together.''}
   \\
   \textit{ \hfill - African proverb}
   \end{tabular}
   \end{center}
   Networking plays an even more vital role in the current age, where socioeconomic status is of utmost importance. It is not only the number of strengths of the connections that matter, but also with whom the connections are made. Given this fact, knowingly or unknowingly, people are constantly in search of making suitable connections and also involved in the process of weakening or breaking connections which they deem unsuitable. 
   \begin{center}
   \begin{tabular}{p{0.9\textwidth}}
   \textit{``The richest people in the world look for and build networks, everyone else looks for work.''}
   \\
   \textit{ \hfill - Robert Kiyosaki, author of the Rich Dad Poor Dad series of books}
   \end{tabular}
   \end{center}
   Figure~\ref{fig:formation_pic} presents an illustration of links being formed and deleted in a social network.
          A more technical introduction to this topic is provided in Chapter~\ref{chap:nfsc}.
      \begin{figure}[t]
      \centering
  \includegraphics[scale=.23]{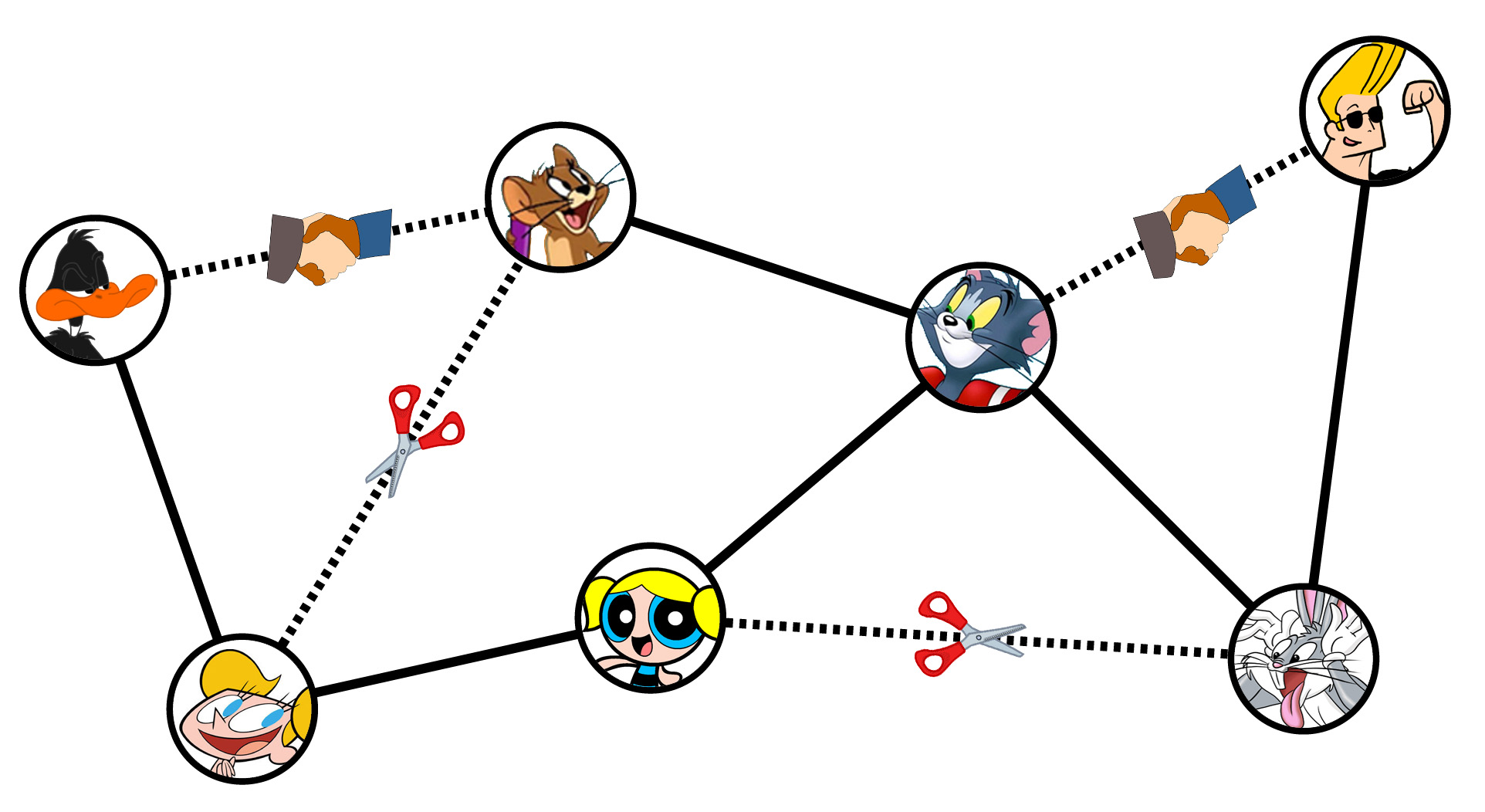}
   \caption{An illustration of links being formed and deleted in a social network
   }
      \label{fig:formation_pic}
      \end{figure}

\subsection{Information Diffusion}

People who are linked in a social network discuss several things and share various pieces of information, be it serious or casual, be it voluntarily or involuntarily, be it with an intention of diffusing it to a wider audience or keeping it private within a friend circle. In short, social networks play an important role in information diffusion and sharing. 
If one's objective is to diffuse certain information (or propagate an influence) so that it reaches a wider audience (or target) for whatever reasons, one cannot ignore the possibility of using social networks. Moreover, in the present age of Internet and the ever-increasing popularity of various social networking and social media sites, social networks have become an effective and efficient media for information diffusion. 

Given this property of a social network, it is natural for companies to exploit it to maximize the sales of their products. A primary method used by companies is based on viral marketing where the existing customers market the product among their friends. Campaigning is another example where a particular idea or a series of ideas is presented to some audience and hence, the idea is spread through the audience. 
   \begin{center}
   \begin{tabular}{p{0.9\textwidth}}
   \textit{``The goal of social media is to turn customers into a volunteer marketing army.''}
   \\
   \textit{ \hfill - Jae Baer, author of Youtility}
   \end{tabular}
   \end{center}
   \begin{figure}[t]
   \centering
  \includegraphics[scale=.22]{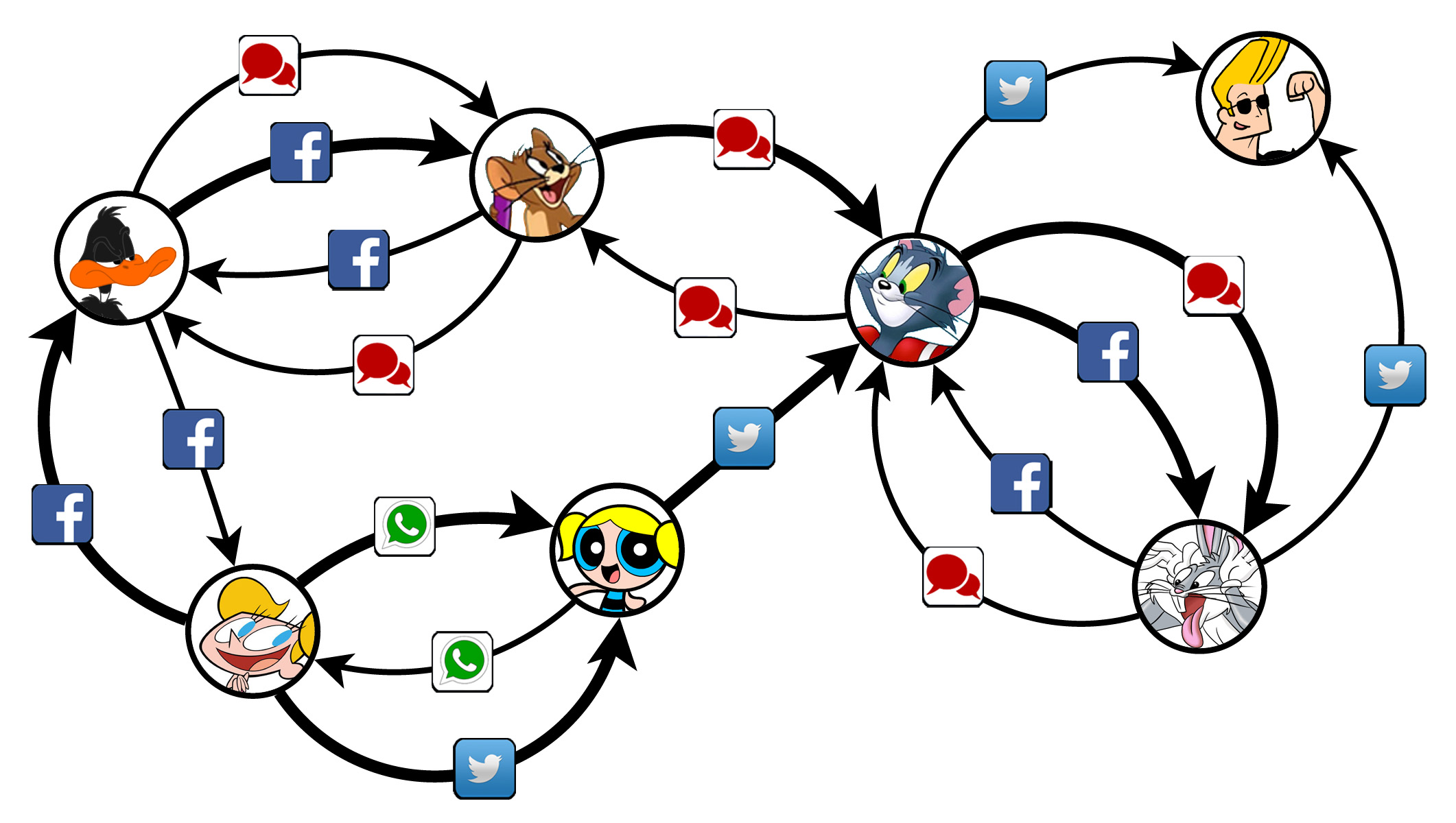}
   \caption{An illustration of information diffusing through social networks such as Facebook, Twitter, WhatsApp, offline interactions, etc. (the edge direction and thickness respectively represent the direction and strength of information propagation)
   }
   \label{fig:diffusion_pic}
   \end{figure}
Figure~\ref{fig:diffusion_pic} presents an illustration of information diffusing (or influence propagating) through various social networks.
A more technical introduction to this topic is provided in Chapter~\ref{chap:mpid}.

\subsection{Development and Spread of Preferences}

{\em Homophily} in social networks arises because of two complementary factors: similar individuals becoming friends and friendships leading to individuals becoming similar. An individual's friendship network and social connections plays a vital role in determining how the individual develops, qualitatively as well as quantitatively, owing to the influences resulting out of regular interactions. 
      \begin{center}
      \begin{tabular}{p{0.9\textwidth}}
      \textit{``Your social networks may matter more than your genetic networks. But if your friends have healthy habits, you are more likely to as well. So get healthy friends.''}
      \\
      \textit{ \hfill - Mark Hyman, founder of the UltraWellness Center}
      \end{tabular}
      \end{center}
            \begin{figure}[t]
         \centering
         \includegraphics[scale=.18]{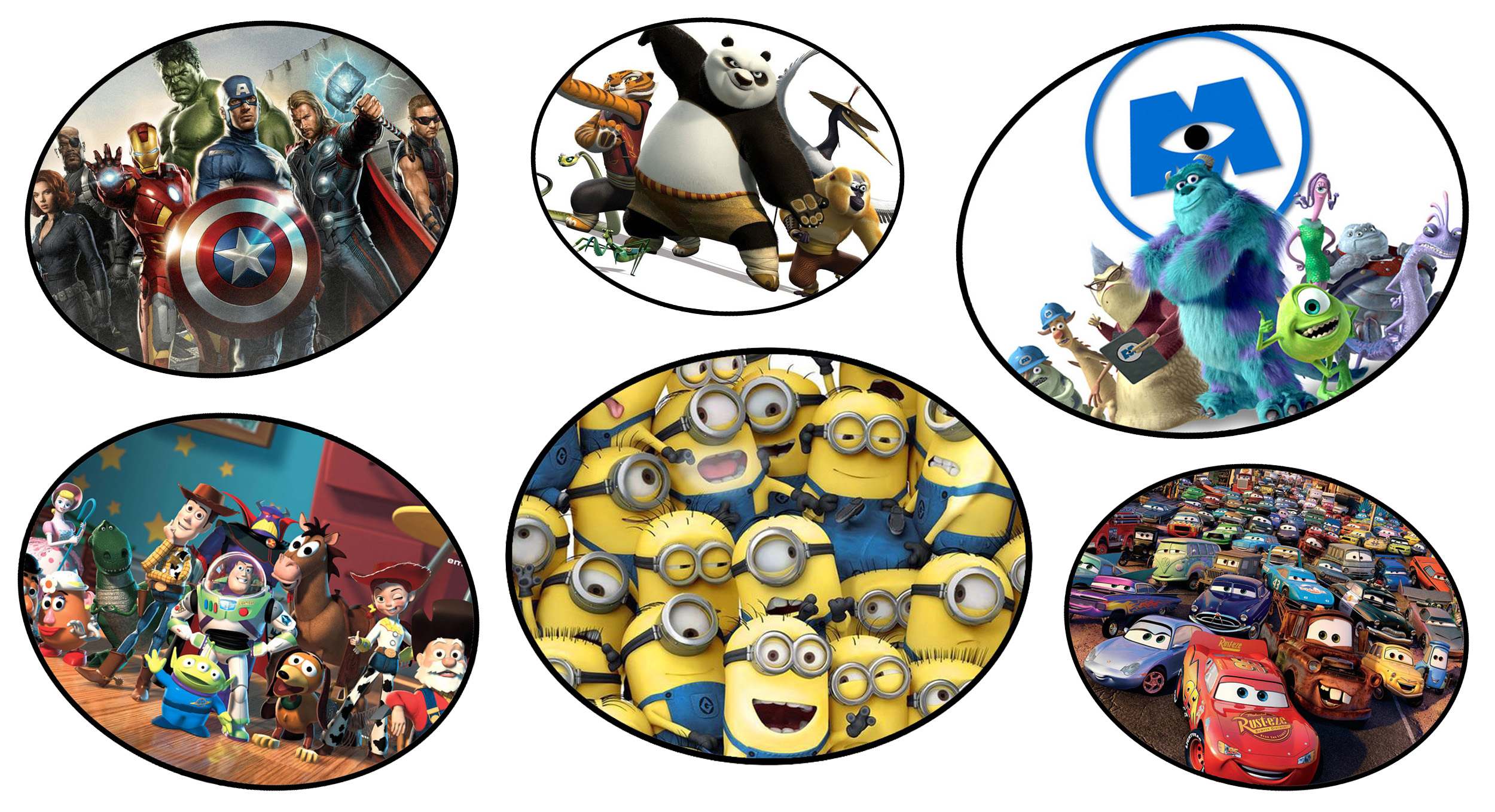}
            \caption{An illustration of the bias in friendships towards similar individuals (similarities lead to friendships and friendships lead to similarities)
            }
            \label{fig:homophily_pic}
            \end{figure}
The social interactions also help develop the preferences of an individual for a variety of topics, ranging from personal ones such as favorite hangout place, to social ones such as favorite political party. Consider the example of political viewpoint itself. Though factors such as mass media (such as news, campaigns, posters, etc.) influence an individual's viewpoints, a significantly bigger factor is the discussions and information shared with people in the individual's social circle, whom he or she trusts and shares common goals and vision with.
   \begin{center}
   \begin{tabular}{p{0.9\textwidth}}
   \textit{``Information is the currency of democracy.''}
   \\
   \textit{ \hfill - Thomas Jefferson, 3rd President of the United States}
   \end{tabular}
   \end{center}
   Figure~\ref{fig:homophily_pic} presents an illustration of the bias in friendships towards similar individuals.
   A more technical introduction to this topic is provided in Chapter~\ref{chap:pasn}.

\section{Overview of the Thesis}

This thesis revolves around social networks. 
In particular, it motivates three novel problems in the context of social networks and attempts to solve them to a great extent. Owing to the novel nature of the studied problems, there is a great scope for future work based on this thesis.
%

\subsection{Orchestrating Network Formation}

      \begin{figure}[t]
   \centering
   \vspace{-3.5mm}
   \includegraphics[scale=.24]{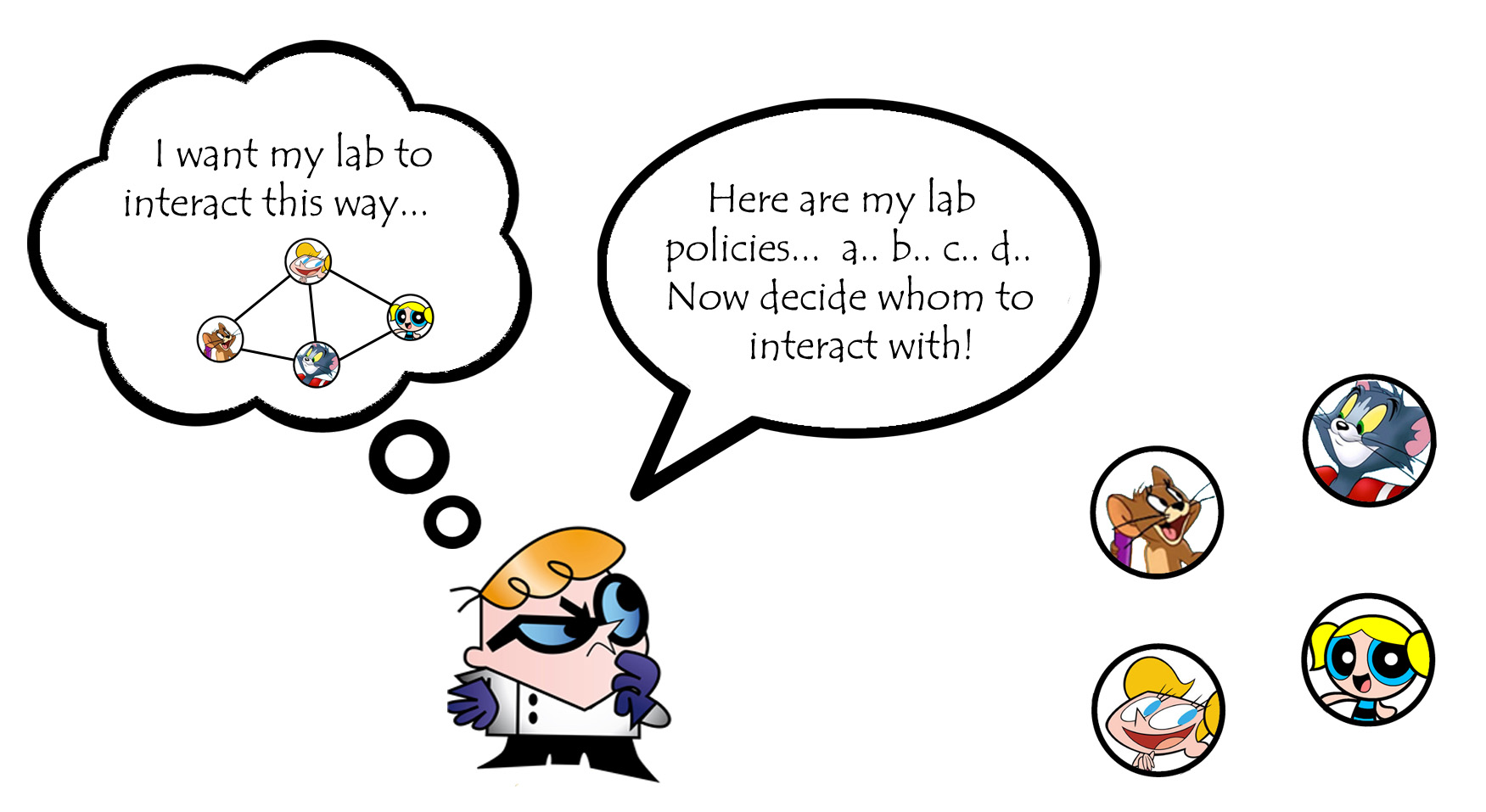}
      \caption{An illustration of orchestrating network formation
      }
      \label{fig:nfsc_problem}
      \end{figure}

Consider a scenario where an organization wants certain tasks to be completed with respect to knowledge management, information extraction, information diffusion, etc. It is known that the structure of the underlying interaction or communication network among the employees plays an important role in determining the ease and speed with which such tasks can be accomplished. In particular, the organization may want the network to be of a certain density, that is, the density should not be so low that it restricts the level of interaction and also not so high that an unreasonable amount of resources are spent for interactions alone. Also, it may want a good degree distribution so that the load of interactions is either borne or not borne by a selected few. In general, there may be a variety of reasons for an organization to prefer a particular network structure over others.

The employees in the organization, among whom the network is to be formed, are strategic and self-interested. While making connections with others, they consider how much they would benefit due to these connections and how much cost is involved in maintaining them. In an organizational setting, the benefits could be in the form of favors, information, discussions, etc., while the costs could be in the form of doing favors, sharing information, spending time and energy in discussion, etc. Employees would want to form connections with other employees such that they maximize the benefits while minimizing the costs at the same time. So if an organization desires to have a particular network structure among its employees, it needs to design its policies such that the employees find it best to direct the network structure towards the one desired by the organization.

Chapter~\ref{chap:nfsc} addresses the problem of deriving conditions under which autonomous link alteration decisions made by strategic agents lead to the formation of a desired network structure.
   Figure~\ref{fig:nfsc_problem} presents an illustration of orchestrating network formation.

\subsection{Multi-Phase Influence Maximization}

Consider a scenario where a company wants to market its newly launched product. There are several means of advertisement which the company can resort to. Viral marketing or word-of-mouth marketing is one such means in which the company offers the product to a few selected individuals for free or at a discounted price. These individuals can then suggest their friends to buy the product if they are satisfied with it. These friends can then decide whether to buy the product based on their individual criteria, who would then suggest their friends to buy the product provided it meets their level of satisfaction. 
The number of individuals to whom free or discounted products can be offered is determined by a certain budget allocated by the company for viral marketing. The company identifies these individuals based on criteria such as their influence and effectiveness of their suggestions on others. 

Viral marketing is known to be one of the most effective means of marketing owing to advertisement of the product by a trustworthy individual or a friend. 
 However, there are several uncertainties involved in viral marketing owing to the uncertainties in the behaviors of individuals involved in viral marketing. So it cannot be definitively said that triggering the viral marketing process at a selected set of individuals would be better than triggering it at some other.
 A natural solution to counter the ill-effects of such a randomness is to trigger the process, not at all the selected individuals, but just a fraction of them; and then make partial observations in the midst of the viral marketing process so as to select the subsequent individuals accordingly. 
However, a disadvantage of this multi-phase approach is that the process slows down owing to the delay involved in selecting subsequent individuals. This may be undesirable in presence of a competing product or when the value of the product decreases with time.

      \begin{figure}[t]
   \centering
   \includegraphics[scale=.25]{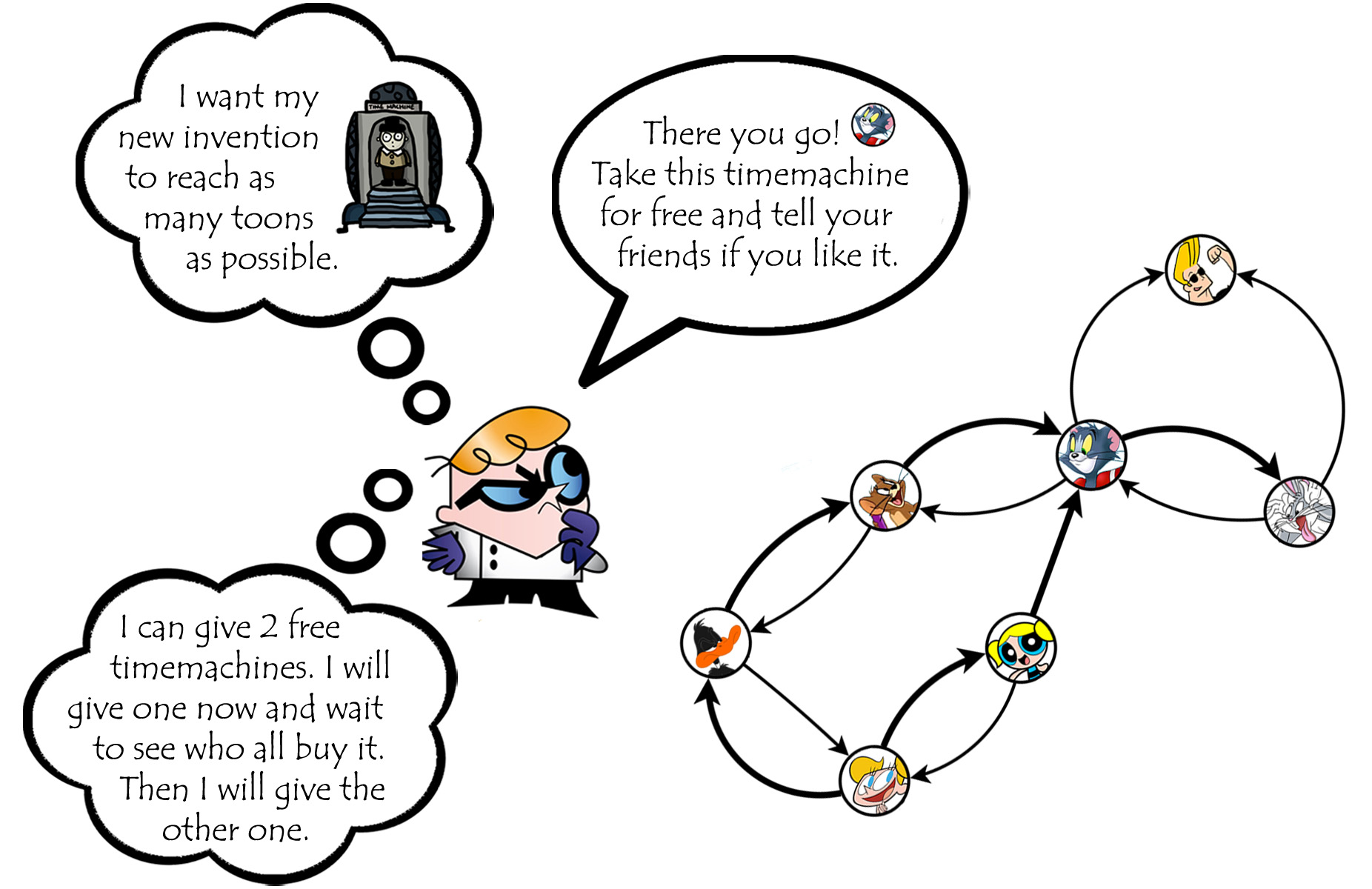}
      \caption{An illustration of multi-phase influence maximization
      }
      \label{fig:mpid_problem}
      \end{figure}

Chapter~\ref{chap:mpid} addresses the problem of effective viral marketing by determining how many and which individuals are to be selected in different phases, and what should be the delay between the phases.
   Figure~\ref{fig:mpid_problem} presents an illustration of multi-phase influence maximization.

\subsection{Scalable Preference Aggregation}

      \begin{figure}[t]
   \centering
   \includegraphics[scale=.27]{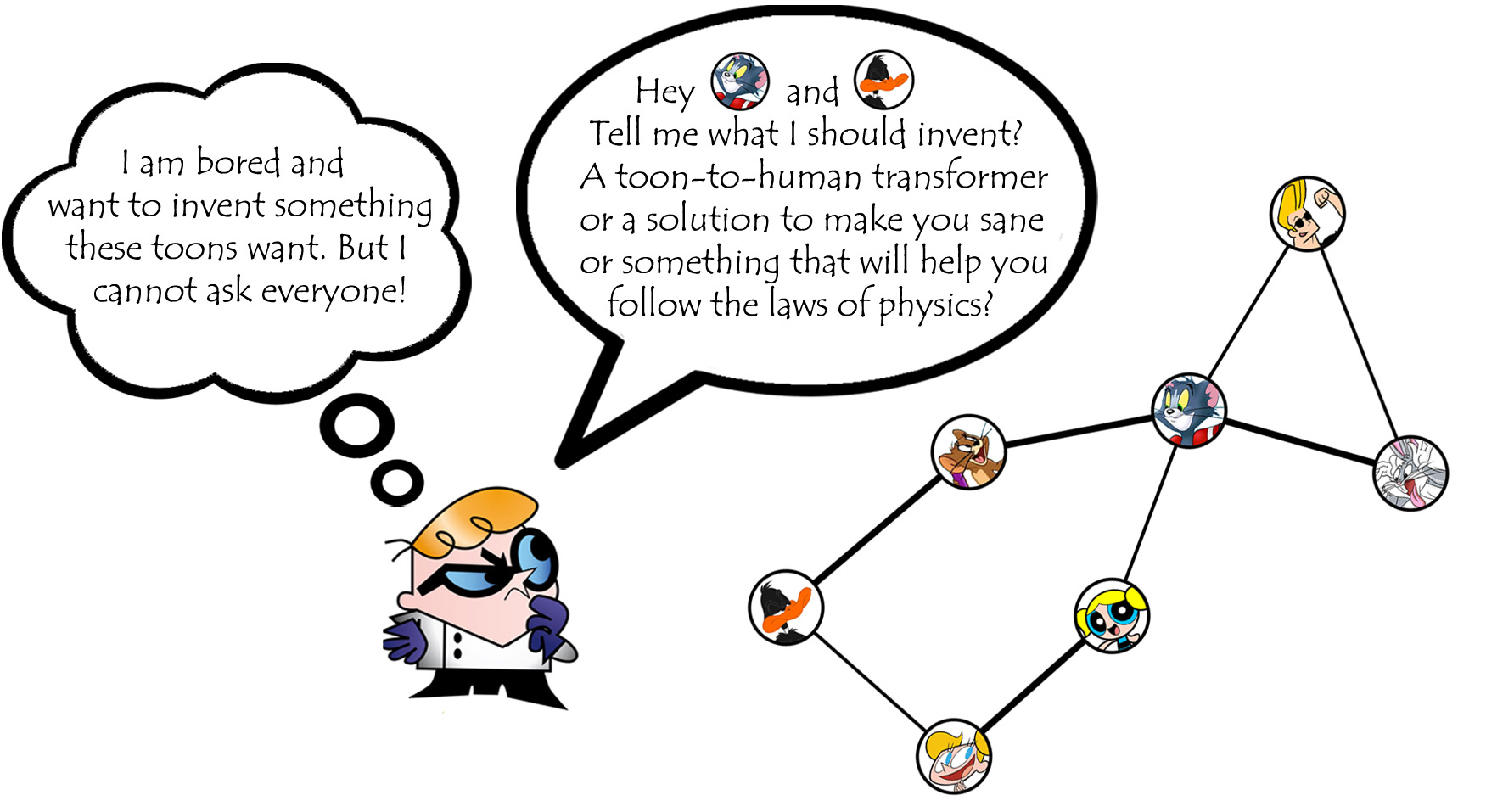}
      \caption{An illustration of scalable preference aggregation
      }
      \label{fig:pasn_problem}
      \end{figure}

Consider a scenario where a company wants to launch a new product based on the past experiences of its existing customers regarding its current products. The company would like to have the opinions of all its customers by sending a feedback request through fast means such as e-mail. However, such requests are not taken seriously by the customers and so the company may end up receiving only a small number of replies. In order to increase the participation, the company may offer incentives in some form to its customers, such as discount coupons and gift vouchers. The incentives need to be good, otherwise not many customers maybe willing to respond to such feedback queries in a prompt and honest way, and devote the required effort to provide a useful feedback. However, the company would have a certain budget for such incentives and would ideally like to arrive at a good balance between the investment for incentives and the level of participation. Even if the company ensures a good enough participation, it is also not clear if the feedback received from the participating customers is a good representative of the opinions of the entire customer base. So the company would ideally like to offer good incentives to a small number of customers, whose opinions would closely reflect the collective opinion of the entire customer base.

In the current age of electronic media, it is a common practice for companies to advise its customers to register their products using a registration website. The customers are given an option to either use their e-mail address or one of their other accounts such as Facebook, Google+, etc. for registration. It is also now becoming a common practice for people to use their online social networking accounts
for registrations on other websites. So a company can potentially obtain the social network underlying its customer base.
If information on the underlying social network is available to the company, it could harness the homophily property, and also deduce which of its customers are good representatives of the population. 

Chapter~\ref{chap:pasn} addresses the problem of determining the best representatives using the underlying social network data, by modeling
the spread of 
preferences among individuals in a social network.
   Figure~\ref{fig:pasn_problem} presents an illustration of scalable preference aggregation.

 \begin{sidewaystable}
 \hspace{-7mm}
 \begin{tabular}{||p{0.1\textwidth}||p{0.29\textwidth}||p{0.29\textwidth}||p{0.29\textwidth}||}
 \hline \hline 
 \vspace{-1mm}
 Topic 
 \vspace{2mm}
 &
 \vspace{-1mm}
 \centering
 Network Formation (Chapter~\ref{chap:nfsc})
 &
 \vspace{-1mm}
 \centering
 Information Diffusion (Chapter~\ref{chap:mpid})
 &
 \vspace{-1mm}
 ~Spread of Preferences (Chapter~\ref{chap:pasn})
 \\ \hline \hline
 \vspace{-1.5mm}
 Literature
 &
 \vspace{-2mm}
 Given conditions on network parameters, which topologies are likely to emerge?
 \vspace{1mm}
 &
 \vspace{-2mm}
 Given a social network, how to select the seed nodes so as to maximize diffusion?
 &
 \vspace{-2mm}
 How to select a best representative set for voting using attributes of nodes and alternatives?
 \\ \hline
 \vspace{-2mm}
 This thesis
 &
 \vspace{-2mm}
 Inverse of this problem
 \vspace{1mm}
 &
 \vspace{-2mm}
 Multi-phase version of this problem
 &
 \vspace{-2mm}
 Network view of this problem
 \\ \hline 
 \vspace{-1mm}
 Problem
 &
 \vspace{-2mm}
 Under what conditions would best response link alteration strategies of strategic agents lead to the formation of a stable network with a desired topology?
 \vspace{2mm}
 &
 \vspace{-2mm}
 For two-phase influence maximization in a social network, what should be the budget split and the delay between the two phases, and how to select the seed nodes?
 &
 \vspace{-2mm}
 How to select a best representative set using information about the underlying social network?
 \\ \hline 
 \vspace{-1mm}
 Approach
 &
 \vspace{-5mm}
 \begin{itemize}[leftmargin=*]
 \itemsep-.25em 
 \item A new network formation model
 \item A very general utility model 
 \item Derivations of these conditions for a range of topologies
 \item Efficiency of these conditions
 \item Robustness of these conditions
 \vspace{-2mm}
 \end{itemize}
 &
 \vspace{-5mm}
 \begin{itemize}[leftmargin=*]
 \itemsep-.25em 
 \item Objective function formulation
 \item Properties and analysis
 \item Algorithms for seed selection
 \item Results on real-world datasets
 \item Combined optimization over budget split, delay, and seed set
 \vspace{-2mm}
 \end{itemize}
 &
 \vspace{-5mm}
 \begin{itemize}[leftmargin=*]
 \itemsep-.25em 
 \item Facebook app for new dataset
 \item Models for spread of preferences in a social network
 \item Objective function formulation 
 \item Algorithms with guarantees
 \item Experimental results
 \vspace{-2mm}
 \end{itemize}
 \\ \hline 
 \vspace{-1mm}
 Conclusions
 &
 \vspace{-5mm}
 \begin{itemize}[leftmargin=*]
 \itemsep-.25em 
 \item Conditions on network entry impact degree distribution
 \item Conditions on link costs impact density
 \item Constraints on intermediary rents owing to contrasting densities of connections
 \vspace{-5mm}
 \end{itemize}
 &
 \vspace{-5mm}
 \begin{itemize}[leftmargin=*]
 \itemsep-.25em 
 \item Strict temporal constraints: use single phase 
 \item Moderate temporal constraints: most budget to first phase with a short delay between phases 
 \item No constraints: $\frac{1}{3}:\frac{2}{3}$ budget split with a long delay between phases
 \vspace{-5mm}
 \end{itemize}
 &
 \vspace{-5mm}
 \begin{itemize}[leftmargin=*]
 \itemsep-.25em 
 \item A sampling based method acts as a good model for spread of preferences
 \item People with high degree serve as good representatives
 \item Using social networks more effective, reliable than random polling
 \vspace{-5mm}
 \end{itemize}
 \\ \hline
 \hline 
 \end{tabular}
 \caption{Summary of the thesis}
 \label{tab:summary}
 \end{sidewaystable}

\section{Contributions and Outline of the Thesis}

This thesis starts by introducing the basics of social networks and the preliminaries required to follow the technical content, followed by the technical contributions and directions for future work. 
 Table~\ref{tab:summary} presents the summary.
We now present the thesis outline.

\subsection*{Chapter~\ref{chap:intro}: \nameref{chap:intro}}

This chapter presents a brief introduction to social networks and descriptions of specific topics of our interest, followed by motivations to the problems addressed in this thesis with the help of real-world examples, and then contributions and outline of this thesis.

\subsection*{Chapter~\ref{chap:prelims}: \nameref{chap:prelims}}

This chapter provides an overview of the following topics which are the prerequisites for
understanding later chapters of the thesis.
\begin{itemize}
\item \nameref{sec:graphtheory}:
isomorphism, automorphism, graph edit distance
\vspace{-2mm} \item \nameref{sec:gametheory}:
an example, subgame perfect equilibrium
\vspace{-2mm} \item \nameref{sec:netform}:
an example model, pairwise stability
\vspace{-2mm} \item \nameref{sec:infodiff}:
independent cascade model, linear threshold model
\vspace{-2mm} \item \nameref{sec:prefaggr}:
dissimilarity measures, aggregation rules
\vspace{-2mm} \item \nameref{sec:setfns}:
non-negativity, monotonicity, submodularity, supermodularity, subadditivity, superadditivity
\vspace{-2mm} \item \nameref{sec:modelprelims}:
KL divergence, RMS error, maximum likelihood estimation
\vspace{-2mm} \item \nameref{sec:opttools}:
greedy hill-climbing, cross entropy method, golden section search
\vspace{-2mm} \item \nameref{sec:cgt}:
the core, Shapley value, nucleolus, Gately point, $\tau$-value
\end{itemize}

\subsection*{Chapter~\ref{chap:nfsc}: \nameref{chap:nfsc}}

In this chapter, we study the problem of determining sufficient conditions
under which, the desired topology uniquely emerge 
when 
agents adopt their best response strategies.  
The chapter is organized as follows:

\begin{itemize} 
\item Section~\ref{sec:intro_nfsc} introduces social network formation, 
 Section~\ref{sec:motiv_nfsc} motivates the problem of orchestrating social network formation,
 Section~\ref{sec:relevant_nfsc} presents some relevant work, and
 Section~\ref{sec:gameinbrief} enlists the contributions of this chapter.

\item In Section~\ref{sec:model}, we propose a recursive model of  network formation 
and a very general
  utility model that captures most key aspects relevant to strategic network formation.
We then present our procedure for deriving sufficient conditions for the formation of a given topology as the unique one. 
\item In Section~\ref{sec:analysis}, using the proposed models, we study common and important network topologies,
and derive sufficient conditions under which these topologies uniquely emerge.
We also investigate the social welfare properties of these topologies. 
\item In Section~\ref{sec:deviation}, we study the effects of deviation 
from the derived sufficient conditions on the resulting network, using the notion of graph edit distance. 
In this process, we develop polynomial time algorithms for computing graph edit distance from certain topologies.
\item Section~\ref{sec:conclusion_nfsc} concludes the chapter.
\end{itemize}

\subsection*{Chapter~\ref{chap:mpid}: \nameref{chap:mpid}}

In this chapter, we study the problem of influence maximization using multiple phases, in particular, determining an optimal budget split among the phases and their scheduling, and also determining an optimal set of seed nodes so that the resulting influence is maximized.
The chapter is organized as follows:
\begin{itemize} 
\item Section~\ref{sec:intro_mpid} introduces information diffusion in social networks,
 Section~\ref{sec:relevant_mpid} presents some relevant work,
 Section~\ref{sec:motiv_mpid} motivates the problem of multi-phase influence maximization, and
 Section~\ref{sec:contrib_mpid} enlists the contributions of this chapter.

\item In Section~\ref{sec:problem_mpid}, focusing on two-phase diffusion process in social networks,
we formulate an appropriate objective function that measures the expected number of influenced nodes, and investigate its properties. 
We then motivate and propose an alternative objective function for ease and efficiency of practical implementation. 
\item In Section~\ref{sec:algo}, we investigate different candidate  algorithms for two-phase diffusion including extensions of
existing algorithms that are popular for single phase diffusion. 
\item In Section~\ref{sec:simulations}, using extensive simulations on real-world datasets, we study the performance of the proposed algorithms to get an idea how two-phase diffusion would perform, even when used most na\"ively. 
\item In Section~\ref{sec:practical}, 
we focus on two constituent problems, namely, how to split the total available budget between the two phases, and when to commence the second phase. 
We then present key insights from a detailed simulation study. 
\item Section~\ref{sec:conclusion_mpid} concludes with discussion and some notes.
\end{itemize}

\subsection*{Chapter~\ref{chap:pasn}: \nameref{chap:pasn}}

In this chapter, we study the problem of determining a set 
of representative nodes of a given cardinality such that, the aggregate preference of the nodes in this set closely approximates the aggregate preference of the entire population. 
The chapter is organized as follows:

\begin{itemize}
\item
 Section~\ref{sec:intro_pasn} introduces preference aggregation,
 Section~\ref{sec:motiv_pasn} motivates the problem of scalable preference aggregation,
 Section~\ref{sec:relevant_pasn} presents some relevant work, and
 Section~\ref{sec:contrib_pasn} enlists the contributions of this chapter.
\item
In Section~\ref{sec:modeling}, we describe the Facebook app that we developed for eliciting the preferences of individuals for a range of topics, while also obtaining the social network among them.
We propose a number of simple yet faithful models with the aim of capturing how preferences are spread in a social network, with the help of the collected Facebook data.

\item
In Section~\ref{sec:problem_pasn},
we formulate an appropriate objective function for the problem of determining an optimal set of representative nodes.
We propose a property, {expected weak insensitivity}, which captures the robustness of an aggregation rule, and 
hence
present two alternate objective functions for computational purposes.

\item In Section~\ref{sec:results},
we propose algorithms for selecting the best representatives. 
We provide a guarantee on the performance of one of the algorithms, 
and study desirable properties of one other algorithm from the viewpoint of cooperative game theory. 
We compare the performance of our proposed algorithms with that of 
the popular method of random polling, and hence justify using social networks for scalable preference aggregation.

\item Section~\ref{sec:conclusion_pasn} concludes with discussion and some notes.

\end{itemize}

\subsection*{Chapter~\ref{chap:conclusion}: \nameref{chap:conclusion}}

This chapter concludes the thesis with a brief summary of the work done, and presents some interesting future directions.\\

\blankpagewithnumber



\chapter[Preliminaries]{Preliminaries}
\label{chap:prelims}

\begin{quote}
This chapter provides an overview of selected topics required to understand the technical content in this thesis.
We cover these topics, up to the requirement of following this thesis, in the following order:
\nameref{sec:graphtheory},
\nameref{sec:gametheory},
\nameref{sec:netform},
\nameref{sec:infodiff},
\nameref{sec:prefaggr},
\nameref{sec:setfns},
\nameref{sec:modelprelims},
\nameref{sec:opttools},
\nameref{sec:cgt}.
\end{quote}

\section{Graph Theory}
\label{sec:graphtheory}

Social networks can be naturally represented as graphs. Depending on the nature of a social network, it can be represented by a directed or undirected, weighted or unweighted graph. So graph theory is an inherent part of any social network study, including this thesis.

A graph $g$ consists of vertices (or nodes) and edges (or links). Let $V(g)$ be the set of its vertices and $E(g)$ be the set of its edges, $n=|V(g)|$ be the number of vertices and $m=|E(g)|$ be the number of edges. 
If $g$ is directed, edges $(u,v)$ and $(v,u)$ are distinct and existence of one does not imply existence of the other.
Furthermore, if $g$ is weighted, there exists a weight function $w \colon E(g) \to \mathbb{R}$, which gives a weight to each edge. For certain graphs, weights are given to vertices also, that is, there exists a weight function $w' \colon V(g) \to \mathbb{R}$.

A {\em path} can be defined as a sequence of distinct vertices $\langle w_1,\ldots,w_k \rangle$ such that there exists an edge between adjacent vertices. It can also be viewed as a sequence of edges which connect a sequence of distinct vertices. We say that nodes $u$ and $v$ are {\em connected} if there exists a path $\langle u,w_1,\ldots,w_k,v \rangle$, that is, if $u$ and $v$ are extreme nodes of a path. 
The {\em length of a path} can be defined as the number of distinct edges constituting the path, while the {\em weight of a path} can be defined as the total weight of the distinct edges constituting the path.

For an unweighted graph, a {\em shortest path} between two nodes is defined as a path having the least length, and this length is termed as the {\em shortest path distance}. 
The {\em diameter} of such a graph is defined as the maximum shortest path distance over all pairs of nodes in the graph.
These definitions for weighted graphs consider weights in place of lengths.

\subsection{Graph Isomorphism and Automorphism}

We need to understand the concepts of graph isomorphism and automorphism for simplifying the network formation analysis in Chapter~\ref{chap:nfsc} by classifying the nodes and connections (or links) into different types. The notion of {\em types\/} will let us perform analysis for all nodes (or connections) of the same type at once, instead of an individual analysis for each of them.

An isomorphism of graphs $g$ and $t$, $f \colon V(g) \to V(t)$,  is a bijection between the vertex sets of $g$ and $t$ such that, any two vertices $u$ and $v$ are adjacent in $g$ if and only if $f(u)$ and $f(v)$ are adjacent in $t$.
It is clear that graph isomorphism is an equivalence relation on graphs. 
Graph isomorphism determines whether the given graphs are structurally same, while ignoring their representations.

%

\begin{figure}[h]
\begin{tabular}{p{0.5\textwidth} p{0.5\textwidth}}
\centering
{
\includegraphics[scale=.65]{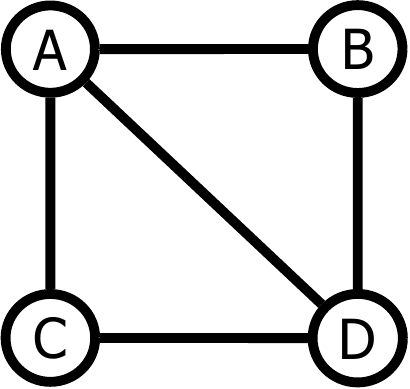}
\\
(a) Graph $g$
}
&
\centering
{
\includegraphics[scale=.65]{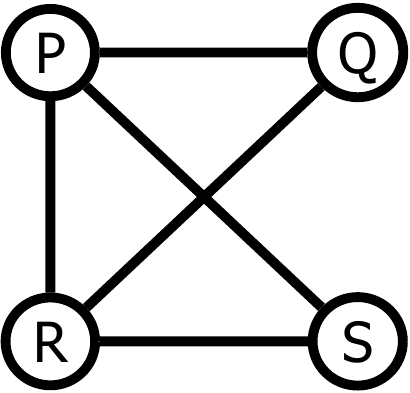}
\\
(b) Graph $t$
}
\end{tabular}
\caption{An example of graph isomorphism}
\label{fig:isomorphism_eg}
\end{figure}

Figure~\ref{fig:isomorphism_eg} shows an example of graph isomorphism, where 
$V(g) = \{A,B,C,D\}$, $V(t) = \{P,Q,R,S\}$,
and the two graphs $g$ and $t$ are isomorphic with an isomorphism $f(A)=P, f(B)=Q, f(C)=S, f(D)=R$.
Note that there may exist multiple isomorphisms for a given pair of graphs, for example, another isomorphism for the considered example is $f(A)=R, f(B)=S, f(C)=Q, f(D)=P$.

An {\em automorphism} $f \colon V(g) \to V(g)$ is a bijection where the vertex set of $g$ is mapped onto itself. An automorphism for graph $g$ in the considered example is $f(A)=D, f(B)=C, f(C)=B, f(D)=A$.

\subsection{Graph Edit Distance}

Graph edit distance (GED) is a standard measure to quantify the distance between two graphs. We will use this notion in Chapter~\ref{chap:nfsc} for studying the deviation of a network from the desired topology, owing to the deviations of network parameters.
Graph edit distance has been defined in several different ways in literature~\cite{gao2010survey}. We will use the following definition for our purpose.

\begin{definition}
Given two graphs $g$ and $h$ having same number of nodes, GED 
between them is the minimum number of link additions and deletions required to transform $h$ into a graph that is isomorphic to $g$.
\end{definition}

\begin{figure}[h]
\begin{tabular}{p{0.5\textwidth} p{0.5\textwidth}}
\centering
{
\includegraphics[scale=.65]{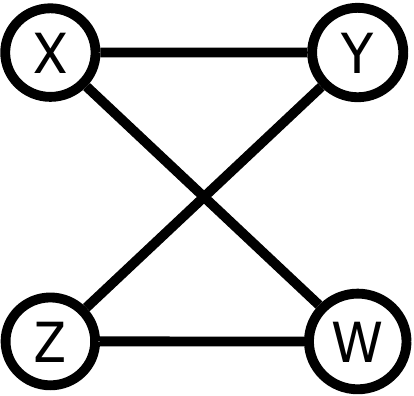}
\\
(b) Graph $h$
}
&
\centering
{
\includegraphics[scale=.65]{original.pdf}
\\
(a) Graph $g$
}
\end{tabular}
\caption{An example for computing graph edit distance}
\label{fig:ged_eg}
\end{figure}

Figure~\ref{fig:ged_eg} shows an example for computing graph edit distance between graphs $g$ and $h$. With visual inspection, one may conclude that graph $h$ can be transformed into graph $g$ by adding links $(X,Z)$ and $(Y,W)$ and deleting link $(Y,Z)$, and so the graph edit distance between them is 3. However, we have seen in Figure~\ref{fig:isomorphism_eg} that graph $t$ is isomorphic to graph $g$, and graph $h$ can be transformed into graph $t$ by adding just link $(X,Z)$ alone. It can be easily seen that this is the minimum transformation required (since graphs $h$ and $t$ are not isomorphic). So the graph edit distance between graphs $g$ and $h$ is 1.

The problem of computing GED between two graphs is NP-Hard, in general~\cite{zeng2009comparing}.
However, structural properties of certain graphs can be exploited to compute GED between them and other graphs, in polynomial time. We will discuss this point in detail in Chapter~\ref{chap:nfsc}.

\section{Extensive Form Games}
\label{sec:gametheory}


In order to analyze the sequential play of a network formation game in Chapter~\ref{chap:nfsc}, and hence direct the play to obtain a particular desired outcome (a desired network topology), we use the extensive form representation of a game.
The extensive form representation captures \cite{narahari2014game}:
\begin{itemize}
\item The ordering in which players play their actions
\item The actions available to each player 
\item The information available to players before playing at each stage
\item The outcomes as a function of the actions of the players
\item The payoff that each player obtains from each outcome
\end{itemize}

\subsection{An Example}

\begin{figure}[h]
\centering
   \iftoggle{clr}{
   \includegraphics[scale=.65]{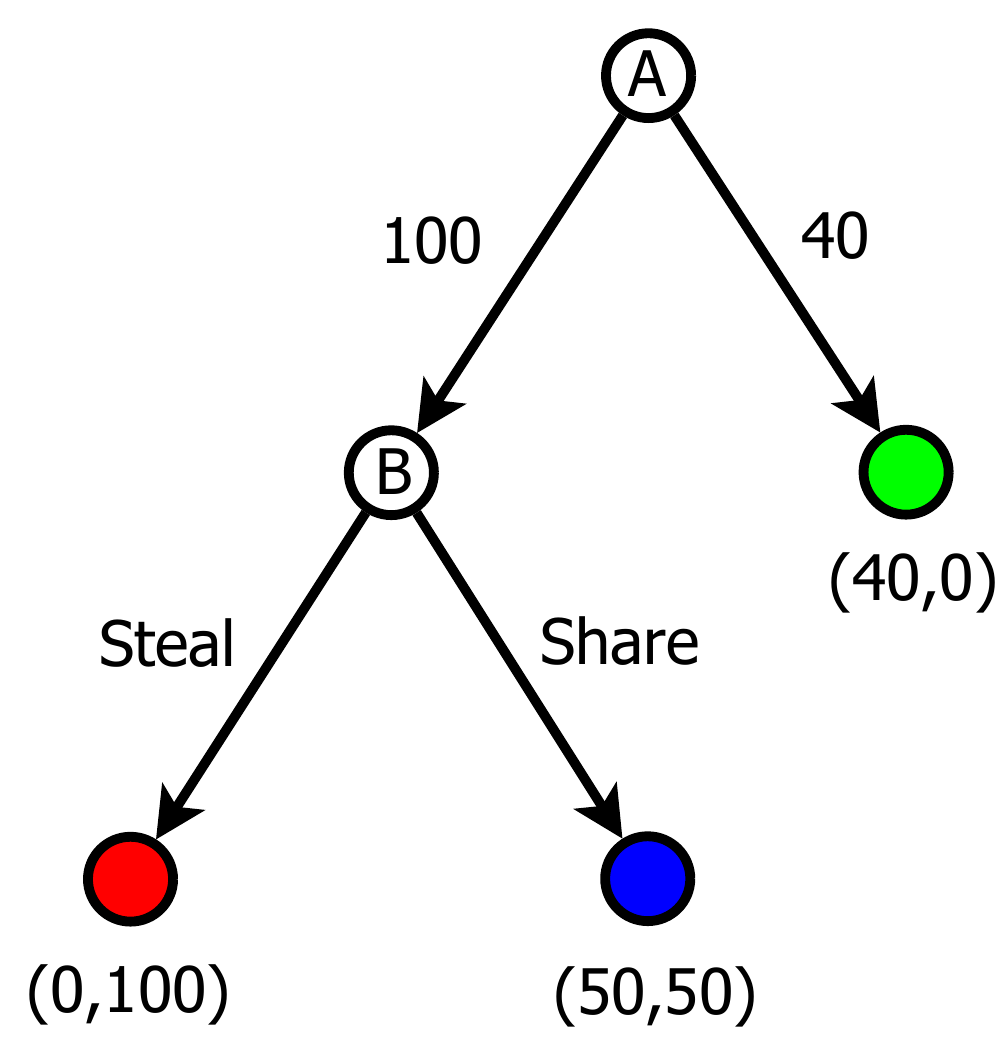}
   }{
\includegraphics[scale=.65]{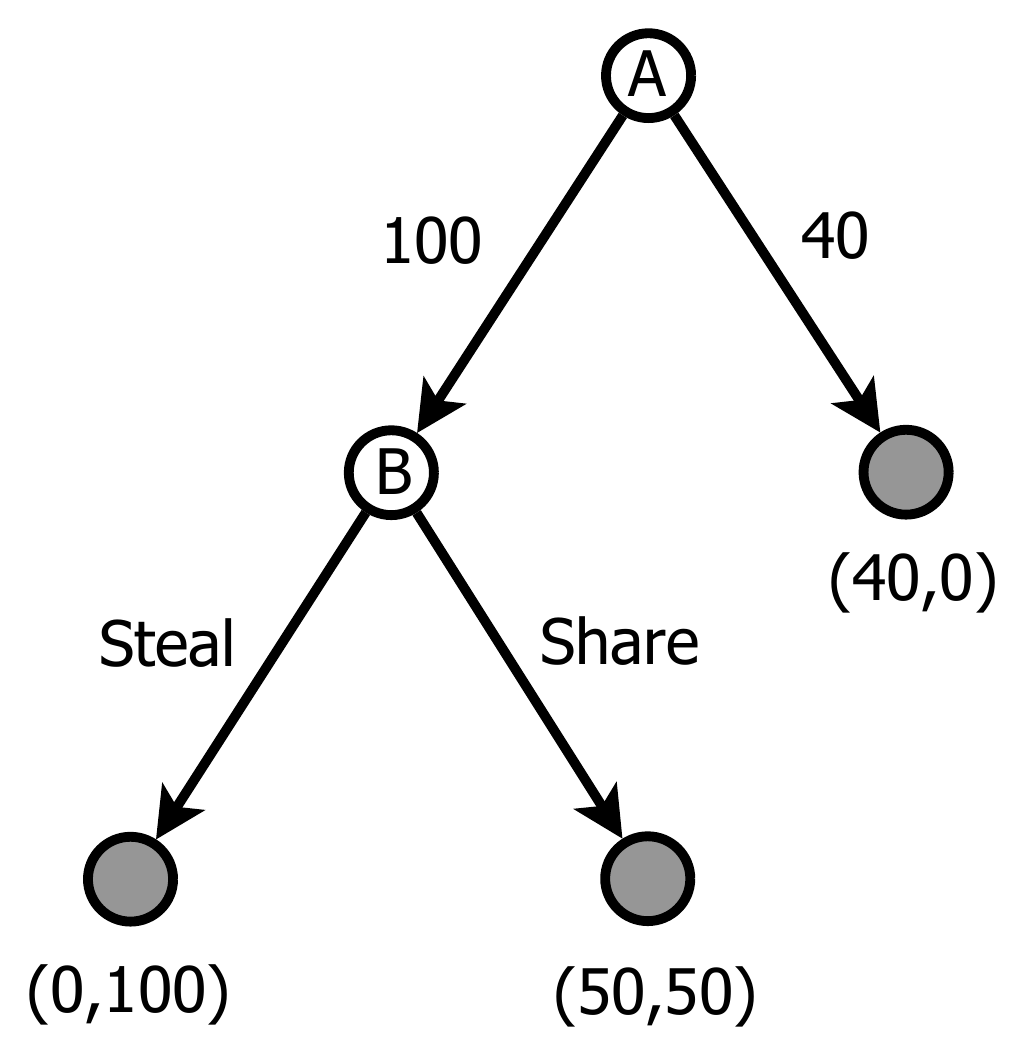}
}
\caption{An extensive form game tree}
\label{fig:extensiveform}
\end{figure}

Consider a two-player sequential game where player $A$ first gets to decide whether (a) to accept 40 units of money, or (b) to choose 100 units and let player $B$ decide whether to share it 50-50 or steal it completely. So, if player $A$ chooses to accept 40 units, the outcome is $(40,0)$, resulting in player $A$ getting the utility of 40 and player $B$ getting 0. Note that $A$ can get a higher utility of 50 on choosing 100, if $B$ chooses to share leading to the outcome $(50,50)$. However, this action of $A$ may lead to 0 utility if $B$ decides to steal and keep the entire 100 units, leading to the outcome $(0,100)$.
This game can be represented as an extensive form game as shown in Figure~\ref{fig:extensiveform}.
It can be notationally represented as $\Gamma=\langle N,(\mathcal{A}_i),\mathcal{H},\mathcal{P},(\mathbb{I}_i),({u}_i) \rangle$ where,

\begin{itemize}
\item 
$N = \{A,B\}$ is the finite set of players 
\item 
$\mathcal{A}_A = \{100,40\} , \mathcal{A}_B = \{{Steal},{Share}\}$ are the sets of actions available to the individual players 
\item 
$\mathcal{H} = \{\langle 100,{Steal}\rangle , \langle 100,{Share}\rangle , \langle 40\rangle\}$ is the set of all terminal histories (a terminal history is a path of actions from the root to a terminal node such that it is not a proper subhistory of any other terminal history) 
\\
$\mathcal{S}_{\mathcal{H}} = \{\epsilon , 100\}$ is the set of all proper subhistories (including the empty history $\epsilon$) of all terminal histories 
\item 
$\mathcal{P}(\epsilon) = A , \mathcal{P}(100) = B$ is the player function that associates each proper subhistory to a certain player 
\item 
$\mathbb{I}_A = \{\{\epsilon\}\} , \mathbb{I}_B = \{\{100\}\}$ is the set of all information sets of the individual players (an information set of a player is a set of that player's decision nodes that are indistinguishable to it). For a game with {\em perfect information\/}, information sets of all the players are singletons.  
\item 
The utilities of the individual players corresponding to each terminal history are \\
$u_A(100,Steal) = 0 , u_A(100,Share) = 50 , u_A(40) = 40 ,$ \\
$u_B(100,Steal) = 100 , u_B(100,Share) = 50 , u_B(40) = 0 $ 
\end{itemize} 
We now present an equilibrium notion for extensive form games, called {\em subgame perfect equilibrium}.

\subsection{Subgame Perfect Equilibrium}

Subgame perfect equilibrium ensures that each player's strategy is optimal given the strategies of other players, after every possible history.

\begin{definition}
Given an extensive form game $\Gamma=\langle N,(\mathcal{A}_i),\mathcal{H},\mathcal{P},(\mathbb{I}_i),({u}_i) \rangle$, a strategy profile $s^*=(s_i^*)_{i\in N}$ is a {\em subgame perfect equilibrium} if $\forall i \in N$,
\begin{displaymath}
u_i(O_h(s_i^*,s_{-i}^*)) \geq u_i(O_h(s_i,s_{-i}^*)) , \;\;\; \forall h \in \{x \in \mathcal{S}_\mathcal{H} : \mathcal{P}(x)=i\}, \;\;\; \forall s_i \in S_i
\end{displaymath}
where $O_h(s_i^*,s_{-i}^*)$ denotes the outcome corresponding to the history $h$ in the strategy profile $(s_i^*,s_{-i}^*)$.
\end{definition}

When it is $B$'s turn to take action, with the history that $A$ has played 100, it is a best response for $B$ to play \textit{Steal} getting a utility of 100 instead of 50 obtained by playing \textit{Share}. Knowing this, $A$ knows that playing 100 will lead to a utility of 0, and so it is $A$'s best response to choose 40. Thus, in the subgame perfect equilibrium of this game, player $A$ decides to accept 40 units of money, denying player $B$ to make a decision. 
\\

\textbf{Note:}
The extensive form game that we study in this thesis does not have any predetermined ordering in which players play their actions. We will explain how to analyze such a game in Chapter~\ref{chap:nfsc}.

\section{Network Formation}
\label{sec:netform}

In this section, we present the basics of network formation required for Chapter~\ref{chap:nfsc}. In particular, we present an example utility model and a well-studied equilibrium notion used in the context of social network formation.

\subsection{An Example Model: Symmetric Connections Model}

Several models have been proposed in literature based on the empirical structure of social networks \cite{jacksonbook}, for instance, Erdos-Renyi random graph model, the small world model, preferential attachment model, etc.
 However, they do not capture the strategic nature of the nodes, who can choose their links based on their utilities.
Several models have been proposed to capture the utilities of nodes in a network or graph \cite{networkscrowdsmarkets}. We now describe one such model - the {\em symmetric connections model} \cite{jackson1996strategic}.

Nodes benefit from their friends or direct links in the form of favors, information, company, etc., while maintaining a link involves some cost in the form of doing favors, giving information, spending time, etc. Nodes also benefit from indirect links like friends of friends, friends of friends of friends, and so on. However, such benefits are of a lesser value than those obtained from direct friends; the most distant the linkage, the lesser are the benefits. 
The symmetric connections model captures this idea using two parameters, $\delta \in (0,1)$ for benefits and $c$ for costs, which are common for the entire network. 

Given an undirected and unweighted network $g$, let $u_j(g)$ be the utility of node $j$, $l_{ij}(g)$ be the length of the shortest path connecting nodes $i$ and $j$, and $d_j(g)$ be the degree of node $j$.
According to this model, a node $j$ gets benefit of $\delta^{l_{ij}(g)}$ from a node $i$ which is at a distance $l_{ij}(g)$ from it, that is, it benefits $\delta$ from each of its friends, $\delta^2$ from each of its friends of friends, and so on. The value of $\delta$ being in the range $(0,1)$ ensures that closer friendships are more beneficial than distant ones.
Also, a node incurs a cost of $c$ for maintaining a link with each of its direct friends, the number of such friends being $d_j(g)$.
So the net utility of a node $j$ in a given graph $g$ is
\begin{displaymath}
u_j(g) = \sum_{i \neq j} \delta^{l_{ij}(g)} - c d_j(g)
\end{displaymath}

Consider the graph in Figure~\ref{fig:ps_eg}(a). Node $A$ benefits $\delta$ from each of its 3 direct friends, namely, $B,C,D$, and incurs a cost $c$ for these links, giving it a net utility of $3\delta - 3c$. 
Node $B$ benefits $\delta$ from its direct friend $A$ and $\delta^2$ from its 2 friends of friends, namely, $C$ and $D$. It also incurs a cost of $c$ for its link with $A$, giving it a net utility of $\delta+2\delta^2-c$. On similar lines, the net utilities of nodes $C$ and $D$ are $\delta+2\delta^2-c$ each.

\subsection{Pairwise Stability}

Pairwise stability is a well-studied notion of equilibrium in the context of social network formation. 
It accounts for bilateral deviations arising from mutual agreement of link creation between two nodes, that Nash equilibrium fails to capture~\cite{jacksonbook}. Deletion is unilateral and a node can delete a link without consent from the other node. 

Let $u_j(g)$ denote the utility of node $j$ when the network formed is $g$.

\begin{definition} 
\label{def:ps}
A network is said to be {\em pairwise stable} if it is a best response for a node not to delete any of its links and there is no incentive for any two unconnected nodes to create a link between them. So $g$  is pairwise stable if \\(a) for each edge $e = (i, j) \in g$, $u_i(g \setminus \{e\}) \leq u_i(g)$ and  $u_j(g \setminus \{e\}) \leq u_j(g)$, and\\
(b) for each edge $e' = (i, j) \notin g$, if $ u_i(g \cup \{e'\})>u_i(g) $, then $u_j(g \cup {e'})<u_j(g)$.
\end{definition}

\begin{figure}[t]
\begin{tabular}{p{0.3\textwidth} p{0.3\textwidth} p{0.3\textwidth} p{0.01\textwidth}}
\centering
{
\includegraphics[scale=.6]{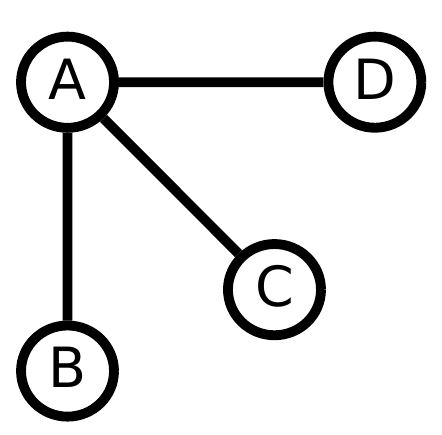}
\\
(a) Graph $g$
}
&
\centering
{
\includegraphics[scale=.6]{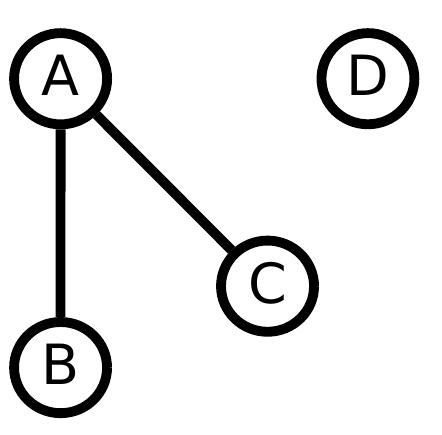}
\\
(b) Graph $g \setminus \{(A,D)\}$
}
&
\centering
{
\includegraphics[scale=.6]{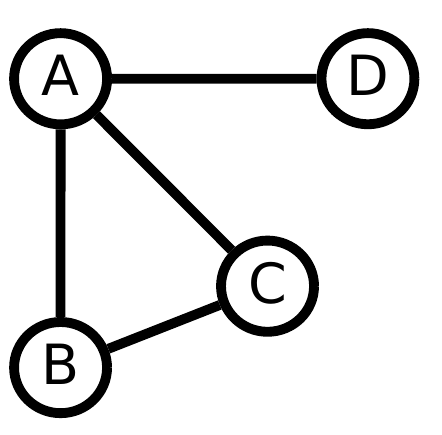}
\\
(c) Graph $g \cup \{(B,C)\}$
}
&
\end{tabular}
\caption{An example of pairwise stability}
\label{fig:ps_eg}
\end{figure}

\begin{table}
\centering
\begin{tabular}{c|c|c|c}
 \hline \hline
\T \B
Node 	&	$g$									&	$g \setminus \{(A,D)\}$	&	$g \cup \{(B,C)\}$
\\ \hline
$A$		&	$3(\delta-c)$					&	$2(\delta-c)$					&	$3(\delta-c)$
\\ \T \B
$B$		&	$\delta-c+2\delta^2$	&	$\delta-c+\delta^2$			&	$2(\delta-c)+\delta^2$
\\ \T \B
$C$		&	$\delta-c+2\delta^2$	&	$\delta-c+\delta^2$			&	$2(\delta-c)+\delta^2$
\\ \T \B
$D$		&	$\delta-c+2\delta^2$	&	0										&	$\delta-c+2\delta^2$
\\ \hline \hline
\end{tabular}
\caption{Utilities of nodes in the networks in Figure~\ref{fig:ps_eg}}
\label{tab:ps}
\end{table}

Consider the example in Figure~\ref{fig:ps_eg}. Let the utility function be defined based on the symmetric connections model; the utilities are presented in Table~\ref{tab:ps}. Consider the values of the parameters $\delta$ and $c$ to be such that $\delta-\delta^2 \leq c \leq \delta$ (with $\delta,c\geq 0$). It can be shown that graph $g$ is pairwise stable under these conditions.
If either nodes $A$ or $D$ delete the link $(A,D)$, their utilities change from $3(\delta-c)$ to $2(\delta-c)$ and from $\delta-c+2\delta^2$ to 0, respectively. That is, either of their utilities do not increase. Owing to symmetry, this is true for links $(A,B)$ and $(A,C)$ as well, and hence condition (a) in Definition~\ref{def:ps} is satisfied.
Also, if nodes $B$ and $C$ create mutual link $(B,C)$, their utilities change from $\delta-c+2\delta^2$ to $2(\delta-c)+\delta^2$. That is, either of their utilities do not increase. Owing to symmetry, this is true for links $(B,D)$ and $(C,D)$ as well, and hence condition (b) in Definition~\ref{def:ps} is satisfied.

\section{Models of Information Diffusion}
\label{sec:infodiff}

The problem of influence maximization in social networks has been extensively studied in the literature, and several models of information diffusion have been proposed in this direction \cite{networkscrowdsmarkets,guille2013information}.
The independent cascade (IC) model and the linear threshold (LT) model are two of the most extensively studied models for information diffusion in social networks.
In Chapter~\ref{chap:mpid}, we will be studying these models in the context of multi-phase diffusion.

\subsection{Independent Cascade Model}
\label{sec:icm}

%
Let $E$ be the set of weighted, directed edges in graph $G$ and denote $|E|=m$.
In the IC model, for every directed edge $(u,v)$, there is an associated probability $p_{uv}$ that represents the probability with which the source node $u$ can influence the target node $v$. The diffusion starts synchronously at time step 0 with a set of initially activated or influenced seed nodes, following which the diffusion proceeds in discrete time steps, one at a time. In each time step, nodes which got influenced in the previous time step ({\em recently activated nodes}) try to influence their neighbors, and succeed in doing so with probabilities that are respectively associated with the corresponding edges. 
These neighbors, if successfully influenced, will now act as recently activated nodes for the next time step. Only recently activated nodes can contribute in diffusing information in any particular time step. After this time step, such nodes are no longer recently activated; instead we call them {\em already activated} nodes. Nodes, once activated, remain activated for the rest of the diffusion process. 
In short, when node $u$ gets activated at a time step, it gets a single chance to activate each of its inactive neighbors, $v$ (that too in the immediately following time step), with the given success 
probability $p_{uv}$.
The diffusion process concludes when no further nodes can be activated.

We now explain the concept of {\em live graph}, which helps simplify the  analysis for IC model.
A live graph $X$ is an instance of graph $G$, obtained by sampling the edges:
an edge $(u,v)$ is present in a live graph with probability $p_{uv}$ and absent with probability $1-p_{uv}$, independent of the presence of other edges in the live graph (so a live graph is a directed graph with no edge probabilities). 
Thus $p(X)$, the probability of occurrence of any live graph $X$, can be obtained as
$\prod_{(u,v) \in X} (p_{uv}) \prod_{(u,v) \notin X} (1-p_{uv})$.
%
It can be seen that as long as a node $u$, when influenced, in turn influences node $v$ with probability $p_{uv}$ that is independent of time, sampling the edge $(u,v)$ in the beginning of the diffusion is equivalent to sampling it when $u$ is activated \cite{kempe2003maximizing}.

\subsection{Linear Threshold Model}
\label{sec:ltm}

In Linear Threshold (LT) model, an influence degree $b_{u,v}$ is associated with every directed edge $(v,u)$, where $b_{u,v} \geq 0$ is the degree of influence that node $v$ has on node $u$, and an influence threshold $\chi_u$ with every node $u$. The weights $b_{u,v}$ are such that $\sum_v b_{u,v} \leq 1$. Also, owing to lack of knowledge about the thresholds, which are held privately by the nodes, it is assumed that the thresholds are chosen uniformly at random from $[0,1]$. The diffusion process starts at time step 0 and proceeds in discrete time steps, one at a time. In each time step, a node is influenced or activated if and only if the sum of influence degrees of the edges incoming from activated neighbors (irrespective of the time of activation of the neighbors) crosses its own influence threshold, that is, when
$
\sum_v b_{u,v} \geq \chi_u
$.
Nodes, once activated, remain activated for the rest of the diffusion process. 
 In any given time step, the recently activated nodes along with previously activated ones contribute to the diffusion process. The diffusion process stops when it is not possible to activate or influence any further nodes.

\section{Preference Aggregation}
\label{sec:prefaggr}

Preference aggregation is a well-studied topic in social choice theory.
In Chapter~\ref{chap:pasn}, we harness the underlying social network to make preference aggregation, more efficient and effective in practice.

Given a set of alternatives, individuals have certain preferences over them. These alternatives can be any entity, ranging from political candidates to food cuisines.
We assume that 
an individual's preference can be represented as a complete ranked list of alternatives.
%
%
We refer to a ranked list of alternatives 
as {\em preference} and the multiset consisting of the preferences of the individuals as {\em preference profile}.
For example, if the set of alternatives is $\{X,Y,Z\}$ and individual $i$ prefers $Y$ the most and $X$ the least, then $i$'s preference 
can be written as
$(Y,Z,X)_i$.
Suppose individual $j$'s preference is $(X,Y,Z)_j$, then the preference profile of the population $\{i,j\}$ is $\{(Y,Z,X),(X,Y,Z)\}$.

\subsection{Measures of Dissimilarity between Preferences}
\label{sec:dissimmeasures}

We now describe Kendall-Tau distance and Footrule distance, two of the most popular measures of dissimilarity between two preferences.

\subsubsection{Kendall-Tau Distance}

A widely used measure of dissimilarity between two preferences is {\em Kendall-Tau distance}. It counts the number of pairwise inversions with respect to the alternatives.
For computing normalized Kendall-Tau distance, given that the number of alternatives is $r$, we normalize Kendall-Tau distance to be in $[0,1]$, by dividing actual distance by \begin{scriptsize}$\dbinom{r}{2}$\end{scriptsize}, the maximum distance between any two preferences on $r$ alternatives.
%
%
For example, the Kendall-Tau distance between preferences $(X,Y,Z)$ and $(Y,Z,X)$ is 2, because two pairs $\{X,Y\}$ and $\{X,Z\}$ are inverted between them. The normalized Kendall-Tau distance is 
$2/$\begin{scriptsize}$\dbinom{3}{2}$\end{scriptsize} $ = \frac{2}{3}$.
%

\subsubsection{Footrule Distance}

Another popular measure of dissimilarity between two preferences is {\em Spearman's Footrule distance}, which sums up the displacements for all the alternatives. In the above example, $X,Y,Z$ are displaced by $2,1,1$ positions, respectively, giving Spearman's Footrule distance of 4.
The normalized Spearman's Footrule distance can be obtained by dividing the absolute distance with the maximum possible distance, which can be shown to be 
$
2 \lceil \frac{r}{2}\rceil \lfloor \frac{r}{2} \rfloor
$,
where $r$ is the number of alternatives.
So the normalized Spearman's Footrule distance in the above example is
$\frac{4}{2 \lceil \frac{3}{2}\rceil \lfloor \frac{3}{2} \rfloor} = 1$.

\subsection{Aggregation Rules}
\label{aggrrules}

An {\em aggregation rule} takes a preference profile as input and outputs the {\em aggregate preference(s)}, which in some sense reflect(s) the collective opinion of all the individuals. 
We consider a wide range of voting rules for our study, namely, Bucklin, Smith set, Borda, Veto, Minmax (pairwise opposition), Dictatorship, Random Dictatorship, Schulze, Plurality, Kemeny, and Copeland.
A survey of voting rules
and related topics 
can be found in \cite{brandt2012computational}.
A more concise table of voting rules and their properties can be found in \cite{wiki:voting}.
Of these rules, only Kemeny, Dictatorship, and Random Dictatorship output the entire aggregate preference; others either determine a winning alternative or give each alternative a score.
For consistency, for all rules except Kemeny, Dictatorship, and Random Dictatorship, we employ the following well accepted approach for converting a series of winning alternatives into an aggregate preference: rank a winning alternative as first, then vote over the remaining alternatives and rank a winning alternative in this iteration as second, and repeat until all alternatives have been ranked
\cite{brandt2012computational}.
As we are indifferent among alternatives, we do not assume any tie-breaking rule so as to avoid bias towards any particular alternative, while determining a winner.
So an aggregation rule may not output a unique aggregate preference, that is, it is a correspondance.

\section{Properties of Set Functions}
\label{sec:setfns}

A {\em set function} is a function whose domain is a collection of sets. In the context of this thesis, let $N$ be a finite set of elements and $2^{N}$ be its power set. Then a set function $f$ takes a subset of $N$ as input and outputs a real number. That is,
\begin{displaymath}
f : 2^{N} \rightarrow \mathbb{R}
\end{displaymath}
 
We now present some properties concerned with set functions, that we will be looking at in Chapters~\ref{chap:mpid} and \ref{chap:pasn}.
\subsection{Non-negativity}
A set function $f$ is said to be {\em non-negative} if 
\begin{displaymath}
f(S) \geq 0, \;\;\; \forall S \subseteq N
\end{displaymath}
This property states that the value of any set should be non-negative.

\subsection{Monotonicity}
A set function $f$ is said to be {\em monotone increasing} if 
\begin{displaymath}
f(S) \leq f(T), \;\;\; \forall S \subset T \subseteq N
\end{displaymath}
This property means that addition of elements to any set should not decrease its value. This is often the case in most real-world applications.

Similarly, a set function $f$ is said to be {\em monotone decreasing} if 
\begin{displaymath}
f(S) \geq f(T), \;\;\; \forall S \subset T \subseteq N
\end{displaymath}

\subsection{Submodularity and Supermodularity}
A set function $f$ is said to be {\em submodular} if 
\begin{displaymath}
f(S \cup \{i\}) - f(S) \geq f(T \cup \{i\}) - f(T), \;\;\; \forall i \in N \setminus T, \;\;\; \forall S \subset T \subset N
\end{displaymath}
That is, submodular functions have a diminishing returns property which means, the marginal value added by an element to a superset of a set is not more than the marginal value added by that element to that set. Informally, the marginal value added by an element to a set decreases as the set grows larger.
For a finite $N$, the above definition is equivalent to
\begin{displaymath}
f(S)+f(T) \geq f(S \cup T) + f(S \cap T), \;\;\; \forall S,T \subseteq N
\end{displaymath}

Submodular functions occur in several real-world applications, since the diminishing returns property is a natural one is several domains.

On the other hand, a set function $f$ is said to be {\em supermodular} if 
\begin{displaymath}
f(S \cup \{i\}) - f(S) \leq f(T \cup \{i\}) - f(T), \;\;\; \forall i \in N \setminus T, \;\;\; \forall S \subset T \subset N
\end{displaymath}

Note that a function can be both submodular and supermodular, for example,
$f(S) = \sum_{i \in S} w_i$.
Also, a function can be neither submodular nor supermodular, for example, 
$f(S) = 1$ when $|S|$ is even and $0$ when $|S|$ is odd.


\subsection{Subadditivity and Superadditivity}
A set function $f$ is said to be {\em subadditive} if 
\begin{displaymath}
f(S \cup T) \leq f(S) + f(T), \;\;\; \forall S , T \subseteq N
\end{displaymath}
That is, the value of a union of any two sets is at most the sum of their individual values.

On the other hand, a set function $f$ is said to be {\em superadditive} if 
\begin{displaymath}
f(S \cup T) \geq f(S) + f(T), \;\;\; \forall S , T \subseteq N
\end{displaymath}

It can be easily seen that a non-negative submodular function is subadditive, while a non-negative superadditive function is supermodular.

\section{Modeling}
\label{sec:modelprelims}

In this section, we present some basics required for modeling the spread of preferences in a social network in Chapter~\ref{chap:pasn}.

\subsection{Kullback-Leibler (KL) Divergence}

The Kullback-Leibler (KL) divergence is a measure of the difference between two probability distributions $D_t$ and $D_m$, where $D_t$ typically represents the true distribution of data, while $D_m$ typically represents a model or approximation of $D_t$.
Informally, the KL divergence of $D_m$ from $D_t$ is the amount of information lost when $D_m$ is used to approximate $D_t$.

For the purpose of this thesis, we will be requiring KL divergence specifically for the case of discrete probability distributions.
Here, the KL divergence of $D_m$ from $D_t$ is defined as
\begin{displaymath}
\sum_i D_t(i) \, \log\frac{D_t(i)}{D_m(i)}
\end{displaymath} 
where $D_t(i)$ is the probability mass function value of $D_t$ at $i$.
Note that KL divergence is not symmetric in $D_t$ and $D_m$,
and is defined only if $D_m(i)=0 \implies D_t(i)=0, \forall i$. Also when $D_t(i)=0$, the contribution of the $i^{th}$ term to the above summation is interpreted as zero.

\subsection{Root Mean Square (RMS) Error}

In the process of model-fitting or validating a model based on its performance, there is an error of the value given by the model against the true value for each test data point. There are different methods available for aggregating the errors 
over all test data points. For our study in this thesis, we use root mean square (RMS) error, which is given by
$\sqrt{\avg_{x} [err(x)]^2}$,
where $err(x)$ is the error of the value given by the model against the true value for the test data point $x$.

RMS error has the following desirable properties:
\begin{itemize}
\item Like any valid measure of average model-fitting error, it ensures that the positive and negative differences are not cancelled out.
\item Like mean square error, which is well-accepted in most machine learning applications, it magnifies larger errors more than smaller errors.
\item It holds an advantage over mean square error in that, it gives an idea about the order of magnitude of the errors for individual test data points.
\end{itemize}

\subsection{Maximum Likelihood Estimation (MLE)}
\label{sec:mle}


Given independent and identically distributed observations $x_1,\ldots,x_n$ known to have been drawn from an assumed probability density or mass function $f(x_i;\theta_1,\ldots,\theta_m)$ with unknown parameters $\theta_1,\ldots,\theta_m$, we aim to find a good estimate of the parameters. 
A good estimate of the parameters would be a value that maximizes the likelihood of getting the observed data.

The {\em likelihood function} is defined as 
\begin{displaymath}
\mathcal{L}(\theta_1,\ldots,\theta_m;x_1,\ldots,x_n) = f(x_1,\ldots,x_n|\theta_1,\ldots,\theta_m) = \prod_{i=1}^n f(x_i|\theta_1,\ldots,\theta_m)
\end{displaymath}
Let $\hat\theta_1,\ldots,\hat\theta_m$ be the respective parameter values that maximize the above function. Then $\hat\theta_i$ is called the {\em maximum likelihood estimator} of $\theta_i$.

It is often more convenient to work with the logarithm of the likelihood function, and logarithm being a monotone increasing function, the above is equivalent to maximizing the {\em log-likelihood},
\begin{displaymath}
\ln\mathcal{L}(\theta_1,\ldots,\theta_m;x_1,\ldots,x_n) = \sum_{i=1}^n \ln f(x_i|\theta_1,\ldots,\theta_m)
\end{displaymath}

\section{Optimization Tools}
\label{sec:opttools}

We will be addressing a variety of optimization problems in Chapters~\ref{chap:mpid} and \ref{chap:pasn}. This section describes some of the general solution approaches.

\subsection{Greedy Hill-climbing}
\label{sec:greedyhill}

The greedy hill-climbing algorithm  is one of the most basic algorithms, often taken as baseline, for solving combinatorial optimization problems.
For maximizing an objective function $f$, it selects elements one at a time, each time choosing an element that provides the largest marginal increase in the value of $f$, until the budget (upper bound on the cardinality of the solution set) is exhausted.
Starting with $\psi^{(0)}=\{\}$, let $\psi^{(t)}$ be the solution set (under construction) after adding $t^{th}$ element to the set.
Then the $t^{th}$ chosen element is 
(until the budget is exhausted)
$
\argmax_{z \in N \setminus \psi^{(t-1)}} f(\psi^{(t-1)} \cup z)
$.

Given $n$ candidate elements, budget $k$, and time $\mathcal{T}$ for computing $f$, the time complexity of the greedy hill-climbing algorithm is $O(kn\mathcal{T})$.

The following result shows the effectiveness of greedy hill-climbing algorithm for solving combinatorial optimization problems, especially those that are NP-hard.

\begin{theorem} 
\label{thm:nemhauser}
For a non-negative, monotone increasing, submodular function $f$, let $S^G$ be a set of size $k$ obtained using greedy hill-climbing. Let $S^O$ be a set that maximizes the value of $f$ over all sets of cardinality $k$. Then 
$f (S^G) \geq (1- \frac{1}{e}) f (S^O )$~\cite{nemhauser1978analysis}.
Furthermore, 
for any $\epsilon>0$, there is a $\gamma>0$ such that by using $(1 + \gamma)$-approximate values for $f$, we obtain a $(1-\frac{1}{e}-\epsilon)$-approximation \cite{kempe2003maximizing}.
\end{theorem}

\subsection{Cross Entropy Method}
\label{sec:ce}


The cross entropy (CE) method is a generic and practical tool for solving combinatorial optimization problems, including those that are NP-hard.
The CE method involves an iterative procedure where each iteration can be broken
down into two phases \cite{de2005tutorial}:
\begin{enumerate}
\item Generate random data samples (trajectories, vectors, etc.) according to a specified mechanism.
\item Update  the  parameters  of  the  random  mechanism  based  on  the  data  to  produce better samples in the next iteration.
\end{enumerate}

We now explain a simple version of the cross entropy method, with an example of determining a set $S \subseteq N$ that maximizes the value of function $f$.
For ease of explanation, we will represent a set $S_i$ as a Bernoulli vector ${\boldsymbol X}_i = (X_{i1},\ldots,X_{in})$ where its $j^{th}$ component $X_{ij}=1$ if $j\in S_i$ and $X_{ij}=0$ if $j\notin S_i$.


We initiate the method by setting the probability vector to $\hat {\boldsymbol p}_0 = (\hat p_{0,1}, \ldots, \hat p_{0,n})$, whose $j^{th}$ component $\hat { p}_{0,j}$ denotes the probability of $X_{ij}=1$.
%
Let $\mathcal{N}$ be the number of samples to be drawn in each iteration, $\rho$ be a performance parameter telling the fraction of samples to be discarded owing to their relatively low function values, and $\alpha$ be a smoothing parameter telling how much weight is to be given to the current iteration (as against the previous iterations).

Following are the steps involved in a simple version of the CE method:\\
Starting with $t=1$, iterate through the following steps while incrementing $t$, until the stopping criterion is met: 
\begin{enumerate}
\item Draw samples ${\boldsymbol X}_1,\ldots,{\boldsymbol X}_{\mathcal{N}}$ of Bernoulli vectors with success probability vector $\hat {\boldsymbol p}_{t-1}$.
Compute $f({\boldsymbol X}_i)$ for all $i$, and order them in descending order of values, say $f_{(1)} \geq \ldots \geq f_{(\mathcal{N})}$. Let $\hat \gamma_t$ be $(1-\rho)$ sample quantile, that is, $\hat \gamma_t = f_{(\lceil (1-\rho)\mathcal{N} \rceil)}$.
\item Use the samples to compute $\hat {\boldsymbol q}_t = (\hat q_{t,1}, \ldots, \hat q_{t,n})$ where
\begin{displaymath}
\hat q_{t,j} = \frac{\sum_{i=1}^{\mathcal{N}}\mathbb{I}_{\{ f({\boldsymbol X}_i)\geq \hat \gamma_t \}} \mathbb{I}_{\{X_{ij}=1\}} f({\boldsymbol X}_i)} {\sum_{i=1}^{\mathcal{N}}\mathbb{I}_{\{ f({\boldsymbol X}_i)\geq \hat \gamma_t \}} f({\boldsymbol X}_i)}
\end{displaymath}
Update the probability vector using
\begin{displaymath}
\hat {\boldsymbol p}_{t} = (1-\alpha) \hat {\boldsymbol p}_{t-1} + \alpha \hat {\boldsymbol q}_{t}
\end{displaymath}
\end{enumerate}
The stopping criterion could be the convergence of $\hat \gamma_t$, or an upper bound on the number of iterations $t$, or something similar.
Once the stopping criterion is met, an optimal set can be chosen based on the probabilities (or appropriateness) of the elements to be included in the set and the application at hand.

For a more detailed and fully adaptive version of the CE method, the reader is referred to \cite{de2005tutorial}.

\subsection{Golden Section Search}
\label{sec:gss}

Golden section search is an efficient method for maximizing a unimodal function. Let the function be unimodal in $[X_{\text{min}},X_{\text{max}}]$. Let $[x_{\text{min}},x_{\text{max}}]$ be the search domain, which is updated after every iteration. These are initialized as $x_{\text{min}} := X_{\text{min}}$ and $x_{\text{max}} := X_{\text{max}}$. Following are the iterative steps of golden section search (refer to Figure~\ref{fig:gss}):
\begin{enumerate}
\item Divide the interval $[x_{\text{min}},x_{\text{max}}]$ into 3 sections using two internal points $x_1$ and $x_2$.
\item 
If  $f(x_1)  > f(x_2)$,
the maximum is in $[x_{\text{min}}, x_2]$,
so redefine $x_{\text{min}} := x_{\text{min}},  x_{\text{max}} := x_2$.\\
If  $f(x_1) <  f(x_2)$, 
the maximum is in $[x_1,x_{\text{max}}]$
so redefine $x_{\text{min}} := x_1,  x_{\text{max}} := x_{\text{max}}$.
\end{enumerate}
The algorithm terminates, either after a fixed number of iterations, or when we obtain a solution $x^{\text{sol}}$ which cannot more than $\epsilon$ far away from the optimal solution $x^{\text{opt}}$, where $\epsilon$ is the desired error (that is, when $|x^{\text{sol}}-x^\text{{opt}}|\leq\epsilon$; note that this error is with respect to the solution and not with respect to the function value).

\begin{figure} [t!]
\centering
\includegraphics[scale=0.5]{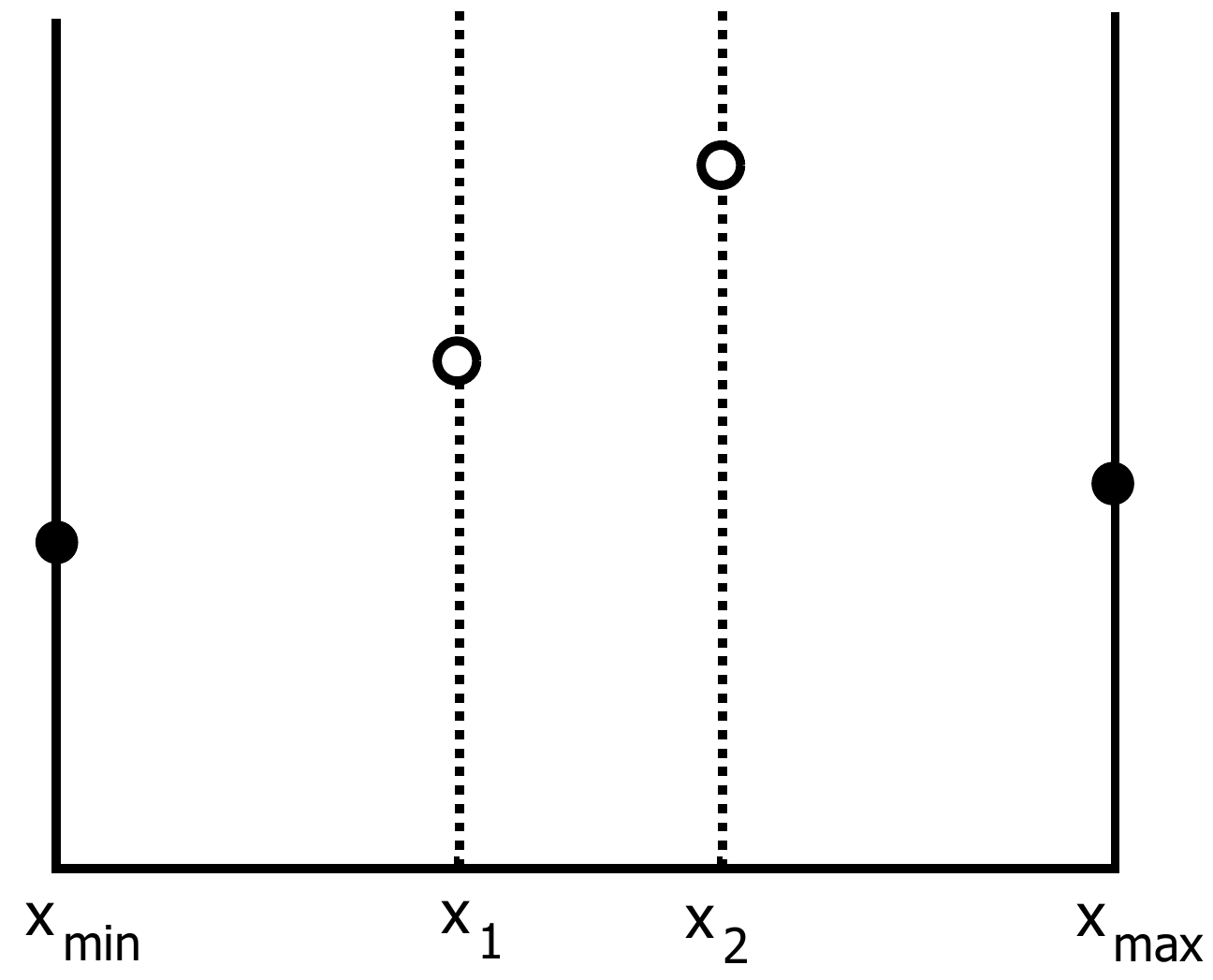}
\caption{
An iteration of the golden section search}
\label{fig:gss}
\end{figure}

In each iteration, golden section search selects the internal points $x_1$ and $x_2$ such that it reuses one of the internal values of the previous iteration, so as to minimize computation. It can be shown that $x_1$ and $x_2$ should be selected such that
\begin{displaymath}
 \frac{x_{\text{max}}-x_1}{x_{\text{max}}-x_{\text{min}}}=\frac{x_2-x_{\text{min}}}{x_{\text{max}}-x_{\text{min}}}=\frac{\sqrt{5}-1}{2} \approx 0.618 \;\;\; \text{(golden ratio)}
\end{displaymath}

\section{Cooperative Game Theory}
\label{sec:cgt}

We will be encountering cooperative game theory in some form or the other, in Chapters~\ref{chap:mpid} and \ref{chap:pasn}.
In this section, we provide a brief insight into the cooperative game theory concepts~\cite{straffin1993game,saad2009coalitional,chun2007coincidence}, namely, the Core, the Shapley value, the Nucleolus, the Gately point, and the $\tau$-value.

A cooperative game or coalitional game or characteristic function form game $(N,\nu)$ consists of two parameters $N$ and $\nu$. $N$ is the set of players and $\nu:2^N\rightarrow \mathbb{R}$ is the characteristic function, which defines the value $\nu(S)$ of any coalition $S\subseteq N$. 

The coalition $N$ consisting of all the players is called the {\em grand coalition}.
Assuming that the grand coalition is formed, the question is how to distribute the total obtained payoff among the individual players.
In what follows, let $x_i$ represent the payoff allocated to player $i$ and $n=|N|$. Cooperative game theory studies several payoff allocations $x = (x_1, . . . , x_n)$, each satisfying a number of certain desirable properties. We now briefly describe some of the payoff allocations, more popularly known as {\em solution concepts}.

\subsection{The Core}
\label{core}

The core consists of all payoff allocations $x = (x_1, . . . , x_n)$ that satisfy the following properties:
\begin{enumerate}
 \item Individual rationality: $x_i \geq \nu(\{i\}) \; \forall \; i \in N$
 \item Collective rationality: $ \sum_{i \in N} x_i = \nu(N)$.
 \item Coalitional rationality: $ \sum_{i \in S} x_i \geq \nu(N) \; \forall S \subseteq N$.
\end{enumerate}
A payoff allocation satisfying individual rationality and collective rationality is called an \textit{imputation}.

\subsection{The Shapley Value}
\label{sec:shapley}
The Shapley value $\phi(\nu) = (\phi_1(\nu), . . . ,\phi_n(\nu))$ is the unique imputation that satisfies the following three axioms which are based on the idea of fairness~\cite{straffin1993game}:

\begin{enumerate}
 \item The Shapley value should depend only on $\nu$, and should respect any symmetries in $\nu$. That is, if players $i$ and $j$ are symmetric, then $\phi_j(\nu)=\phi_i(\nu)$.
\item If $\nu(S)=\nu(S \setminus \{i\})\ \forall S \subseteq N$, then $\phi_i(\nu)=0$. In other words, if player $i$ contributes nothing to any coalition, then the player can be considered as a dummy. Furthermore, adding a dummy should not affect the original game.
\item Consider two games defined on the same set of players, represented by $(N,\nu)$ and $(N,w)$. Define a sum game $(N,(\nu+w))$ where $(\nu+w)(S)$ = $\nu(S)+w(S)\ \forall S \subseteq N$. Also, if $\phi(\nu)$ and $\phi(w)$ represent the Shapley values of the two games, then the Shapley value of the sum game should satisfy $\phi(\nu+w) = \phi(\nu) +\phi(w)$.
\end{enumerate}

%
For any general coalitional game with transferable utility $(N,\nu)$, the Shapley value of player $i$ is given by
\begin{eqnarray*}
\nonumber
 \phi_i(\nu) &=& \frac{1}{n!} \sum_{i\in S} (|S|-1)!(n-|S|)![\nu(S)-\nu(S\setminus\{i\})] \\
 &=&\frac{1}{n!}\sum_{\pi\in \Pi}x_i^\pi
\end{eqnarray*}
$\Pi=$ set of all permutations on $N$\\
$x_i^\pi=$ contribution of player $i$ to permutation $\pi$

\subsection{The Nucleolus}
The basic motivation behind the nucleolus is that, instead of applying Shapley value (having general fairness axiomization), one can provide an allocation that minimizes the dissatisfaction of the players from the allocation they can receive in a game~\cite{saad2009coalitional}.

Let $x$ be any payoff vector (or allocation) and $q(x)$ be a vector whose components are the numbers $(\nu(S)-x(S))$ arranged in non-increasing order, where $S$ runs over all coalitions in $N$ except the grand coalition. Then, payoff vector $x$ is at least as acceptable as payoff vector $y$, if $q(x)$ is lexicographically less than $q(y)$; write it as $x \succeq y$. The nucleolus of a game is the set 
$\{x \in X :x \succeq y, \; \forall y \in X \}$, where $X$ is the set of all  payoff vectors. It is shown that every game possesses a non-empty nucleolus and is unique~\cite{schmeidler1969nucleolus}. 
%
In other words, nucleolus is an allocation that minimizes the dissatisfaction of the players from the allocation they can receive in a game~\cite{schmeidler1969nucleolus}. 

For every imputation $x$, consider the excess defined by
\begin{displaymath}
 e_S(x) = \nu(S) - \sum_{i \in S} x_i
\end{displaymath}
$e_S(x)$ is a measure of unhappiness of $S$ with $x$. The goal of nucleolus is to minimize the most unhappy coalition (that is, the largest of the $e_S(x)$). 
The linear programming formulation is as follows:
\begin{displaymath}
\text{min }Z 
\end{displaymath}
subject to
\begin{displaymath}
 Z + \sum_{i \in S} x_i \geq \nu(S) \; \;\; \forall S \subseteq N
\end{displaymath}
\begin{displaymath}
\sum_{i \in N} x_i = \nu(N)
\end{displaymath}
The nucleolus $Nu(\nu)=(Nu_1(\nu),\ldots,Nu_n(\nu))$ of a game $(N,\nu)$ has the following properties~\cite{saad2009coalitional}:
\begin{enumerate}
 \item The nucleolus depends only on $\nu$, and respects any symmetries in $\nu$. That is, if players $i$ and $j$ are symmetric, then $Nu_j(\nu)=Nu_i(\nu)$.
\item If $\nu(S)=\nu(S\setminus \{i\})\ \forall S \subseteq N$, then $Nu_i(\nu)=0$. In other words, if player $i$ contributes nothing to any coalition, then the player can be considered as a dummy. 
\item If players $i$ and $j$ are in the same coalition, then the highest excess that $i$ can make in a coalition without $j$ is equal to the highest excess that $j$ can make in a coalition without $i$. This is derived from the fact that nucleolus lies in the \textit{Kernel} of the game, which is the set of all allocations $x$ such that 
\begin{displaymath}
\max_{\substack{S\subseteq N\setminus \{j\} \\ i \in S}} e_S(x)= \max_{\substack{T\subseteq N\setminus \{i\}\\ j \in T}} e_T(x)
\end{displaymath}

\end{enumerate}

Geometrically, nucleolus is the point in the core whose distance from the closest wall of the core is as large as possible.
The reader is referred to \cite{saad2009coalitional} for the detailed properties of nucleolus. 

\subsection{The Gately Point}
\label{sec:gately}
Player $i$'s \textit{propensity to disrupt} the grand coalition is defined to be the following ratio~\cite{straffin1993game}.
\begin{equation}
\nonumber
\label{gatelyeq}
d_i(x) = \frac{\sum_{j\neq i}x_j - \nu(N\setminus\{i\})}{x_i - \nu(\{i\})}
\end{equation}
If $d_i(x)$ is large, player $i$ may lose something by deserting the grand coalition, but others will lose a lot more. The Gately point $Gv(\nu)=(Gv_1(\nu),\ldots,Gv_n(\nu))$ of a game is the imputation which minimizes the maximum propensity to disrupt. The general way to minimize the largest propensity to disrupt is to make all of the propensities to disrupt, equal. When the game is normalized so that $\nu(\{i\}) = 0$ for all $i$, the way to set all the $d_i(x)$'s equal is to choose $x_i$ in proportion to $\nu(N) - \nu(N\setminus\{i\})$.
That is,
\begin{equation}
\nonumber
Gv_i(\nu) = \left(\frac{\nu(N) - \nu(N\setminus\{i\})}{\sum_{j\in N}(\nu(N) - \nu(N\setminus\{j\}))}\right)\nu(N)
\end{equation}

\subsection{The $\tau$-value}

For each $i\in N$, let
$
M_i(\nu) = \nu(N) - \nu(N\setminus\{i\}) \text{ and }m_i(\nu) = \nu(\{i\})
$,
and let $M(\nu) = (M_i(\nu))_{i\in N}$ and $m(\nu) = (m_i(\nu))_{i\in N}$.
The $\tau$-value $\tau(\nu)=(\tau_1(\nu),\ldots,\tau_n(\nu))$ of a game is the unique solution concept which is {\em efficient} (or collectively rational) and has the following properties~\cite{tijs1987axiomatization}:
\begin{enumerate}
\item
The {\em minimal right property\/}, which implies that $\tau(\nu) = m(\nu) + \tau(\nu-m(\nu))$. 
So it does not matter for a player $i$,
whether $i$ gets the $\tau$-value payoff allocation in the game $(N,\nu)$,
or whether $i$ obtains first the minimal right payoff $m_i(\nu)$ in the game $(N,\nu)$ and then the $\tau$-value payoff allocation in the right reduced game $(N,(\nu-m(\nu)))$.
This property is a weaker form of the additivity property: $(\nu+w)(S)$ = $\nu(S)+w(S)\ \forall S \subseteq N$, which plays a role in the axiomatic characterization of the Shapley value.
\item
The {\em restricted proportionality property\/}, which implies that $\tau(\nu)$ is a multiple of the vector $M(\nu)$. 
So for games with minimal right payoff vector $m(\nu)$ as zero, the payoff allocation to the players is proportional to the marginal contribution of the players to the grand coalition.
\end{enumerate}

The $\tau$-value selects the maximal feasible allocation on the line connecting $M(\nu) = (M_i(\nu))_{i\in N}$ and $m(\nu) = (m_i(\nu))_{i\in N}$~\cite{chun2007coincidence}.
For each convex game $(N,\nu)$,
\begin{equation}
\nonumber
\label{taueq}
\tau(\nu) = \lambda M(\nu) + (1-\lambda)m(\nu)
\end{equation}
where $\lambda \in [0,1]$ is chosen so as to satisfy
\begin{equation}
\nonumber
\label{tausatis}
\sum_{i\in N} [\lambda(\nu(N)-\nu(N\setminus\{i\})) + (1-\lambda)\nu(\{i\})] = \nu(N)
\end{equation}

There exist a number of solution concepts in the literature on cooperative game theory; the ones explained above suffice for the purpose of this thesis.

\vspace{10mm}

With the required conceptual preliminaries and tools in hand, we now move on to the technical contributions of this thesis. The next chapter deals with the problem of orchestrating social network formation. The classical network problem focuses on predicting which network topologies are likely to emerge, given the conditions on the network parameters. So the problem studied in the next chapter is the inverse of the classical problem, since we aim to derive the conditions on the network parameters, given that we want the network to have a particular desired topology.

\chapter[Formation of Stable Strategic Networks with Desired Topologies]{Formation of Stable Strategic Networks with Desired Topologies
  \blfootnote{A part of this chapter is published as \cite{dhamalwine}:
Swapnil Dhamal and Y. Narahari. Forming networks of strategic agents with desired topologies. In Paul W. Goldberg, editor, {\em Internet and Network Economics (WINE)}, Lecture Notes in Computer Science, pages 504--511. Springer Berlin Heidelberg, 2012.}
  \blfootnote{A significant part of this chapter is published as \cite{dhamalsim}:
Swapnil Dhamal and Y. Narahari. Formation of stable strategic networks with desired topologies. {\em Studies in Microeconomics}, 3(2):158--213, 2015.}
}

\label{chap:nfsc}

\begin{quote}
Many real-world networks such as social networks consist of strategic agents. The topology of these networks often plays a crucial role in determining the ease and speed with which certain information driven tasks can be accomplished. Consequently, growing a stable network having a certain desired topology is of interest. Motivated by this, we study the following important problem: given a certain desired topology, under what conditions would best response link alteration  strategies adopted by strategic agents, uniquely lead to formation of a stable network having the given topology. This problem is the inverse of the classical network formation problem where we are concerned with determining stable topologies, given the conditions on the network parameters. We study this interesting inverse problem by proposing (1) a recursive model of network formation and (2) a utility model that captures key determinants of network formation. Building upon these models, we explore relevant topologies such as star graph, complete graph, bipartite Tur\'an graph, and multiple stars with interconnected centers. We derive a set of sufficient conditions under which these topologies uniquely emerge, study their social welfare properties, and investigate the effects of deviating from the derived conditions.
%
\end{quote}


\newpage
\section{Introduction}
\label{sec:intro_nfsc}

A primary reason for networks such as social networks to be formed is that every person (or agent or node)  gets certain benefits from the network. These benefits assume different forms in different types of networks. These benefits, however, do not come for free. Every node in the network has to  incur
a certain cost for maintaining links with its immediate neighbors or direct friends. 
This cost takes the form of time, money, or effort, depending on the type of network. 
Owing to the tension between benefits and costs, self-interested or rational nodes 
think strategically while choosing their immediate neighbors. 
A stable network that forms out
of this process will have a topological structure that is dictated by the individual
utilities and the resulting best response strategies of the nodes.

The underlying social network structure plays a key role in determining the dynamics of several processes 
such as, the spread of epidemics~\cite{ganesh2005effect} and the diffusion of information~\cite{jacksonbook}. 
This, in turn, affects the decision of which nodes should be selected to be vaccinated \cite{abbassi2011toward}, 
or to trigger a campaign so as to either maximize the spread of certain information \cite{kempe2003maximizing} 
or minimize the spread of an already spreading misinformation \cite{budak2011limiting}. 
Often, stakeholders such as a social network owner or planner, who work with the networks so formed, would like the network to have a certain desirable topology to facilitate efficient handling of 
information driven tasks using the network. Typical examples of these tasks include spreading certain information to nodes (information diffusion), extracting certain critical information from nodes (information extraction), enabling optimal communication among nodes for maximum efficiency (knowledge management), etc. If a particular topology is an ideal one for the set of tasks to be handled, it would be useful to orchestrate network formation in a way that only the desired topology emerges as the unique stable topology.

A network in the current context can be naturally represented as a 
graph consisting of strategic agents called {\em nodes} 
and connections 
among them called {\em links}. 
Bloch and Jackson~\cite{bloch2006definitions} examine a variety of stability and equilibrium notions that have been used to study strategic network formation.
%
Our analysis in this chapter is based on the notion of {\em pairwise stability\/} which accounts for bilateral deviations arising from mutual agreement of link creation between two nodes, that Nash equilibrium fails to capture~\cite{jacksonbook}. Deletion is unilateral and a node can delete a link without consent from the other node. 
Consistent with the definition of pairwise stability, we consider that all nodes are homogeneous and they have global knowledge of the network (this is a common, well accepted assumption in the literature on social network formation~\cite{jacksonbook}).

Before we proceed further, 
we present two important definitions from the literature~\cite{jacksonbook} for ease of discussion.
Let $u_j(g)$ denote the utility of node $j$ when the network formed is $g$.

\begin{definition} 
A network is said to be {\em pairwise stable} if it is a best response for a node not to delete any of its links and there is no incentive for any two unconnected nodes to create a link between them. So $g$  is pairwise stable if \\(a) for each edge $e = (i, j) \in g$, $u_i(g \backslash \{e\}) \leq u_i(g)$ and  $u_j(g \backslash \{e\}) \leq u_j(g)$, and\\
(b) for each edge $e' = (i, j) \notin g$, if $ u_i(g \cup \{e'\})>u_i(g) $, then $u_j(g \cup {e'})<u_j(g)$.
\end{definition}

\begin{definition} 
A network is said to be {\em efficient} if the sum of the utilities of the nodes in the network is maximal. So given a set of nodes $N$, $g$ is efficient if it maximizes $\sum_{j\in N} u_j(g)$, that is, for all networks $g'$ on $N$, $\sum_{j\in N} u_j(g) \geq \sum_{j\in N} u_j(g') $.
\end{definition}

Every network has certain parameters that influence its evolution process.
We refer to the tuple of values of these parameters as {\em conditions on the network}.
%
By conditions on a network, we mean a listing of the range of values taken by the various parameters that influence network formation, including the relations between these parameters.
For example, let $b_1$ be the benefit that a node gets from each of its direct neighbors, $b_2$ be the benefit that it gets from each node that is at distance two from it, and $c$ be the cost it pays for maintaining link with each of its direct neighbors. In real-world networks, it is often the case that $0 \leq b_2 \leq b_1$ and $c \geq 0$. The list of relations, say (1) $0 < b_2 < b_1$ and (2) $b_1-b_2 < c < b_1$, are the conditions on the network. Based on these conditions, the utilities of the involved nodes are determined, which in turn affect their (link addition/deletion) strategies, hence influencing the process of formation of that network.
Throughout this chapter, we ignore enlisting trivial conditions such as $0 \leq b_2 \leq b_1$ and $c \geq 0$.

In general, the evolution of a real-world social network would depend on several other factors such as the information diffusing through the network~\cite{ehrhardt2006diffusion,zhang2013strategic}.
For simplicity, we make a well accepted assumption that the network evolves purely based on the conditions on it and does not
depend on any other factor.

\section{Motivation}
\label{sec:motiv_nfsc}
One of the key problems addressed in the literature on social network formation is: 
given a set of self-interested nodes and a model of social network formation, 
which topologies are  stable and which ones are  efficient.
The trade-off between stability and efficiency
is a key topic of  interest and concern in the literature on social network formation~\cite{jackson2005survey,jacksonbook}.

This work focuses on the inverse problem, namely, 
given a certain desired topology, under what conditions would  best response 
(link addition/deletion) strategies played 
by self-interested agents, uniquely lead to the formation of a stable (and perhaps efficient)
network with that topology. 
The problem becomes important because networks, such as an organizational network of a global company, play an important role in a
variety of knowledge management, information extraction, and information diffusion tasks. The topology of these networks is one of the major factors that decides the ease and speed with which the above tasks can be accomplished. In short, a certain topology might serve the 
interests of the network owner better.
%

%
In social networks, in general, it is difficult to figure out what the desired topology is. Moreover, it is possible that the social network is being formed for more than one reason. It can, however, be argued that given a set of individuals, there may not exist a unique social network amongst them. For instance, there may exist several networks like friendship network, collaboration network, organizational network, etc. on the same set of nodes. 
Different networks have different cost and benefit parameters, for example, from a mutual connection, two nodes may gain more in collaboration network than in friendship network, and also pay more cost. 
Furthermore, in real-world networks, a link between two nodes in one network may influence the corresponding link in another network. The influence may be positive (friendship trust leads to business trust) or negative (time spent for maintaining link in one network may adversely affect the corresponding link in another network).
For simplicity, we consider these various networks to be formed independently of each other.
A way to look at the problem under consideration is that, we focus on one such network at a time and derive conditions so that it has the desired topology or structure. 

\begin{figure}[t!]
\begin{tabular}{p{5cm} p{5cm} p{5cm}}
\centering
\includegraphics[scale=0.5]{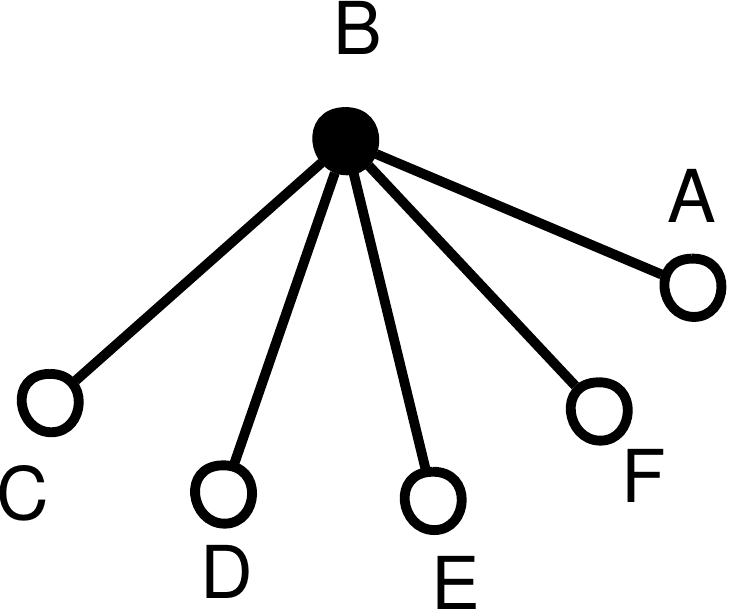}
&
\centering
\includegraphics[scale=0.5]{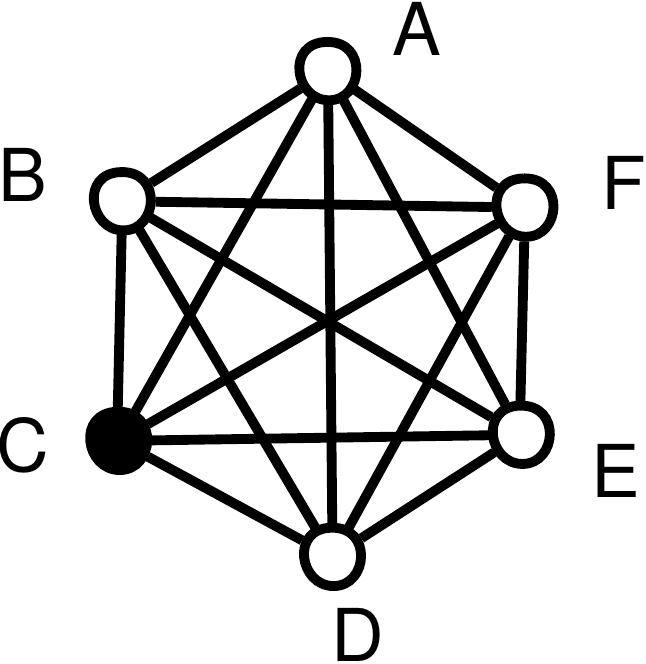}
&
\hspace{7mm}
\includegraphics[scale=0.5]{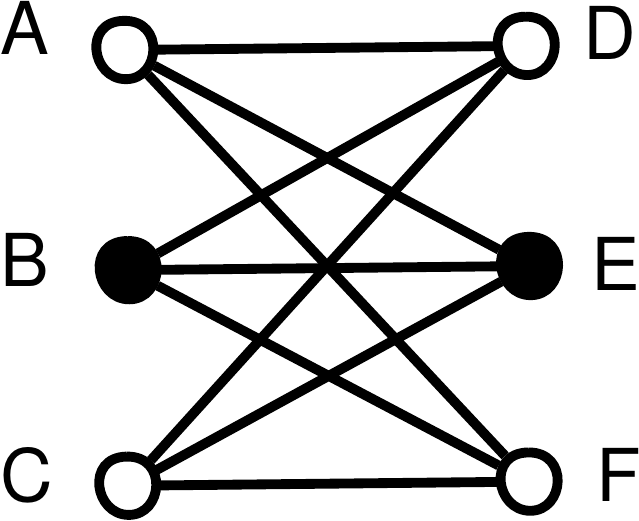}
\\
\centering (a) Star & 
\centering (b) Complete & 
\centering (c) Bipartite Tur\'an 
\vspace{5mm}
\end{tabular}

\begin{tabular}{p{7.5cm} p{7.5cm}}
\centering
\includegraphics[scale=0.5]{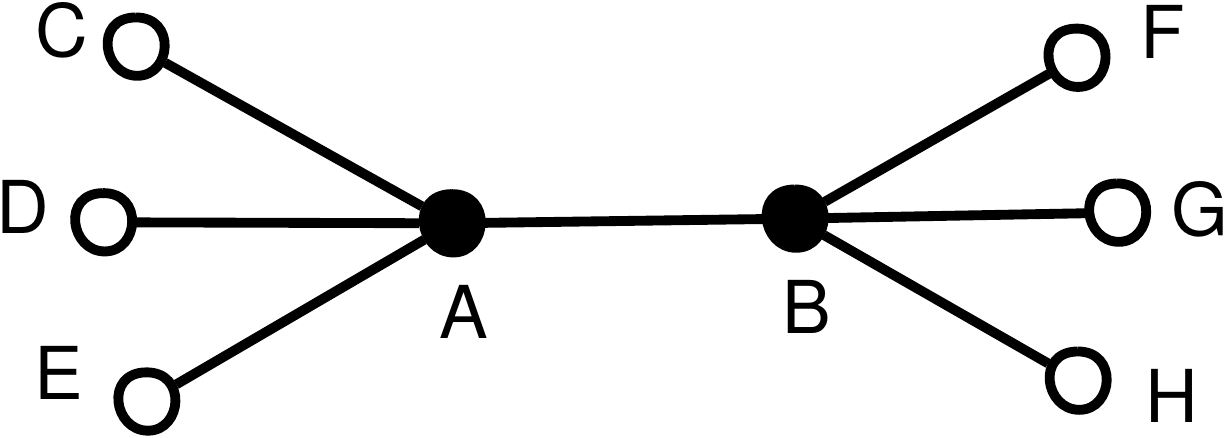}
&
\hspace{8.7mm}
\includegraphics[scale=0.5]{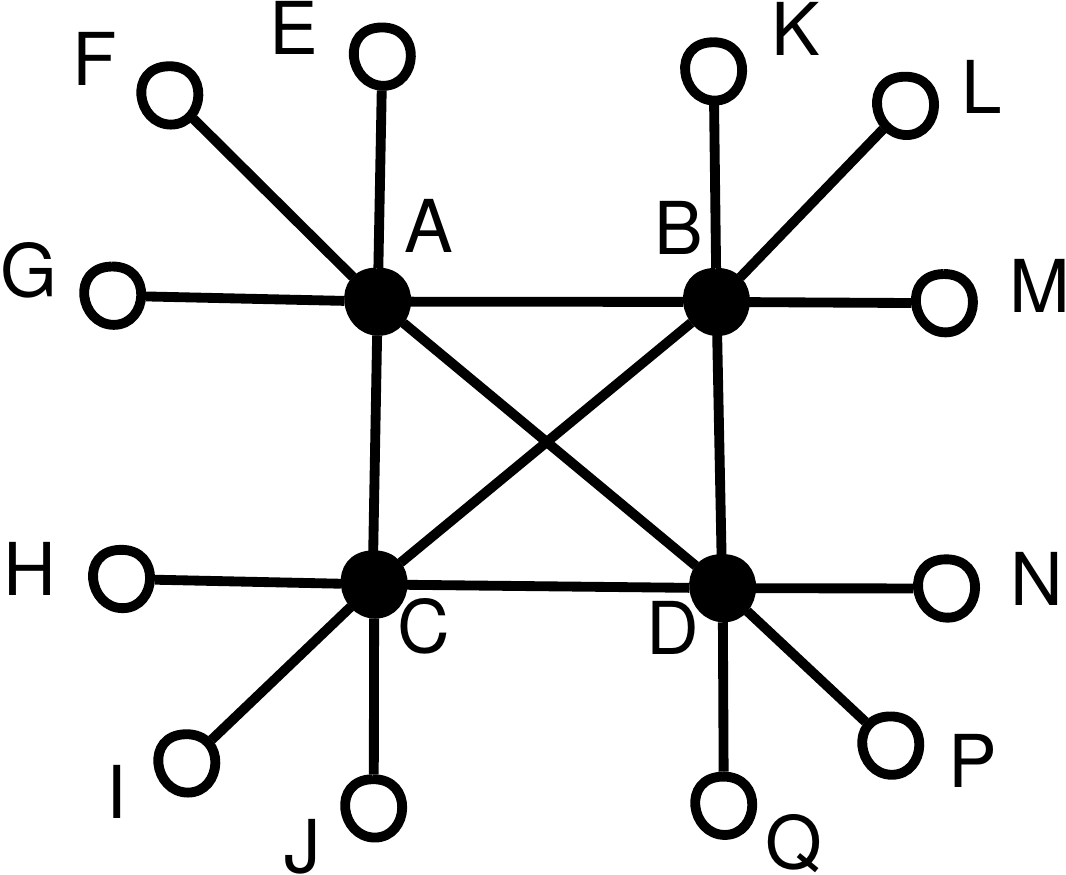}
\\ 
\centering (d) 2-star & 
\centering (e) $k$-star ($k=4$)
\end{tabular}

\caption{Relevant topologies investigated in this work}
\label{fig:motiv_nfsc}
\end{figure}

In this chapter, for the sake of brevity, we consider only a representative set of commonly encountered topologies for investigation. However, our approach is general and can be used to study other topologies, albeit with more involved analysis.
We motivate our investigation further with the help of several 
relevant topologies shown
in Figure~\ref{fig:motiv_nfsc}.

Consider a network where there is a need to rapidly spread certain critical information, requiring redundant ways of communication to account for any link failures. 
The information may be received by any of the nodes and it is important that all other 
nodes also get the information at the earliest.
In such cases, a complete network (Figure~\ref{fig:motiv_nfsc}(b)) would be ideal. 
In general, if the information received by any node is required to be propagated throughout the network within a certain number of steps $d$, the network's diameter should be bounded by the number $d$.

Consider a different scenario where the time required to spread the information is critical, but
there is also a need for moderation to verify the authenticity of the information before 
spreading it to the other nodes in the network (for example, it could be a rumor).
 Here a star network (Figure~\ref{fig:motiv_nfsc}(a)) would be desirable since the center would act as a moderator and any 
information that originates in any part of the network has to flow through the
moderator before it reaches other nodes in the network. 
Virus inoculation is a related example where a star network would be desirable 
since vaccinating the center may be sufficient to prevent spread of the virus
to other parts of the network, thus reducing the cost of vaccination.

Our next example concerns  two sections
 of a society where some or all members of a section
receive certain information simultaneously. The objective here is to forward
the information to the other section. Moreover, it is desirable to not have intra-section links to save on resources. In this case, it would be desirable to have 
a bipartite network. Moreover, if the information is critical and urgent, 
requiring redundancy, a complete bipartite network would be desirable. 
A bipartite Tur\'an network (Figure~\ref{fig:motiv_nfsc}(c)) is a practical special case where both sections
are required to be nearly of equal sizes.

Consider a generalization of the star network where there are $k$ centers and 
the leaf nodes are evenly distributed among them, that is, the difference between the number of leaf nodes connected to any two centers, is at most one. Such a network would be desirable 
when the number of nodes is expected to be very large and there is a need for 
decentralization for efficiently controlling information in the network. 
We call such a network, $k$-star network (Figures~\ref{fig:motiv_nfsc}(d-e)).

For similar reasons, if fast information extraction is the main criterion, certain topologies may be better than others. Information extraction in social networks can be thought of as the reverse of information diffusion. Also, an information extraction or search algorithm would work better on some topologies than others.

The problem under study also assumes importance in knowledge management. McInerney~\cite{mcinerney2002knowledge} defines knowledge management as an effort to increase useful knowledge within an organization, and highlights that the ways to do this include encouraging communication, offering opportunities to learn, and promoting the sharing of appropriate knowledge artifacts. An organization may want to develop a particular network within, so as to make the most of knowledge management. A complete network would be desirable if the nodes are trustworthy with no possibility of manipulation.
For practical reasons, an organization may want nodes of different sections to communicate with each other and not within sections so that each node can aggregate knowledge received from nodes belonging to the other section, in its own way. A bipartite Tur\'an network would be desirable in such a case. Such a network may also be more desirable than the complete network in order to prevent inessential investment of time for communication within a section.

Similarly, for a variety of reasons, there may be a need to form networks having certain other structural properties.
So depending on the tasks for which the network would be used, a certain topology might
be more desirable than others. This provides the motivation for our work.

\section{Relevant Work}
\label{sec:relevant_nfsc}

Models of network formation in literature can be broadly classified as either simultaneous move models or sequential move models.
Jackson and Wolinsky~\cite{jackson1996strategic} propose a simultaneous move game model where nodes simultaneously propose the set of nodes with whom they want to create a link, and a link is created between any two nodes if they mutually propose a link to each other.
Aumann and Myerson~\cite{myerson20} provide a sequential move game model where nodes are farsighted, whereas
Watts~\cite{watts618} considers a sequential move game model where nodes are myopic. In both of these approaches and in any sequential network formation model in general, the resulting network is based on the ordering in which links are altered and owing to the assumed random ordering, it is not clear which networks would emerge.
 
 The modeling of strategic formation in a general network setting was first studied by
  Jackson and Wolinsky~\cite{jackson1996strategic} by proposing a utility model called {\em symmetric connections model}. This widely cited model, however, does not capture many key determinants involved in strategic network formation.
Since then, several utility models have been proposed in literature in the effort of capturing these determinants. 
 Jackson~\cite{jackson2003stability} reviews several such models in the literature and highlights that pairwise stable networks may not exist in some settings. 
 Hellmann and Staudigl \cite{hellmann2014evolution} provide a survey of random graph models and game theoretic models for analyzing network evolution. 

Given a network, Myerson value~\cite{myerson1977graphs} gives an allocation to each of the involved nodes based on certain desirable properties.
 Jackson~\cite{jackson2005allocation} proposes a family of allocation rules that consider alternative network structures when allocating the value generated by the network to the individual nodes.
 Narayanam and Narahari~\cite{ramasuri1} investigate the topologies of networks formed with a generic model based on value functions and analyze resulting networks using 
 Myerson value.
 %
 There have also been studies on stability and efficiency of specific networks such as R\&D networks \cite{konig2012efficiency}. 
 Atalay \cite{atalay2013sources} studies sources of variation in social networks by extending the model in \cite{jackson2007meeting} by allowing agents to have varying abilities to attract contacts. 
 %

    Goyal and Joshi~\cite{goyal2006unequal} explore two specific models of network formation and arrive at circumstances under which networks exhibit an unequal distribution of connections across agents.
 Goyal and Vega-Redondo~\cite{goyal2007structural} propose a non-cooperative game model capturing bridging benefits wherein they introduce the concept of {\em essential nodes}, which is a part of our proposed utility model. Their model, however, does not capture the decaying of benefits obtained from remote nodes.
Kleinberg et al.~\cite{kleinberg2008strategic} propose a localized model that considers benefits that a node gets by bridging any pair of its neighbors separated by a path of length 2. Their model does not capture indirect benefits and bridging benefits that nodes can gain by being intermediaries between non-neighbors which are separated by a path of length greater than 2.
%
Under another localized model where a node's bridging benefits depend on its clustering coefficient, Vallam et al.~\cite{vallam2013topologies} study stable and efficient topologies.

%
Hummon~\cite{hummon2000utility} uses agent-based simulation approaches to explore the dynamics of network evolution based on the symmetric connections model. 
Doreian~\cite{doreian2006actor}, given some conditions on a network, analytically arrives at specific networks that are pairwise stable using the same model. However, the complexity of analysis increases exponentially with the number of nodes and the analysis in the paper is limited to a network with only five nodes.
Some gaps in this analysis are addressed by
Xie and Cui~\cite{xie2008cost,xie2008note}.

Most existing models of social network formation assume that all nodes are present throughout the evolution of a network, thus allowing nodes to form links that may be inconsistent with the desired network.
For instance, if the desired topology is a star, 
it is desirable to have conditions that ensure a link between two nodes, of which one would play the role of the center. But with the same conditions, links between other pairs would be created with high probability, leading to inconsistencies with the star topology.
Also, with all nodes present in an unorganized network, a random ordering over them in sequential network formation models adds to the complexity of analysis.
However, in most social networks, not all nodes are present from beginning itself. A network starts building up from a few nodes and gradually grows to its full form. Our model captures such a type of network formation.

There have been  a few approaches earlier to design incentives for nodes so that the resulting network is efficient.
Woodard and Parkes~\cite{woodard2003strategyproof} 
use mechanism design to
ensure that the outcome is an efficient network. 
Mutuswami and Winter~\cite{mutuswami2002subscription} design a mechanism that ensures efficiency, budget balance, and equity. 
Though it is often assumed that the welfare of a network is based only on its efficiency, there are many
situations where this may not be true. A network
may not be efficient in itself, but it may be desirable for reasons
external to the network, as explained in Section~\ref{sec:motiv_nfsc}.

\section{Contributions of this Chapter}
\label{sec:gameinbrief}
In this chapter, we study the inverse of the classical network formation problem, that is, 
under what conditions would the desired topology uniquely emerge 
when 
agents adopt their best response strategies.  
Our specific contributions are summarized below.

\begin{itemize} 
\item We propose a recursive model of  network formation, with which we can guarantee that a network being formed retains a designated topology in each of its stable states. Our model ensures that, for common network topologies, the analysis can be carried out independent of the current number of nodes in the network and also independent of the upper bound on the 
number of nodes in the network.
 The utility model we propose captures most key aspects relevant to strategic network formation:
(a) benefits from immediate neighbors, (b) costs of maintaining links with immediate neighbors, (c) benefits from indirect neighbors, (d) bridging benefits, (e) intermediation rents, and (f) an entry fee  for entering the network. 
We then present our procedure for deriving sufficient conditions for the formation of a given topology as the unique one. (Section~\ref{sec:model})
\item Using the proposed models, we study common and important networks, namely, star network, complete network, bipartite Tur\'an network, and $k$-star network, 
and derive sufficient conditions under which these topologies uniquely emerge.
We also investigate the efficiency (or social welfare) properties of the above network topologies. (Section~\ref{sec:analysis})
\item We introduce the concept of dynamic conditions on a network and study the effects of deviation 
from the derived sufficient conditions on the resulting network, using the notion of graph edit distance. 
In this process, we develop a polynomial time algorithm for computing graph edit distance between a given graph and a corresponding $k$-star graph. 
(Section~\ref{sec:deviation})
\end{itemize}

To the best of our knowledge, this is the first detailed effort in investigating the problem of obtaining a desired topology uniquely in social network formation.
%

\begin{figure} [!t]
\centering
\includegraphics[scale=0.42]{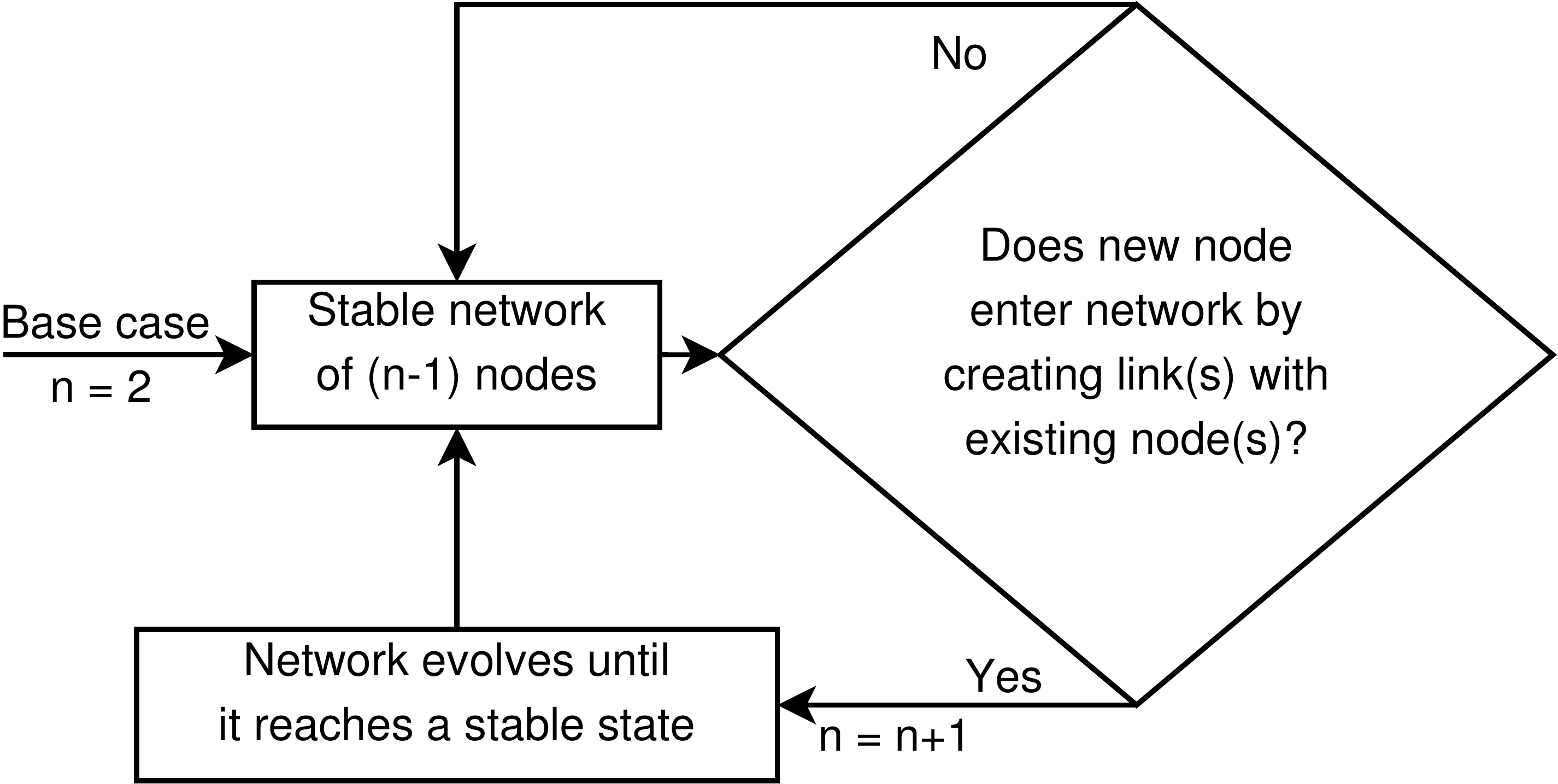}
\caption{
Proposed recursive model of network formation}
\label{fig:model}
\end{figure}

\section{The Model}
\label{sec:model}

We consider the
process of formation of a network consisting of 
strategic nodes, where each node aims at maximizing its utility it gets from the network. 

\subsection{A Recursive Model of Network Formation}

The network consists of $n$ nodes at any given time, where $n$ could vary over time. 
The process starts with one node, whose only strategy is to remain in its current state. The strategy of the second node is to either (a) not enter the network or (b) propose a link with the first node. We make a natural assumption that in order to be a part of the network, the second node has to propose a link with the first node and not vice versa. 
Based on the model under study, the first node may or may not get to decide whether to accept this link. 
If this link is created, the second node successfully enters the network. Following this, the network evolves to reach a stable state after which, 
the third node considers entering the network. 
The third node can enter the network by successfully creating link(s) with one or both of the first two nodes. In this chapter, we consider that at most one link is altered at a time, and so the third node can enter the network by successfully creating a link with exactly one of the already present nodes in the network. If it does, 
the network of these three nodes evolves. 
Once the network reaches a stable state, 
the fourth node considers entering the network, and
this process continues.
Note that in the above process, 
no node in the network of $n-1$ nodes can create a link with the newly entering $n^{th}$ node until the latter proposes and successfully creates a link 
in order to enter
the network.
After the new node enters the network successfully, 
the network evolves 
until it reaches a stable state consisting of $n$ nodes. Following this, a new ${(n+1)}^{th}$ node considers entering the network and the process goes on recursively. The assumption that a node considers entering the network only when it is stable may seem unnatural in general networks, but can be justified in networks where entry of nodes can be controlled by a network administrator.
This recursive model is depicted in Figure~\ref{fig:model}.
Note that the model is not based on any utility model, network evolution model, or equilibrium notion.

%
It can be observed at first glance that, if at some point of time, a new node fails to enter the network by failing to create a link with some existing node, the network will cease to grow. In such cases, it may seem that Figure~\ref{fig:model} 
goes into infinite loop for no reason, while it may have just pointed to an exit. The argument holds for the current social network models where the cost and benefit parameters, and hence the conditions on the network, are assumed to remain unchanged throughout the network formation process. 
But in real-world networks, this is often not the case and the conditions may vary over time or evolve owing to some internal or external factors. For instance, if the individual workload on the employees increases, the cost of maintaining link with each other also increases. On the other hand, if the workload is of collaborative nature, then the benefit parameters attain an increased value.
It is possible that no node successfully enters the network for some time, but with changes in the conditions, nodes may resume entering and the network may start to grow again. We explore this concept of {\em dynamic conditions} on a network in Section~\ref{sec:deviation}.

\subsection{Dynamics of Network Evolution}
\label{sec:directing}

The model of network evolution considered in this chapter is based on a sequential move game \cite{watts618}.
During the evolution phase, nodes which get to make a move are chosen at random at all time. Each node has a set of strategies at any given time and when it gets a chance to make a move, it chooses its {\em myopic best response} strategy which maximizes its immediate utility. 
A strategy can be of one of the three types, namely (a) creating a link with a node that is not its immediate neighbor, (b) deleting a link with an immediate neighbor, or (c) maintaining status quo. 
Note that a node will compute whether a link it proposes, decreases utility of the other node, because if it does, it is not its myopic best response as the link will not be accepted by the latter. 
Moreover, consistent with the notion of pairwise stability, if a node gets to make a move and altering a link does not strictly increase its utility, then it prefers not to alter it.
%
%
The aforementioned sequential move evolution process can be represented as an extensive form game tree. 

\subsubsection{Game Tree}
\label{sec:tree}

The entry of each node in the network results in one game tree, and so the network formation process results in a series of game trees, each tree corresponding to a sequential move game
(see Figure~\ref{fig:star}).
Each branch represents a possible transition from a network state, owing to decision made by a node.
So, the root of a game tree represents the network state in which a new node considers entering the network.

A way to find an equilibrium in an extensive form game consisting of farsighted players, is to use backward induction~\cite{osborne}. 
However, in our game, 
the players have bounded rationality, that is, their best response strategies are myopic. So instead of the regular backward induction approach or the bottom-up approach, we take a top-down approach for ease of understanding.
%
We now recall the definition of an {\em improving path}~\cite{jackson2002evolution}.

\begin{definition}
\label{def:improving} 
An {\em improving path} is a sequence of networks, where each transition is obtained by either any two nodes choosing to add a mutual link or any node choosing to delete any of its links. 
\end{definition}

Thus, a pairwise stable network is one from which there is no improving path leaving it.
%
The notion of improving paths is based on the assumption of myopic agents, who make their decisions 
without considering how their actions affect the decisions of other nodes and hence the evolution of the network. 

\subsubsection{Notion of Types}
\label{sec:types}
As the order in which nodes take decisions is random, in a general game, the number of branches arising from each state in the game tree depends on the number of nodes, $n$, as well as the number of possible direct connections each node can be involved in (or number of possible direct connections with respect to each node), $n-1$. 
The complexity of analysis can, however, be significantly reduced by the notion of {\em types} using which, several nodes and links can be analyzed at once. This is a widely used technique in analyzing pairwise stability of a network. 
%
%
We now explain the notion of types in detail.

%
\begin{definition}
\label{def:typenodes}
Two nodes $A$ and $C$ of a graph $g$ are of the same type if there exists an automorphism
$f:V(g)\rightarrow V(g)$ such that $f(A)=C$, where $V(g)$ is the vertex set of $g$.
\end{definition}
The implication of nodes being of the same type is that, for any automorphism $f$, if a best response strategy of node $A$ is to alter its link with node $D$, then a best response strategy of $f(A)$ is to alter its link with $f(D)$. So at any point of time, it is sufficient to consider the best response strategies of one node of each type.
\begin{definition}
\label{def:typeconnections}
Two connections with respect to a node $B$, connections $BA$ and $BC$, are of the same type if there exists an automorphism
$f$ such that $f(A)=C$ and $f(B)=B$.
\end{definition}
The implication of connections being of the same type with respect to a node is that, the node is indifferent between the connections, irrespective of the underlying utility model. Different types of connections with respect to a node form different branches in the game tree.

\begin{wrapfigure}{l}{75mm}
  \centering
    \includegraphics[scale=0.45]{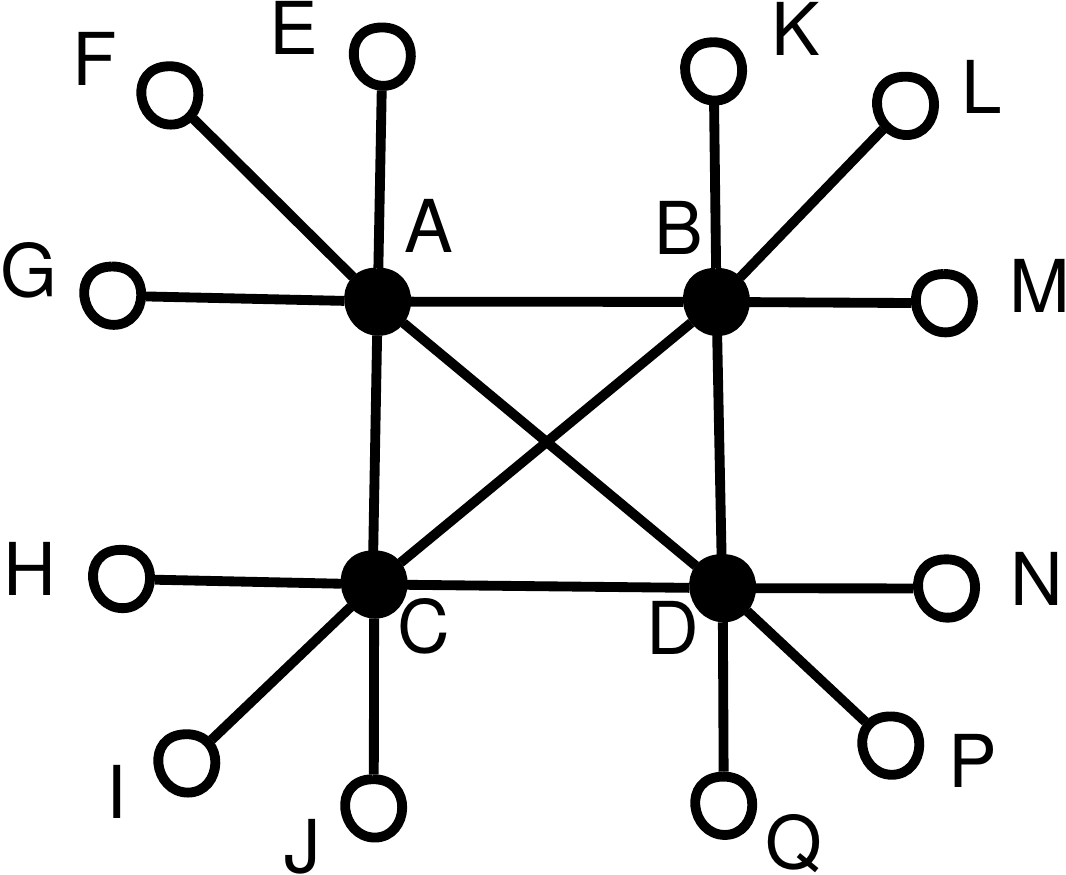} 
  \caption{A 4-star graph}
  \label{fig:kstar}
\end{wrapfigure}

%
For example, in Figure~\ref{fig:kstar}, nodes $G$ and $H$ are of the same type. Also, the two possible connections $MG$ and $MH$ with respect to node $M$, are of the same type.
But the possible connections $EG$ and $EH$ with respect to node $E$, are not of the same type. So, these two strategies of node $E$, namely, connecting with node $G$ and connecting with node $H$, form different branches in the game tree, implying that the utilities arising from these two types of connections are not necessarily equal.

\subsubsection{Directing Network Evolution}
\label{directing_dymanics}

Our procedure for deriving sufficient conditions for the formation of a given topology as the unique topology, is modeled on the lines of {\em mathematical induction}. Consider a base case network with very few nodes (two in our analysis).
We derive conditions so that the network formed with these few nodes has the desired topology. Then using induction, we assume that a network with $n-1$ nodes has the desired topology, and derive conditions so that, the network with $n$ nodes, also has that topology. 
Without loss of generality, we explain this procedure with the example of star topology, referring to the game tree in Figure~\ref{fig:star}. 
Assuming that the network formed with $n-1$ nodes is a star, our objective is to derive conditions so that the network of $n$ nodes is also a star.

In Figure~\ref{fig:star}, at the root of the game tree, node $A$ is the newly entering $n^{th}$ node and the network is in state 0, where a star with $n-1$ nodes is already formed. 
%
Recall that the complexity of analyzing a network
depends on the number of different types of nodes as well as the number of different types of possible connections with respect to a node in that network.
Note that in state 0, with respect to node $A$, there are two types of possible connections: (a) with the center and (b) with a leaf node. 
In states 1, 3, 4 and 5, there are two types of nodes, and two types of possible connections with respect to a leaf node and one with respect to the center. It will be seen that, the network is directed to not enter state 2, so even though there are four types of nodes in that state, it is not a matter of concern.

\begin{figure}[t]

\begin{tabular}{cc}
\hspace{1cm}
\begin{minipage}{3.7cm}
\vspace{-2.8in}
\centering
\includegraphics[scale=0.37]{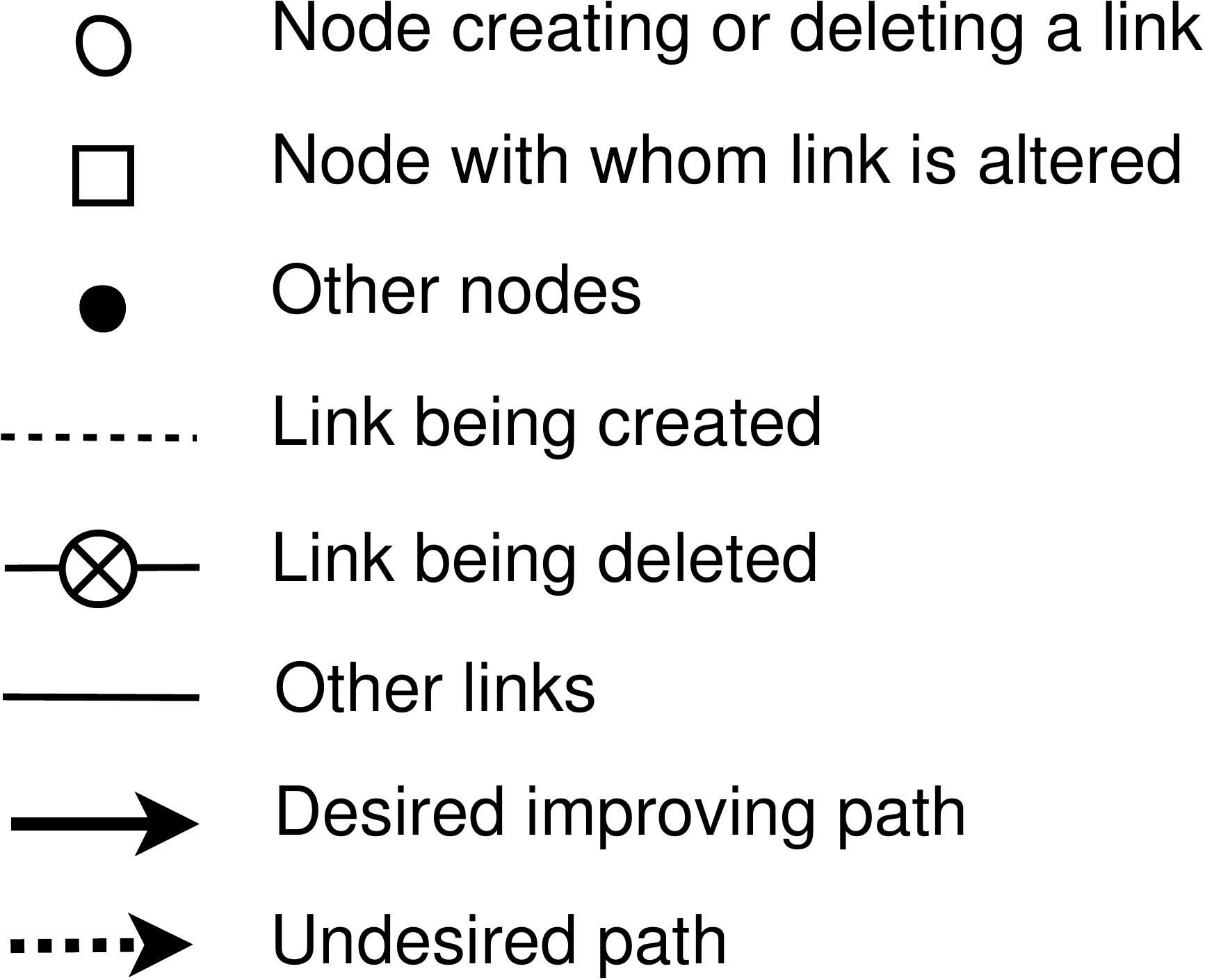}
\end{minipage}
&
\begin{minipage}{7cm}
\centering
\includegraphics[scale=0.37]{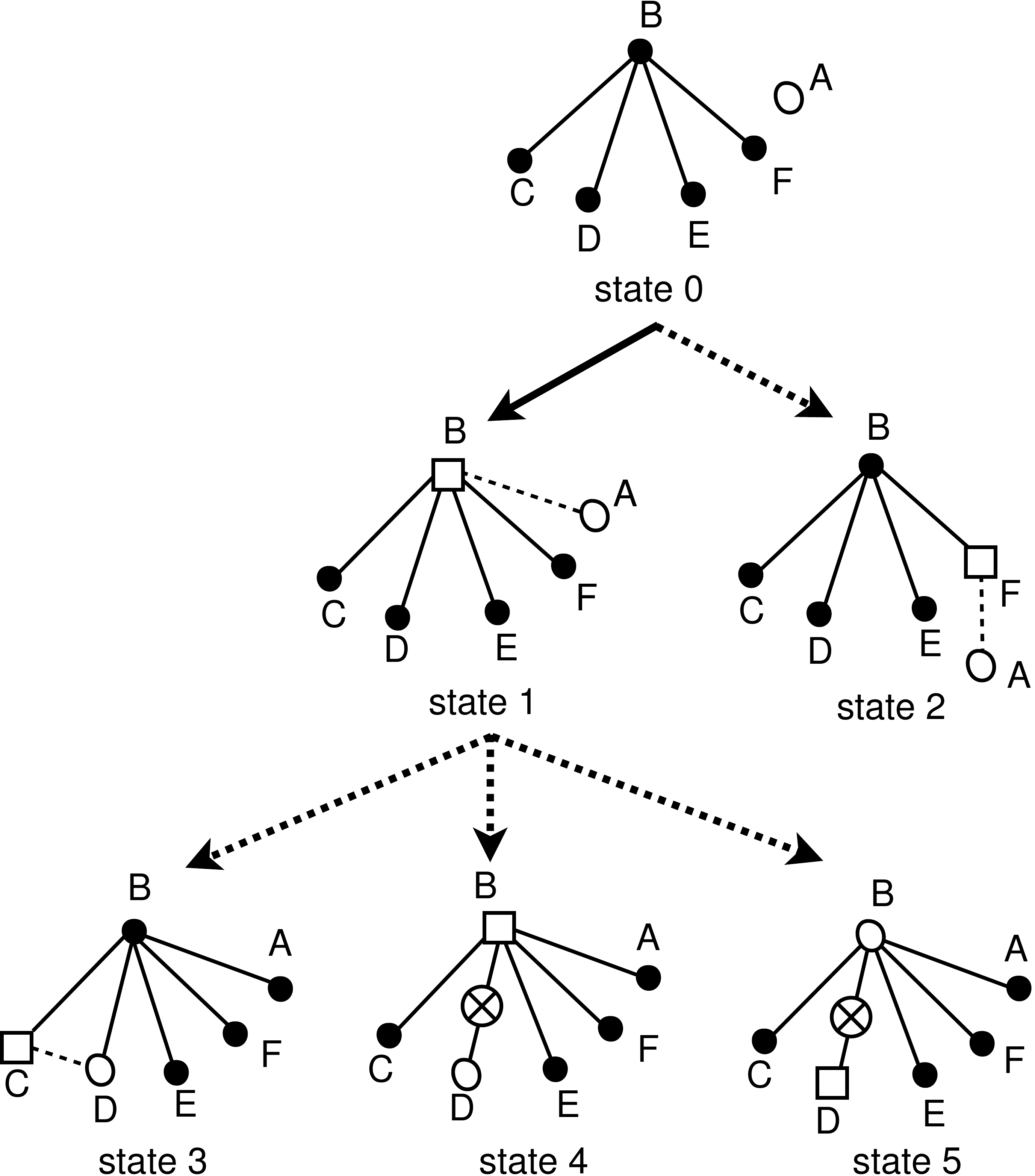}
\end{minipage}
\end{tabular}
\caption{Directing the network evolution for the formation of star topology uniquely}
\label{fig:star}
\end{figure}

Let $u_j(s)$ be the utility of node $j$ when the network is in state $s$.
In state 0, as the newly entering node $A$ gets to make the first move,
we want it to connect to the center by choosing the improving path that transits from state 0 to state 1. So utility of node $A$ in state 1 should be greater than that in state 0, that is, $u_A(1) > u_A(0)$.
Similarly, for node $B$ to accept the link from node $A$, $B$'s utility should not decrease, that is $u_B(1) \geq u_B(0)$. 
We do not want node $A$ to connect to any of the leaf nodes, that is, we do not want the network to enter state 2. Note that as we are interested in sufficient conditions, we are not concerned if there exists an improving path from state 2 that eventually results in a star (we discard state 2 in order to shorten the analysis). 
One way to ensure that the network does not enter state 2, irrespective of whether it lies on an improving path, is by making it less favorable for node $A$ than the desired state 1, that is, $u_A(2) < u_A(1)$. 
Another way to ensure the same is by a condition for a leaf node such that, accepting a link from node $A$ decreases its utility, and so the leaf node does not accept the link, thus forcing node $A$ to connect to the center. That is, $u_j(2) < u_j(0)$ for any leaf node $j$. Thus the network enters state 1, which is our desired state.

To ensure pairwise stability of our desired state, no improving paths should lead out of it, for which we need to consider two cases. First, when node $B$ gets to make its move, it should not break any of its links (state 5),
that is $u_B(1) \geq u_B(5)$.
Second, when any of the leaf nodes is chosen at random, it should neither create a link with some other leaf node (state 3), nor delete its link with the center (state 4).
The corresponding conditions are $u_j(1) \geq u_j(3)$ and $u_j(1) \geq u_j(4)$ for any leaf node $j$.

Thus we direct the network evolution along a desired improving path by imposing a set of conditions, ensuring that the resulting network is in the desired state or has the desired topology uniquely. 
In the evolution process of a network consisting of homogeneous nodes, the number of branches from a state of the game tree depends on the number of different types of nodes and the number of different types of possible connections with respect to a node, at that particular instant.
As we are primarily interested in the formation of special topologies in a recursive manner (nodes are already organized according to the topology and the objective is to extend the topology to that with one more node, so the existing nodes play the same role as before, and most or all of the existing links do not change), the number of different types of nodes as well as the number of different types of possible connections with respect to a node, are small constants at any instant, thus simplifying the analysis.

\subsection{The Utility Model}
\label{sec:utility}

Keeping in view the necessity of solving the problem in a setting that reflects real-world networks in a reasonably general way, we propose a utility model that captures several key determinants of social network formation.
In particular, our model is a considerable generalization of the extensively explored symmetric connections model~\cite{jackson1996strategic} and also builds upon other well known models in literature~\cite{goyal2007structural,kleinberg2008strategic}.
Furthermore, as nodes have global knowledge of existing nodes in the network while making their decisions (for instance, proposing a link with a faraway node), we propose a utility model that captures the global view of indirect and bridging benefits. 

\begin{definition} \cite{goyal2007structural}
\label{def:essential}
A node $j$ is said to be {\em essential} for nodes $y$ and $z$ if $j$ lies on every path joining $y$ and $z$.
\end{definition}

Whenever nodes $y$ and $z$ are directly connected, they get the entire benefits arising from the direct link. On the other hand, when they are indirectly connected with the help of other nodes, of which at least one is essential, $y$ and $z$ lose some fraction of the benefits arising from their communication, in the form of intermediation rents paid to the essential nodes without whom the communication is infeasible. 

Let $E(x,y)$ be the set of essential nodes connecting nodes $y$ and $z$.
The model proposed by Goyal and Vega-Redondo~\cite{goyal2007structural} suggests that the benefits produced by $y$ and $z$ be divided in a way that $x$, $y$, and the nodes in $E(x,y)$ get fraction $\frac{1}{|E(x,y)|+2}$ each.
However, in practice, if nodes $y$ and $z$ can communicate owing to the essential nodes connecting them, that pair would want to enjoy at least some fraction of the benefits obtained from each other, since that pair is the real producer of these benefits (and possess human characteristics such as ego and prestige). That is, the pair would not agree to give away more than some fraction, say $\gamma$, to the corresponding set of essential nodes. As this fact is known to all nodes, in particular, to the set of essential nodes, they as a whole will charge the pair exactly $\gamma$ fraction as intermediation rents. As each essential node in the set is equally important for making the communication feasible, it is reasonable to assume that the intermediation rents are equally divided among them.

It can be noted that nodes which lie on every shortest path connecting $y$ and $z$, but are not essential for connecting them, also have bargaining power, since without them, the indirect benefits obtained from the communication would be less. And so, they should get some fraction proportional to their differential contribution, that is, the indirect benefits produced through the shortest path minus the indirect benefits produced through the second shortest path.
But, for simplicity of analysis, we ignore this differential contribution and assume that nodes that lie on path(s) connecting $y$ and $z$, but are not essential, do not get any share of the intermediation rents. So, when $y$ and $z$ are indirectly connected with the help of other nodes of which none is essential, they get the entire indirect benefits arising from their communication.

We now describe the determinants of network formation that our model captures, and thus obtain expression for the utility function.
Let $N$ be the set of nodes present in the given network,
$d_j$ be the degree of node $j$,
$l(j,w)$ be the shortest path distance between nodes $j$ and $w$,
$b_i$ be the benefit obtained from a node at distance $i$ in absence of rents (assume $b_\infty=0$),
and
$c$ be the cost for maintaining link with an immediate neighbor.

\textbf{\textit{(1) Network Entry Fee:}} 
Since nodes enter a network one by one, we introduce the notion of {\em network entry fee}. 
This fee corresponds to some cost a node has to bear in order to be a part of the network. 
It is clear that, if a newly entering node wants its first connection to be with an existing node which is of high importance or degree, then it has to spend more time or effort. 
So we assume the entry fee that the former pays to be an increasing function of the latter's degree, say $d_{\text{T}}$. 
For simplicity of analysis, we assume the fee to be directly proportional to $d_{\text{T}}$ and call the proportionality constant, {\em network entry factor} $c_0$.

%

\textbf{\textit{(2) Direct Benefits:}}
These benefits are obtained from immediate neighbors  in a network.
For a node $j$, these benefits equal $b_1$ times $d_j$.

\textbf{\textit{(3) Link Costs:}}
These costs are the amount of resources like time, money, and effort a node has to spend in order to maintain links with its immediate neighbors.
For a node $j$, these costs equal $c$ times $d_j$.

\textbf{\textit{(4) Indirect Benefits:}}
These benefits are obtained from indirect neighbors, and these decay with distance $(b_{i+1} < b_i)$.
In the absence of rents, the total indirect benefits that a node $j$ gets is $\sum_{w \in N, \text{ } l(j,w)>1}{b_{l(j,w)}}$.


\textbf{\textit{(5) Intermediation Rents:}}
Nodes pay a fraction $\gamma$ ($0 \leq \gamma < 1$) of the indirect benefits, in the form of additional favors or monetary transfers to the corresponding set of essential nodes, if any. The loss incurred by a node $j$ due to these rents is $\sum_{w \in N , \text{ } E(j,w)\neq \phi}{\gamma b_{l(j,w)}}$.


\textbf{\textit{(6) Bridging Benefits:}}
Consider a node $j \in E(y,z)$.
Both $y$ and $z$ benefit $b_{l(y,z)}$ each and so this indirect connection produces a total benefit of $2 b_{l(y,z)}$. 
As described earlier, each node from the set $E(y,z)$ gets a fraction $\frac{\gamma}{|E(y,z)|}$, the absolute benefits being $\left( \frac{\gamma}{|E(y,z)|} \right) 2 b_{l(y,z)}$.
So the bridging benefits obtained by a node $j$ from the entire network is $\sum_{j \in E(y,z),\text{ }\{y,z\}\subseteq N}{ \left( \frac{\gamma}{|E(y,z)|} \right) 2 b_{l(y,z)} }$.


\textbf{\textit{Utility Function:}}
The utility of a node $j$ is a function of the network, that is, $u_j:g \rightarrow \mathbb{R}$. We drop the notation $g$ from the following equation for readability.
Summing up all the aforementioned determinants of network formation that our model captures, we get
%
\begin{equation}
\label{eqn:utility}
\begin{split}
u_j =& -c_0d_{\text{T}(j)}\textbf{I}_{\{j=\text{NE}\}} + d_j(b_1-c) +\sum_{\substack{w \in N \\l(j,w)>1}}{b_{l(j,w)}} \\
& - \sum_{\substack{w \in N \\E(j,w)\neq \phi}}{\gamma b_{l(j,w)}}  
    + \sum_{\substack{j \in E(y,z) \\ \{y,z\}\subseteq N}}{ \left( \frac{\gamma}{|E(y,z)|} \right) 2 b_{l(y,z)} } 
    \end{split}
\end{equation}
where 
$\text{T}(j)$ is the node to which node $j$ connects to enter the network, and
$\textbf{I}_{\{j=\text{NE}\}}$ is 1 when $j$ is a newly entering node about to create its first link, else it is 0.

%
\section{Analysis of Relevant Topologies}
\label{sec:analysis}

Using the proposed model of recursive and sequential network formation and the proposed utility model,
we provide sufficient conditions under which several relevant network topologies, namely star, complete graph, bipartite Tur\'an graph, 2-star, and $k$-star, uniquely emerge as pairwise stable networks.
Note that as the conditions derived for any particular topology are sufficient, 
 there may exist alternative conditions that result in the same topology uniquely.

\subsection{Sufficient Conditions for the Formation of Relevant Topologies Uniquely}

We use Equation~(\ref{eqn:utility}) for mathematically deriving the conditions.

\begin{proposition}
\label{thm:star}
For a network, if $b_1-b_2 + \gamma b_2 \leq c < b_1$ and $c_0 < \left( 1-\gamma \right) \left( b_2-b_3 \right)$, 
the unique resulting topology is star.
\end{proposition}
\begin{proof}
Refer to Figure~\ref{fig:star} throughout the proof. For the base case of $n=2$, the requirement for the second node to propose a link 
to the first is that its utility should become strictly positive. Also as the first node has degree $0$, there is no entry fee.
\begin{equation}
\label{B1for2} 
0 < b_1-c  \iff c < b_1
\end{equation}
Now, consider a star consisting of $n-1$ nodes. Let the newly entering $n^{th}$ node get to make a decision of whether to enter the network. For $n\geq3$, if the entering node connects to the center, it gets indirect benefits of $b_2$ each from $n-2$ nodes. But as the center is essential for enabling communication between newly entering node and other leaf nodes, the new node has to pay $\gamma$ fraction of these benefits to the center. Also, it has to pay an entry fee of $(n-2)c_0$ as the degree of center is $n-2$. So in Figure~\ref{fig:star}, 
$u_A(0) < u_A(1)$ gives
\begin{equation}
\nonumber
0 < b_1-c+(n-2) \left( 1-\gamma \right) b_2-(n-2)c_0
\end{equation}
\begin{equation}
\nonumber
\iff  c < b_1+(n-2) \left( \left( 1-\gamma \right) b_2-c_0 \right)
\end{equation}
As it needs to be true for all $n \geq 3$, we set the condition to
\begin{equation}
\nonumber
 c < \min_{n \geq 3} \Big\{ b_1+(n-2) \left( \left( 1-\gamma \right) b_2-c_0 \right) \Big\}
\end{equation}
\begin{equation}
\label{B1}
\Longleftarrow  c <b_1+  \left( 1-\gamma \right)  b_2-c_0
\end{equation}
The last step is obtained so that the condition for link cost is independent of the upper limit on the number of nodes, by enforcing
\begin{equation}
\label{B1forc0}
c_0 \leq  \left( 1-\gamma \right) b_2
\end{equation}
which enables us to substitute $n=3$ and the condition holds for all $n\geq 3$.\\
For the center to accept a link from the newly entering node, we need to have $u_B(0) \leq u_B(1)$.
For $n=2$, the requirement for the first node to accept link from the second node is $0 \leq b_1-c$ which is satisfied by Inequality~(\ref{B1for2}).
For $n=3$, as the center is essential for connecting the other two nodes separated by distance two, it gets $\gamma$ fraction of $b_2$ from both the nodes. So it gets bridging benefits of $2 \gamma b_2$.
\begin{equation}
\nonumber
 b_1-c \leq 2(b_1-c)+2\gamma b_2
\end{equation}
\begin{equation}
\nonumber
\iff c \leq b_1+ 2\gamma b_2
\end{equation}
This condition is satisfied by Inequality~(\ref{B1for2}).
For $n\geq 4$, prior to entry of the new node, the center alone bridged $\dbinom{n-2}{2}$
\normalsize pairs of nodes at distance two from each other, while after connecting with the new node, the center is the sole essential node for \small{$\dbinom{n-1}{2}$} 
\normalsize such pairs.
So the required condition:
\begin{equation}
\nonumber
 (n-2)(b_1-c)+\gamma \dbinom{n-2}{2}2 b_2 \leq (n-1)(b_1-c)+\gamma \dbinom{n-1}{2}2 b_2
\end{equation}
This condition is satisfied by Inequality~(\ref{B1for2}) for all $n \geq 4$.\\
For the newly entering node to prefer the center over a leaf node as its first connection (not applicable for $n=2$ and $3$), we need $u_A(1) > u_A(2)$.
\begin{equation}
\nonumber
\begin{split}
 b_1-c +(n-2) \left( 1-\gamma \right) b_2 -(n-2)c_0  > b_1-c+ \left( 1-\gamma \right) b_2 -c_0 +(n-3) \left( 1-\gamma \right) b_3
 \end{split}
\end{equation}
\begin{equation}
\label{B3a}
\iff c_0 <  \left( 1-\gamma \right) \left( b_2-b_3 \right)
\end{equation}
Alternatively, the newly entering node may want to connect to the leaf node, but the leaf node's utility decreases. In that case, the alternative condition can be $u_j(2)<u_j(0)$ for $j=C,D,E,F$.
Note that this leaf node gets bridging benefits of $2\gamma b_2$ for being essential for indirectly connecting the new node with the center. Also, as it is one of the two essential nodes for indirectly connecting the new node with the other $n-3$ leaf nodes (the other being the center), it gets bridging benefits of $(n-3) (\frac{\gamma}{2})2 b_3 = (n-3) \gamma b_3$.
\begin{equation}
\nonumber
\begin{split}
b_1-c +(n-3) \left( 1-\gamma \right) b_2 > 2(b_1-c)+(n-3) \left( 1-\gamma \right) b_2  + 2\gamma b_2  + (n-3) \gamma b_3
 \end{split}
\end{equation}
which gives $c>b_1+ 2\gamma b_2 + (n-3) \gamma b_3$. But this is inconsistent with the condition in Inequality~(\ref{B1for2}). So in order to ensure that the newly entering node connects to the center and not to any of the leaf nodes, we use Inequality~(\ref{B3a}).

Now that a star of $n$ nodes is formed, we ensure its pairwise stability by deriving conditions for the same. 
Firstly, we ensure that the center does not delete any of its links. So we need $u_B(1) \geq u_B(5)$. Note that from the center's point of view, state $5$ is same as state $0$ and as we have seen earlier that $u_B(0) \leq u_B(1)$, the required condition $u_B(5) \leq u_B(1)$ is already ensured.\\
Next, no two leaf nodes should form a link between them. So we should ensure that, not creating a link between them is at least as good for them as creating, that is $u_j(1) \geq u_j(3)$ for any leaf node $j$. This condition is applicable for $n\geq 3$. 
\begin{equation}
\nonumber
b_1-c+(n-2) \left( 1-\gamma \right) b_2 \geq 2(b_1-c)+(n-3) \left( 1-\gamma \right) b_2
\end{equation}
\begin{equation}
\label{B4}
\iff c \geq b_1-b_2+\gamma b_2
\end{equation}
For a leaf node to not delete its link with the center, we need $u_j(1) \geq u_j(4)$ for any leaf node $j$. For $n \geq 2$, we have
\begin{equation}
\nonumber
b_1-c+(n-2)  \left( 1-\gamma \right) b_2  \geq 0
\end{equation}
\begin{equation}
\nonumber
\iff c \leq b_1+(n-2)  \left( 1-\gamma \right) b_2
\end{equation}
which is a weaker condition than Inequality~(\ref{B1for2}) for $n\geq 2$.

Note that Inequalities~(\ref{B1for2}) and (\ref{B3a}) put together are stronger than Inequalities~(\ref{B1}) and (\ref{B1forc0}) combined. 
We get the required result using Inequalities~(\ref{B1for2}), (\ref{B3a}) and (\ref{B4}).
\end{proof}


We provide the proofs of the remaining results of this section in Appendices~\ref{app:smallworld} through \ref{app:kstar}. 


\begin{proposition}
\label{thm:smallworld}
For a network, if $c<b_1-b_{d+1}$ 
and $c_0\leq(1-\gamma)b_2$, the resulting diameter is at most $d$.
\end{proposition}
%

The following corollary results when $d=1$.

\begin{corollary}
\label{thm:complete}
For a network, if $c < b_1-b_2$ and $c_0 \leq \left( 1-\gamma \right) b_2$, the unique resulting topology is complete graph.
\end{corollary}
%


\begin{proposition}
\label{thm:bipartite}
For a network with $\gamma <   \frac{b_2 - b_3}{3b_2 - b_3} $, if $b_1-b_2+ \gamma \left( 3b_2 - b_3 \right) <  c < b_1 - b_3$ 
and $\left( 1-\gamma \right) \left( b_2-b_3 \right) < c_0 \leq \left( 1-\gamma \right) b_2$, the unique resulting topology is 
bipartite Tur\'an graph.
\end{proposition}


%
\begin{proposition}
\label{thm:2star}
Let $\sigma$ be the upper bound on the number of nodes that can enter the network and $\lambda = \lceil \frac{\sigma}{2} -1 \rceil \left( 2b_2-b_3 \right)$.
Then, if $\left( 1-\gamma \right) \left( b_2-b_3 \right) < c_0 < \left( 1-\gamma \right) \left( b_2-b_4 \right) $ and either \\
(i) $\gamma < \min \Big\{  \frac{b_2-b_3}{\lambda-b_3} ,  \frac{b_3}{b_2+b_3} \Big\}$ and $b_1-b_3+\gamma(b_2+b_3) \leq c < b_1$, or\\
(ii) $\frac{b_2-b_3}{\lambda-b_3} \leq \gamma < \min \Big\{ \frac{b_2}{\lambda+b_2} , \frac{b_3}{b_2+b_3}  \Big\}$ and $b_1-b_2+\gamma b_2 + \gamma \lambda \leq c < b_1$,
\\
the unique resulting topology is 
2-star.
\end{proposition}


The following corollary transforms the above conditions in (i) to be independent of the upper bound on the number of nodes that can enter the network.

\begin{corollary}
\label{cor:2star}
For a network with $\gamma=0$, if $b_1-b_3 \leq  c < b_1$ and $b_2-b_3< c_0 < b_2-b_4$,
the unique resulting topology is 
2-star.
\end{corollary}



We define {\em base graph} of a network formation process as the graph from which the process starts. The conditions derived for the formation of the above networks are obtained starting from the graph consisting of a single node (corresponding to the base case of formation of a network with $n=2$). 
Now for certain topologies to be well-defined, it is required that the network has a certain minimum number of nodes. For instance, for a network to have a well-defined $k$-star topology, it should consist of at least $2k$ nodes (complete network on $k$ centers with one leaf node connected to each center). So it is reasonable to consider this network of $2k$ nodes as a base graph for forming a $k$-star network.
Moreover, in case of some topologies (under a given utility model), the conditions required for its formation on discretely small number of nodes, may be inconsistent with that required on arbitrarily large number of nodes. We will now see that, under the proposed network formation and utility models, $k$-star ($k\geq 3$) is one such topology; and a way to circumvent this problem is to start the network formation process from the aforementioned base graph.

Note that in a real-world network, the upper bound on the number of nodes 
is unknown to the network owner. So it is essential that, irrespective of the number of nodes, the desired topology is formed and is stable. That is, the conditions on the network must be set such that the entire family of networks having that topology, is stable.


\begin{lemma}
\label{lem:kstar0}
Under the proposed utility model, for the entire family of $k$-star networks (given some $k\geq 3$) to be pairwise stable, it is necessary that 
$\gamma=0$ and $c=b_1-b_3$.
\end{lemma}

It can be seen that 
the conditions necessary for the family of $k$-star networks to be pairwise stable (Lemma~\ref{lem:kstar0})
are sufficient conditions for the formation of a 2-star network uniquely, 
when $b_2-b_3< c_0 < b_2-b_4$ (Corollary~\ref{cor:2star}).
%
When $c_0 < b_2-b_3$, these conditions $\gamma=0$ and $c=b_1-b_3$, are sufficient for the formation of a star topology uniquely (Proposition~\ref{thm:star}).
When $b_2-b_4 < c_0 < b_2$, these necessary conditions form a cycle among the initially entered nodes, but fails to form a clique among $k$ nodes even as more nodes enter the network, thus making it inconsistent with the $k$-star topology.
It can be similarly seen that for other values of $c_0$ including the boundary cases $c_0 = b_2-b_3$ and $c_0 = b_2-b_4$, the network so formed is not consistent with $k$-star topology for any $k\geq 3$. 
So we have that,
%
under the proposed network formation and utility models, with the requirement that the entire family be pairwise stable, no $k$-star network (given some $k\geq3$) can be formed starting with a network consisting of a single node.
%


A reasonable solution to overcome this problem is to start the network formation process from some other base graph. Such a graph can be obtained by external methods such as providing additional incentives to its nodes.
For 
initializing the formation of $k$-star, 
as mentioned earlier,
the base graph can be taken to be the complete network on the $k$ centers, with the centers connected to one leaf node each. As the base graph consists of $2k$ nodes, the induction starts with the base case for formation of $k$-star network with $n=2k+1$.


\begin{proposition}
\label{thm:kstar}
For a network starting with the base graph for $k$-star (given some $k \geq 3$), and $\gamma =0 $, if $c =b_1-b_3 $ 
and $ b_2-b_3  < c_0 < b_2-b_4$, the unique resulting topology is 
$k$-star.
\end{proposition}

\subsection{Intuition Behind the Sufficient Conditions}
\label{sec:explain}

The network entry fee has an impact on the resulting topology as seen from the above propositions. For instance, in Propositions~\ref{thm:star} and \ref{thm:bipartite}, the intervals spanned by the values of $c$ and $\gamma$ may intersect, but the values of network entry factor $c_0$ span mutually exclusive intervals separated at $(1-\gamma)(b_2-b_3)$. In case of star, $c_0$ is low and so a newly entering node can afford to connect to the center, which in general, has very high degree. In case of bipartite Tur\'an graph, it is important to ensure that the sizes of the two partitions are as equal as possible. As $c_0$ is high, a newly entering node connects to a node with a lower degree (whenever applicable), that is, to a node that belongs to the partition with more number of nodes. Hence the newly entering node potentially becomes a part of the partition with fewer number of nodes, thus maintaining a balance between the sizes of the two partitions. 
In case of $k$-star, as the objective is to ensure that a newly entering node connects to a node with moderate degree, the network entry factor is not so high that a newly entering node prefers connecting to a leaf node and not so low that it prefers connecting to a center with the highest degree. This intuition is clearly reflected in Propositions~\ref{thm:2star} and \ref{thm:kstar} where $c_0$ takes intermediate values.
In general, {\em network entry factor} $c_0$ plays an important role in dictating the degree distribution of the resulting network; 
a higher value of $c_0$ lays the foundation for formation of a more regular graph.

As $c$ increases, the desirability of a node to form links decreases.
This is clear from Proposition~\ref{thm:smallworld} which says that, as $c$ decreases, nodes would create more links, hence effectively reducing the network diameter.
In particular, a complete network is formed when the costs of maintaining links is extremely low, as reflected in Corollary~\ref{thm:complete}. The remaining topologies are formed in the intermediate ranges of $c$. 

From Propositions~\ref{thm:bipartite}, \ref{thm:2star} and \ref{thm:kstar}, it can be seen that the feasibility of a network being formed depends on the values of $\gamma$ as well, which arises owing to contrasting densities of connections in a network. 
For instance, in a bipartite Tur\'an network, nodes belonging to different partitions are densely connected with each other, while that within the same partition are not connected at all. Similarly, in a $k$-star network, there is an extreme contrast in the densities of connections (dense amongst centers and sparse for leaf nodes).

\subsection{Connection to Efficiency}
\label{sec:efficiency}

We now analyze efficiency of the considered networks. 
As the derived conditions are sufficient, there may exist other sets of conditions that uniquely result in a given topology. We analyze the efficiency 
assuming that the networks are formed using the derived conditions.
%

From Equation~(\ref{eqn:utility}), the intermediation rents are transferable among the nodes, and so do not affect the efficiency of a network. Furthermore, the network entry fee is paid by any node at most once, and so does not account for efficiency in the long run. So the expression for efficiency of a network is
\begin{equation}
\nonumber
\sum_{j\in N} \Bigg( d_j(b_1-c) + \sum_{\substack{w \in N \\l(j,w)>1}}{b_{l(j,w)}} \Bigg)
 \end{equation}

The following result follows from the analysis by 
Narayanam and Narahari~\cite{ramasuri1}.
\begin{lemma}
\label{lem:efficient}
Let $\mu$ be the number of nodes in  network. \\
(a) If $c < b_1-b_2$, complete graph is uniquely efficient.\\
(b) If $b_1-b_2< c \leq b_1 + \left( \frac{\mu-2}{2} \right) b_2$, star is the unique efficient topology.\\
(c) If $c >b_1+ \left( \frac{\mu-2}{2}\right)b_2 $, null graph is uniquely efficient.
\end{lemma} 

The null network in the proposed model of recursive network formation corresponds to a single node to which no other node prefers to connect, and so the network does not grow. 


\begin{proposition}
\label{thm:eff_star}
Based on the derived sufficient conditions, null network, star network, and complete network are efficient.
\end{proposition}
\begin{proof}
It is easy to see that irrespective of the value of $c_0$, if $c>b_1$, no node, external to the network, connects to the only node in the network and hence, does not enter the network. Such a network is trivially efficient as in the range $c>b_1$, it is a star of one node and also a null network. It is also clear that the star network and the complete network are efficient as the conditions on $c$ from Proposition~\ref{thm:star} and Corollary~\ref{thm:complete}, respectively, form a subset of the range of $c$ in which these topologies are respectively efficient.
\end{proof}


It can be seen that when the number of nodes in the network is small, the absolute difference between the efficiency of the resulting network and that of the efficient network is also small, and hence the network owner will not be too concerned about the efficiency of the network. 
So for the following propositions, we make a reasonable assumption that the number of nodes in the network is sufficiently large. 


\begin{proposition}
\label{thm:eff_bipartite}
Based on the derived sufficient conditions, for sufficiently large number of nodes, the efficiency of a bipartite Tur\'an network is half of that of the efficient network in the worst case and the network is close to being efficient in the best case.
\end{proposition}
%
\begin{proof} 
As $\mu$ is large, $\mu$ can be assumed to be even without loss of accuracy. The sum of utilities of nodes in a bipartite Tur\'an network with even number of nodes is approximately
\begin{equation}
\nonumber
\left( \frac{\mu}{2} \right)^2 2(b_1-c)+2 \dbinom{\frac{\mu}{2}}{2}2b_2
\end{equation}
From Lemma~\ref{lem:efficient}, star network is efficient in the range of $c$ derived in Proposition~\ref{thm:bipartite}. So, to get the efficiency of the bipartite Tur\'an network relative to the star network, we divide the above expression by the sum of utilities of nodes in a star network, which is
\begin{equation}
\label{eqn:star_eff}
2(\mu-1)(b_1-c)+\dbinom{\mu -1}{2}2b_2
\end{equation}
Using the assumption that $\mu$ is large and the fact from the derived sufficient conditions that $b_2$ is comparable to $b_1-c$, it can be shown that the efficiency relative to the star network, approximately is
$\frac{1}{2}+\frac{b_1-c}{2b_2}$.
As the range of $c$ in Proposition~\ref{thm:bipartite} depends on the value of $\gamma$, the values of $c$ are bounded by $b_1-b_2$ and $b_1-b_3$. So the efficiency is bounded by 1 on the upper side and $\left( \frac{1}{2}+\frac{b_3}{2b_2} \right)$ on the lower side, of that of the star network; $\left( \frac{1}{2}+\frac{b_3}{2b_2} \right)$ can take a minimum value of $\frac{1}{2}$ when $b_3<<b_2$. 
\end{proof}


\begin{proposition}
\label{thm:eff_kstar}
Based on the derived sufficient conditions, for sufficiently large number of nodes, the efficiency of a $k$-star network is $\frac{1}{k}$ of that of the efficient network in the worst case and the network is close to being efficient in the best case.
\end{proposition}
%
\begin{proof}
As $\mu$ is large, in particular, $\mu >> k$ (not necessarily $>>k^2$), $\mu$ can be assumed to be divisible by $k$ without loss of accuracy. The sum of utilities of nodes in such a $k$-star network is approximately
\begin{equation}
\nonumber
\begin{split}
\left\{\dbinom{k}{2}+(\mu-k) \right\}2(b_1-c)+ \left\{k (k-1)\left( \frac{\mu-k}{k} \right)+k \dbinom{\frac{\mu-k}{k} }{2} \right\} 2b_2  + \dbinom{k}{2} \left( \frac{\mu-k}{k} \right) ^2 2b_3
\end{split}
\end{equation}
From Lemma~\ref{lem:efficient}, star network is efficient in the range of $c$ derived in Propositions~\ref{thm:2star} and \ref{thm:kstar}. So, to get the efficiency of the $k$-star network relative to the star network, we divide the above expression by Expression~(\ref{eqn:star_eff}).
Using the assumption that $\mu$ is large and the fact from the derived sufficient conditions that $b_2$ and $b_3$ are comparable to $b_1-c$, it can be shown that the efficiency relative to the star network, approximately is
$\frac{1}{k}+ \left( 1- \frac{1}{k} \right) \frac{b_3}{b_2}$.
As $b_3$ is bounded by $0$ and $b_2$, the efficiency of $k$-star is bounded by $\frac{1}{k}$ and 1 of that of the star network.
\end{proof} 

\section{Deviation from the Derived Sufficient Conditions: A Simulation Study}
\label{sec:deviation}
We have derived sufficient conditions under which various network topologies uniquely emerge. In this section, we investigate the robustness of the derived sufficient conditions by studying the deviation in network topology when there is a slight deviation in these sufficient conditions. This problem is of practical interest since it may be difficult to maintain the conditions on a network throughout its formation process.

We use the notion of {\em graph edit distance} (GED)~\cite{gao2010survey} to measure 
the deviation in network topology.

\begin{definition}
Given two graphs $g$ and $h$ having same number of nodes, the {\em graph edit distance} 
between them is the minimum number of link additions and deletions required to transform $g$ into a graph that is isomorphic to $h$.
\end{definition}

\subsection{Computation of Graph Edit Distance}
\label{sec:ged}

The problem of computing GED between two graphs is NP-hard, in general~\cite{zeng2009comparing}.
However, we can exploit structural properties of certain graphs to compute GED between them and other graphs, in polynomial time;
we state three such results. 

\begin{theorem}
\label{thm:gedstar}
The graph edit distance between a graph $g$ and a star graph with same number of nodes as $g$, is $\mu+\xi-2\Delta-1$, where $\mu$ and $\xi$ are the number of nodes and edges in $g$, respectively, and $\Delta$ is the 
 highest degree 
in $g$.
\end{theorem}
\begin{proof}
While transforming $g$ into a corresponding star graph, we need to map one node of $g$ to the center while the others to the leaf nodes. Let $d$ be the degree of the node which is mapped to the center. In order to transform $g$ into a star graph, the node mapped to the center must be connected to $\mu-1$ nodes. So the number of edges to be added is $(\mu-1)-d$. Also all edges connecting any two nodes, that are mapped to the leaf nodes, must be deleted, that is, all edges except the ones incident to the node mapped to the center, must be removed. These account for $\xi - d$ edges. Thus, total number of edges to be added and deleted is $\mu+\xi-2d-1$. This is minimized when $d=\Delta$.
\end{proof}

\begin{theorem}
\label{thm:gedcomplete}
The graph edit distance between a graph $g$ and a complete graph with same number of nodes as $g$, is $\frac{\mu(\mu-1)}{2} - \xi$, where $\mu$ and $\xi$ are the number of nodes and edges in $g$, respectively.
\end{theorem}
\begin{proof}
Graph $g$ can be transformed into the corresponding complete graph in minimum number of steps by adding the edges which are absent.
\end{proof}

%

\begin{theorem}
\label{thm:gedkstar}
There exists an $O(\mu^{k+2})$ polynomial time algorithm to compute the graph edit distance 
 between a graph $g$ and a $k$-star graph with same number of nodes as $g$, where $\mu$ is the number of nodes in $g$.
\end{theorem}

We provide the proof of Theorem~\ref{thm:gedkstar} in Appendix~\ref{app:gedkstar}.

\subsection{Simulation Setup}
\label{sec:simsetup}

In order to study the robustness of the derived sufficient conditions, we observed the effects of deviation from these conditions, on the resulting networks, using GED as the measure of topology deviation. We first observed the effect when the conditions were made to deviate throughout the network formation process. The results were, however, uninteresting since the deviation from the sufficient conditions for the formation of one topology, lead to the formation of a completely different topology. 
A primary reason for such observations is that, under the deviated conditions, some other networks are pairwise stable and these networks have a very different topology than the desired one.
In some cases, these deviated conditions were sufficient conditions for other topologies, which were, however, not the desired ones.

In fact, it is unreasonable to assume that the conditions remain deviated throughout the entire network formation process.
It is possible that the conditions deviate at some point of time, 
but the network owner will observe the resulting network under such deviations and take necessary actions to rectify this problem.
This lets us introduce the concept of {\em dynamic conditions} on the network.

In simulations, we assume that the conditions deviate during the entry of a new node and remain deviated throughout the evolution of the network until it reaches pairwise stability. Once stability is reached, the network owner observes the deviation of the network from the desired one, and takes actions to restore the original conditions. As it is undesirable for the network to remain stagnant, any node which wants to enter the deviated network next, is allowed to do so immediately, and the original conditions take effect during the entry of such a node and evolution thereafter.

We observe how the topology deviates when the conditions deviate, and if, how, and when the topology is restored, once the sufficient conditions are restored. We also observe the values within the sufficient conditions which are more robust than others, that is, when the conditions are restored to these values, the topology is restored at the earliest.

For simulations, we set the benefit parameters as per the symmetric connections model~\cite{jackson1996strategic}, that is, we set $b_i=\delta^i$, where $\delta \in (0,1)$; we set $\delta=0.8$ in our simulations.
We consider three types of values within the sufficient conditions, namely, \{low($L$), moderate($M$), high($H$)\}
 for each of the parameters $c$, $c_0$ and $\gamma$ (whenever applicable) and observe the combination of their values which are the most robust to deviations.
 In our simulation study, low values correspond to value around the lower 10\% of the range in sufficient conditions, moderate to around 50\% mark, and high to around higher 10\%.
Also, for each combination, we run the network formation process several times in order to account for the effects of randomization in the order in which nodes take decisions.

Owing to sequential entry of nodes, there is an inherent ordering on nodes and they can be numbered from 1 to the current number of nodes in the network, in the order in which they enter. We call the node number at which the sufficient conditions deviate, as the {\em deviation node}. The sufficient conditions are restored during the entry of the node immediately following the deviation node. 
We say that the deviation from sufficient conditions on a parameter is {\em negative} if the deviated value of the parameter is less than its lower bound in the sufficient conditions, and {\em positive} if its deviated value is greater than its upper bound.
In our simulation study, the amount of deviation for each parameter was 2\% of the length of its range in sufficient conditions. The results observed for 5\% and 10\% deviations were almost same. For parameters whose range in sufficient conditions is a singleton, the results were studied for an absolute deviation of 0.01 on the scale where $b_i = 0.8^i$.

\subsection{Simulation Results}
\label{sec:simresults}

We observe the effects of deviation from the derived sufficient conditions for $c$ and $c_0$ on the resulting network. The observations can be primarily classified into the following four cases, in the decreasing order of desirability to network owner:
\begin{enumerate}
\item[(A)] The network does not deviate during the entry and also during the evolution after the entry of deviation node.
\item[(B)] The network deviates after the entry of the deviation node, and perhaps remains deviated during the entry and evolution for the entry of nodes following the deviation node, but after a certain number of such node entries, the network regains its original topology.
\item[(C)] The network deviates after the entry of the deviation node and remains deviated during the entry and evolution for the entry of nodes following the deviation node; the network does not regain its original topology, but the deviation is constant and so a near-desired topology is obtained.
\item[(D)] The network deviates after the entry of the deviation node and the deviation increases monotonically during the entry and evolution for the entry of nodes following the deviation node.
\end{enumerate}

\begin{figure} [t!]
\begin{tabular}{c}
\hspace{-.7cm}
\begin{minipage}{.5\textwidth}
\centering
\iftoggle{clr}{
\includegraphics[scale=0.62]{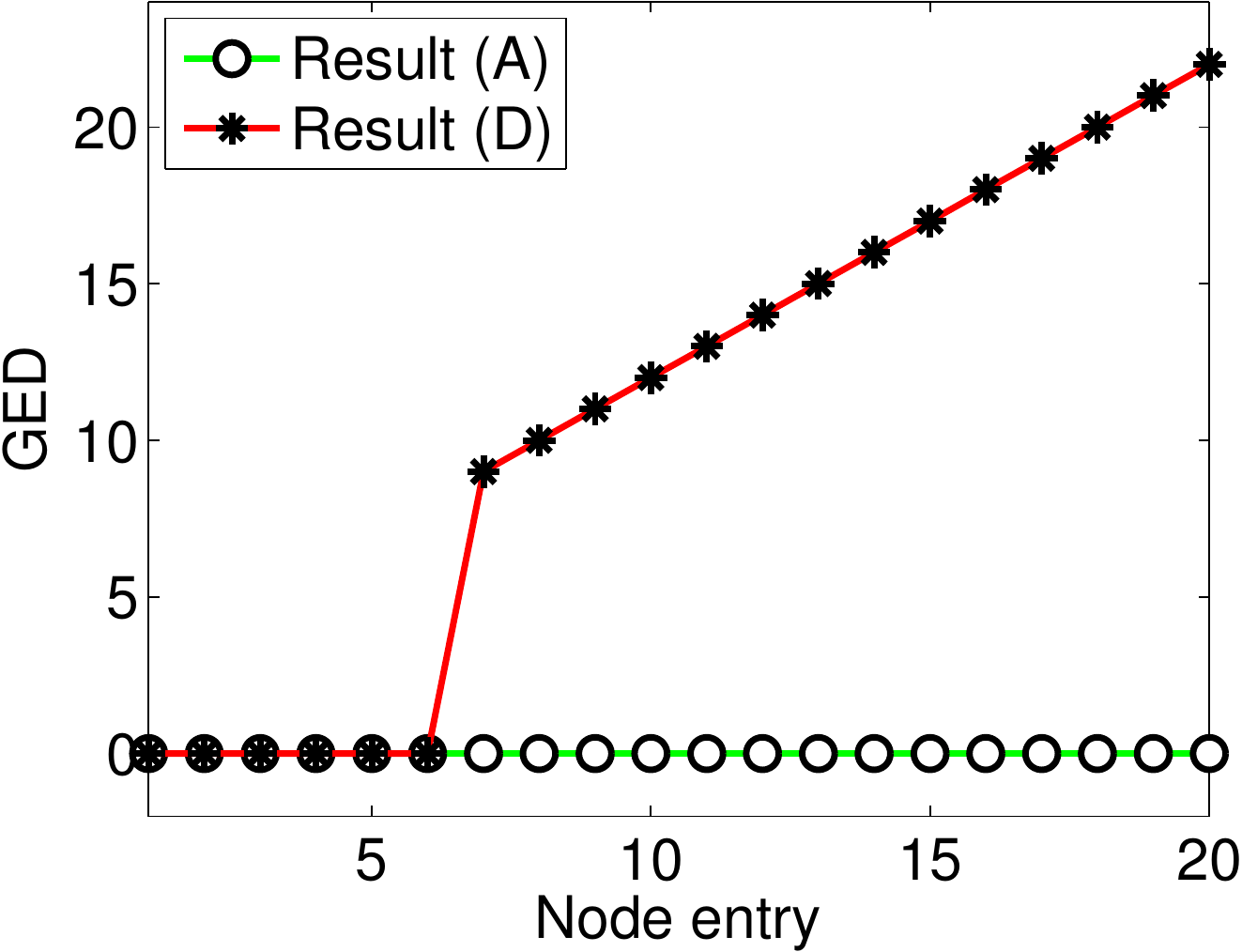}
}{
\includegraphics[scale=0.62]{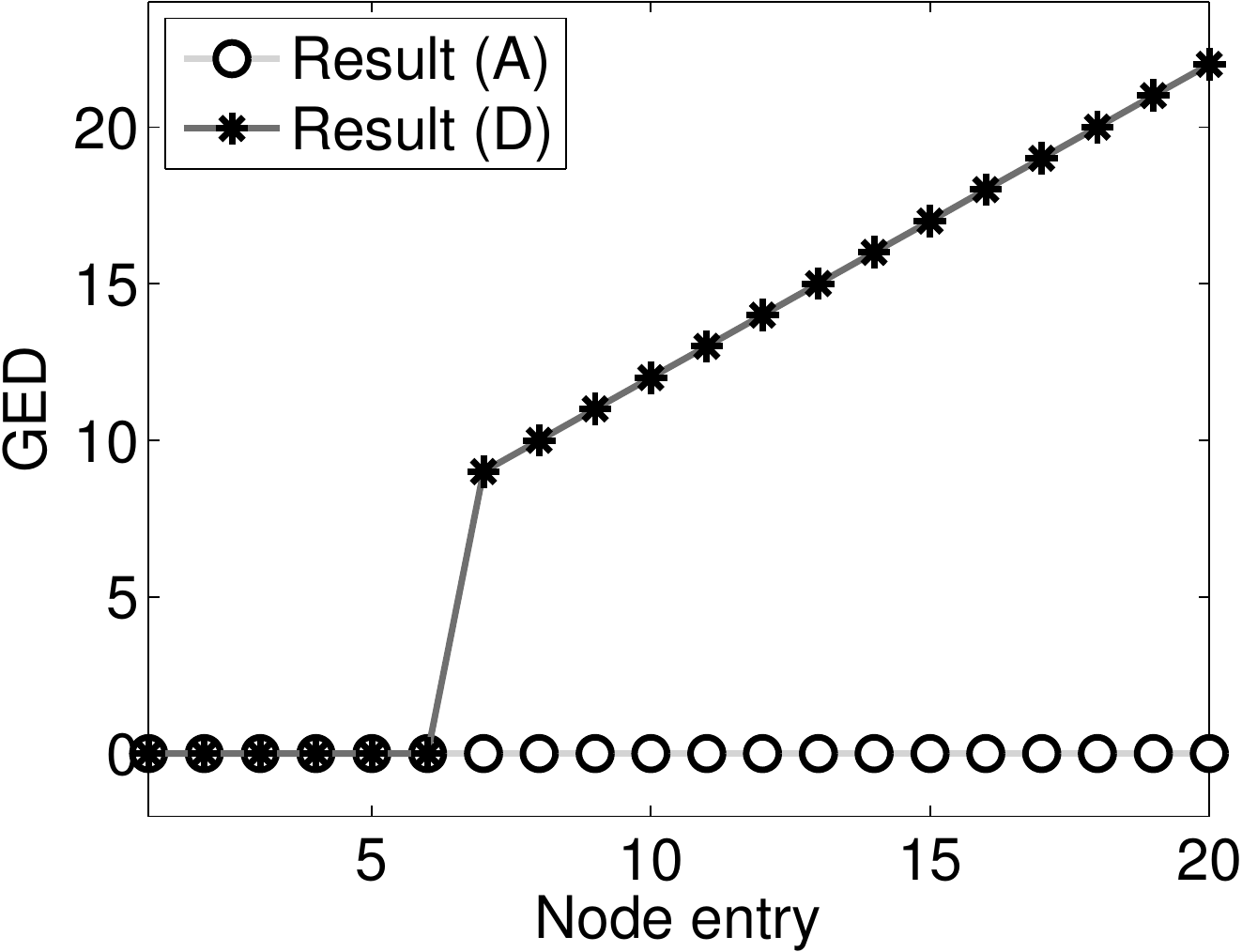}
}
\end{minipage}
\begin{minipage}{.5\textwidth}
\centering
\iftoggle{clr}{
\includegraphics[scale=0.62]{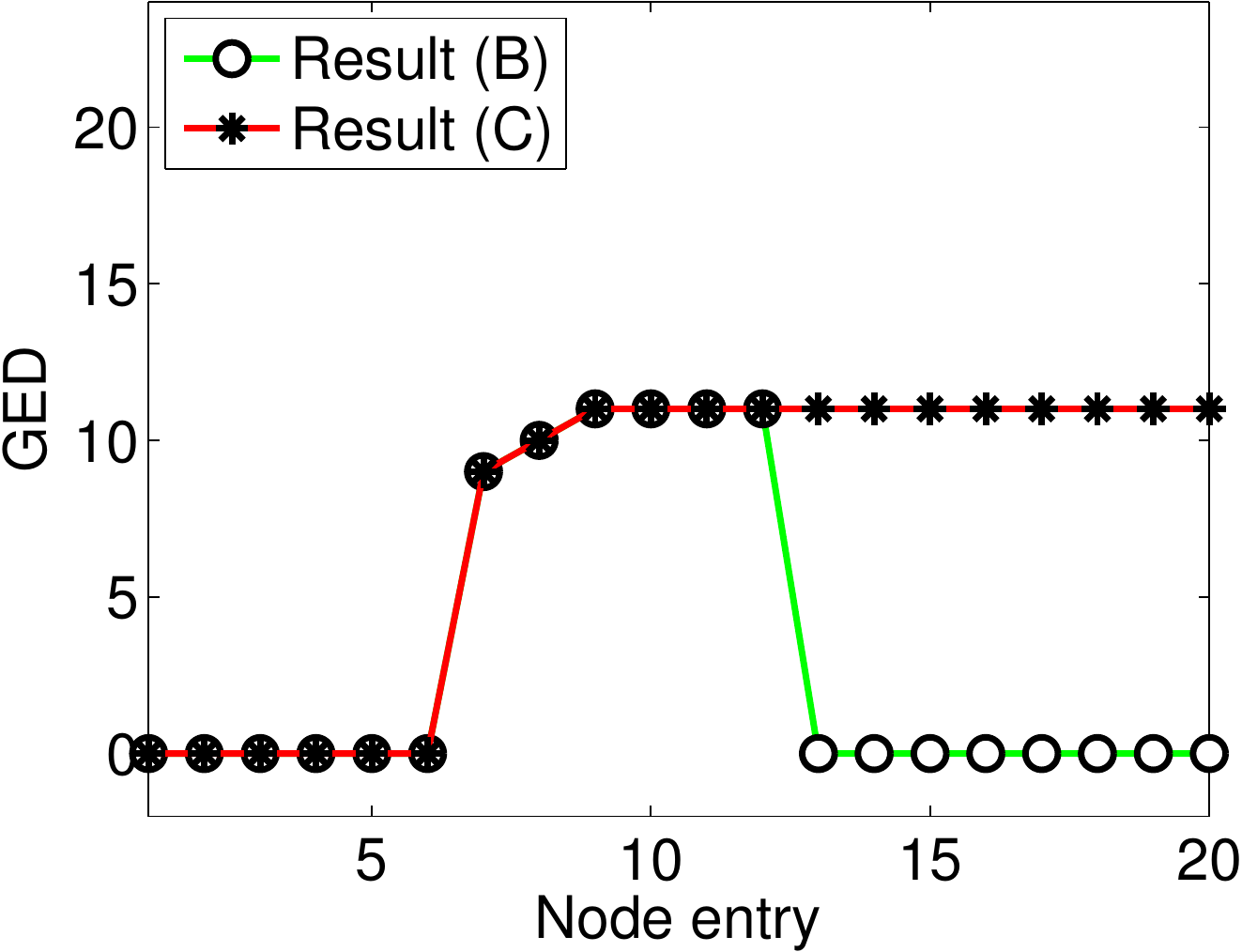}
}{
\includegraphics[scale=0.62]{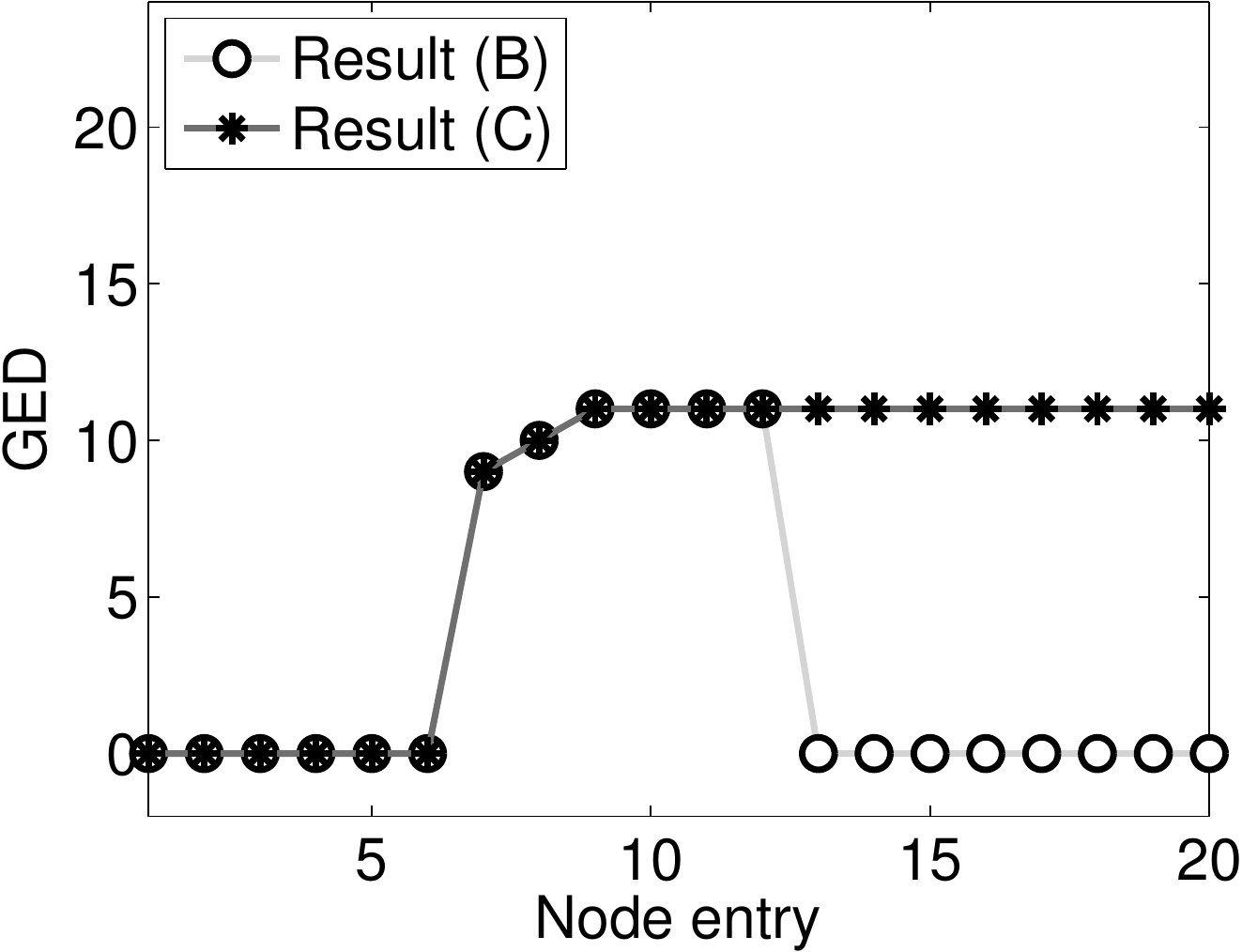}
}
\end{minipage}
\\
\begin{minipage}{.5\textwidth}
\centering
\small{(a)}
\end{minipage}
\begin{minipage}{.5\textwidth}
\centering
\small{(b)}
\end{minipage}
\end{tabular}
\caption{(a-b) Typical results of deviation from the derived sufficient conditions 
for deviation node 7
(Y-axis gives the deviation when the network consists of number of nodes given on X-axis)
}
\label{fig:resulttypes}
\end{figure}

Figures~\ref{fig:resulttypes}(a-b) give typical plots of the above four cases.
The plots are split into two parts for clarity.
Result (A) is the most desirable but can be obtained only for some particular deviation nodes depending on the topology for which the sufficient conditions are derived. Result (B) is very common and
this is the result the network owner should be looking at. Result (C) is good from a practical viewpoint as the resulting network need not be exactly the desired one, but it may still serve the purpose almost entirely. Result (D) is the one that any network owner should avoid.

Recall that $c$ is the cost incurred by a node in order to maintain a link with each of its immediate neighbors. So as $c$ increases, the desirability of a node to form links decreases.
Also as discussed earlier, a higher value of {\em network entry factor} $c_0$ lays the foundation for formation of a more regular graph. In general, it plays an important role in dictating the degree distribution of the resulting network.
In what follows, we study the effects of all valid deviations from sufficient conditions on cost parameters $c$ and $c_0$, on the resulting network.
In the tables that follow, if there were very few instances in which the network did not deviate, we ignore them since such cases are remote when nodes take decisions in some particular order.
For observing deviations from $k$-star topology ($k \geq 3$), the network is assumed to start with the corresponding base graph consisting of $2k$ nodes as discussed earlier.

Enlisted are the major findings of the simulations:
\begin{itemize} 
\item Certain values of parameters within the derived sufficient conditions may be more robust than others, that is, the value to which the conditions are restored during the entry of the node immediately following the entry of the deviation node, may directly affect the restoration of the topology.
\item Network with certain number of nodes may be bottleneck for the range of sufficient conditions (can be seen from the derivations of these conditions). In such cases, the topology deviates only for discretely few deviation nodes, while it does not for others. So the network owner may relax the conditions for most of the network formation process.
\item The sufficient conditions on $c$ are more sensitive than those on $c_0$, that is, the network deviates more from the desired topology when the value of $c$ deviates than when the value of $c_0$ deviates by similar margins.
\item Results obtained owing to deviation from sufficient conditions during the entry of a deviation node may be very different from that obtained owing to deviation during the entry of some other deviation node.
\item It may be possible to uniquely form some interesting topologies which may not be feasible using any static sufficient conditions.
\item In most scenarios, the order in which nodes take decisions plays an important role in deciding the resulting topology. Deviations from sufficient conditions may cause large deviations from the desired topology due to some ordering, while no deviation at all due to some other.
\end{itemize}


\begin{figure} [t!]
\begin{tabular}{c}
\hspace{-.7cm}
\begin{minipage}{.5\textwidth}
\centering
\iftoggle{clr}{
\includegraphics[scale=0.62]{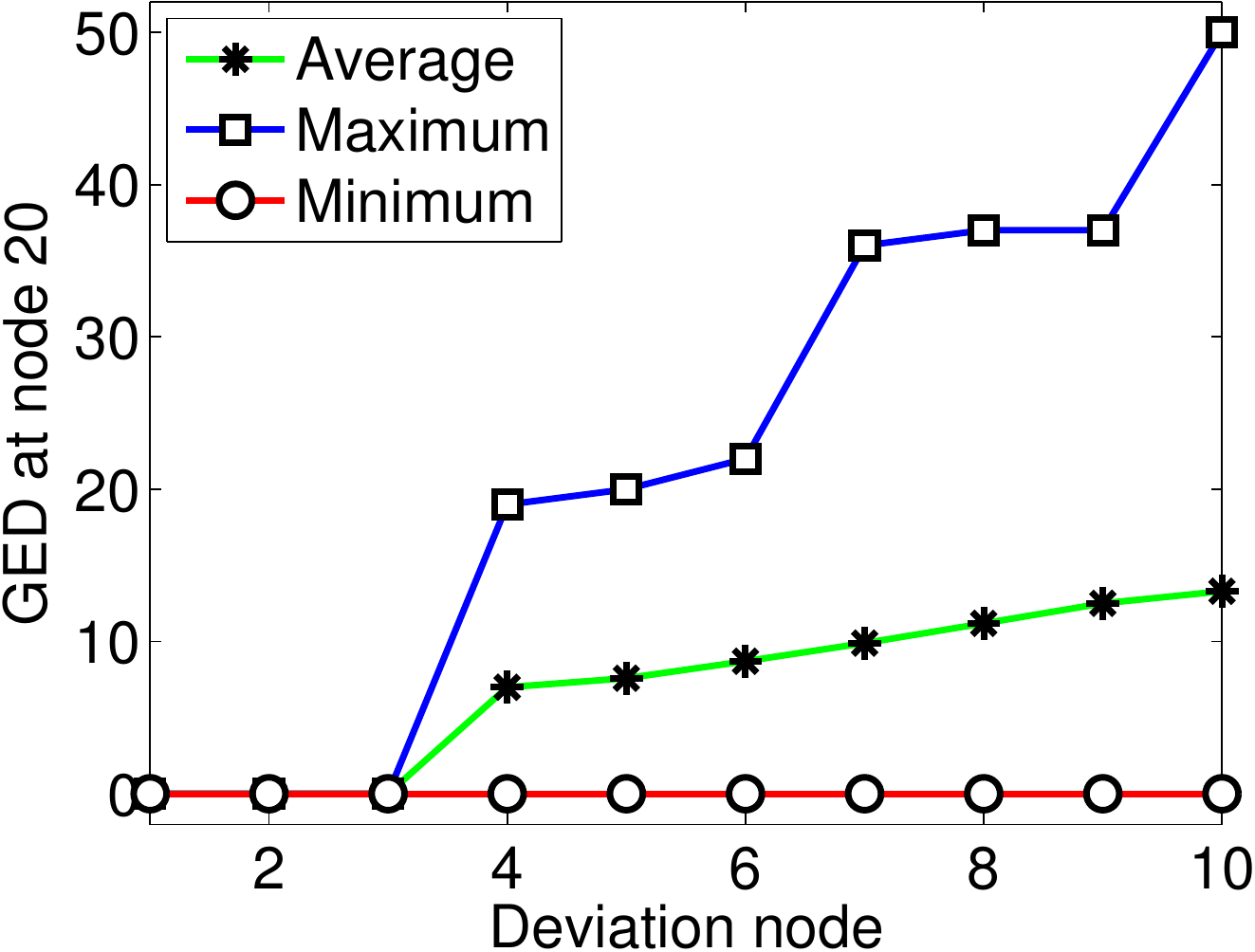}
}{
\includegraphics[scale=0.62]{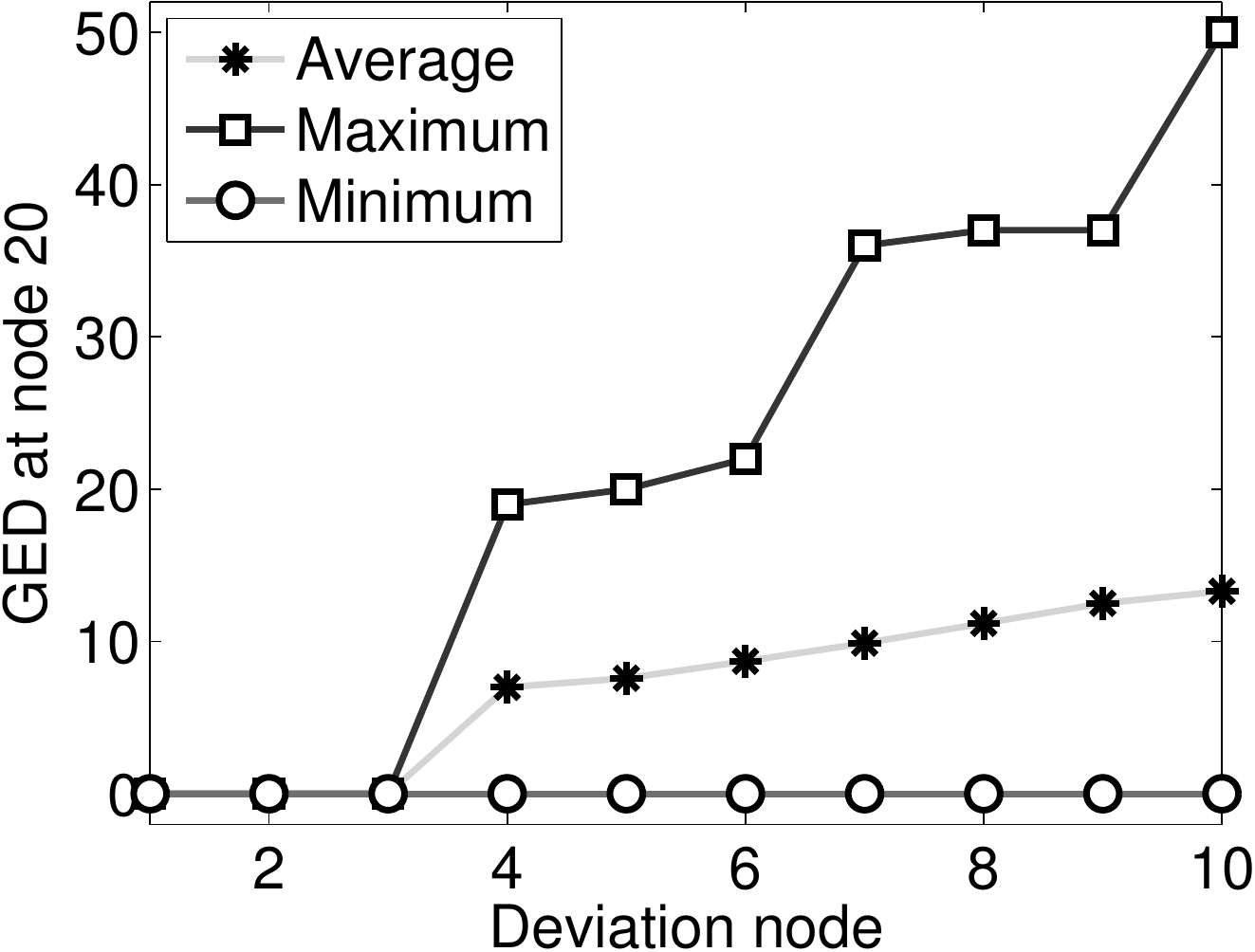}
}
\end{minipage}
\begin{minipage}{.5\textwidth}
\centering
\iftoggle{clr}{
    \includegraphics[scale=0.85]{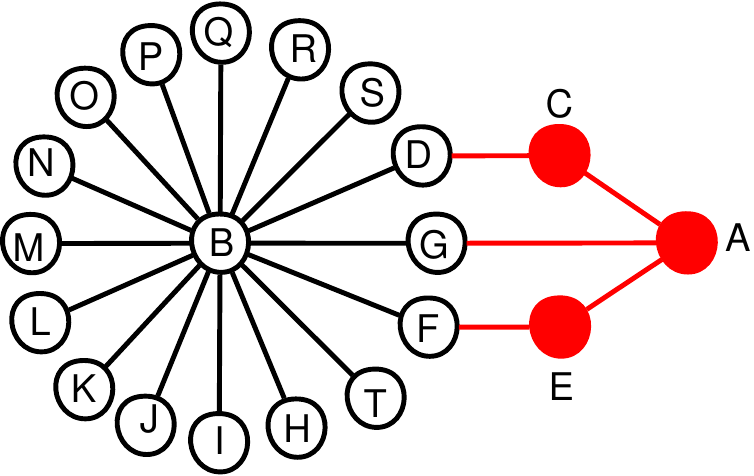}
}{
    \includegraphics[scale=0.85]{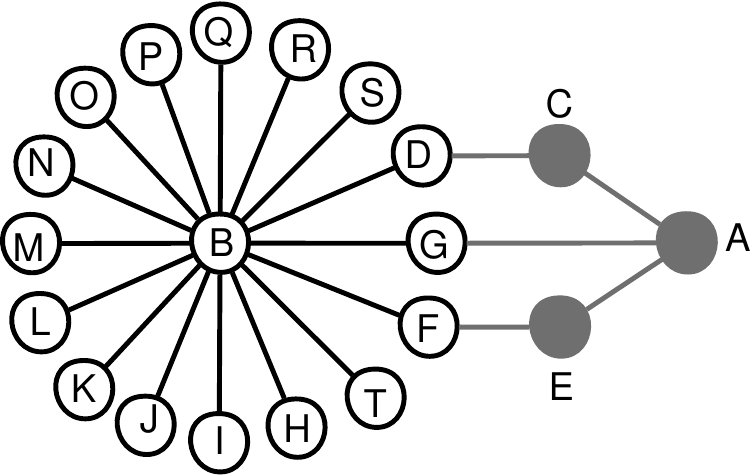}
}
\end{minipage}
\\
\begin{minipage}{.5\textwidth}
\centering
\small{(a)}
\end{minipage}
\begin{minipage}{.5\textwidth}
\centering
\small{(b)}
\end{minipage}
\end{tabular}
\caption{
(a) Results of negative deviation of $c$ from the sufficient conditions for star topology when the network consists of 20 nodes
and 
(b) A near-star network
}
\label{fig:star_costneg}
\end{figure}

The reader should note the difference in labels on the X and Y axes of the different plots in this chapter.

\subsection[Results for Deviation with Respect to ${c}$]{Results for Deviation with Respect to $\boldsymbol{c}$}
\label{sec:devcost}

\subsubsection*{{Negative deviation of $\boldsymbol{c}$ from sufficient conditions for star network:}}
These results are shown qualitatively in Table~\ref{tab:devstarcostneg} and quantitatively in Figure~\ref{fig:star_costneg}(a). 
Figure~\ref{fig:star_costneg}(a) plots the deviation from network as observed for a network with 20 nodes, if the conditions were deviated at a given deviation node.
For deviation nodes 2 and 3, no deviation in network was observed. 
For other deviation nodes, Table~\ref{tab:devstarcostneg} shows the type of result obtained owing to deviation from sufficient conditions on $c$ at a deviation node, following which, the values of $\gamma$, $c_0$ and $c$ are restored to one of \{$L,M,H$\}. 
The results are invariant with respect to the restored value of $c_0$.
The table shows that $\gamma=L$ coupled with $c=H$, and $\gamma=M$ coupled with $c=M \text{ or } H$, give the best results, where the star topology is restored as per result (B). 
$\gamma=L$ coupled with $c=M$, and $\gamma=H$ coupled with $c=M \text{ or } H$, give decent results for practical purposes, where a near-star network (Figure
\ref{fig:star_costneg}(b)) 
is obtained as per result (C). 
$c=L$ is unacceptable and should be avoided by network owner desiring to form a star network, as these values are not robust to deviations from sufficient conditions.
Typical observations 
are shown in Figures~\ref{fig:resulttypes}(a-b). 

\begin{figure*} [t!]
\begin{tabular}{c}
\hspace{-.7cm}
\begin{minipage}{.5\textwidth}
\centering
\iftoggle{clr}{
\includegraphics[scale=0.62]{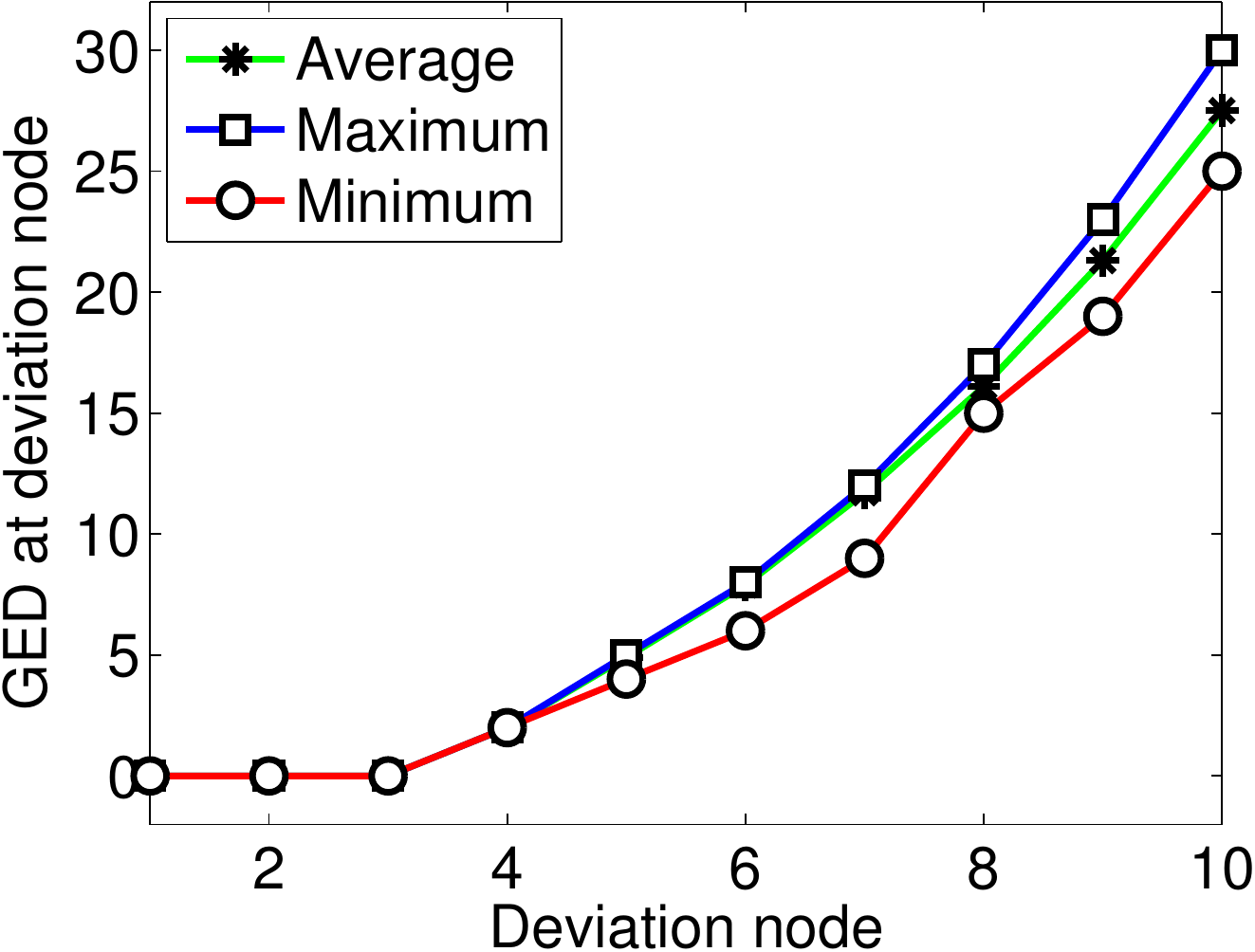}
}{
\includegraphics[scale=0.62]{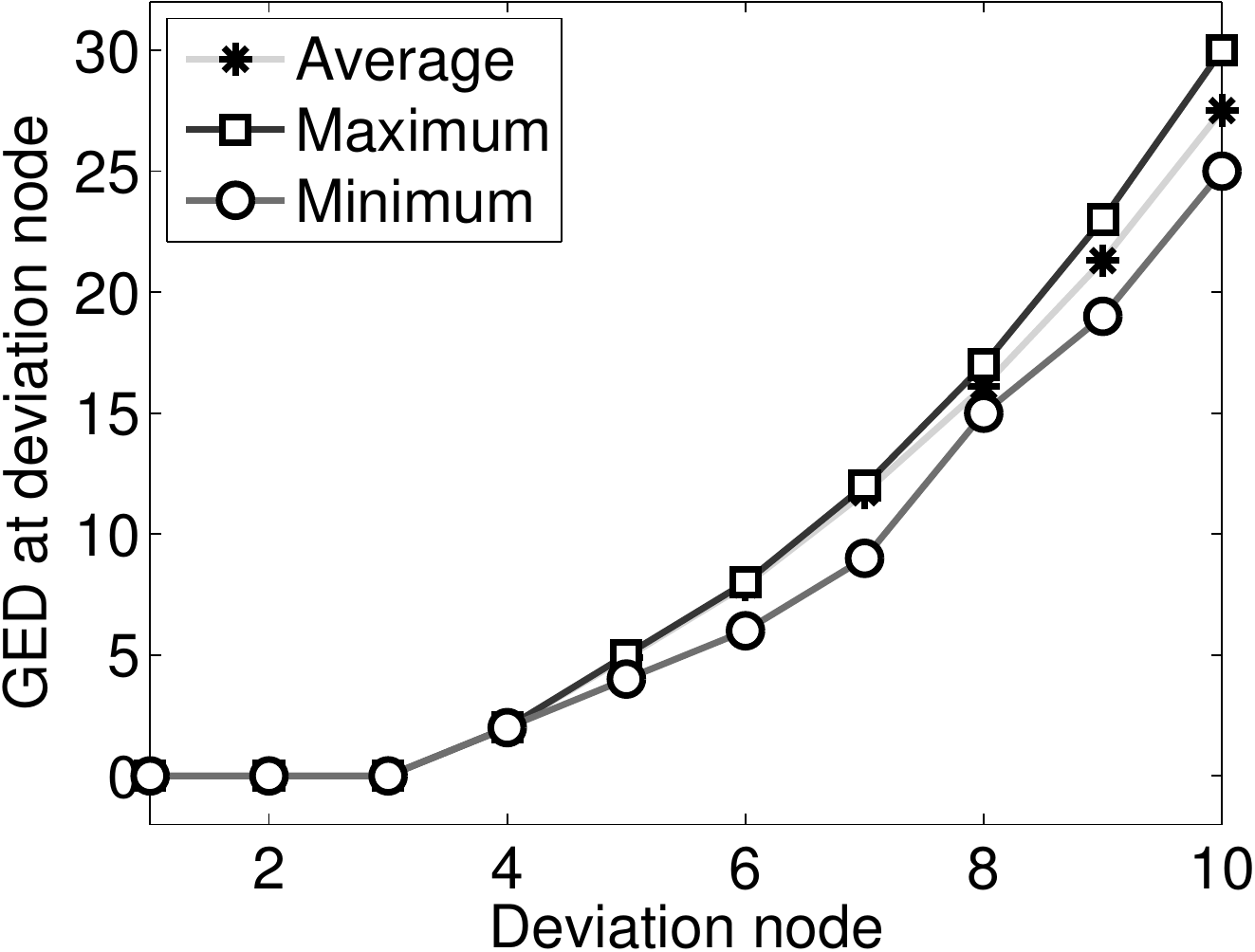}
}
\end{minipage}
\begin{minipage}{.5\textwidth}
\centering
\iftoggle{clr}{
\includegraphics[scale=0.62]{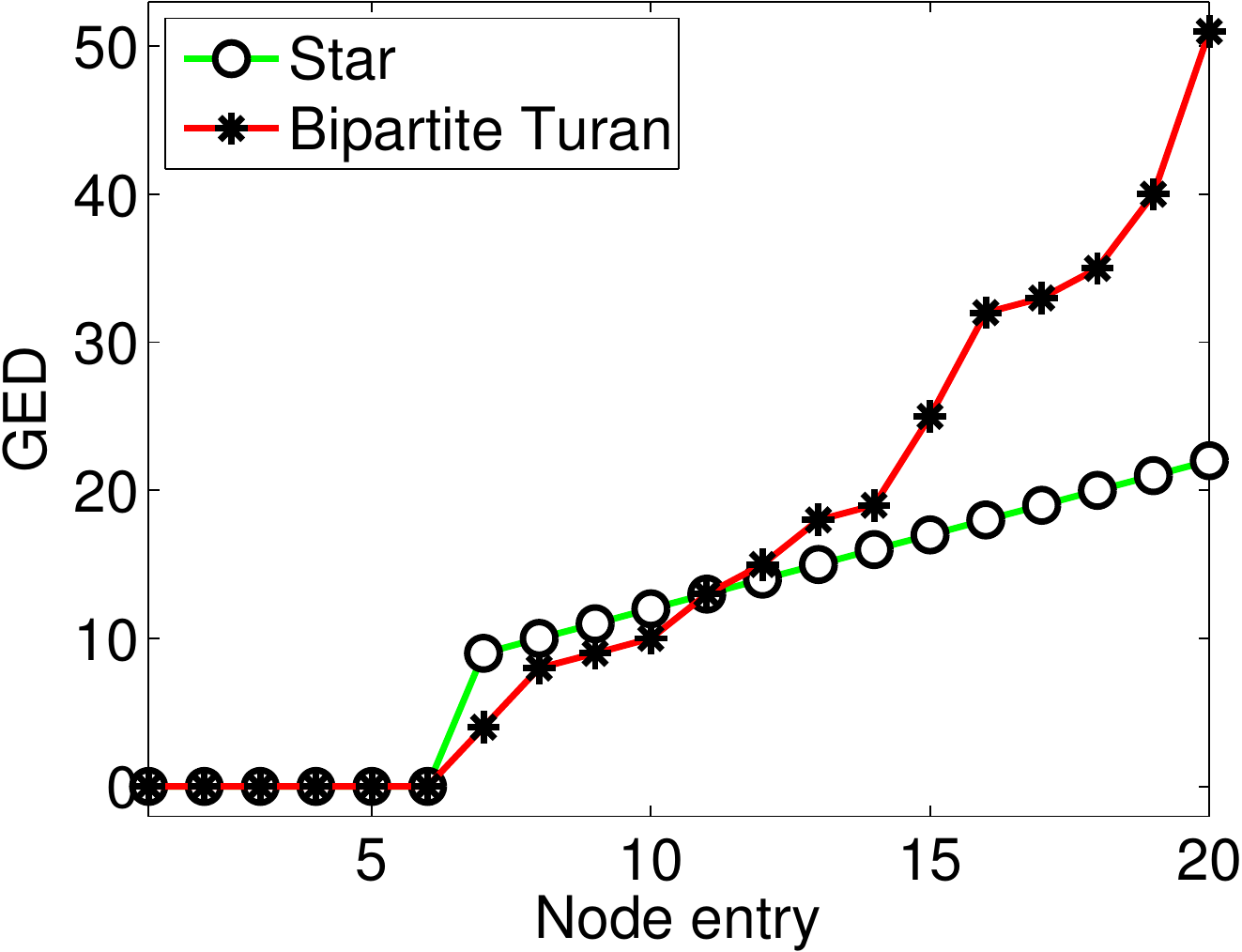}
}{
\includegraphics[scale=0.62]{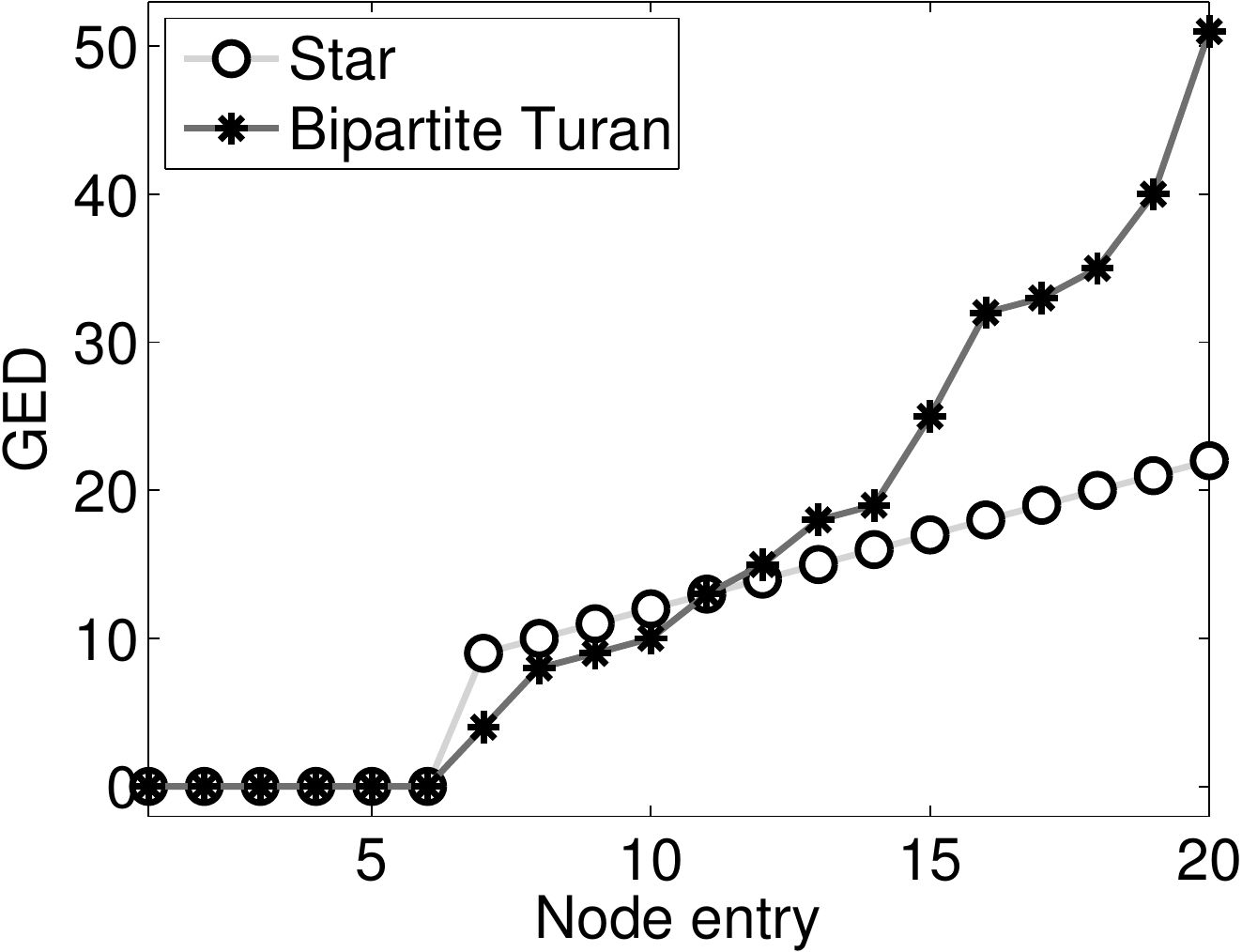}
}
\end{minipage}
\\
\begin{minipage}{.5\textwidth}
\centering
\small{(a)}
\end{minipage}
\begin{minipage}{.5\textwidth}
\centering
\small{(b)}
\end{minipage}
\end{tabular}
\caption{(a) Results of positive deviation of $c$ from the sufficient conditions for complete network and
(b) Comparison between result (D) for star network and bipartite Tur\'an network for deviation node 7
(Y-axis gives the deviation when the network consists of number of nodes given on X-axis)
}
\label{fig:devc}
\end{figure*}

\begin{table}[t]
\centering
  \begin{tabular}{ p{2.2cm}  p{.25cm}  p{.25cm}  p{.75cm}  p{.25cm}  p{.25cm}  p{.75cm} p{.25cm}  p{.25cm} p{.64cm} }
  \hline  \hline
 \T \B 
 & \multicolumn{3}{c}{\hspace{-.3cm}$\gamma=L$}  & \multicolumn{3}{c}{\hspace{-.3cm}$\gamma=M$}  & \multicolumn{3}{c}{\hspace{-.3cm}$\gamma=H$} \\ \hline
  \backslashbox{$c_0$}{$c$} & $L$ & $M$ & $H$  & $L$ & $M$ & $H$  & $L$ & $M$ & $H$ 
  \\ \hline 
\T \B $L/M/H$ &
D & C & B &
D & B & B &
D & C & C
\\ \hline  \hline
  \end{tabular}
  \caption{Results of negative deviation of $c$ for star network
  }
  \label{tab:devstarcostneg}
\end{table}

\subsubsection*{{Positive deviation of $\boldsymbol{c}$ from sufficient conditions for star network:}}
No node enters the network at deviation node 2, while for all other deviation nodes, the network does not deviate at all and so result (A) is obtained. The same is clear from the derivation of sufficient conditions for star network, that entry of node 2 is the bottleneck on the upper bound for $c$ ($c<b_1$). So node 2 stays out of the network until the sufficient conditions are restored so that they are favorable for it to enter the network, and hence the network builds up as desired. 
These results are desirable if the network owner is not too concerned about the delay of node 2's entry into the network. 

\subsubsection*{{Positive deviation of $\boldsymbol{c}$ from sufficient conditions for complete network:}}
No deviation in network was observed for deviation nodes 2 and 3. 
For other deviation nodes, deviations in network were observed only during the entry of the deviation node until the stabilization of the network henceforth (Figure~\ref{fig:devc}(a)). Following this, the sufficient conditions were restored and the network regained the desired topology, after the entry of the node following the deviation node and the stabilization henceforth (result (B)), since the condition $c<b_1-b_2$ ensures that the network so formed has diameter at most 1 (Proposition~\ref{thm:smallworld}), and this is irrespective of the preceding network states. 

\subsubsection*{{Negative deviation of $\boldsymbol{c}$ from sufficient conditions for bipartite Tur\'an network:}}
The desired network was obtained for all deviation nodes except 4, as clear from the derivation of  sufficient conditions (the 4-node network is the bottleneck for the lower bound on $c$).
For deviation node 4, GED between the resulting network of 4 nodes and the corresponding bipartite Tur\'an network was 3. The topology was restored from the entry of the following node onwards in most instances, while it took up to 9 node entries for some.

\subsubsection*{{Positive deviation of $\boldsymbol{c}$ from sufficient conditions for bipartite Tur\'an network:}}
No deviation in network was observed 
for deviation nodes 2 to 5. However, deviation node 6 onwards, result (D) was observed regularly for all combinations of values \{$L,M,H$\} assigned to $\gamma$, $c_0$ and $c$, apart from when nodes take decisions in a particular order (in which case, no deviation was observed).
For each deviation node 6 onwards, the average GED when the network reached the size of 20 nodes was around 50 and was increasing rapidly as shown in Figure~\ref{fig:devc}(b).
This GED is expected to be more than that in the case of star network, owing to its relatively high edge density.
Such deviations from the desired network were observed even for extremely minor deviations of $c$ from the derived sufficient conditions.
So restoring the sufficient conditions is not a viable solution for this case.
The network owner should ensure that the values of $c$ are on the lower side so as to stay away from the upper bound. 

\subsubsection*{{Negative deviation of $\boldsymbol{c}$ from sufficient conditions for $k$-star network:}}
GED for all deviation nodes were strictly positive and monotonically increasing, qualitatively looking like result (D) in Figure~\ref{fig:resulttypes}(a). 

\subsubsection*{{Positive deviation of $\boldsymbol{c}$ from sufficient conditions for $k$-star network:}}
Result (A) was observed for all deviation nodes except $2k$ through $3k-1$. 
The reason for the deviation in network for these deviation nodes is that, in the $k$-star network consisting of number of nodes between $2k$ and $3k-1$, both inclusive, there exists at least one center with only one leaf node linked to it.
When there is a positive deviation of $c$ from the sufficient conditions for $k$-star network, it is beneficial for any other center to delete link with a center that is linked to only one leaf node, and this link deletion leads to other link alterations among other nodes, thus deviating the network from the desired topology.
For deviation nodes $2k$ through $3k-1$, result (D) was observed consistently, which qualitatively looked like the one in Figure~\ref{fig:resulttypes}(a).

%

\subsection[Results for Deviation with Respect to ${c_0}$]{Results for Deviation with Respect to $\boldsymbol{c_0}$}
\label{sec:devc0}

\subsubsection*{{Positive deviation of $\boldsymbol{c_0}$ from sufficient conditions for star network:}}
These results are shown qualitatively in Table~\ref{tab:devstarc0pos} and quantitatively in Figure~\ref{fig:star_c0pos}(a). 
The graph in Figure~\ref{fig:star_c0pos}(a) plots the deviation from network as observed when the network reached the size of 20 nodes, if the conditions were deviated at a given deviation node.
For deviation nodes 2 and 3, no deviation in network was observed. 
For other deviation nodes, Table~\ref{tab:devstarc0pos} shows the type of result obtained owing to deviation from sufficient conditions on $c_0$ at a deviation node, following which, the values of $\gamma$, $c_0$ and $c$ are restored to one of \{$L,M,H$\}. 
%
%
%
When the sufficient conditions are restored to low values of $c$ after deviating from the sufficient conditions, the resulting network is a $(2,n-2)$-complete bipartite network (result (D)) similar to that in Figure~\ref{fig:star_c0pos}(b), where node $Y$ was the original center and the conditions were deviated during entry of node $X$. 

\begin{figure} [t!]
\begin{tabular}{c}
\hspace{-.7cm}
\begin{minipage}{.5\textwidth}
\centering
\iftoggle{clr}{
\includegraphics[scale=0.62]{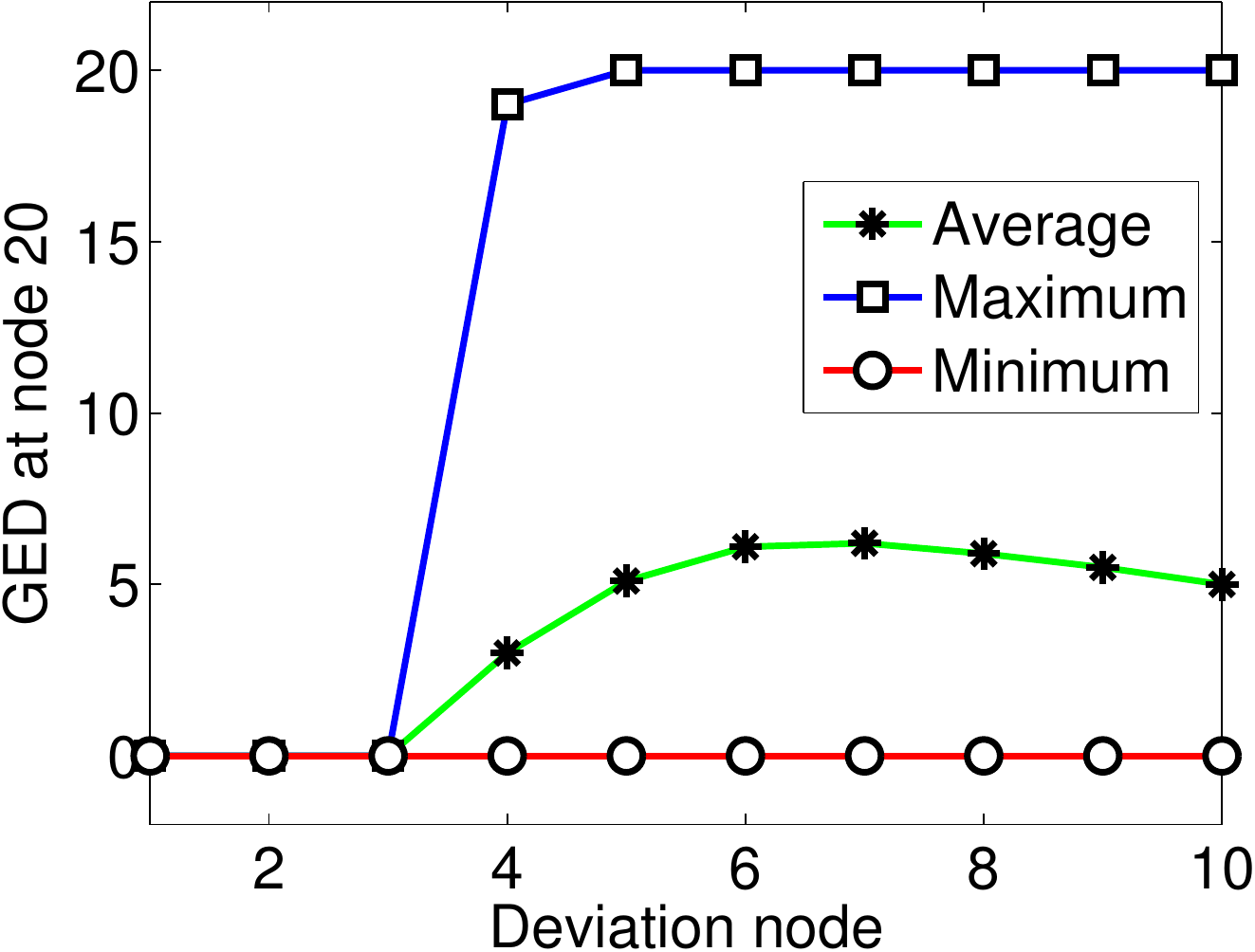}
}{
\includegraphics[scale=0.62]{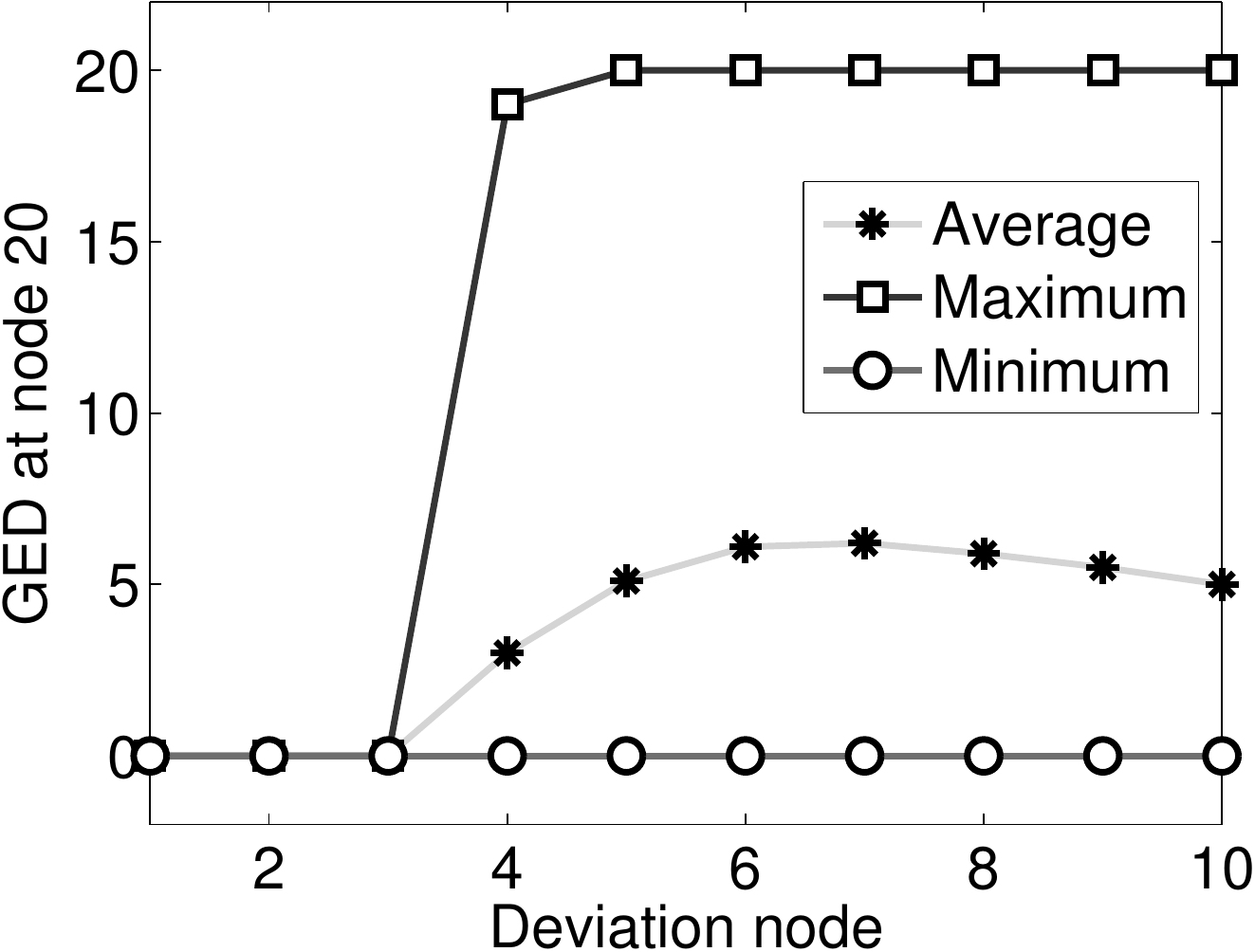}
}
\end{minipage}
\begin{minipage}{.5\textwidth}
\centering
\iftoggle{clr}{
    \includegraphics[scale=0.85]{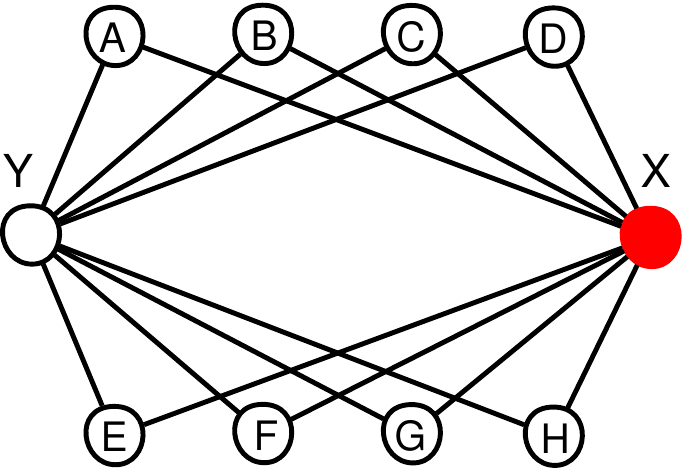}
}{
    \includegraphics[scale=0.85]{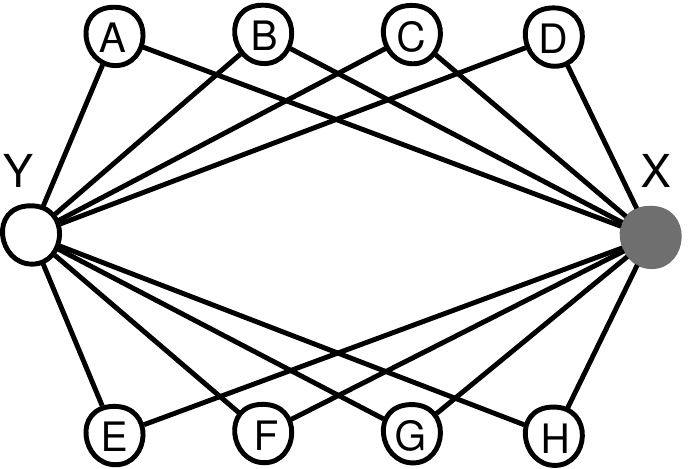}
}
\end{minipage}
\\
\begin{minipage}{.5\textwidth}
\centering
\small{(a)}
\end{minipage}
\begin{minipage}{.5\textwidth}
\centering
\small{(b)}
\end{minipage}
\end{tabular}
\caption{
(a) Results of positive deviation of $c_0$ from the sufficient conditions for star network when the network consists of 20 nodes
and 
(b) A $(2,8)$-complete bipartite network
}
\label{fig:star_c0pos}
\end{figure}

\begin{table}[t]
\centering
  \begin{tabular}{ p{2.2cm}  p{.25cm}  p{.25cm}  p{.75cm}  p{.25cm}  p{.25cm}  p{.75cm} p{.25cm}  p{.25cm} p{.64cm} }
    \hline  \hline
   \T \B 
   & \multicolumn{3}{c}{\hspace{-.3cm}$\gamma=L$}  & \multicolumn{3}{c}{\hspace{-.3cm}$\gamma=M$}  & \multicolumn{3}{c}{\hspace{-.3cm}$\gamma=H$} \\ \hline
  \backslashbox{$c_0$}{$c$} & $L$ & $M$ & $H$  & $L$ & $M$ & $H$  & $L$ & $M$ & $H$ 
  \\ \hline 
\T \B $L/M/H$ &
D & B & B &
D & B & B &
D & C & C
\\ \hline \hline
  \end{tabular}
  \caption{Results of positive deviation of $c_0$ for star network
  }
  \label{tab:devstarc0pos}
\end{table}

\subsubsection*{{Positive deviation of $\boldsymbol{c_0}$ from sufficient conditions for complete and bipartite Tur\'an networks:}}
No deviation was observed for early deviation nodes, that is, if the conditions were deviated when the network consisted of less number of nodes.
Let $d_T$ be the degree of the node to which a new node desires to connect in order to enter the network. For both complete and bipartite Tur\'an networks, beyond a certain limit on the number of nodes, the minimum value of $d_T$ is very high.
So during positive deviation of $c_0$, the term $d_T((1-\gamma)b_2 - c_0)$ becomes extremely negative, overpowering other benefits, thus making it undesirable for a new node to enter the network. A new node enters once the sufficient conditions are restored.
These results are desirable if the network owner is not concerned about the delay of node entry.

\begin{figure} [t]
\begin{tabular}{c}
\begin{minipage}{.4\textwidth}
\centering
\iftoggle{clr}{
\includegraphics[scale=0.87]{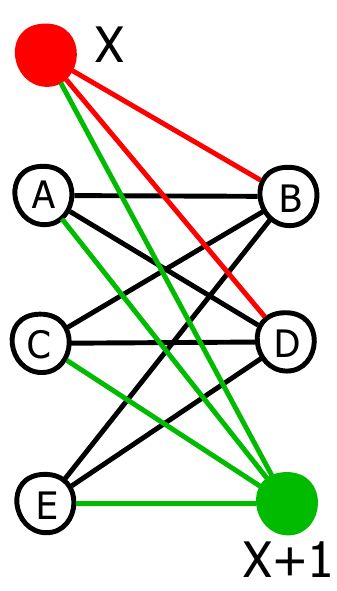}
}{
\includegraphics[scale=0.87]{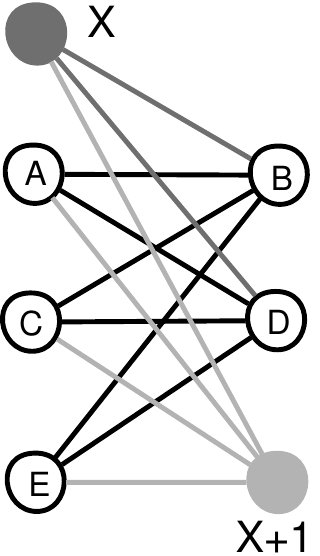}
}
\end{minipage}
\begin{minipage}{.6\textwidth}
\centering
\iftoggle{clr}{
\includegraphics[scale=0.82]{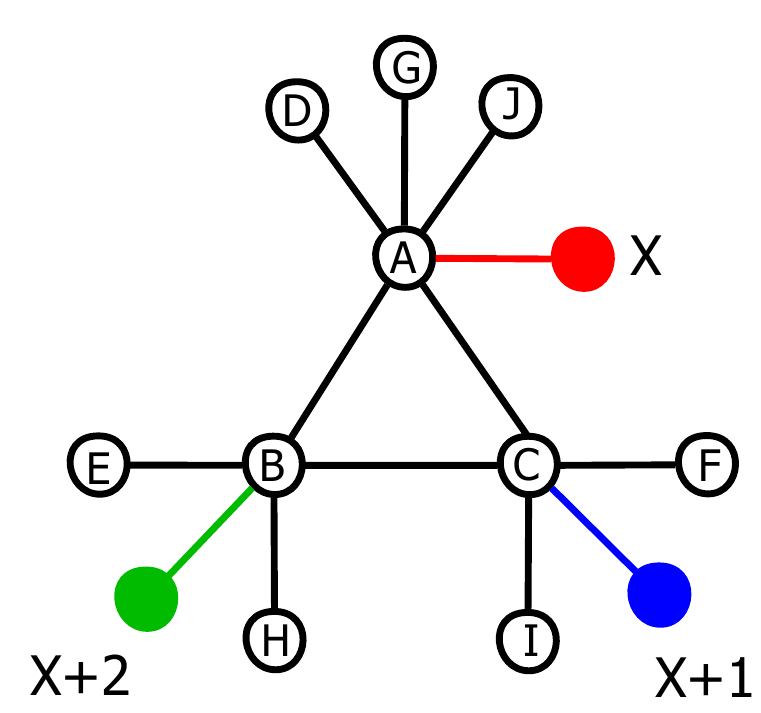}
}{
\includegraphics[scale=0.82]{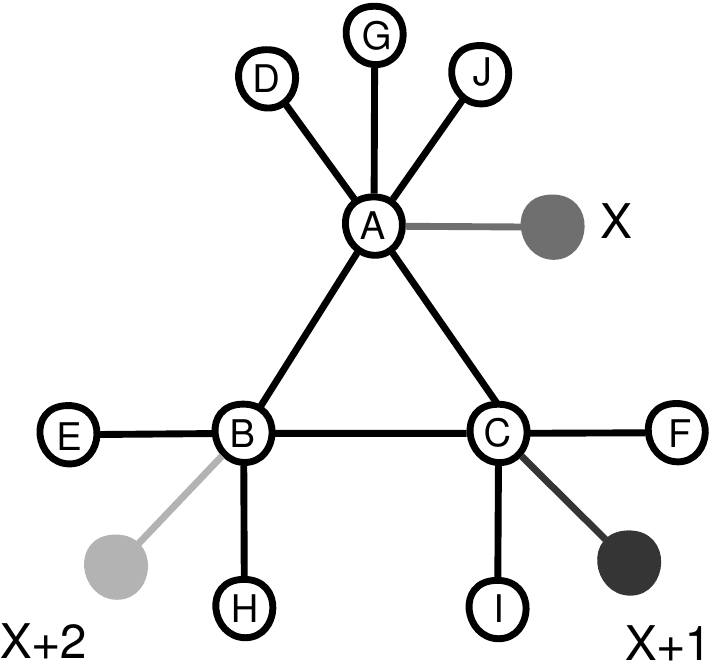}
}
\end{minipage}
\\
\begin{minipage}{.4\textwidth}
\centering
(a)
\end{minipage}
\begin{minipage}{.6\textwidth}
\centering
(b)
\end{minipage}
\end{tabular}
\caption{Restorations of (a) bipartite Tur\'an and (b) 3-star network topologies
}
\label{fig:devc0}
\end{figure}

\subsubsection*{{Negative deviation of $\boldsymbol{c_0}$ from sufficient conditions for bipartite Tur\'an network:}}
The desired network was obtained for all odd numbered deviation nodes and deviation node 2. For deviation node 4, GED between the resulting network of 4 nodes and the corresponding bipartite Tur\'an network was 3. For most instances, the topology was restored from the entry of the following node onwards; but some instances took up to 9 node entries to settle back to a bipartite Tur\'an network (very similar to the case of negative deviation of $c$). 
For every even-numbered deviation node $n \geq 4$, deviations in network were observed only during the entry of the deviation node until the stabilization of the network henceforth, with GED $=n-1$. Following this, the sufficient conditions were restored and the network regained the desired topology, after the entry of the node following the deviation node and the stabilization henceforth. 
Figure~\ref{fig:devc0}(a) shows the result when node $X$ tries to enter the bipartite Tur\'an network consisting of nodes $A,B,C,D,E$, as the $6^{th}$ node, during negative deviation of $c_0$. It creates links with nodes $B,D$ instead of $A,C,E$, thus giving graph edit distance of 5. Following this, the sufficient conditions are restored and the following node $X+1$ 
forms links with low degree nodes, forming a bipartite Tur\'an network of 7 nodes, thus restoring the topology. 

\subsubsection*{{Negative deviation of $\boldsymbol{c_0}$ from sufficient conditions for $k$-star network:}}
For deviation node $n$ such that $(n \mod k)=1$, the network did not deviate and so result (A) was observed. 
For all other deviation nodes, result (B) was observed.
In general, for deviation node $n$, GED was observed to be 2, and it took $\left[ (k+1-z)\mod k \right]$ node entries for the topology to be restored once the sufficient conditions were restored, where $z = (n \mod k)$.
Figure~\ref{fig:devc0}(b) shows the result when node $X$ tries to enter the 3-star network consisting of nodes $A$ through $J$, as the $11^{th}$ node, during negative deviation of $c_0$.
It creates a link with node $A$ instead of either $B$ or $C$, thus giving GED of 2. Following this, the sufficient conditions are restored and so the following node $X+1$ 
forms links with a lowest degree center, say $C$; but GED remains 2. 
Then the next node $X+2$ tries to enter, which forms a link with the only lowest degree center $B$, forming a 3-star network of 13 nodes, thus restoring the topology.
In this example, $k=3$ and $n=11$ and so it takes 2 node entries for the topology to be restored. 

\subsubsection*{{Positive deviation of $\boldsymbol{c_0}$ from sufficient conditions for $k$-star network:}}
Let $C$ be a center with the lowest degree and $m_j$ be the number of leaf nodes already connected to center $j$.
It can be shown that result (A) will be obtained if the positive deviation of $c_0$ is less that the threshold:
\begin{equation}
\nonumber
(b_3-b_4)\left( \frac{\sum_{j \neq C}m_j + \textbf{I}_{n \neq pk+1}}{k+m_C-2}-1 \right)
\end{equation}
where $n$ is the deviation node, and $\textbf{I}_{n \neq pk+1}$ is 1 if $n \neq pk+1$ for any integer $p$, else it is 0.
If the deviation crossed this threshold in simulations, result (D) was observed consistently, which qualitatively looked like the one in Figure~\ref{fig:resulttypes}(a).
The result is owing to the fact that a high value of $c_0$ would force a new node to prefer connecting to a leaf node which is linked to a center with the highest degree, rather than any center directly; this leads to other link alterations among other nodes, thus deviating the network from the desired topology.

\section{Conclusion}
\label{sec:conclusion_nfsc}

We proposed a model of recursive network formation where nodes enter a network sequentially, thus triggering evolution 
each time a new node enters.
We considered a sequential move game model with myopic nodes under a very general utility model, and pairwise stability as the equilibrium notion; however the proposed model (Figure~\ref{fig:model}) is independent of the network evolution model, the equilibrium notion, as well as the utility model.
The recursive nature of our model enabled us to analyze the network formation process using an elegant induction-based technique.
For each of the relevant topologies, by directing network evolution as desired, we derived sufficient conditions 
under which that topology uniquely emerges. 
  The derived conditions suggest that conditions on network entry impact degree distribution, while
   conditions on link costs impact density;
   also there arise constraints on intermediary rents owing to contrasting densities of connections in the desired topology.
We then analyzed the social welfare properties of the considered topologies,
and studied the effects of deviating from the derived conditions.

\vspace{-.2cm}


%

\vspace{10mm}
This chapter dealt with the problem of forming a social network. Now assuming that a social network is already formed and available to us, and also there is some level of trust developed among connected individuals, the next chapter will focus on the problem of maximizing the diffusion of an information using the social network. In particular, we study the multi-phase version of the well-studied problem of influence maximization in social networks.

\begin{subappendices}


\chapter*{Appendix for Chapter~\ref{chap:nfsc}}
\addcontentsline{toc}{chapter}{Appendix for Chapter~\ref{chap:nfsc}}
\label{chap:appendix_nfsc}

\section{Proof of Proposition~\ref{thm:smallworld}
}
\label{app:smallworld}
\begin{customprop}{\ref{thm:smallworld}}
For a network, if $c<b_1-b_{d+1}$ ($d \geq 1$) and $c_0\leq(1-\gamma)b_2$, the resulting diameter is at most $d$.
\end{customprop}
\begin{proof}
The conditions $c<b_1$ and $c_0\leq(1-\gamma)b_2$ ensure that any new node successfully enters the network, that is, it gets a positive utility by doing so, and the node to which it connects to in order to enter the network, also gets a higher utility.

Now consider a network where $c<b_1-b_{d+1}$ and there exist two nodes, say $A$ and $B$, which are at a distance $x>d$ from each other. The indirect benefit they get from each other is  $b_x \leq b_{d+1}$. In the case where there exist essential nodes connecting these nodes, each has to pay an additional rent of $\gamma b_{x}$. By establishing a connection between them, each node gets an additional direct benefit of $b_1$ and incurs an additional cost $c$. Also this connection may decrease the distances between either of these nodes and other nodes, for instance, direct neighbors of node $B$ which were at distance $b_{x-1}$, $b_x$ or $b_{x+1}$ from node $A$, are now at distance  $\min\{b_2,b_{x-1}\}$, resulting in increase in indirect benefits for node $A$.

It can be easily seen that if either (or both) of these nodes acted as an essential node for some pair of nodes, it remains to do so even after the connection is established. Furthermore, it is possible that the established connection shortens the path between this pair, resulting in higher bridging benefits for the node under consideration.

Summing up, by establishing a mutual connection between nodes which are at distance $x>d$ from each other, the overall increase in utility for either node is at least $b_1-c$ and the overall decrease is at most $b_{d+1}$. So the condition sufficient for link creation is $b_1-c>b_{d+1}$. As this is true for any such pair, without loss of generality, the network will evolve until distance between any pair is at most $d$.
\end{proof}

%
%

\section{Proof of Proposition~\ref{thm:bipartite}
}
\label{app:bipartite}
\begin{customprop}{\ref{thm:bipartite}}
For a network with $\gamma <   \frac{b_2 - b_3}{3b_2 - b_3} $, if $b_1-b_2+ \gamma \left( 3b_2 - b_3 \right) <  c < b_1 - b_3$ 
and $\left( 1-\gamma \right) \left( b_2-b_3 \right) < c_0 \leq \left( 1-\gamma \right) b_2$, the unique resulting topology is a 
bipartite Tur\'an graph.
\end{customprop}
\begin{proof}
We first derive conditions for ensuring pairwise stability of a bipartite Tur\'an network, that is, assuming that such a network is formed, what conditions are required so that there are no incentives for any two unconnected nodes to create a link between them and for any node to delete any of its links. Note that these conditions can be integrated in the later part of the proof within different scenarios that we consider.\\
In what follows, $p_1$ is the size of the partition constituting the node taking its decision, $p_2$ is the size of the other partition and $n=p_1+p_2$ is the number of nodes in the network.
We need to consider cases for some discretely small number of nodes owing to the nature of essential nodes, after which, the analysis holds for arbitrarily large number of nodes. For brevity, we present the analysis for the base case and a generic case in each scenario, omitting presentation of discrete cases. \\

\noindent
\textbf{No two nodes belonging to the same partition should create a link between them:} Their utility should not increase by doing so. This is not applicable for $n=2$. \\
For $n=3$, 
\begin{equation}
\nonumber
2(b_1-c) \leq b_1-c+(1-\gamma)b_2
\vspace{-5mm}
\end{equation}
\begin{equation}
\label{E16}
\iff c \geq b_1-b_2+\gamma b_2
\end{equation}
For $n\geq 4$,
\begin{equation}
\nonumber
(p_2+1)(b_1-c)+(p_1-2)b_2 \leq p_2(b_1-c)+(p_1-1)b_2
\vspace{-5mm}
\end{equation}
\begin{equation}
\nonumber
\iff c \geq b_1-b_2
\end{equation}
which is a weaker condition that Inequality~(\ref{E16}).\\

\noindent
\textbf{No node should delete its link with any node belonging to the other partition:} That is, their utility should not increase by doing so. \\
For $n=2$,
\begin{equation}
\nonumber
0\leq b_1-c
\vspace{-5mm}
\end{equation}
\begin{equation}
\label{E17for2}
\iff c \leq b_1
\end{equation}
For $n\geq 6$,
\begin{equation}
\nonumber
(p_2-1)(b_1-c)+(p_1-1)b_2+b_3 \leq p_2(b_1-c)+(p_1-1)b_2
\vspace{-5mm}
\end{equation}
\begin{equation}
\label{E17}
\iff c \leq b_1-b_3
\end{equation}
It can be shown that conditions for the discrete cases $n=3,4,5$  are satisfied by Inequality~(\ref{E17}).\\  \\
In the process of formation of a bipartite Tur\'an network, at most four different types of nodes exist at any point in time.
  \begin{center}
\begin{tabular}{l l }
    \hline \hline
\T \B  
I & newly entered node \\ \hline
\T \B  
II & nodes connected to the newly entered node\\ \hline
\T \B  
III & nodes in the same partition as Type II nodes, but not connected to newly entered node\\ \hline
\T \B  
IV & rest of the nodes\\ \hline \hline
  \end{tabular}
  \end{center}
The notation we use while deriving the sufficient conditions are as follows:
\begin{center}
\begin{tabular}{l l}
    \hline  \hline
\T \B 
$k$ & number of nodes of Type II \\ \hline
\T \B 
$n$ & number of nodes in network, including new node\\ \hline
\T \B 
$m_1$ & number of nodes of Types II and III put together\\ \hline
\T \B 
$m_2$ & number of nodes of Type IV\\ \hline \hline
  \end{tabular}
\end{center}

\noindent
\textbf{For the newly entering node to enter the network:} Its utility should be positive after doing so. Also, in case of even $n$, for the new node to be a part of the smaller partition, its first connection should be a node belonging to the larger partition. So for $k=0$, we have \\
For $n\geq 2$,
\begin{equation}
\nonumber
b_1-c+ \lceil \frac{n}{2}-1 \rceil \left( (1-\gamma)b_2 - c_0 \right) + \lfloor \frac{n}{2}-1 \rfloor (1-\gamma)b_3 >0
\end{equation}
It can be seen that the condition is the strongest when $n=2$ whenever
\begin{equation}
\label{E1b}
c_0 \leq (1-\gamma)b_2
\end{equation}
The condition thus becomes
\begin{equation}
\nonumber
c< b_1
\end{equation}
which is satisfied by Inequality~(\ref{E17}).\\

\noindent 
\textbf{The utility of a node in the larger partition, whenever applicable, should not decrease after accepting link from the new node:}\\
For $n=2$,
\begin{equation}
\nonumber
b_1-c \geq 0
\vspace{-5mm}
\end{equation}
\begin{equation}
\nonumber
\iff c \leq b_1
\end{equation}
For $n \geq 5$,
\begin{align}
\nonumber
\begin{split}
&
\lceil \frac{n}{2} \rceil (b_1 -c) + \lfloor \frac{n}{2}-1 \rfloor b_2 + \gamma \lceil \frac{n}{2}-1 \rceil 2 b_2 + \gamma \lfloor \frac{n}{2}-1 \rfloor 2 b_3 \\
&\geq \lceil \frac{n}{2}-1 \rceil (b_1 -c) + \lfloor \frac{n}{2}-1 \rfloor b_2
\end{split}
\end{align}
\begin{equation}
\nonumber
\iff c \leq b_1+ 2\gamma \lceil \frac{n}{2}-1 \rceil b_2 + \lfloor \frac{n}{2}-1 \rfloor b_3
\end{equation}
The conditions for these as well as the discrete cases $n=3,4$ are satisfied by Inequality~(\ref{E17}).\\

\noindent
\textbf{The new node should connect to a node in the larger partition, whenever applicable:} One way to see this is by ensuring that this strategy strictly dominates connecting to a node in the smaller partition.
This scenario arises for even values of $n\geq 4$.
\begin{align}
\begin{split}
\nonumber
&
b_1-c+ \left( \frac{n}{2}-1 \right) \left( (1-\gamma)b_2 - c_0 \right) + \left( \frac{n}{2}-1 \right) (1-\gamma)b_3 \\
&> b_1-c+ \left( \frac{n}{2} \right) \left( (1-\gamma)b_2 - c_0 \right) + \left( \frac{n}{2}-2 \right) (1-\gamma)b_3 
\end{split}
\end{align}
\begin{equation}
\label{E2a}
\iff c_0 > (1-\gamma)(b_2-b_3)
\end{equation}
An alternative condition would be such that the utility of a node in the smaller partition decreases if it accepts the link from the new node, thus forcing the latter to connect to a node in the other partition. But it can be seen that this condition is inconsistent with Inequality~(\ref{E17}) and so we use Inequality~(\ref{E2a}) to meet our purpose.\\

\noindent
\textbf{Type I node should prefer connecting to a Type III node, if any, than remaining in its current state:}
For $k\geq 2$, this scenario does not arise for $n<6$. 
For $n\geq 6$,
\begin{equation}
\nonumber
\begin{split}
(k+1)(b_1-c)+m_2b_2+(m_1-k-1)b_3 > k(b_1-c)+m_2b_2
+(m_1-k)b_3
\end{split}
\end{equation}
\begin{equation}
\label{E6}
\iff c<b_1-b_3
\end{equation}
Now for $k=1$, this scenario does not arise for $n=2,3$.\\
For $n \geq 4$,
\begin{align}
\nonumber
\begin{split}
2(b_1-c)+m_2b_2+(m_1-2)b_3 
> b_1-c+(1-\gamma)m_2b_2
+(1-\gamma)(m_1-1)b_3
\end{split}
\end{align}
\vspace{-5mm}
\begin{equation}
\nonumber
\iff c<b_1-b_3+\gamma(m_2b_2+(m_1-1)b_3)
\end{equation}
Note that as $n \geq 4$, we have $m_1 \geq 2$ and $m_2 \geq 1$ and so the above condition is weaker that Inequality~(\ref{E6}).\\
It is also necessary that utility of Type III node does not decrease on accepting link from Type I node. In fact, when the former gets a chance to move, we derive conditions so that it also volunteers to create a link with the later.\\

\noindent
\textbf{The utility of Type III node should increase if it successfully creates a link with Type I node:}
When $k=1$, the case does not arise for $n=2,3$.\\
For $n \geq 6$,
\begin{equation}
\nonumber
(m_2+1)(b_1-c)+(m_1-1)b_2>m_2(b_1-c)+(m_1-1)b_2+(1-\gamma)b_3
\vspace{-5mm}
\end{equation}
\begin{equation}
\nonumber
\iff c<b_1-b_3+\gamma b_3
\end{equation}
The conditions obtained from discrete cases $n=4,5$ are weaker than this one.\\
For $k\geq 2$, this case does not arise for $n <6$. \\
For $n\geq 6$,
\begin{equation}
\nonumber
(m_2+1)(b_1-c)+(m_1-1)b_2>m_2(b_1-c)+(m_1-1)b_2+b_3
\vspace{-5mm}
\end{equation}
\begin{equation}
\nonumber
\iff c<b_1-b_3
\end{equation}
The conditions for all cases are satisfied by Inequality~(\ref{E6}).\\

\noindent
\textbf{Type III node should not delete its link with Type IV node:} This can be assured if this strategy is dominated by its strategy of forming a link with Type I node.
This scenario does not arise for $n=2,3$. 
The conditions for the discrete cases $n=4,5,6$ are weaker than that for $n\geq 7$.\\
For $n\geq 7$,
\begin{equation}
\nonumber
\begin{split}
(m_2+1)(b_1-c)+(m_1-1)b_2>(m_2-1)(b_1-c)+b_3
+(m_1-1)b_2+(1-\gamma)b_3
\end{split}
\end{equation}
\vspace{-5mm}
\begin{equation}
\nonumber
\iff c<b_1-b_3+\frac{\gamma}{2}b_3
\end{equation}
For $k\geq 2$, the cases applicable are $n\geq 6$.
The condition for discrete case $n=6$ is weaker than the following condition.\\
For $n\geq 7$,
\begin{equation}
\nonumber
\begin{split}
(m_2+1)(b_1-c)+(m_1-1)b_2>(m_2-1)(b_1-c)+b_3+b_3+(m_1-1)b_2
\end{split}
\end{equation}
\vspace{-5mm}
\begin{equation}
\nonumber
\iff c<b_1-b_3
\end{equation}
Hence, all conditions for this scenario are satisfied by Inequality~(\ref{E6}).\\

\noindent
\textbf{Type III node should prefer connecting to Type I node than to another Type III node:} This does not arise for $n<6$.
When $k=1$, \\
For $n\geq 6$,
\begin{align}
\nonumber
\begin{split}
(m_2+1)(b_1-c)+(m_1-1)b_2
>(m_2+1)(b_1-c)
+(m_1-2)b_2+(1-\gamma)b_3
\end{split}
\end{align}
\vspace{-5mm}
\begin{equation}
\nonumber
\iff b_2>(1-\gamma)b_3
\end{equation} 
which is always true. For $k\geq 2$,\\
For $n \geq 6$,
\begin{equation}
\nonumber
(m_2+1)(b_1-c)+(m_1-1)b_2>(m_2+1)(b_1-c)+(m_1-2)b_2+b_3
\end{equation}
\vspace{-5mm}
\begin{equation}
\nonumber
\iff b_2>b_3
\end{equation} 
which is always true.\\

\noindent
\textbf{Type IV node should not delete its link with Type III node:} That is, its utility should not increase by doing so.
This does not arise for $n<4$.\\
For $n\geq 7$,
\begin{equation}
\nonumber
\begin{split}
(m_1-1)(b_1-c)+(m_2-1)b_2+(1-\gamma)b_2+b_3 \\ \leq m_1(b_1-c)+(m_2-1)b_2+(1-\gamma)b_2
\end{split}
\end{equation}
\begin{equation}
\nonumber
\iff  c\leq b_1-b_3
\end{equation} 
The conditions for discrete cases $n=4,5,6$ are weaker than the above condition.
For $k\geq 2$, the new cases are $n\geq 6$, where the discrete case $n=6$ result in conditions weaker than the following one.\\
For $n \geq 7$,
\begin{equation}
\nonumber
(m_1-1)(b_1-c)+(m_2-1)b_2+b_2+b_3 \leq m_1(b_1-c)+(m_2-1)b_2+b_2
\end{equation}
\vspace{-5mm}
\begin{equation}
\nonumber
\iff c\leq b_1-b_3
\end{equation} 
It can be seen that all conditions of this scenario are satisfied by Inequality~(\ref{E6}).\\

\noindent
\textbf{Type IV node should also not break its link with Type II node:} That is, its utility should not increase by doing so.
For $k=1$,\\
For $n\geq 6$,
\begin{align}
\nonumber
\begin{split}
&
(m_1-1)(b_1-c)+(m_2-1)b_2+(1-\gamma)b_4+(1-\gamma)b_3 \\ 
&\leq m_1(b_1-c)+(m_2-1)b_2+(1-\gamma)b_2 
\end{split}
\end{align}
\begin{equation}
\nonumber
\iff c\leq b_1-b_3+(1-\gamma)(b_2-b_4)+\gamma b_3
\end{equation} 
The discrete cases $n=3,4,5$ result in weaker conditions than this.
For $k\geq 2$,\\
For $n\geq 6$,
\begin{equation}
\nonumber
(m_1-1)(b_1-c)+(m_2-1)b_2+b_2+b_3  \leq m_1(b_1-c)+(m_2-1)b_2+b_2 
\end{equation}
\vspace{-5mm}
\begin{equation}
\nonumber
\iff c\leq b_1-b_3
\end{equation} 
The conditions are satisfied by Inequality~(\ref{E6}).\\

\noindent
\textbf{Type I node should not propose a link to a Type IV node:} One way is to ensure that this strategy of Type I node is dominated by its strategy to propose a link to a Type III node.
It can be seen that for $k \geq 2$ and $n \geq 6$, this translates to
\begin{align}
\nonumber
\begin{split}
&
(k+1)(b_1-c)+m_2b_2+(m_1-k-1)b_3\\ 
&> (k+1)(b_1-c)+(m_2-1)b_2+(m_1-k)b_2
\end{split}
\end{align}
\begin{equation}
\nonumber
\iff b_2-b_3>(m_1-k)(b_2-b_3)
\end{equation}
which is not true for $m_1>k$.\\
So we look at the alternative condition that the utility of Type IV node decreases if it accepts the link from Type I node, and as Type I node computes this decrease in utility, it will not propose a link to Type IV node.
First, we consider $k=1$. The discrete case $n=4$ gives the following condition.\\
\begin{equation}
\nonumber
3(b_1-c)+2\gamma b_2+2\gamma b_2 < 2(b_1-c)+(1-\gamma)b_2+2\gamma b_2+\gamma b_3 
\vspace{-5mm}
\end{equation}
\begin{equation}
\label{E7b}
\iff c>b_1-b_2+\gamma(3b_2-b_3)
\end{equation}
The other discrete cases $n=3,5$ result in weaker conditions than the above.\\
For $n\geq 6$,
\begin{equation}
\nonumber
(m_1+1)(b_1-c)+(m_2-1)b_2<m_1(b_1-c)+(m_2-1)b_2+(1-\gamma)b_2
\vspace{-5mm}
\end{equation}
\begin{equation}
\nonumber
\iff c>b_1-b_2+\gamma b_2 
\end{equation}
which is a weaker condition than Inequality~(\ref{E7b}).
Now for $k \geq 2$, $n=4,5$ correspond to pairwise stability conditions and cases $n<4$ are not applicable.\\
For $n\geq 6$,
\begin{equation}
\nonumber
(m_1+1)(b_1-c)+(m_2-1)b_2 < m_1(b_1-c)+(m_2-1)b_2+b_2
\vspace{-5mm}
\end{equation}
\begin{equation}
\nonumber
\iff c>b_1-b_2
\end{equation}
which is satisfied by Inequality~(\ref{E7b}).\\

\noindent
\textbf{Type IV node should not propose a link to Type I node:} This scenario is essentially equivalent to the previous one scenario of utility of Type IV node decreasing due to link with Type I node, with the equalities permitted. So these result in weaker and hence no additional conditions.\\

\noindent
\textbf{Type III node should not propose a link to Type II node:} One way is to ensure that for Type III node, connecting to Type II node is strictly dominated by connecting to Type I node.
It can be seen that for $k\geq 2$ and $n\geq 6$, this translates to
\begin{equation}
\nonumber
(m_2+1)(b_1-c)+(m_1-2)b_2+b_2<(m_2+1)(b_1-c)+(m_1-1)b_2
\end{equation}
which gives $0>0$. 
So we need to use the alternative condition that the utility of Type II node decreases on accepting link from Type III node.
For $k=1$,\\
For $n=4$,
\begin{equation}
\nonumber
3(b_1-c)+4\gamma b_2 < 2(b_1-c)+(1-\gamma)b_2+2\gamma b_2+\gamma b_3
\vspace{-5mm}
\end{equation}
\begin{equation}
\nonumber
\iff c > b_1-b_2+\gamma(3b_2 - b_3)
\end{equation}
which is same as Inequality~(\ref{E7b}). \\
For $n\geq 5$,
\begin{align}
\nonumber
\begin{split}
&
(m_2+2)(b_1-c)+(m_1-2)b_2+2\gamma(m_2+1)b_2+2\gamma(m_1-2)b_3 \\
&< (m_2+1)(b_1-c)+(m_1-1)b_2+2\gamma m_2b_2+ 2\gamma(m_1-1)b_3
\end{split}
\end{align}
\begin{equation}
\nonumber
\iff c > b_1-b_2+2\gamma(b_2-b_3)
\end{equation}
which is a weaker condition than Inequality~(\ref{E7b}). 
Now for $k \geq 2$, the only new case is the following.\\
For $n\geq 6$,
\begin{equation}
\nonumber
(m_2+2)(b_1-c)+(m_1-2)b_2<(m_2+1)(b_1-c)+(m_1-1)b_2
\vspace{-5mm}
\end{equation}
\begin{equation}
\nonumber
\iff c>b_1-b_2
\end{equation}
which is satisfied by Inequality~(\ref{E7b}).\\

\noindent
\textbf{Type II node should not propose a link with Type III node:} This is essentially equivalent to the above scenario of utility of Type II node decreasing due to link with Type III node, with the equalities permitted. So these result in weaker and hence no additional conditions.\\

\noindent
\textbf{No Type II node should delete link with Type IV node:} First, we consider $k=1$.\\
For $n \geq 7$, 
\begin{align}
\nonumber
\begin{split}
&
m_2(b_1-c)+(b_1-c) + (m_1-1)b_2+2\gamma m_2 b_2 + 2\gamma(m_1-1)b_3 \\
&\geq (m_2-1)(b_1-c) + (b_1-c) + (m_1-1)b_2+b_3
+2\gamma (m_2-1)b_2+2\gamma (m_1-1)b_3 +2\gamma b_4
\end{split}
\end{align}
\begin{equation}
\nonumber
\iff c \leq b_1 - b_3 + 2 \gamma (b_2 - b_4)
\end{equation}
This as well as all discrete cases $n<7$ are satisfied by Inequality~(\ref{E6}).\\
For $k \geq 2$,  the cases of $n=4,5$ correspond to pairwise stability condition that we have already considered, while cases $n<4$ are not applicable.\\
For $n \geq 6$,
\begin{equation}
\nonumber
m_2(b_1-c)+(m_1-1)b_2+b_3 \leq (m_2+1)(b_1-c)+(m_1-1)b_2
\vspace{-5mm}
\end{equation}
\begin{equation}
\nonumber
\iff c \leq b_1-b_3
\end{equation}
which is satisfied by Inequality~(\ref{E6}).\\

\noindent
\textbf{Two Type IV nodes should not create a mutual link:} That is their utilities should not increase by doing so.
When $k=1$, it is not applicable for $n<5$. Also, the discrete case $n=5$ results in the same condition as below.\\
For $n\geq 6$,
\begin{align}
\nonumber
\begin{split}
&
(m_1+1)(b_1-c)+(m_2-2)b_2+(1-\gamma)b_2 \\ 
&\leq m_1(b_1-c)+(m_2-1)b_2+(1-\gamma)b_2
\end{split}
\end{align}
\begin{equation}
\nonumber
\iff c \geq b_1-b_2
\end{equation} 
For $k\geq 2$, $n=5$ corresponds to pairwise stability condition. \\
For $n\geq 6$,
\begin{equation}
\nonumber
(m_1+1)(b_1-c)+(m_2-2)b_2+b_2 \leq m_1(b_1-c)+(m_2-1)b_2+b_2
\vspace{-5mm}
\end{equation}
\begin{equation}
\nonumber
\iff c \geq b_1-b_2
\end{equation} 
These are weaker conditions than Inequality~(\ref{E7b}).\\

\noindent
\textbf{No two Type II nodes should create a link between themselves:} This only applies to $k\geq 2$.
Also $n=4,5$ result in pairwise stability condition.\\
For $n\geq 6$,
\begin{equation}
\nonumber
(m_2+2)(b_1-c)+(m_1-2)b_2 \leq (m_2+1)(b_1-c)+(m_1-1)b_2
\vspace{-5mm}
\end{equation}
\begin{equation}
\nonumber
\iff c\geq b_1-b_2
\end{equation} 
which is a weaker condition than Inequality~(\ref{E7b}).\\

\noindent
\textbf{Link between Type I node and Type II node  should not be deleted:} It is clear that it will not be deleted as such a link is just formed with no other changes in the network.\\

Inequalities~(\ref{E6}) and (\ref{E7b}) are stronger conditions than Inequalities~(\ref{E16}), (\ref{E17for2}) and (\ref{E17}).
Furthermore, for non-zero range of $c$, from Inequalities~(\ref{E6}) and (\ref{E7b}), we have
\begin{equation}
\label{E0}
\gamma < \frac{b_2-b_3}{3b_2-b_3}
\end{equation} 
The required sufficient conditions are obtained by combining Inequalities~(\ref{E1b}), (\ref{E2a}), (\ref{E6}), (\ref{E7b}) and (\ref{E0}).
\end{proof}

\section{Proof of Proposition~\ref{thm:2star}
}
\label{app:2star}
\begin{customprop}{\ref{thm:2star}}
Let $\sigma$ be the upper bound on the number of nodes that can enter the network and $\lambda = \lceil \frac{\sigma}{2} -1 \rceil \left( 2b_2-b_3 \right)$.
Then, if $\left( 1-\gamma \right) \left( b_2-b_3 \right) < c_0 < \left( 1-\gamma \right) \left( b_2-b_4 \right) $ and either \\
(i) $\gamma < \min \Big\{  \frac{b_2-b_3}{\lambda-b_3} ,  \frac{b_3}{b_2+b_3} \Big\}$ and $b_1-b_3+\gamma(b_2+b_3) \leq c < b_1$, or\\
(ii) $\frac{b_2-b_3}{\lambda-b_3} \leq \gamma < \min \Big\{ \frac{b_2}{\lambda+b_2} , \frac{b_3}{b_2+b_3}  \Big\}$ and $b_1-b_2+\gamma b_2 + \gamma \lambda \leq c < b_1$,\\
the unique resulting topology is a 
2-star.
\end{customprop}
\begin{proof}
We derive sufficient conditions for the formation of a 2-star network by forming its skeleton of four nodes first, that is, a network with two interconnected centers, connected to one leaf node each. Once this is formed, we ensure that a newly entering node connects to the center with fewer number of leaf nodes, whenever applicable, so as to maintain the load balance between the two centers.\\

\noindent
\textbf{Forming the skeleton of the 2-star network:}
With one node in the network, the second node should successfully create a link with the former. The condition for ensuring this is
\begin{equation}
\label{F0.1}
c<b_1
\end{equation}
For the third node to enter, it should propose a link to any of the two existing nodes in the network, that is, it should get a positive utility by doing so. This gives
\begin{equation}
\nonumber
c<b_1+(1-\gamma)b_2-c_0
\end{equation}
This is ensured by Inequality~(\ref{F0.1}) and 
\begin{equation}
\label{F0.2}
c_0\leq (1-\gamma)b_2
\end{equation}
Also the existing node to which the link is proposed, should accept it, that is, its utility should not decrease by doing so.
\begin{equation}
\nonumber
2(b_1-c) + 2\gamma b_2 \geq b_1-c
\vspace{-5mm}
\end{equation}
\begin{equation}
\nonumber
\iff c \leq b_1+2\gamma b_2
\end{equation}
which is a weaker condition than Inequality~(\ref{F0.1}).
We have to also ensure that this V-shaped network of three nodes is pairwise stable. It is clear that no node will delete any of its links since such a link is just formed. However, we have to ensure that the two leaf nodes of this V-shaped network do not create a mutual link. This can be ensured by
\begin{equation}
\nonumber
2(b_1-c) \leq b_1-c + (1-\gamma)b_2
\vspace{-5mm}
\end{equation}
\begin{equation}
\label{F5for3}
\iff c \geq b_1-b_2+\gamma b_2
\end{equation}
Following this, the fourth node should propose a link to one of the two leaf node in the V-shaped network. For ensuring that its utility increases by doing so,
\begin{equation}
\nonumber
c<b_1+(1-\gamma)b_2-c_0 +(1-\gamma)b_3
\end{equation}
which is satisfied by Inequalities~(\ref{F0.1}) and (\ref{F0.2}).
Also, it should prefer connecting to a leaf node than the center of the V-shaped network, that is,
\begin{equation}
\nonumber
b_1-c +(1-\gamma)b_2 - c_0 +(1-\gamma)b_3 > b_1-c +2(1-\gamma)b_2 -2c_0
\vspace{-5mm}
\end{equation}
\begin{equation}
\label{F0.3}
\iff c_0 > (1-\gamma)(b_2-b_3)
\end{equation}
The leaf node to which the link is proposed, should accept the link.
\begin{equation}
\nonumber
2(b_1-c)+(1-\gamma)b_2+2\gamma b_2 +\gamma b_3 \geq b_1-c +(1-\gamma)b_2
\vspace{-5mm}
\end{equation}
\begin{equation}
\nonumber
\iff c \leq b_1 + \gamma(2b_2+b_3)
\end{equation}
which is satisfied by Inequality~(\ref{F0.1}).\\
We have to also ensure that this network is pairwise stable. We derive sufficient conditions for pairwise stability of a general 2-star network with number of nodes $n \geq 4$, which includes the sufficient conditions for pairwise stability of the skeleton thus formed.\\

Let the centers of the 2-star be labeled $C_1$ and $C_2$. Also, let the number of leaf nodes connected to $C_1$ be $m_1$ and that connected to $C_2$ be $m_2$. \\

\noindent
\textbf{Leaf nodes that are connected to different centers, should not create a mutual link:} 
This scenario is valid for $n\geq 4$. Without loss of generality, for a leaf node connected to $C_1$,
\begin{align}
\nonumber
\begin{split}
&
2(b_1-c)+(m_1-1)(1-\gamma)b_2 +b_2 +(m_2-1)(1-\gamma)b_3\\
&\leq b_1-c +m_1(1-\gamma)b_2 +m_2(1-\gamma)b_3
\end{split}
\end{align}
\begin{equation}
\label{F5}
\iff c \geq b_1-b_3+\gamma(b_2+b_3)
\end{equation}

\noindent
\textbf{Link between one center and a leaf node of the other center should not be created:}
One option to ensure this is to see that the utility of center $C_1$ decreases owing to its link with a leaf node of $C_2$. This is valid for $n\geq 4$.
\begin{align}
\nonumber
\begin{split}
&
(m_1+2)(b_1-c) +(m_2-1)(1-\gamma)b_2 + \gamma(2)(m_1)2b_2 
+\frac{\gamma}{2}(m_1)(m_2-1)2b_3\\
&< (m_1+1)(b_1-c) +m_2(1-\gamma)b_2 + \gamma(1)(m_1)2b_2 
+ \frac{\gamma}{2}(m_1)(m_2)2b_3
\end{split}
\end{align}
\begin{equation}
\nonumber
\iff c > b_1-b_2+\gamma b_2 +\gamma m_1 (2b_2-b_3)
\end{equation}
As it needs to be true for all $n \geq 4$, we set the condition to
\begin{equation}
\nonumber
 c > \max_{n \geq 4} \Big\{ b_1-b_2+\gamma b_2 +\gamma m_1 (2b_2-b_3) \Big\}
\end{equation}
Since $\max\{m_1\} = \lceil \frac{\sigma}{2}-1 \rceil$, where $\sigma$ is the upper bound on the number of nodes that can enter the network,
\begin{equation}
\label{F6a}
c > b_1-b_2+\gamma b_2 +\gamma \lceil \frac{\sigma}{2}-1 \rceil (2b_2-b_3)
\end{equation}
An alternative option to the above condition is to ensure that the utility of leaf node connected to $C_2$ decreases when it establishes a link with $C_1$.
\begin{align}
\nonumber
\begin{split}
&
2(b_1-c) +(m_2-1)(1-\gamma)b_2 +m_1(1-\gamma)b_2 \\
&< b_1-c +m_2(1-\gamma)b_2  + m_1(1-\gamma)b_3
\end{split}
\end{align}
\begin{equation}
\nonumber
\iff c > b_1-(1-\gamma)b_2+m_1(1-\gamma)(b_2-b_3)
\end{equation}
As it needs to be true for all $n \geq 4$ and $\max\{m_1\} = \lceil \frac{\sigma}{2}-1 \rceil$, we set the condition to
\begin{equation}
\label{F6b}
 c > b_1-(1-\gamma)b_2+(1-\gamma) \lceil \frac{\sigma}{2}-1 \rceil(b_2-b_3)
\end{equation}

\noindent
\textbf{No link is broken in the 2-star network:}
It is easy to check that, as 2-star is a tree graph, the condition $c<b_1$ in Inequality~(\ref{F0.1}) is sufficient to ensure this.\\

\noindent
\textbf{Two leaf nodes of a center should not create a mutual link:}
This case arises for $n\geq 5$. It can be easily checked that the condition $c \geq b_1-b_2+\gamma b_2$ in Inequality~(\ref{F5for3}) is sufficient to ensure this.\\

This completes the sufficient conditions for pairwise stability of a 2-star network.
In what follows, we ensure that any new node successfully enters an existing 2-star network such that the topology is maintained.\\

\noindent
\textbf{A newly entering node should prefer connecting to the center with less number of leaf nodes, whenever applicable:}
This case arises when $n$ is even and $n\geq 6$, that is, when a new node tries to enter a 2-star network with odd number of nodes. Without loss of generality, let $m_1=m_2+1$. So the new node should prefer connecting to $C_2$ over $C_1$.
\begin{align}
\nonumber
\begin{split}
&
b_1-c +(m_2+1)((1-\gamma)b_2 - c_0) +m_1(1-\gamma)b_3\\
&> b_1-c +(m_1+1)((1-\gamma)b_2 -c_0) +m_2(1-\gamma)b_3
\end{split}
\end{align}
\begin{equation}
\nonumber
\iff c_0 > (1-\gamma)(b_2-b_3)
\end{equation}
which is same as Inequality~(\ref{F0.3}).\\

\noindent
\textbf{The new node should not stay out of the network:}
Its utility should be positive when it enters the network by connecting to the center with less number of leaf nodes, whenever applicable.
\begin{equation}
\nonumber
b_1-c +(m_2+1)((1-\gamma)b_2 - c_0) +m_1(1-\gamma)b_3 > 0
\end{equation}
It can be easily seen that, as $m_1,m_2 \geq 1$, the above is always true when Inequalities~(\ref{F0.1}) and (\ref{F0.3}) are satisfied.\\

\noindent
\textbf{The center with less number of leaf nodes, whenever applicable, should accept the link from the newly entering node:}
The condition $c<b_1$ in Inequality~(\ref{F0.1}) is sufficient to ensure this.\\

\noindent
\textbf{The newly entering node should prefer connecting to the center with less number of leaf nodes, whenever applicable, over connecting to any leaf node:}
It is easy to see that, as $b_3>b_4$, whenever the number of leaf nodes connected to the centers are different, a newly entering node prefers connecting to a leaf node connected to $C_1$ over that connected to $C_2$ (assuming $m_1=m_2+1$). 
Hence we have to ensure that connecting to the center with less number of leaf nodes, whenever applicable, is more beneficial to a newly entering node than connecting to a leaf node that is connected to $C_1$. Without loss of generality, we want the new node to prefer connecting to $C_2$ (irrespective of whether $m_1=m_2$ or $m_1=m_2+1$).
\begin{align}
\nonumber
\begin{split}
&
b_1-c +(m_2+1)((1-\gamma)b_2 - c_0) +m_1(1-\gamma)b_3 \\
&> b_1-c +(1-\gamma)b_2-c_0 +m_1(1-\gamma)b_3 +m_2(1-\gamma)b_4
\end{split}
\end{align}
\begin{equation}
\nonumber
\iff m_2(1-\gamma)b_2 -m_2 c_0 > m_2(1-\gamma)b_4
\end{equation}
As $m_2 \geq 1$,
\begin{equation}
\label{F4}
c_0 < (1-\gamma)(b_2-b_4)
\end{equation}

The conditions on $c$ can be obtained from Inequalities~(\ref{F0.1}), (\ref{F5for3}), (\ref{F5}), and either (\ref{F6a}) or (\ref{F6b}). 
Suppose we choose Inequality~(\ref{F6b}) over Inequality~(\ref{F6a}). So, for $c$ to have a non-empty range of values, from Inequalities~(\ref{F0.1}) and (\ref{F6b}), we must have
\begin{equation}
\nonumber
b_1-(1-\gamma)b_2+(1-\gamma) \lceil \frac{\sigma}{2}-1 \rceil(b_2-b_3) < b_1
\end{equation}
As $\gamma<1$, the above is equivalent to
\begin{equation}
\nonumber
b_2>\lceil \frac{\sigma}{2}-1 \rceil(b_2-b_3) 
\end{equation}
which is not true for arbitrarily large values of $\sigma$. So we cannot use Inequality~(\ref{F6b}).
Suppose we choose Inequality~(\ref{F6a}). So, for $c$ to have a non-empty range of values, from Inequalities~(\ref{F0.1}) and (\ref{F6a}), we must have
\begin{equation}
\nonumber
b_1-b_2+\gamma b_2 +\gamma \lceil \frac{\sigma}{2}-1 \rceil (2b_2-b_3) < b_1
\end{equation}
Let $\lambda = \lceil \frac{\sigma}{2}-1 \rceil (2b_2-b_3)$. So the above is equivalent to
\begin{equation}
\label{2stargamma1}
\gamma < \frac{b_2}{\lambda+b_2}
\end{equation}
which is a valid range of $\gamma$ as $\gamma \in {[0,1)}$.
So we use Inequality~(\ref{F6a}) instead of Inequality~(\ref{F6b}). Also, Inequality~(\ref{F5for3}) is weaker than Inequality~(\ref{F6a}).
For $c$ to have a non-empty range of values, it is also necessary, from Inequalities~(\ref{F0.1}) and (\ref{F5}), that
\begin{equation}
\nonumber
b_1-b_3+\gamma(b_2+b_3) < b_1
\vspace{-5mm}
\end{equation}
\begin{equation}
\label{2stargamma2}
\iff \gamma < \frac{b_3}{b_2+b_3}
\end{equation}
Both Inequalities~(\ref{F5}) and (\ref{F6a}) lower bound $c$. So we need to determine the stronger condition of the two. It can be seen that Inequality~(\ref{F5}) is at least as strong as Inequality~(\ref{F6a}) if and only if
\begin{equation}
\nonumber
b_1-b_3+\gamma(b_2+b_3) \geq b_1-b_2+\gamma b_2 + \gamma \lambda
\end{equation}
\begin{equation}
\label{2stargamma3}
\iff \gamma \leq \frac{b_2-b_3}{\lambda-b_3}
\end{equation}
We consider the cases when either is a stronger condition.

\noindent
\textbf{Case (i)} If Inequality~(\ref{2stargamma3}) is true:\\
Inequalities~(\ref{F0.1}) and (\ref{F5}) are the strongest conditions. So the sufficient condition on $c$ is
\begin{equation}
\label{2starfinalc1}
b_1-b_3+\gamma(b_2+b_3) \leq c < b_1
\end{equation}
and Inequalities~(\ref{2stargamma1}), (\ref{2stargamma2}) and (\ref{2stargamma3}) give
\begin{equation}
\nonumber
\gamma < \min \Big\{ \frac{b_2-b_3}{\lambda-b_3} , \frac{b_2}{\lambda+b_2} ,   \frac{b_3}{b_2+b_3}  \Big\}
\end{equation}
It can also be shown that for $\lambda b_3 \geq b_2^2$, 
$\min \Big\{ \frac{b_2-b_3}{\lambda-b_3} , \frac{b_2}{\lambda+b_2} ,   \frac{b_3}{b_2+b_3}  \Big\} = \frac{b_2-b_3}{\lambda-b_3}$ \\
 and for $\lambda b_3 \leq b_2^2$, 
$\min \Big\{ \frac{b_2-b_3}{\lambda-b_3} , \frac{b_2}{\lambda+b_2} ,   \frac{b_3}{b_2+b_3}  \Big\} = \frac{b_3}{b_2+b_3}$. So the above reduces to
\begin{equation}
\label{2starfinalgamma1}
\gamma < \min \Big\{ \frac{b_2-b_3}{\lambda-b_3} , \frac{b_3}{b_2+b_3}  \Big\}
\end{equation}

\noindent
\textbf{Case (ii)} If Inequality~(\ref{2stargamma3}) is not true:\\
Inequalities~(\ref{F0.1}) and (\ref{F6a}) are the strongest conditions. So the sufficient condition on $c$ is
\begin{equation}
\label{2starfinalc2}
b_1-b_2+\gamma b_2 + \gamma \lambda \leq c < b_1
\end{equation}
and Inequalities~(\ref{2stargamma1}), (\ref{2stargamma2}) and the reverse of (\ref{2stargamma3}) give
\begin{equation}
\label{2starfinalgamma2}
\frac{b_2-b_3}{\lambda-b_3} \leq \gamma < \min \Big\{ \frac{b_2}{\lambda+b_2}  ,  \frac{b_3}{b_2+b_3}  \Big\}
\end{equation}
Furthermore, Inequalities~(\ref{F0.2}), (\ref{F0.3}) and (\ref{F4}) give the sufficient conditions on $c_0$.
\begin{equation}
\label{2starfinalc0}
(1-\gamma)(b_2-b_3) < c_0 < (1-\gamma)(b_2-b_4)
\end{equation}
Inequalities~(\ref{2starfinalc1}), (\ref{2starfinalgamma1}) and (\ref{2starfinalc0}) give the sufficient conditions $(i)$ in the proposition, while Inequalities~(\ref{2starfinalc2}), (\ref{2starfinalgamma2}) and (\ref{2starfinalc0}) give the sufficient conditions $(ii)$.
\end{proof}

\section{Proof of Lemma~\ref{lem:kstar0}
}
\label{app:kstar0}
\begin{customlem}{\ref{lem:kstar0}}
Under the proposed utility model, for the entire family of $k$-star networks (given some $k\geq 3$) to be pairwise stable, it is necessary that 
$\gamma=0$ and $c=b_1-b_3$.
\end{customlem}
\begin{proof}
We consider two scenarios sufficient to prove this.\\

\noindent
\textbf{I) No center should delete its link with any other center:} Here, only one case is enough to be considered, that is, when each center has just one leaf node,
since in all other cases, the benefits obtained by each center from the connection with other centers is at least as much. For $k=3$, 
\begin{align}
\nonumber
\begin{split}
&
3(b_1-c) + 2(1-\gamma)b_2 + \gamma(1)(2)2b_2 + \frac{\gamma}{2}(1)(2)b_3 \\
&\geq
2(b_1-c) + 2(1-\gamma)b_2 + (1-\gamma)b_3 + \gamma(1)(1)2b_2 
 + 2\left( \frac{\gamma}{2}(1)(1)2b_3 \right) +  \frac{\gamma}{3}(1)(1)2b_4
\end{split}
\end{align}
\begin{equation}
\label{eq:kstarineq1}
\iff c \leq b_1-b_3+\gamma(2b_2+b_3)-\frac{2\gamma}{3}b_4
\end{equation}
For $k\geq 4$,
\begin{align}
\nonumber
\begin{split}
&
(k-1+1)(b_1-c) + (k-1)(1-\gamma)b_2 + \gamma(1)(k-1)2b_2  
+ \frac{\gamma}{2}(1)(k-1)2b_3\\
 &\geq (k-2+1)(b_1-c) + (k-2)(1-\gamma)b_2 + b_2 + (1-\gamma)b_3  \\
 &\;\;\;\;\; +  \gamma(1)(k-2)2b_2+ \frac{\gamma}{2}(1)(k-2)2b_3 + \gamma(1)(1)2b_3 + \frac{\gamma}{2}(1)(1)2b_4
\end{split}
\end{align}
\begin{equation}
\label{eq:kstarineq2}
\iff c \leq b_1-b_3+\gamma(b_2-b_4)
\end{equation}

\noindent
\textbf{II) Leaf nodes of different centers should not form a link with each other:} Consider a leaf node. Let $m_i$ be the number of leaf nodes connected to the center to which the leaf node under consideration, is connected. For $k\geq 3$,
\begin{align}
\nonumber
\begin{split}
&
2(b_1-c) + (m_i-1)(1-\gamma)b_2 + (1-\gamma)b_3 (\sum_{j \neq i}m_i - 1) + (k-1)b_2 \\
&\leq b_1-c + (m_i-1)(1-\gamma)b_2 + (1-\gamma)b_3 \sum_{j \neq i}m_i + (k-1)(1-\gamma)b_2
\end{split}
\end{align}
\begin{equation}
\label{eq:kstarineq3}
\iff c \geq b_1-b_3+\gamma((k-1)b_2+b_3)
\end{equation}
The only way to satisfy Inequalities~(\ref{eq:kstarineq1}), (\ref{eq:kstarineq2}) and (\ref{eq:kstarineq3}) simultaneously is by setting 
\begin{equation}
\label{eq:kstargamma}
\gamma=0
\end{equation}
and
\begin{equation}
\label{eq:kstarcost}
c=b_1-b_3
\end{equation}
thus proving the lemma.
\end{proof}

\section{Proof of Proposition~\ref{thm:kstar}
}
\label{app:kstar}
\begin{customprop}{\ref{thm:kstar}}
For a network starting with the base graph for $k$-star (given some $k \geq 3$), and $\gamma =0 $, if $c =b_1-b_3 $ 
and $ b_2-b_3  < c_0 < b_2-b_4$, the unique resulting topology is a 
$k$-star.
\end{customprop}
\begin{proof}
It is clear from Lemma~\ref{lem:kstar0} that under the proposed utility model, for the family of $k$-star networks ($k\geq 3$) to be pairwise stable, it is necessary that $\gamma=0$ and $c=b_1-b_3$ in order to stabilize all possible $k$-star networks for a given $k$, and hence forms the necessary part of sufficient conditions for the formation of a $k$-star network. Hence, for the rest of this proof, we will assume that
\begin{equation}
\label{eq:appkstargamma}
\gamma=0
\end{equation}
and
\begin{equation}
\label{eq:appkstarcost}
c=b_1-b_3
\end{equation}

 Without loss of generality, assume some indexing over the $k$ centers from $1$ to $k$.
Let $C_i$ be the center with index $i$ and $m_i$ be the number of leaf nodes it is linked to. 
Also we start with a base graph in which every center is linked to one leaf node and the number of leaf nodes linked to each center increases as the process goes on. 
So we have, $m_i \geq 1\text{ for } 1 \leq i \leq k$.\\

\noindent
\textbf{For the newly entering node to propose entering the network:}
Our objective is to ensure that the newly entering node connects to a center with the least number of leaf nodes, in order to maintain balance over the number of leaf nodes linked to the centers. Without loss of generality, assume that we want the newly entering node to connect to $C_1$. The utility of the newly entering node should be positive after doing so.
\begin{equation}
\nonumber
b_1 - c + (m_1+k-1)\left(b_2-c_0\right) +b_3\sum_{i=2}^{k} m_i > 0
\end{equation}
Since the minimum value of $m_i$ is $1$ for any $i$, the above condition is true if
\begin{equation}
\nonumber
c < b_1 + k \left(b_2-c_0\right) + (k-1)b_3
\end{equation}
This is satisfied by Equation~(\ref{eq:appkstarcost}) and
\begin{equation}
\label{G1forc0}
c_0 < b_2+b_3
\end{equation}

\noindent
\textbf{The newly entering node should connect to a center with the least number of leaf nodes, whenever applicable:}
This case does not arise when all centers have the same number of leaf nodes. Moreover, the way we direct the evolution of the network, the number of leaf nodes connected to any two centers differs by at most one. Without loss of generality, assume that we want the newly entering node to connect to $C_1$. Consider a center $C_p$ such that $m_p = m_1+1$. So the newly entering node should prefer connecting to $C_1$ over connecting to $C_p$.
\vspace{-5mm}
\begin{align}
\nonumber
\begin{split}
&
b_1-c + (m_1+k-1) \left(b_2 - c_0 \right) +b_3\sum_{i=2}^{k} m_i \\
&> b_1-c + (m_p+k-1) \left( b_2 - c_0 \right) + b_3\sum_{\substack{1 \leq i \leq k \\ i \neq p}} m_i
\end{split}
\end{align}
As $m_p = m_1 + 1$, we have
\begin{equation}
\label{G3}
c_0 > b_2-b_3
\end{equation}

\noindent
\textbf{For a center with the least number of leaf nodes to accept the link from the newly entering node:}
It can be easily seen that this is ensured by Equation~(\ref{eq:appkstarcost}).\\

\noindent
\textbf{The newly entering node should not connect to any leaf node:}
It can be easily seen that owing to benefits degrading with distance, for the newly entering node, connecting to any leaf node which is connected to a center with the most number of leaf nodes, strictly dominates connecting to any other leaf node, whenever applicable. So it is sufficient to ensure that the newly entering node does not connect to any leaf node which is connected to a center with the most number of leaf nodes.
This can be done by ensuring that for the newly entering node, connecting to a center with the least number of leaf nodes strictly dominates connecting to any leaf node which is connected to a center with the most number of leaf nodes.
Say we want the newly entering node to prefer connecting to center $C_1$ over a leaf node that is linked to center $C_p$.
\vspace{-5mm}
\begin{align}
\nonumber
\begin{split}
&
b_1-c + (m_1+k-1)(b_2-c_0) +  b_3\sum_{i=2}^k m_i \\
&> b_1-c +b_2 - c_0 + (m_p + k-2)b_3 +  b_4\sum_{\substack{1 \leq i \leq k \\ i \neq p}} m_i 
\end{split}
\end{align}
We need to consider two cases (i) $m_p = m_1+1$ and (ii) $m_p=m_1$\\
Case (i) $m_p = m_1+1$: Substituting  this value of $m_p$ gives
\begin{equation}
\nonumber
\begin{split}
(m_1+k-1)(b_2-c_0) +  (b_3-b_4)\sum_{\substack{2 \leq i \leq k \\ i \neq p}} m_i + m_1 (b_3-b_4) \\+ b_3
> b_2 - c_0 + (m_1 + k-1)b_3 
\end{split}
\end{equation}
As the minimum value of $\sum_{2 \leq i \leq k, i \neq p}m_i$ is $k-2$, the above remains true if we replace $\sum_{2 \leq i \leq k, i \neq p}m_i$ with $k-2$. Further simplification gives
\begin{equation}
\nonumber
(m_1+k-2)(b_2-b_4-c_0)>0
\end{equation}
Since $m_1+k-2>0$ is positive (as $m_1 \geq 1$ and $k \geq 3$), we must have
\begin{equation}
\label{G4forc0}
c_0 < b_2-b_4
\end{equation}
Case (ii) $m_k=m_1$: It can be similarly shown that Equation~(\ref{G4forc0}) is the sufficient condition.\\

\noindent
Now that the newly entering node enters in a way such that $k$-star network is formed, we have to ensure that no further modifications of links occur so that the network thus formed, is pairwise stable.\\

\noindent
\textbf{For centers and the corresponding leaf nodes to not delete the link between them:} It can be easily seen that $c<b_1$, a weaker condition than Equation~(\ref{eq:appkstarcost}), is a sufficient condition to ensure this.\\

\noindent
\textbf{No center should delete its link with any other center:} This is ensured by the inequalities in the proof of Lemma~\ref{lem:kstar0}, which are weaker than Equations~(\ref{eq:appkstargamma}) and (\ref{eq:appkstarcost}). \\

\noindent
\textbf{Leaf nodes of a center should not form a link with each other:} The net benefit that a leaf node would get by forming such a link should be non-positive.
\begin{equation}
\nonumber
b_1-c - b_2 \leq 0
\vspace{-5mm}
\end{equation}
\begin{equation}
\nonumber
\iff c \geq b_1-b_2
\end{equation}
which is satisfied by Equation~(\ref{eq:appkstarcost}).\\

\noindent
\textbf{Leaf nodes of different centers should not form a link with each other:} This is ensured by the inequality in the proof of Lemma~\ref{lem:kstar0}, which is weaker than Equations~(\ref{eq:appkstargamma}) and (\ref{eq:appkstarcost}). \\

\noindent
\textbf{Link between a center and a leaf node of any other center should not be created:} Let $C_i$ be the center under consideration and the leaf node under consideration be linked to $C_j$ ($j\neq i$). There are two ways to ensure this. First is to ensure that a center neither proposes nor accepts a link with a leaf node of any other center. This mathematically is
\begin{align}
\nonumber
\begin{split}
&
(k-1+m_i+1)(b_1-c)+  b_2(\sum_{\substack{1 \leq q \leq k \\ q \neq i}}m_q-1) \\ &< (k-1+m_i)(b_1-c) + b_2 \sum_{\substack{1 \leq q \leq k \\ q \neq i}}m_q 
\end{split}
\end{align}
\begin{equation}
\nonumber
\iff c > b_1-b_2
\end{equation}

An alternative to this condition is to ensure that a leaf node neither proposes nor accepts a link with a center to which it is not connected, but since this condition is already satisfied by Equation~(\ref{eq:appkstarcost}), this alternative need not be considered.

Equations~(\ref{eq:appkstargamma}), (\ref{eq:appkstarcost}), (\ref{G1forc0}), (\ref{G3}) and (\ref{G4forc0}) give the required sufficient conditions for the $k$-star network topology.
%
\end{proof}

\section{Proof of Theorem~\ref{thm:gedkstar}
}
\label{app:gedkstar}
\begin{customthm}{\ref{thm:gedkstar}}
There exists an $O(\mu^{k+2})$ polynomial time algorithm to compute the graph edit distance 
 between a graph $g$ and a $k$-star graph with same number of nodes as $g$, where $\mu$ is the number of nodes in $g$.
\end{customthm}
\begin{proof}
Assume that the mapping of the $k$ centers of the $k$-star network to the nodes in $g$, is known. Let us call these nodes of $g$ as {\em pseudo-centers}. The graph edit distance can be computed by taking the minimum number of edge edit operations over all possible mappings. In a $k$-star graph, each node, other than centers, is allotted to exactly one center. Hence, our objective is to allot nodes, other than pseudo-centers, (call them {\em pseudo-leaves}) in $g$ to pseudo-centers such that the graph edit distance is minimized.
Let $\mu$ and $\xi$ be the number of nodes and edges in $g$, respectively. 
Let {\em vacancy} of a pseudo-center at any point of time be defined as the maximum number of nodes that can be allotted to it, given the current allotment.
Note that if $\mu$ is not a multiple of $k$, the vacancy of a pseudo-center depends not only on the number of pseudo-leaves allotted to it, but also on the number of pseudo-leaves allotted to other pseudo-centers.

It is clear that if the mapping of the $k$ centers is known, for transforming $g$ to a corresponding $k$-star, it is necessary that all missing links between any two pseudo-centers be added (let $\beta_1$ be the number of such links) and all existing links between any two pseudo-leaves be deleted (let $\beta_2$ be the number of such links).
The only other links that need to be computed for additions or deletions, in order to minimize graph edit distance, are those interlinking pseudo-leaves with pseudo-centers.
The number of links that already interlink pseudo-leaves with pseudo-centers in $g$ is $\beta_3 = (\xi-\beta_2-($\begin{footnotesize}$\dbinom{k}{2}$\end{footnotesize}$-\beta_1))$.
Say the number of these edges that are retained during the transformation to $k$-star, is $f$, that is, exactly $f$ pseudo-leaves are allotted a pseudo-center and $(\mu-k-f)$ are not. So the number of edges interlinking pseudo-leaves with pseudo-centers, that are deleted during the transformation, is $(\beta_3-f)$. Also, the number of edges to be added in order to allot the pseudo-leaves, that are not allotted to any pseudo-center, to some pseudo-center having a positive vacancy, is $(\mu-k-f)$.
So the number of edge edit operations is $(\beta_1+\beta_2+\beta_3+\mu-k-2f) = (\mu+\xi+2\beta_1-$\begin{normalsize}$\frac{k}{2}$\end{normalsize}$(k+1)-2f)$.
Given a mapping of the $k$ centers, the only variable in this expression is $f$. 
So in order to minimize its value, we need to maximize the number of edges interlinking pseudo-leaves and pseudo-centers, that remain intact after the transformation to $k$-star.
We now address this problem of maximizing $f$.

Let the number of nodes in $g$ be $\mu=pk+q$ where $p$ and $q$ are integers such that $p \geq 0$ and $1\leq q< k$. 
In a $k$-star graph with $\mu$ nodes, $q$ centers are linked to $p$ leaf nodes and the remaining $k-q$ are linked to $p-1$ leaf nodes. So for transforming $g$ to a corresponding $k$-star graph, $q$ pseudo-centers should be allotted $p$ pseudo-leaves and the remaining $k-q$ should be allotted $p-1$.
So, at most $q$ pseudo-centers should be allotted $p$ nodes, that is, the vacancy of at most $q$ pseudo-centers should be $p$, while that of the remaining $k-q$ should be $p-1$.
In other words, to start with, the sum of vacancies of any $q+1$ pseudo-centers should be at most $(q+1)p-1$. 

\begin{figure}[t!]
\centering
\includegraphics[scale=0.57]{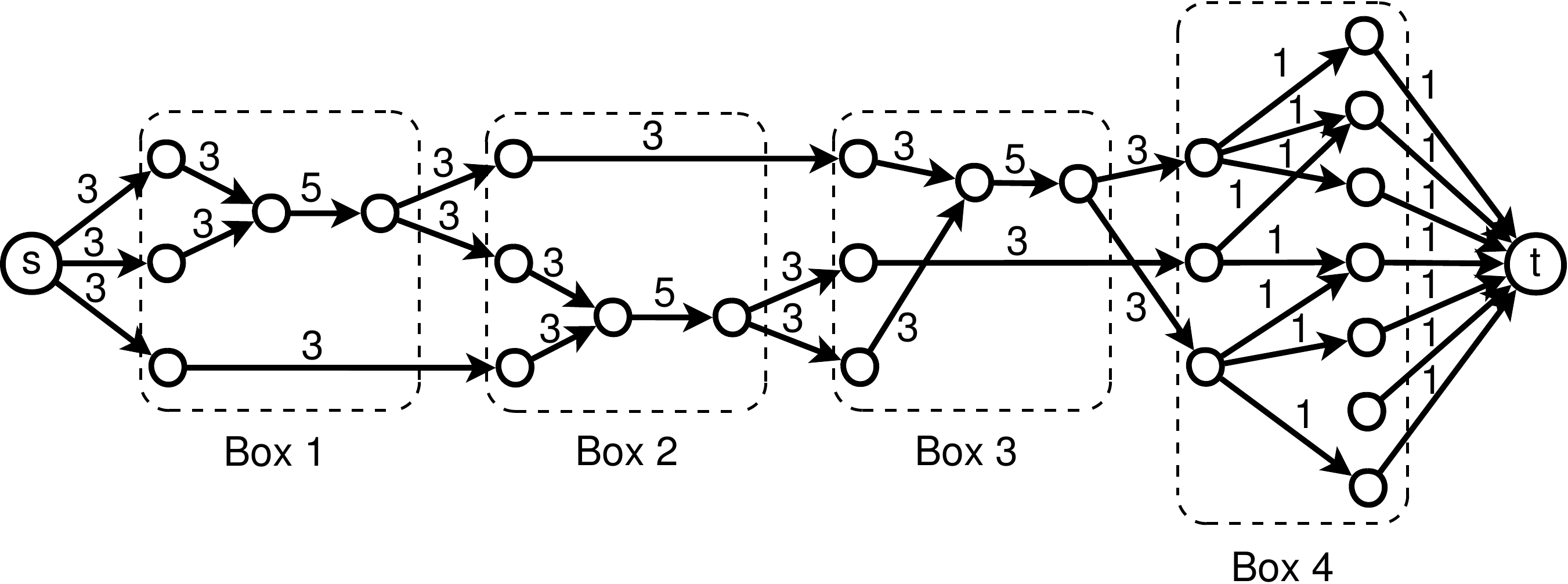}
\caption{Formulation of graph edit distance between graph $g$ ($\mu=10$) and a $3$-star graph with same number of nodes as $g$, as a max-flow problem}
\label{fig:max_flow_eg}
\end{figure}

The above problem can be formulated as an application of max-flow in a directed network.
Figure~\ref{fig:max_flow_eg} shows the formulation for a graph $g$ with 10 nodes and a 3-star graph, where $p=3$ and $q=1$. The edges directing from the source node $s$ to the left $k$ nodes in Box 1 (here $k=3$) and those in Boxes 1, 2 and 3, formulate the vacancy of each of these pseudo-centers to be $p$. Boxes 1, 2 and 3 formulate the constraint that the sum of vacancies of any $q+1$ pseudo-centers should be at most $(q+1)p-1$.
The rightmost Box 4 is obtained by considering edges only interlinking any pseudo-centers (left nodes) and pseudo-leaves (right nodes). 

As all the edges have integer capacities, the Ford-Fulkerson algorithm constructs an integer maximum flow.
The number of constraints concerning the sum of vacancies of pseudo-centers is \begin{footnotesize}$\dbinom{k}{q+1}$\end{footnotesize} and number of edges added per such constraint is $2(q+1)+2$.

So the maximum number of edges, say $\chi$, in the max-flow formulation, is $k$ (from source node to left $k$ nodes in Box 1) $+$ $\left( 2(q+1)+2 \right)$\begin{footnotesize}$\dbinom{k}{q+1}$\end{footnotesize} (from the above calculation) $+$ $k(\mu-k)$ (upper limit on the number of edges in Box 4, interlinking pseudo-centers and pseudo-leaves) $+$ $(\mu-k)$ (number of edges directing towards target node). 
Since $1 \leq q < k$, we have $\chi=O(k^{\frac{k}{2}+1}+\mu k)$.
As the value of the maximum flow is upper bounded by $\mu-k$, the Ford-Fulkerson algorithm runs in $O(\chi \mu) = O(\mu k^{\frac{k}{2}+1}+\mu^2 k)$ time. 
Furthermore, as $k$ is a constant, the asymptotic worst-case time complexity is $O(\mu^2)$.

So given a mapping of the $k$ centers, the number of edge edit operations, $(\mu+\xi+2\beta_1-$\begin{normalsize}$\frac{k}{2}$\end{normalsize}$(k+1)-2f)$, is minimized since $f$ is maximized.
The time complexity of the above algorithm is dominated by the max-flow algorithm. The above analysis was assuming that the mapping of the $k$ centers of the $k$-star network to the nodes in $G$, is known. The graph edit distance can, hence, be computed by taking the minimum edit distance over all \begin{footnotesize}$\dbinom{\mu}{k}$\end{footnotesize} $= O(\mu^k)$ possible mappings. So the asymptotic worst-case time complexity of the algorithm is $O(\mu^{k+2}) = O(\mu^{O(1)})$, since $k$ is a constant.
%
\end{proof}

\end{subappendices}

\blankpagewithnumber


\newcolumntype{L}[1]{>{\raggedright\let\newline\\\arraybackslash\hspace{0pt}}m{#1}}
\newcolumntype{C}[1]{>{\centering\let\newline\\\arraybackslash\hspace{0pt}}m{#1}}
\newcolumntype{R}[1]{>{\raggedleft\let\newline\\\arraybackslash\hspace{0pt}}m{#1}}

\newcommand{\cmark}{\ding{51}}
\newcommand{\xmark}{\ding{55}}

\chapter[Information Diffusion in Social Networks in Multiple Phases]{Information Diffusion in Social Networks in Multiple Phases
  \blfootnote{A part of this chapter is published as \cite{dhamal2015multiphase}:
  Swapnil Dhamal, Prabuchandran K. J., and Y. Narahari. A multi-phase approach for improving information diffusion in social networks. In {\em Proceedings of the 14th International Conference on Autonomous Agents and Multiagent Systems (AAMAS)}, pages 1787--1788, 2015.}
\blfootnote{A significant part of this chapter is published as \cite{dhamal2016information}:
Swapnil Dhamal, Prabuchandran~K.~J., and Y. Narahari. Information diffusion in social networks in two phases. {\em Transactions on Network Science and Engineering}, 3(4):197--210, 2016.}
}

\label{chap:mpid}

\begin{quote}
The problem of maximizing information diffusion, given a certain budget constraint expressed in terms of the number of seed nodes, is an important topic in social networks research. Existing literature focuses on single phase diffusion where (a) all seed nodes are selected at the beginning of diffusion and (b) all the selected nodes are activated simultaneously. This chapter undertakes a detailed investigation of the effect of selecting and activating seed nodes in multiple phases. Specifically, we study diffusion in two phases assuming the well-studied independent cascade model. First, we formulate an objective function for two-phase diffusion, investigate its properties, and propose efficient algorithms for finding the seed nodes in the two phases. Next, we study two associated problems: (1) {\em budget splitting} which seeks to optimally split the total budget between the two phases and (2) {\em scheduling} which seeks to determine an optimal delay after which to commence the second phase. Our main conclusions include: (a) under strict temporal constraints, use single phase diffusion, (b) under moderate temporal constraints, use two-phase diffusion with a short delay allocating more of the budget to the first phase, and (c) when there are no temporal constraints, use two-phase diffusion with a long delay allocating roughly one-third of the budget to the first phase.
%
\end{quote}


\newpage
\section{Introduction}
\label{sec:intro_mpid}
Social networks play a fundamental role in the spread of information on a large scale. In particular, online social
networks have become very popular in recent times, and so is the trend of using them for information diffusion.  
An information can be of various types, namely, opinions, ideas, behaviors, innovations, diseases, rumors, etc. 
%
The objective of whether to maximize or restrict the spread of information would depend on the type of information, and thus the objective function is defined accordingly.
One of the central questions in information diffusion is the following: given a certain budget $k$ expressed in terms of the number of seed nodes, which $k$ nodes in the social network should be selected to trigger the diffusion so as to maximize a suitably defined objective function?

For example, if a company wishes to do a viral marketing of a particular product via, for instance, word-of-mouth,
the objective is to spread the information through the network such that number of nodes influenced at the end of the diffusion process, is maximized. So the company would try to select the seed nodes (nodes to whom free samples, discounts, or other such incentives are provided) such that the number of nodes influenced by viral marketing, and hence the sales of that product, would be maximized. 
%
On the other hand, if an organization wishes to contain the spread of certain 
rumor that is already spreading in the network, it would want to trigger a competing positive campaign at selected seed nodes with the objective of  minimizing the effects of rumor.
%

In this chapter, we focus on {\em influence maximization} in which, given a budget $k$, our objective is to select at most $k$ seed nodes where the diffusion should be triggered,
so as to maximize the spread of influence at the end of the diffusion.
Throughout this chapter, we call $k$, the number of seed nodes, as the budget.

\subsection{Model for Information Diffusion}
We represent a social network as a graph $G$, having $N$ as its set of $n$ nodes and $E$ as its set of $m$ weighted and directed edges.
For studying diffusion in such a network, several models  have been proposed in the literature \cite{networkscrowdsmarkets}. 
The Independent Cascade (IC) model and the Linear Threshold (LT) model are two of the most well-studied models.
In this chapter, our focus will be on the IC model for most part; we later provide a note on the LT model.

\subsubsection{The Independent Cascade (IC) model}
In the IC model, for each directed edge $(u,v) \in E$, there is an associated {\em influence probability\/} $p_{uv}$ that specifies the probability with which the source node $u$  influences the target node $v$. The diffusion starts at time step $0$ with simultaneous triggering of a set of initially activated or influenced seed nodes, following which, the diffusion proceeds in discrete time steps. In each time step, nodes which got influenced in the previous time step (call them {\em recently activated nodes}) attempt to influence their neighbors, and succeed in doing so with the influence probabilities that are associated with the corresponding edges. 
These neighbors, if successfully influenced, will now become the recently activated nodes for the next time step. In any given time step, only recently  activated nodes contribute to diffusing information. After this time step, such nodes are no longer recently activated and we call them {\em already activated} nodes. Nodes, once activated, remain activated for the rest of the diffusion process. 
In short, when node $u$ gets activated at a certain time step, it gets exactly one chance to activate each of its inactive neighbors (that too in the immediately following time step), with the given influence 
probability $p_{uv}$ for each neighbor $v$.
The diffusion process terminates when no further nodes can be activated.

\subsubsection{Notion of Live Graph}
The notion of {\em live graph} is crucial to the analysis of the IC model.
A live graph $X$ of a graph $G$ is an instance of graph $G$, obtained by sampling the edges;
an edge $(u,v)$ is present in the live graph with probability $p_{uv}$ and absent with probability $1-p_{uv}$, independent of the presence of other edges in the live graph (so a live graph is a directed graph with no edge probabilities). 
The probability $p(X)$ of occurrence of any live graph $X$, can be obtained as
$\prod_{(u,v) \in X} (p_{uv}) \prod_{(u,v) \notin X} (1-p_{uv})$.
%
%
It can be seen that as long as a node $u$, when influenced, in turn influences node $v$ with probability $p_{uv}$ that is independent of time, sampling the edge $(u,v)$ in the beginning of the diffusion is equivalent to sampling it when $u$ is activated \cite{kempe2003maximizing}. 

\subsubsection{Special Cases of the IC model}
In this chapter, when there is a need for transforming an undirected, unweighted network (dataset) into a directed and weighted network for studying the diffusion process, we consider two popular, well-accepted special cases of the IC model, namely, the {\em weighted cascade (WC) model} and the {\em trivalency (TV) model}. The weighted cascade model does the transformation by making all edges bidirectional and assigning a weight to every directed edge $(u,v)$ equal to the reciprocal of $v$'s degree in the undirected network~\cite{kempe2003maximizing}.
The trivalency model makes all edges bidirectional and assigns a weight to every directed edge by uniformly sampling from the set of values $\{0.001, 0.01, 0.1\}$.


\section{Relevant Work}
\label{sec:relevant_mpid}

The problem of influence maximization in social networks has been extensively studied in the literature \cite{networkscrowdsmarkets,guille2013information}.
The impact of recommendations and word-of-mouth marketing on product sales revenue is also well-studied in marketing; see for example,
\cite{godes2012strategic,van2010viral,aral2011creating,aral2012identifying,reichheld2003one}.

It has been shown that obtaining the exact value of the objective function for a seed set (that is, the expected number of influenced nodes at the end of the diffusion process that was triggered at the nodes of that set), is \#P-hard under the IC model \cite{chen2010scalable} as well as the LT model \cite{chen2010scalablelt}.
However, the value be obtained with high accuracy using a sufficiently large number of Monte-Carlo simulations.
Kempe, Kleinberg, and Tardos \cite{kempe2003maximizing} show that maximizing the objective function under the IC model is NP-hard,
and present a $(1-\frac{1}{e}-\epsilon)$-approximate algorithm,
where $\epsilon$ is small for sufficiently large number of Monte-Carlo simulations.
Chen, Wang, and Yang \cite{chen2009efficient} propose fast heuristics for influence maximization under the IC model as the greedy algorithm is computationally intensive. 
%
Another line of work 
attempts to relax the assumption that the influence probabilities are known, for example, 
\cite{goyal2010learning}.

There exist fundamental generalizations of these basic models, for instance, general threshold model \cite{kempe2003maximizing} and decreasing cascade model \cite{kempe2005influential}.
Borodin, Filmus, and Oren \cite{borodin2010threshold} provide several natural extensions to the LT model and show that 
for a broad family of competitive influence models, it is NP-hard to achieve an approximation that is better than a square root of the optimal solution.
%
Jiang et al. \cite{jiang2013evolutionary} propose a multiagent model where each agent evolves its trust network consisting of other agents whom it trusts (or who can influence it).
Gabbriellini and Torroni \cite{gabbriellini2013arguments} simulate the propagation and evolution of opinions by proposing a model where agents belonging to a social network reason and interact argumentatively and decide whether and how to revise their own beliefs.
Subbian et al. \cite{subbian2013social} suggest finding the influencers in a social network based on their model which considers the individual social values generated by collaborations in the network.
Yu et al. \cite{yu2013emergence} suggest that based on the inherent voting rule adopted by the agents to aggregate the opinions of their neighbors, they are more influenced by the opinion adopted either by most of their neighbors or by a person who has a higher reputation.

Narayanam and Narahari \cite{narayanam2010shapley,suri2008determining} provide a Shapley value based algorithm 
that gives satisfactory performance irrespective of whether or not the objective function is submodular.
%
Franks et al. \cite{franks2013manipulating} use influencer agents effectively to manipulate the emergence of conventions and increase convention adoption and quality.
Shakarian et al. \cite{shakarian2013mancalog} introduce a logical framework designed to describe cascades in complex networks.
%
Franks et al. \cite{franks2013learning} propose a general methodology for learning the network value of a node in terms of influence.
Ghanem et al. \cite{ghanem2012agents} study the different patterns of interaction behavior in online social networks wherein, they identify four primary classes of social agents and analyze the influence of agents from each class in the viral spread of ideas under various conditions. 
Mohite and Narahari \cite{mohite2011incentive} use a mechanism design approach to elicit influence values truthfully from the agents since these values are usually not known to the social planner and the strategic agents may not reveal them truthfully.
Bakshy et al. \cite{bakshy2012role} discuss the importance of weak ties in information diffusion.



Another well-studied problem in the topic of information diffusion in social networks is the problem of influence limitation \cite{budak2011limiting, premm2012influence}, where the objective is to minimize the spread of a negative campaign by triggering a positive campaign.
Bharathi, Kempe, and Salek \cite{bharathi2007competitive} study the problem of competitive influence maximization wherein multiple companies market competing products using viral marketing; they provide an approximation algorithm for computing the best response to the strategy of competitors.
Pathak, Banerjee, and Srivastava \cite{pathak2010generalized} provide a generalized version of the LT model for multiple cascades on a network while allowing nodes to switch between them wherein, the steady state distribution of a Markov chain is used to estimate highly likely states of the cascades' spread in the network.
Myers and Leskovec \cite{myers2012clash} develop a more realistic and practical model where contagions not only propagate at the same time but they also interact, that is, compete or cooperate with each other as they spread over the network.
%
Goyal and Kearns \cite{goyal2012competitive} develop a game-theoretic framework for competitive diffusion in a social network and hence analyze the game in detail.

Time related constraints in the context of diffusion have also been studied in the literature. 
Chen, Lu, and Zhang \cite{chen2012time} consider the problem where the goal is to maximize influence spread within a given deadline.
Nguyen et al. \cite{nguyen2012containment} aim to find the smallest set of influential nodes whose decontamination with good information would help contain the viral spread of misinformation, to a desired ratio in a given number of time steps,
given the seed nodes of the misinformation.

The above papers address only single phase diffusion.
The idea of using multiple phases for maximizing an objective function has been presented in \cite{golovin2011adaptive}; the study is a preliminary one.
To the best of our knowledge, ours is the first detailed effort to study multi-phase information diffusion in social networks.
In the next section, we bring out the motivation for this work, present a motivating example, and
describe the agenda of this work.

\section{Motivation}
\label{sec:motiv_mpid}

Most of the existing literature on information diffusion works with the assumption that the diffusion is triggered at all the selected $k$ seed nodes in one go, that is, in the beginning of the process; in other words, the budget $k$ is exhausted in one single instalment. We consider triggering the diffusion in multiple phases by splitting the total budget
$k$ in an appropriate way across the multiple phases. 
A tempting advantage of multi-phase diffusion is that the seed nodes in second and subsequent phases could be chosen based on the spread observed so far, thus reducing uncertainty while selecting the seed nodes.
However, the disadvantage could be that the diffusion may  slow down owing to instalment-based triggering of seed nodes.

In the IC model, where influence probabilities are crucial to the diffusion, the diffusion process is a random process and the general problem addressed in the literature is to maximize influence spread in expectation; it is possible the spread in certain instances may be much less than the expected one. This is a vital practical issue because a company or organization investing in triggering seed nodes for diffusion cannot afford awkward instances where the spread is disappointingly lower than the expected one. Multi-phase diffusion seems an attractive and  natural approach wherein, the company can modulate its decisions at intermediate times during the diffusion process, in order to avoid such instances. This happens because the company would be more certain about the diffusion process and hence would hopefully select better seed nodes in the second and subsequent phases.  However, as already noted, there is a delay in activating the second and subsequent seed sets and the overall diffusion process may be delayed, leading to compromise of time. 
This may be undesirable when the value of the product or information decreases with time, or when there is a competing diffusion and people get influenced by the product or information which reaches them first.

There is thus a natural trade-off between (a) using better knowledge of influence spread to increase the number of influenced nodes at the end of the diffusion and (b) the accompanying delay in the activation of seed sets from the second phase onwards.

For multi-phase diffusion to be implemented effectively, it is necessary that the company is able to observe the status of nodes in the social network (inactive, recently activated, or already activated), that is, the company needs to link its customers to the corresponding nodes in the social network. To make such an observation, it would be useful to get the online social networking identity (say Facebook ID) of a customer as soon as the customer buys the product. 
This could be done using a product registration website (say for activating warranty) where a customer, on buying the product, needs to login using a popular social networking website (say Facebook), or needs to provide an email address that can be linked to a Facebook ID.
Thus the time step when the node buys the product, is obtained, and hence the node can be classified as already activated or recently activated.

In this chapter, to obtain a firm grounding on multi-phase diffusion, we focus our attention on two-phase diffusion. We believe that many of the findings of this work carry over to multi-phase diffusion.
It is to be noted that the start of the second phase does not kill the diffusion that commenced in the first phase. When the second phase commences, the recently activated nodes (activated due to diffusion in the first phase) effectively become additional seed nodes for second phase (in addition to the seed nodes that are separately selected for second phase).

\begin{figure}[t]
\centering
\begin{tabular}{cc}
\includegraphics[scale=.65]{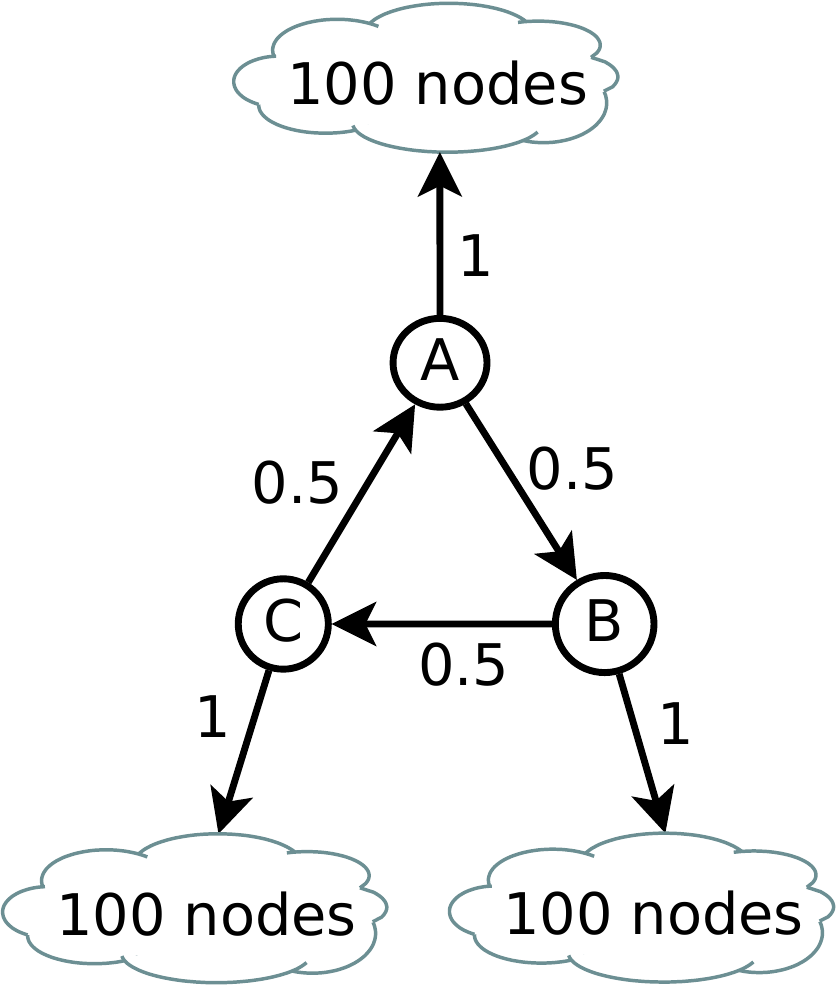}
&
\hspace{2cm}
\includegraphics[scale=.65]{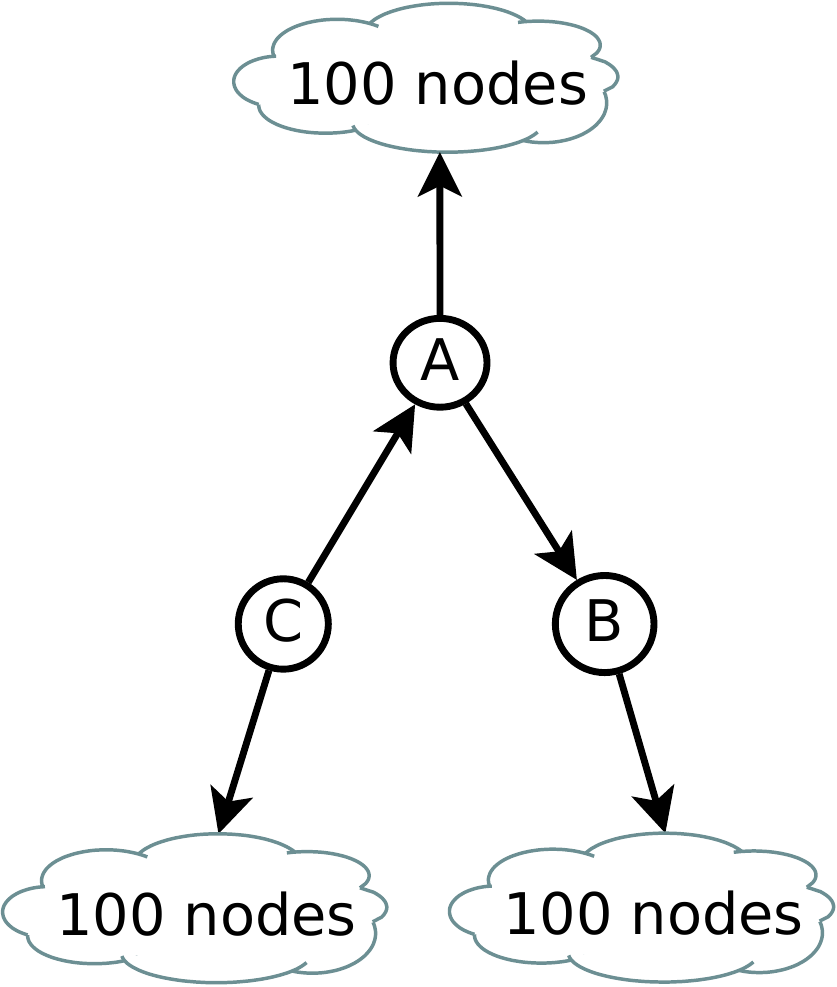}
\\ 
(a) input graph & 
\hspace{2cm}
(b) a live graph
\end{tabular}
\caption{\mbox{Multi-phase diffusion: a motivating example}}
\label{fig:motiv_mpid}
\end{figure}

\subsection{A Motivating Example}
\label{sec:example}
We now illustrate the usage of two-phase diffusion  with a simple stylized example.
Consider the graph in Figure~\ref{fig:motiv_mpid}(a) where the influence probabilities are as shown. Activation of node $A$ or $B$ or $C$ results in activation of 100 additional nodes each, in the following time step, with probability 1.
Consider a total budget of $k=2$. 
Assume that the live graph in Figure~\ref{fig:motiv_mpid}(b) is destined to occur (we do not have this information at time step 0 when we are required to select the seed nodes for triggering the diffusion).
Consider an influence maximization algorithm $\mathbb{A}$.

Let us study single-phase diffusion on this graph. Let $A$ and $B$ be the two seed nodes selected by 
algorithm $\mathbb{A}$ in time step 0. 
In time step 1 as per the IC model, 200 additional nodes get influenced owing to recently activated nodes $A$ and $B$.
Since the realized live graph is as shown in Figure~\ref{fig:motiv_mpid}(b), the diffusion does not proceed any further as there is no outgoing edge from the recently activated nodes to any inactive node. So the diffusion stops at time step \textbf{1}, with \textbf{202} influenced nodes.

For two-phase diffusion, let the total budget $k=2$ be split as $1$ each for the two phases,
and let the second phase be scheduled to start in time step 3 (the seed nodes for second phase are to be selected in time step 3). 
%
Now let us say that at time step 0, algorithm $\mathbb{A}$ selects node $A$ as the only seed node for triggering diffusion in the first phase. In time step 1, it influences its set of 100 nodes and also node $B$.
In the following time step (step 2), 
$B$'s set of 100 nodes get influenced. But more importantly, we know that $C$ is not influenced, thus deducing the absence of edge $BC$ in the live graph. So we are more certain about which live graph 
is likely to occur than we were in the beginning of the diffusion, as we have eliminated the possibilities of occurrence of live graphs containing edge $BC$.
Now based on this observation, algorithm $\mathbb{A}$ would select 
 $C$ as the seed node for second phase (in time step 3), which in turn, would influence its set of 100 nodes in the following time step. Thus the process stops at time step \textbf{4} with \textbf{303} influenced nodes. 
%
%
Note that during its first phase, the two-phase diffusion is expected to be slower than the single phase one, because of using only a part of the budget.

If algorithm $\mathbb{A}$ had selected $B$ as the seed node for the first phase, the diffusion observed after 2 time steps would have guided the algorithm to select $C$ as the seed node for the second phase, since it would influence $A$ with probability 0.5 (also, given that $B$ is already influenced, selection of $A$ as the seed node would not influence $C$), thus leading to all 303 nodes getting influenced.
In another case, if $C$ gets selected as the seed node for the first phase, it would influence all the nodes without having to utilize the entire budget of $k=2$. So multi-phase diffusion can also help achieve a desired spread with a reduced budget.
%

In short, the idea behind using the two-phase diffusion is that, for influence maximization algorithms (especially those predicting expected spread over live graphs), reducing the space of possible live graphs results in a better estimate of expected spread, leading to selection of a better seed set.
In fact, two-phase diffusion would facilitate an improvement while using a general influence maximization algorithm, owing to knowledge of already and recently activated nodes, and hence a refined search space for seed nodes to be selected for second phase. We discuss this point throughout this chapter.

\section{Contributions of this Chapter}
\label{sec:contrib_mpid}

With the objective of multi-phase influence maximization in social networks, this work makes the
following specific contributions.
\begin{itemize} 
\item Focusing on two-phase diffusion process in social networks under the IC model, we formulate an appropriate objective function that measures the expected number of influenced nodes, and investigate its properties. 
We then motivate and propose an alternative objective function for ease and efficiency of practical implementation. (Section~\ref{sec:problem_mpid})
\item We investigate different candidate  algorithms for two-phase diffusion including extensions of
existing algorithms that are popular for single phase diffusion. In particular, we propose the use of the cross entropy method and a Shapley value based method as promising algorithms for influence maximization in social networks.
Selecting seed nodes for the two phases using an influence maximization algorithm could be done in two natural ways: (a) myopic or (b) farsighted. (Section~\ref{sec:algo})
\item Using extensive simulations on real-world datasets, we study the performance of the proposed algorithms to get an idea how two-phase diffusion would perform, even when used most na\"ively. (Section~\ref{sec:simulations})
\item To achieve the best performance out of  two-phase diffusion, we focus on two constituent problems, namely, (a) {\em budget splitting}: how to split the total available budget between the two phases and (b) {\em scheduling}: when to commence the second phase. 
Through a deep investigation of the nature of our observations, we propose efficient algorithms for the combined optimization problem of budget splitting, scheduling, and seed sets selection.
We then present key insights from a detailed simulation study. (Section~\ref{sec:practical})
\item We conclude the chapter with (a) a note on how the value of diffusion would decay with time, (b) a note on subadditivity of the objective function, and (c) an overview of how two-phase diffusion could be used under the linear threshold model. (Section~\ref{sec:conclusion_mpid})
\end{itemize}

\section{Two Phase Diffusion: A Model and Analysis}
\label{sec:problem_mpid}

As mentioned earlier, we concentrate on two-phase diffusion in this chapter. 
Let $k$ be the total budget, that is, the sum of the number of seed nodes that can be selected in the two phases put together. At the beginning of the process (time step 0), suppose $k_1$ seed nodes are selected for the first phase and at time step, say $d$, $k_2$ (where $k_2 = k - k_1$) seed nodes are selected for the second phase. 
Our objective is to maximize the expected number of influenced nodes at the end of the two-phase diffusion process. 
In what follows, we assume $k_1,k_2,d$ to be given.
We study the problem of optimizing over these parameters in Section~\ref{sec:practical}.

\subsection{Objective Function}
\label{sec:objectivefn}

Let $X$ be a live graph obtained by sampling edges for a given graph $G$.
Let $\sigma ^X (S)$ be the number of nodes reachable from seed set $S$ in $X$, that is, the number of nodes influenced at the end of the diffusion process that starts at $S$, if the resulting live graph is $X$. Let $p(X)$ be the probability of occurrence of $X$. So the number of influenced nodes at the end of the process, in expectation, is
$\sigma (S) = \sum_X p(X) \sigma ^X (S)$.
%
It has been shown that $\sigma ^X (S)$, and hence $\sigma (S)$, are non-negative, monotone increasing, and submodular \cite{kempe2003maximizing}.

%
We now formulate an appropriate objective function that measures the expected number of influenced nodes at the end of two-phase diffusion.
Let $S_1$ be the seed set for the first phase
and $X$ be the live graph that is destined to occur ($X$ is not known at the beginning of diffusion, but we know $p(X)$ from edge probabilities in $G$).
Let $Y$ be the partial observation at time step $d$, owing to the observed diffusion. 
As we will be able to classify activated nodes as already activated and recently activated at time step $d$, we assume that $Y$ conveys this information.
That is, from $Y$, the set of already activated nodes $\mathcal{A}^Y$ and the set of recently activated nodes $\mathcal{R}^Y$ at time step $d$, can be determined.
Given $Y$, we can now update the probability of occurrence of a live graph $X$ by $p(X|Y)$.

Now at time step $d$, we should select a seed set that maximizes the final influence spread, considering that nodes in $\mathcal{R}^Y$
will also be effectively acting like seed nodes for second phase.
Let $S_2 ^{O(Y,k_2)}$ be an optimal set of $k_2$ nodes to be selected as seed set, given the occurrence of partial observation $Y$ (which implicitly gives $\mathcal{A}^Y,\mathcal{R}^Y$). We can also write the above optimal set as $S_2 ^{O(X,S_1,d,k_2)}$, as $Y$ can be uniquely obtained for a given $d$ and particular $X$ and $S_1$. So
for all $S_2 ' \subseteq N \setminus S_1$ such that $|S_2 ' | \leq k_2$ (note that it is optimal to have $|S_2 ' | = k_2$ owing to monotone increasing property of $\sigma(\cdot)$),
\begin{displaymath}
 \sum_{X} p(X|Y) \sigma^{X \setminus \mathcal{A}^Y} (\mathcal{R}^Y \cup S_2 ^{O(X,S_1,d,k_2)}) 
 \geq \sum_{X} p(X|Y) \sigma^{X \setminus \mathcal{A}^Y} (\mathcal{R}^Y \cup S_2 ')
\end{displaymath}
where ${X \setminus \mathcal{A}^Y}$ is the graph derived from $X$ by removing nodes belonging to $\mathcal{A}^Y$.
%
Now, adding $\sum_{X} p(X|Y) |\mathcal{A}^Y|$ on both sides to account for the already activated nodes in the first phase, we get
\begin{small}
\begin{displaymath}
\sum_{X} p(X|Y) \sigma^X (S_1 \cup S_2 ^{O(X,S_1,d,k_2)})  \geq \sum_{X} p(X|Y) \sigma^X (S_1 \cup S_2 ')
\end{displaymath}
\end{small}
We call this inequality, the {\em optimality of} $S_2 ^{O(X,S_1,d,k_2)}$ throughout this chapter.
So assuming that, given a $Y$, we will select an optimal seed set for the second phase, our objective is to select an optimal $S_1$ (seed set for first phase). Now, as $Y$ is unknown at the beginning of the first phase, the objective function, say $\mathbb{F}(S_1,d,k_2)$,
is an expected value with respect to all such $Y$'s.
Until Section~\ref{sec:practical}, we assume $k_2$ and $d$ to be given, and so we write $\mathbb{F}(S_1,d,k_2)$ as $f(S_1)$.
So,
\vspace{-7mm}
\begin{align}
f(S_1)&= 
\sum_{Y} p(Y)  \big\{ |\mathcal{A}^{Y}| + \sum_{X} p(X|Y) \sigma^{X \setminus \mathcal{A}^Y}(\mathcal{R}^Y \cup S_2 ^{O(Y,k_2)})   \big\}
\nonumber
\\ &= 
\sum_{Y} p(Y)  \big\{ |\mathcal{A}^{Y}| + \sum_{X} p(X|Y) \sigma^{X \setminus \mathcal{A}^Y}(\mathcal{R}^Y \cup S_2 ^{O(X,S_1,d,k_2)})   \big\}
\nonumber
\\ &= 
\sum_{Y} p(Y) \sum_{X} p(X|Y) \big\{ |\mathcal{A}^{Y}| +  \sigma^{X \setminus \mathcal{A}^Y}(\mathcal{R}^Y \cup S_2 ^{O(X,S_1,d,k_2)})   \big\}
\nonumber
\\ &= 
\sum_{Y} p(Y) \sum_{X} p(X|Y) \sigma^X (S_1 \cup S_2 ^{O(X,S_1,d,k_2)}) 
\label{eqn:basic}
\\ &=
\sum_{X} \sum_{Y} p(Y) p(X|Y) \sigma^X (S_1 \cup S_2 ^{O(X,S_1,d,k_2)})  
\nonumber
\end{align}
\vspace{-5mm}
\begin{equation}
\label{eqn:f}
\therefore \; f(S_1) = \sum_X p(X) \sigma^X (S_1 \cup S_2 ^{O(X,S_1,d,k_2)}) 
\end{equation}
Note that at time step $d$, the choice of $S_2 ^{O(X,S_1,d,k_2)}$ depends 
not only on $X$, but
on partial observation $Y$, and hence on all live graphs that could result from $Y$ (just as in single phase, choice of the best seed set depends on all live graphs that could result from the given graph $G$).
It is easy to prove on similar lines as \cite{kempe2003maximizing} that, the problem of maximizing $f(\cdot)$ is NP-hard.
%
We now show how to compute $f(\cdot)$ using an example.

\begin{figure}[h] 
\centering
\includegraphics[scale=.55]{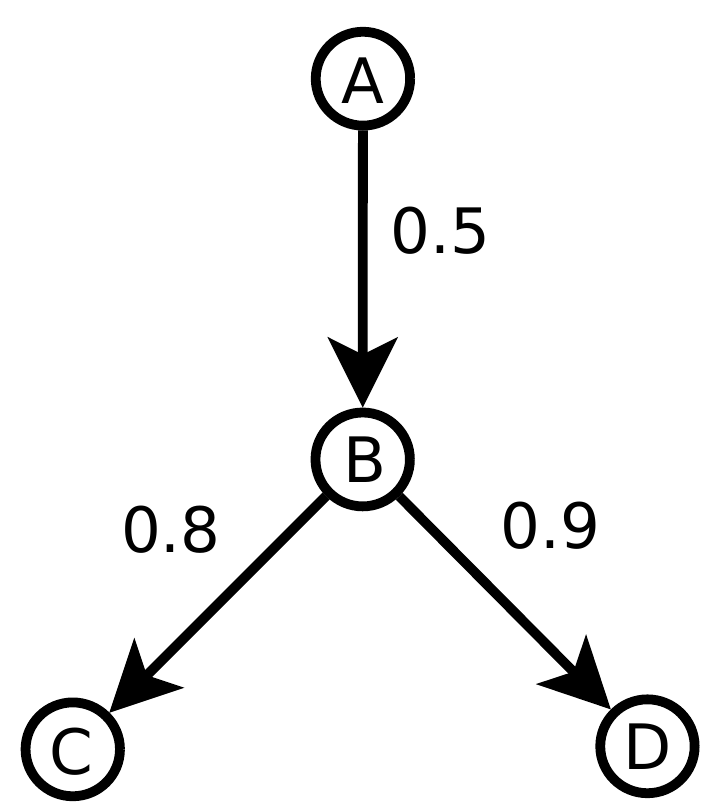}
\caption{An example network
}
\label{fig:submod_counter}
\end{figure}

\begin{example}
\label{eg:obj_fn}
%
(Figure~\ref{fig:submod_counter})
Consider 
$S_1=\{A\}$, $k_2=1$, and $d=1$.
Table~\ref{tab:mpid_example} lists the two possibilities of $Y$ ($\mathcal{A}^Y,\mathcal{R}^Y$) at $d=1$. 
The set $S_2^{O(Y,k_2)}$ is easy to compute for both the cases.
So $f(\{A\})= \mathbb{F}(\{A\},1,1) = \sum_X p(X) \sigma^X (\{A\} \cup S_2 ^{O(X,\{A\},1,1)}) = 3.8$.
\end{example}

\begin{table}[b!]
\begin{center}
\begin{tabular}{c|c|c|c|c|c}
\hline \hline  
\multicolumn{5}{c}{\T \B ~~~~~~~~~~~~~~~~~~$S_1=\{A\}, k_2=1, d=1$} \\
\hline \T \B \hspace{-.4cm}
  \multirow{3}{*}{$X$} & \multirow{3}{*}{$p(X)$} & \multicolumn{2}{c|}{\T \B $Y$} & \multirow{3}{*}{$S_2^{O(Y,k_2)}$}  & \multirow{3}{*}{$f(S_1)$}  \\ 
  \cline{3-4} \hspace{-.4cm} 
 & & \multirow{2}{*}{$\mathcal{A}^Y$} & \multirow{2}{*}{$\mathcal{R}^Y$}  & & \\
 \hspace{-.4cm} 
 & & & & & \\
  \hline \hline \T \B  \hspace{-.4cm}
  $\{AB,BC,BD\}$ & 0.36 & \multirow{4}{*}{$\{A\}$}  & \multirow{4}{*}{$\{B\}$} &  \multirow{4}{*}{$\{C\}$} & 4\\
   \cline{1-2}\cline{6-6} \hspace{-.4cm} \T \B
  $\{AB,BC\}$  & 0.04 &  & &   & 3 \\
   \cline{1-2}\cline{6-6} \hspace{-.4cm} \T \B
     $\{AB,BD\}$ & 0.09 & & &  & 4 \\
   \cline{1-2}\cline{6-6} \hspace{-.4cm} \T \B
     $\{AB\}$ & 0.01 &  & &  & 3\\
  \hline \T \B \hspace{-.4cm}
  $\{BC,BD\}$ &  0.36 & \multirow{4}{*}{$\{A\}$} & \multirow{4}{*}{$\{\}$} & \multirow{4}{*}{$\{B\}$} 
  &  4 \\
  \cline{1-2}\cline{6-6} \hspace{-.4cm} \T \B
   $\{BC\}$ & 0.04  &  &  & & 3\\
   \cline{1-2}\cline{6-6} \hspace{-.4cm} \T \B
   $\{BD\}$ & 0.09  & &  & & 3\\
   \cline{1-2}\cline{6-6} \hspace{-.4cm} \T \B
   $\{\}$ & 0.01  & &  & & 2\\ 
   \hline \hline
\end{tabular}
\end{center}
\caption{Table to aid the computation of objective function for Example~\ref{eg:obj_fn}}
\label{tab:mpid_example}
\end{table}


\subsection{Properties of the Objective Function}
\label{sec:props}

\begin{property}
\label{prop:monotone}
$f(\cdot)$ is non-negative and monotone increasing.
\end{property}
\begin{proof}
$f(\cdot)$ is non-negative since $\sigma ^X (\cdot)$ is non-negative.
Consider $S_1 \subset T_1$. Then,
\begin{align*}
f(T_1)
&=
\; \sum_{X} p(X) \sigma^X (T_1 \cup T_2 ^{O(X,T_1,d,k_2)})
\\ & \geq
\; \sum_{X} p(X)   \sigma^X (T_1 \cup S_2 ^{O(X,S_1,d,k_2)})
   \\ & \geq 
\; \sum_{X} p(X)   \sigma^X (S_1 \cup S_2 ^{O(X,S_1,d,k_2)})
\;   =f(S_1)
\end{align*}
The first inequality is from optimality of $T_2 ^{O(X,T_1,d,k_2)}$  and the second one from monotonicity of $\sigma ^X (\cdot)$.
\end{proof}

\noindent
As $f(\cdot)$ is monotone increasing and $|S_1| \leq k_1$, given a fixed $k_1$, it is optimal to select $S_1$ such that $|S_1|=k_1$.

\begin{property}
\label{prop:nonsubmodular}
\mbox{$f(\cdot)$ is neither submodular nor supermodular.}
\end{property}
\begin{proof}
We prove this using a simple counterexample network in Figure~\ref{fig:submod_counter}.
Consider $d = 3$ and $k_2 = 1$. 

Considering $S_1 = \{\}$, $T_1 = \{D\}$, $i = C$, 
we get $f(S_1 \cup \{i\} ) = 2.95$, $f(S_1) = 2.7$, $f(T_1 \cup \{i\} ) = 3.5$, $f(T_1) = 2.9$.
So we have 
$f(S_1 \cup \{i\} ) - f(S_1) < f(T_1 \cup \{i\} ) - f(T_1)$
for some $T_1$, $S_1 \subset T_1$,  $i \notin T_1$, which proves non-submodularity of $f(\cdot)$.

Considering $S_1 = \{\}$, $T_1 = \{B\}$, $i = A$, 
we get $f(S_1 \cup \{i\} ) = 3.84$, $f(S_1) = 2.7$, $f(T_1 \cup \{i\} ) = 3.98$, $f(T_1) = 3.7$.
So we have 
$f(S_1 \cup \{i\} ) - f(S_1) > f(T_1 \cup \{i\} ) - f(T_1)$
for some $T_1$, $S_1 \subset T_1$,  $i \notin T_1$, which proves its non-supermodularity.
\end{proof}

\begin{remark}
It was observed using simulations on the test graphs that the diminishing returns property is satisfied in most cases. That is, for most $S_1$ and $T_1$ such that $S_1 \subset T_1 \subset N$ and  $i \in N\setminus T_1$, it was observed that $f(S_1 \cup \{i\} ) - f(S_1) \geq f(T_1 \cup \{i\} ) - f(T_1)$.
Furthermore, it can be shown that $f(\cdot)$ is subadditive; we provide a note regarding this in Section~\ref{sec:subadditive}.
\end{remark}

Owing to NP-hardness of the single phase influence maximization problem, it is impractical to compute $S_2^{O(X,S_1,d,k_2)}$ in Equation~(\ref{eqn:f}). We surmount this difficulty by maximizing an alternative function instead of $f(\cdot)$. To emphasize this point, note that this impractical computation is for finding the objective function value itself, which makes finding an optimal $S_1$, a computationally infeasible task. So the alternative function must be several orders of magnitude faster to compute than $f(\cdot)$.
We now address this problem.

\subsection{An Alternative Objective Function}
\label{sec:compute}

\subsubsection{Using Greedy Hill-climbing Algorithm}
\label{sec:greedy}
The greedy hill-climbing algorithm selects nodes one at a time, each time choosing a node that provides the largest marginal increase in the function value, until the budget is exhausted.

Now given the occurrence of the partial observation $Y$,
let $S_2 ^{G(Y,k_2)} = S_2 ^{G(X,S_1,d,k_2)}$ be a set of size $k_2$ obtained using the greedy hill-climbing algorithm.
%
Let
\begin{displaymath}
\label{eqn:g}
g(S_1) \overset{\mathcal{MC}}=
\sum_X p(X) \sigma^X (S_1 \cup S_2 ^{G(X,S_1,d,k_2)}) 
\end{displaymath}
that is, $g(\cdot)$ is obtained using Monte-Carlo (${\mathcal{MC}}$) simulations.

\begin{theorem} 
\label{thm:nemhauser}
For a non-negative, monotone increasing, submodular function $\mathcal{F}$, let $S^G$ be a set of size $k$ obtained using greedy hill-climbing. Let $S^O$ be a set that maximizes the value of $\mathcal{F}$ over all sets of cardinality $k$. Then 
for any $\epsilon>0$, there is a $\gamma>0$ such that by using $(1 + \gamma)$-approximate values for $\mathcal{F}$, we obtain a $(1-\frac{1}{e}-\epsilon)$-approximation~\cite{kempe2003maximizing}.
\end{theorem}

The $(1 + \gamma)$-approximate values for $\mathcal{F}$ with small $\gamma$ can be obtained using sufficiently large number of Monte-Carlo simulations for measuring influence spread in the IC model~\cite{kempe2003maximizing}.

\begin{lemma}
\label{lem:f_approx_g}
$g(\cdot)$ gives a $\left( 1-\frac{1}{e}-\epsilon \right)$ approximation to $f(\cdot)$.
\end{lemma}
\begin{proof}
Let $\Phi_T (S) = \sigma (T \cup S)$ and $\Phi_T^X (S) = \sigma ^X (T \cup S)$.
It can be easily shown that $\Phi_T ^X (S)$, and hence $\Phi_T (S)$, are non-negative, monotone increasing, and submodular.
%
So we have
\begin{align*}
g(S_1)
&\overset{\mathcal{MC}}{=} 
\sum_X p(X) \Phi_{S_1}^X (S_2 ^{G(X,S_1,d,k_2)}) 
 \\ & \geq 
\text{\scriptsize{$\left( 1-\frac{1}{e}-\epsilon \right)$}}  \sum_X p(X) \Phi_{S_1}^X (S_2 ^{O(X,S_1,d,k_2)}) 
 \\ &=
\text{\scriptsize{$ \left( 1-\frac{1}{e}-\epsilon \right)$}} f(S_1) 
\end{align*}
where the first inequality results from Theorem~\ref{thm:nemhauser}.
\end{proof}

So one can aim to maximize $g(\cdot)$ instead of $f(\cdot)$.
However, greedy hill-climbing algorithm itself is expensive in terms of running time (even after diffusion specific optimizations such as in \cite{chen2009efficient}), so we aim to maximize yet another function which would act as a proxy for $g(\cdot)$.
%

%
%

\subsubsection{Using Generalized Degree Discount Heuristic}
\label{sec:gdd}
Consider the process of selecting seed nodes one at a time.
At a given time in the midst of the process, let $\mathcal{X}=$ set of in-neighbors of node $v$ already selected as seed nodes and $\mathcal{Y}=$ set of out-neighbors of $v$ not yet selected as seed nodes.
We develop Generalized Degree Discount (GDD) Heuristic
as an extension to 
the argument for Theorem~2 in \cite{chen2009efficient}: 
if $v$ is not (directly) influenced by any of the already selected seeds, which occurs with probability
$ \prod_{x \in \mathcal{X}} (1-p_{xv}) $, then the additional expected number of nodes that it influences directly (including itself) is $\left(1+\sum_{y \in \mathcal{Y}} p_{vy}\right)$.
So until the budget is exhausted, GDD heuristic iteratively selects a node $v$ having the largest value of
\begin{equation}
\label{eqn:gdd_value}
w_v = \left( \prod_{x \in \mathcal{X}} (1-p_{xv}) \right) \left( 1+\sum_{y \in \mathcal{Y}} p_{vy} \right)
\end{equation}
Its time complexity is $ O( k n \Delta )$, where $\Delta$ is the maximum degree in the graph.

Given the occurrence of the partial observation $Y$,
let $S_2 ^{H(Y,k_2)} = S_2 ^{H(X,S_1,d,k_2)}$ be a set of size $k_2$ obtained using the 
GDD heuristic.
Let 
\begin{equation}
\label{eqn:h}
h(S_1) \overset{\mathcal{MC}}{=} 
\sum_X p(X) \sigma^X (S_1 \cup S_2 ^{H(X,S_1,d,k_2)}) 
\end{equation}
%
%
We conducted simulations for checking how well $h(\cdot)$ acts as a proxy for $g(\cdot)$, using both weighted cascade and trivalency models. We observed the following. 
\begin{observation}
\label{obs:gdd_approx_greedy}
For almost all $S,T$ pairs: 
\begin{enumerate}
\item[(a)] If $g(T)>g(S)$, then $h(T)>h(S)$ (in particular, this is satisfied for almost all pairs of sets that give excellent objective function values), which ensures that the selected seed set remains unchanged in most cases when we have $h(\cdot)$ as our objective function instead of $g(\cdot)$.
\item[(b)] $\frac{h(S)}{h(T)} \approx \frac{g(S)}{g(T)}$, which ensures that the seed set selected by algorithms, which implicitly depend on the ratios of the objective function values given by any two sets, remains unchanged in most cases when we have $h(\cdot)$ as our objective function instead of $g(\cdot)$; two of the algorithms we consider, namely, FACE (Section~\ref{sec:ce_method}) and SPIC (Section~\ref{sec:shapley_method}) belong to this category of algorithms. 
\end{enumerate}
\end{observation}

\begin{remark}
One could question, why not use a function $\hat h(S_1) \overset{\mathcal{MC}}{=}  \sum_X p(X) \sigma^X (S_1 \cup S_2 ^{\hat H(X,S_1,d,k_2)})$ instead of $h(S_1)$, where $S_2 ^{\hat H(X,S_1,d,k_2)}$ is a set of size $k_2$ obtained using the PMIA algorithm (it has been observed that PMIA performs very close to greedy algorithm on practically all relevant datasets, while running orders of magnitude faster \cite{chen2010scalable}). However, it is to be noted that, though PMIA is an efficient algorithm for single phase influence maximization, it is highly undesirable to use it for computation of objective function value alone. On the other hand, GDD is orders of magnitude faster than PMIA. Though we use moderately sized datasets for making Observation~\ref{obs:gdd_approx_greedy}, we could stretch the size of datasets by aiming to observe how well $h(S_1)$ acts as a proxy for $\hat h(S_1)$.
\end{remark}

Owing to the above justifications, we aim to maximize $h(\cdot)$ instead of $f(\cdot)$, for two-phase influence maximization in the rest of this chapter.

\section{Algorithms for Two-Phase Influence Maximization}
\label{sec:algo}

In the previous section, we formulated the objective function for two-phase influence maximization $f(\cdot)$ and studied its properties. In addition to their theoretical relevance, these properties have implications for as to which algorithms are likely to perform well. We present them while describing the algorithms.


Let graph $G = (N, E, \mathcal{P})$ be the input directed graph where $N$ is the set of $n$ nodes, $E$ is the set of $m$ edges, and $\mathcal{P}$ is the set of probabilities associated with the edges. Let $k \leq n$ be the total budget (sum of the number of seed nodes for both the phases put together), $k_1 \leq k$ be the budget for the first phase ($k_1=k$ corresponds to single-phase diffusion), $d$ be the time step in which the second phase starts,
and $k_2=k-k_1$ be the budget for the second phase.
%
Let $\mathcal{T}$ be the time taken for computing the objective function value for a given set.
%
\begin{itemize} 
\item For single phase objective function $\sigma(\cdot)$, $\mathcal{T} = O(m \mathcal{M})$,
where $\mathcal{M}$ is the number of Monte-Carlo simulations.
\item For two-phase objective function $h(\cdot)$, $\mathcal{T} = O(k_2 n \Delta m \mathcal{M}_1 \mathcal{M}_2)$,
where $\mathcal{M}_1$ and $\mathcal{M}_2$ are the numbers of Monte-Carlo simulations for first and second phases, respectively, and $\Delta$ is the maximum out-degree in the graph.
\end{itemize}

\subsection{Candidate Algorithms for Seed Selection}

Now we present the algorithms that we consider for seed selection for single phase influence maximization, 
which we later explain how to extend to the two-phase case.

\subsubsection{Greedy Algorithm}

As described earlier, the greedy (hill-climbing) algorithm for maximizing a function $\mathcal{F}$, selects nodes one at a time, each time choosing a node that provides the largest marginal increase in the value of $\mathcal{F}$, until the budget is exhausted.
Its time complexity is $O(kn\mathcal{T})$.
As noted earlier, though our two-phase objective function is not submodular, we observed that the diminishing returns property was satisfied in most cases; so even though the condition in Theorem~\ref{thm:nemhauser} is not satisfied, the greedy algorithm is likely to perform well.
Further, unlike in the case of single phase influence maximization, we cannot use CELF optimization for the two-phase case owing to its objective function being non-submodular
(though satisfiability of the diminishing returns property in most cases may make it a reasonable approach,
we do not use it so as
to preserve performance accuracy of the greedy algorithm for two-phase influence maximization).


\subsubsection{Single/Weighted Discount Heuristics (SD/WD)}

The single discount (SD) heuristic~\cite{chen2009efficient} for a graph $G$ can be described as follows: select the node having the largest number of outgoing edges in $G$, then remove that node and all of its incoming and outgoing edges to obtain a new graph $G'$, again select the node having the largest number of outgoing edges in the new graph $G'$, and continue until the budget is exhausted.
%
Weighted discount (WD) heuristic is a variant of SD heuristic where, sum of outgoing edge probabilities is considered instead of number of outgoing edges.
The time complexity of these heuristics is $ O( k n \Delta )$.
These heuristics run extremely fast and hence can be used for efficient seed selection for very large networks. 


\subsubsection{Generalized Degree Discount (GDD) Heuristic}
The generalized degree discount (GDD) heuristic is as described in Section~\ref{sec:gdd}.


\subsubsection{PMIA}
\label{sec:pmia}
This heuristic, based on the arborescence structure of influence spread,
is shown to perform close to greedy algorithm and runs orders of magnitude faster \cite{chen2010scalable}.

\subsubsection{Fully Adaptive Cross Entropy Method (FACE)}
\label{sec:ce_method}

It has been shown that the cross entropy (CE) method provides a simple, efficient, and general method for solving combinatorial optimization problems~\cite{de2005tutorial}.
In our context, the CE method involves an iterative procedure where each iteration consists of two steps, namely,
(a) generating data samples (a vector consisting of a sampled 
candidate seed set) according to a specified distribution and
(b) updating the distribution based on the sampled data to produce better samples in the next iteration.
 We use an adaptive version of the CE method called the {\em fully adaptive cross entropy} (FACE) algorithm~\cite{de2005tutorial}.
 Its time complexity is $O(n\mathcal{T} \mathcal{I})$, where $\mathcal{I}$ is the number of iterations taken for the algorithm to terminate.
 However, the running time can be drastically reduced for single phase diffusion using preprocessing similar to that for greedy algorithm as in \cite{chen2009efficient}.
An added advantage of this algorithm is that it would not only find an optimal seed set, but also implicitly determine how to split the total budget between the two phases and also the delay after which the second phase should be triggered (see Section~\ref{sec:practical}).
%

\subsubsection{Shapley Value based - IC Model (SPIC)}
\label{sec:shapley_method}

We consider a Shapley value based method because it is shown to perform well even when the  objective function is non-submodular~\cite{narayanam2010shapley}.
It has been observed that,
in order to obtain the seed nodes after computing Shapley values of the nodes, 
some post-processing is required.
We present a number of post-processing methods in Appendix~\ref{app:ppmethods}.
%
As the post-processing step under the IC model, we propose the following discounting scheme for SPIC: 
\\
(a) Since node $x$ would get directly activated because of node $y$ with probability $p_{yx}$, we discount the value of $x$ by multiplying it with $(1-p_{yx})$ whenever any of its in-neighbors $y$ gets chosen in the seed set. 
%
\\
(b) As node $z$ influences node $y$ directly with probability $p_{zy}$, it gets a fractional share of $y$'s value (since $z$ would be influencing other nodes indirectly, through $y$).
So when $y$ is chosen in the seed set, we subtract $y$'s share ($p_{zy}\phi_y$ where $\phi_y$ is the value of $y$ during its selection) from the current value of $z$.
If the value becomes negative because of oversubtraction, we assign zero value to it.
%
%
\\
A node, not already in the seed set, with the highest value after discounting, is then added to the seed set in a given iteration.
In our simulations, we observed that this discounting scheme outperforms the SPIN algorithm (choosing seed nodes one at a time while eliminating neighbors of already chosen nodes~\cite{narayanam2010shapley}).
Assuming $O(n)$ permutations for approximate computation of Shapley value~\cite{narayanam2010shapley} (since exact computation is \#P-hard),
the algorithm's time complexity is approximately 
$O(n \mathcal{T})$.
It is to be noted, however, that the SPIN algorithm~\cite{narayanam2010shapley} is not scalable to very large networks even for single phase influence maximization~\cite{chen2010scalable}, and so isn't SPIC.

\subsubsection{Random Sampling and Maximizing (RMax)}

Here, 
we sample $O(n)$ number of sets that satisfy the budget constraint, 
and then assign that set as the seed set which gives the maximum function value among the sampled sets.
Note that this method is different from the random set selection method~\cite{kempe2003maximizing}, 
where only one sample is drawn.
Its time complexity is $O(n\mathcal{T})$.
We consider this method as it is very generic and agnostic to the properties of the objective function, and can be used for optimizing functions with arbitrary or no structure. This method is likely to perform well when the number of samples is sufficiently large.

\subsection{Extension of Algorithms to Two-phase Influence Maximization}

Now we present how the aforementioned single phase influence maximization algorithms can be extended for two-phase influence maximization.
Let $\mathcal{F}_1 (\cdot)$ and $\mathcal{F}_2 (\cdot)$ be objective functions corresponding to 
seed selection in
 first and second phases, respectively.
Consider an influence maximization algorithm $\mathbb{A}$.

%
\begin{algorithm}[t!]

\KwIn{$G = (N, E, \mathcal{P})$, $k_1$, $k_2$, $d$}
\KwOut{Seed nodes for the first and second phases at time steps 0 and $d$, respectively }

\textbf{First phase:} 

~ Find set of size $k_1$ using algorithm $\mathbb{A}$ for maximizing $\mathcal{F}_1 (\cdot)$ on $G$\;

~ Run the diffusion using IC model until time step $d$ \;

\textbf{Second phase:} 

~ On observing $Y$ at time step $d$, construct $G^d$ from $G$ by deleting $\mathcal{A}^Y$\;

~ With $\mathcal{R}^Y$ forming partial seed set, find set of size $k_2$ using $\mathbb{A}$ for maximizing $\mathcal{F}_2 (\cdot)$ on $G^d$\;

~ Continue running the diffusion using IC model until no further nodes can be influenced\;

\caption{{Two-phase general algorithm (IC model)}}
\label{alg:generic_algo}

\end{algorithm}

We explore two special cases of 
Algorithm \ref{alg:generic_algo}
(the notation can be recalled from Sections~\ref{sec:objectivefn} and \ref{sec:compute}):
\begin{tabbing}
1.  Farsighted \= : $\;\mathcal{F}_1 (S_1) = h(S_1) \; , \;\mathcal{F}_2 (S_2) = \sigma(\mathcal{R}^Y \cup S_2)$ \\
2.   Myopic \> : $\; \mathcal{F}_1 (S_1) = \sigma(S_1) \; , \;\mathcal{F}_2 (S_2) = \sigma(\mathcal{R}^Y \cup S_2)$ 
\end{tabbing}
As explained earlier, the second phase objective function assumes that $\mathcal{R}^Y$ forms a partial seed set, hence the above form of $\mathcal{F}_2(\cdot)$.
The farsighted objective function looks ahead and accounts for the fact that there is going to be a second phase and hence attempts to maximize $h(\cdot)$, while the myopic function does not.
Note that heuristics such as PMIA, GDD, WD, SD do not consider the actual objective function for seed selection, and so the myopic and farsighted algorithms are the same for these heuristics.

We now formally prove the effectiveness of two-phase diffusion for influence maximization.

\begin{theorem}
For any given values of $k_1$ and $k_2$, the expected influence achieved using optimal two-phase algorithm is at least as much as that achieved using optimal single phase one. 
\end{theorem}
\begin{proof}
Let $S^*$ be the optimal seed set of cardinality $k=k_1+k_2$ selected in single phase diffusion. Let sets $S_1^*$ and $S_2^*$ be such that $|S_1|=k_1$, $|S_2|=k_2$, $S^*=S_1^* \cup S_2^*$, and $S_1^* \cap S_2^* =\emptyset$. 
Now assuming 
any $d$, it is clear from the optimality of $S_2 ^{O(X,S_1^*,d,k_2)}$ (see derivation of $f(\cdot)$ preceding Equation~(\ref{eqn:f})) that,
\begin{small}
\begin{align*}
\sum_X p(X) \sigma^X (S_1^* \cup S_2 ^{O(X,S_1^*,d,k_2)}) \geq \sum_X p(X) \sigma^X (S_1^* \cup S_2^* ) 
\end{align*}
\end{small}
Note that 
%
the left hand side is $f(S_1^*)$ (Equation~(\ref{eqn:f})) and right hand side is $\sigma(S^*)$. So we have,
\begin{small}
\begin{align*}
\max_{S_1} f(S_1) \geq f(S_1^*) \geq \sigma(S^*)
\end{align*}
\end{small}
The leftmost
and rightmost expressions are the expected spreads using 
 two-phase
and single phase optimal algorithms, respectively, hence the result.
Note that this holds for any $d$.
\end{proof}

\section{A Study to Demonstrate Efficacy of Two-phase Diffusion}
\label{sec:simulations}

In this section, we study how much improvement one can expect by diffusing information in two phases over a social network, even with the simplest of approaches. To start with, we assume that $k_1,k_2,d$ are known and our objective is to find the seed sets for the two phases (we study the problem of optimizing over these parameters in Section~\ref{sec:practical}). 
As a simple and na\"ive first approach, we consider an equal budget split between the two phases, that is, $k_1=k_2=\frac{k}{2}$. 
Furthermore, we consider $d=D$, where $D$ is the length of the longest path in the network, so that by time step $D$, the first phase would have completed its diffusion. In practice, $D$ could be the maximum delay that we are ready to incur in absence of any temporal constraints.
Intuitively, it is clear that one should wait for as long as possible before selecting the seed nodes for second phase, as it would give a larger observation and a reduced search space. We now prove this formally.

\begin{lemma}
\label{lem:high_d_better}
For any given values of $k_1$ and $k_2$, the number of nodes influenced using an optimal two-phase influence maximization algorithm is a non-decreasing function of $d$.
\end{lemma}
\begin{proof}
Starting from a given first phase seed set $S_1$, let 
$Y_i$'s be the partial observations at time step $d$.
Also, let $Y_{ij}$'s be the partial observations at time step $d^+ > d$ resulting from a given $Y_i$ at time step $d$. 
%
%
By enumerating the partial observations at time step $d$, the expected number of nodes influenced at the end of diffusion, as given in Equation~(\ref{eqn:basic}), can be written~as
\begin{align*}
&
\sum_{i} p(Y_i) \sum_{X} p(X|Y_i) \sigma^X (S_1 \cup S_2 ^{O(Y_i,k_2)}) 
\\&=
\sum_i \sum_{j} p(Y_{ij}) \sum_{X} p(X|Y_{ij}) \sigma^X (S_1 \cup S_2 ^{O(Y_i,k_2)}) 
\\&\leq
\sum_i \sum_{j} p(Y_{ij}) \sum_{X} p(X|Y_{ij}) \sigma^X (S_1 \cup S_2 ^{O(Y_{ij},k_2)}) 
\end{align*}
which is the expected number of nodes influenced at the end of diffusion, if the second phase starts at time step $d^+ > d$.
The last inequality results from the optimality of $S_2 ^{O(Y_{ij},k_2)}$ for partial observation $Y_{ij}$.
\end{proof}

The following result now follows directly.

\begin{theorem}
\label{thm:D_is_opt}
For any given values of $k_1$ and $k_2$, the number of nodes influenced using an optimal two-phase influence maximization algorithm is maximized when $d=D$.
\end{theorem}

\begin{remark}
Determining $D$ exactly may be infeasible in practice. For instance, checking whether the first phase has completed its diffusion requires polling at every time step. Also, finding the length of the longest path in the network is known to be an NP-hard problem. 
However, for all practical purposes, $D$ can be approximated by a large enough value based on the network in consideration.
\end{remark}

\subsection{Simulation Setup}

For computing the objective function value and evaluating performance using single phase diffusion, we ran $10^4$ Monte-Carlo iterations (standard in the literature). To set a balance between running time and variance, we ran $10^3$ Monte-Carlo iterations for each of the phases in two-phase diffusion (equivalent to $10^6$ live graphs); the observed variance was negligible. 

As mentioned earlier, for transforming an undirected and unweighted network (dataset) into a directed and weighted network for studying the diffusion process, we consider two popular, well-accepted special cases of the IC model, namely, the {weighted cascade (WC) model} and the {trivalency (TV) model}.
We first conduct simulations on the Les Miserables (LM) dataset  \cite{knuth1993stanford} consisting of 77 nodes and 508 directed edges in order to study the performances of computationally intensive farsighted algorithms for two-phase influence maximization.
%
%
%
%
For studying two-phase diffusion on a larger dataset, we consider an academic collaboration network obtained from co-authorships in the ``High Energy Physics - Theory'' papers published on the e-print arXiv from 1991 to 2003. It contains 15,233 nodes and 62,774 directed edges, and is popularly denoted as NetHEPT. This network exhibits many structural features of large-scale social networks  and is widely used for experimental justifications, for example, in \cite{kempe2003maximizing, chen2009efficient, chen2010scalable}.
We also conducted experiments on a smaller collaboration network Hep-Th having 7,610 nodes and 31,502 directed edges \cite{newman2001structure}. As the results obtained were very similar, we present the results for only the NetHEPT dataset.
%
%
%
For two-phase diffusion, as a na\"ive first approach as mentioned earlier, we consider equal budget split, $k_1=k_2=\frac{k}{2}$, and $d=D$. 

\begin{remark}
In the two-phase influence maximization problem, seed selection is not computationally intensive, but seed evaluation is. At the end of the first phase, only one $Y$ is possible in practice; however, for the purpose of evaluation as part of simulations, we need to consider $\mathcal{M}_1$ (Monte-Carlo iterations for first phase) number of $Y$'s. This severely restricts the size of network under study. We believe the NetHEPT dataset suffices for our study owing to its social networks-like features and its wide usage for experimentation in the literature. 
\end{remark}

\begin{remark}
The simulations can also be run using a single level of Monte-Carlo iterations instead of two levels as described above.
For instance, instead of deciding the diffusion over each edge dynamically as in the above approach, one can decide an entire live graph in advance so that there is no requirement of separate Monte-Carlo iterations for the two phases.
However, the number of Monte-Carlo iterations (live graphs) required to compute the value with same variance as the above approach would be $\Theta(\mathcal{M}_1 \mathcal{M}_2)$.
\end{remark}

We now list the parameter values for the considered algorithms, specifically for the LM dataset.
For the detailed FACE algorithm, the reader is referred to \cite{de2005tutorial}.
We initialize the method with distribution $(\frac{\gamma}{n},\ldots,\frac{\gamma}{n})$, that is, each node has a probability of $\frac{\gamma}{n}$ of getting selected in any sample set in the first iteration (where $\gamma$ is the budget which is $k,k_1,k_2$ for single phase diffusion, first phase, and second phase, respectively).
In any iteration, the number of samples (satisfying budget constraint) is bounded by
 $\mathcal{N}_{\text{min}}=n$ and
$\mathcal{N}_{\text{max}}=20n$, the number of elite samples (samples that are deemed to have good enough function value) is
$\mathcal{N}_{\text{elite}}=\lceil \frac{n}{4} \rceil$.
%
We use a weighted update rule for the distribution where, in any given iteration, the weight of any elite sample is proportional to its function value. 
The smoothing factor (telling how much weight is to be given to the current iteration as against the previous iterations) that we consider is $\alpha = 0.6$.
%
In our simulations, we observed that in most cases, the FACE algorithm converged in 5 iterations (extending till 7 at times) by giving a reliable solution ({\em reliable} refers to the case wherein the method deduces that it has successfully solved the problem). Also, the total number of samples drawn in any iteration was $n$ in almost all cases (it did not exceed $2n$ in any iteration).
That is, the total number of samples over all iterations was approximately $5n$.
So for direct comparison with SPIC and RMax, we consider $5n$ permutations in order to compute the approximate Shapley values of all the nodes~\cite{narayanam2010shapley},
and $5n$ sampled sets for RMax.

\begin{table}[t!]
\centering
\begin{tabular}{c|c|c|c|c|c|c}
\hline \hline
\T \B
\multirow{4}{*}{\hspace{-4mm} Method} & \multicolumn{3}{c|}{Expected spread} &  \multicolumn{3}{c}{Running time for} \\
\T \B
	& \multicolumn{3}{c|}{} & \multicolumn{3}{c}{seed selection (seconds)} \\
\cline{2-7}
\T \B
	& Single & Two- & \%  & Single & Myopic & Farsighted  \\
\T \B
					& phase & phase & gain & phase  & 2-phase & 2-phase \\
\hline
\T \B
 FACE 	& 46.2 & 50.7 &  9.7 & 15 & 29  & 1209 \\
\T \B
  SPIC 	& 45.9 & 50.4 &  9.8 & 16 & 31   & 1272 \\
\T \B
  Greedy & 46.2 & 49.7 &  7.6 & 10  &  11 &  390  \\
\T \B
  PMIA 	& 46.2 & 49.4 &  6.9 & 0.2 & 0.2 & 0.2  \\
\T \B
  GDD 	& 45.8 & 49.3 &  7.6 & 0.002 & 0.002 & 0.002 \\
\T \B
  WD 	& 45.7 & 48.7 &  6.6 & 0.002 & 0.002 & 0.002 \\
\T \B
  SD	& 40.5 & 44.5 &  9.9 & 0.002 & 0.002 & 0.002  \\
\T \B
  RMax 	& 35.9 & 46.6 &  29.8 & 6 & 12  & 751 \\
\hline \hline
\end{tabular}
\caption{{Gain of two-phase diffusion 
over single phase} one on LM dataset (WC model) with ${k = 6, k_1 = k_2 = 3, d = D}$}
\label{tab:percent_improve}
\end{table}

\subsection{Simulation Results}

Throughout the rest of this chapter, we present results for a few representative settings. We have conducted simulations over a large number of settings and the results presented here are very general in nature.

\begin{observation}
FACE algorithm is very effective for single phase influence maximization, performing at par with greedy and PMIA or better for most values of $k$. SPIC also performs almost at par with them. To justify the effectiveness of two-phase diffusion process, it was necessary to consider these high performing single phase algorithms. Furthermore, GDD heuristic performs very closely to these algorithms, while taking orders of magnitude less time.
\end{observation}

Table~\ref{tab:percent_improve} shows the improvement of the na\"ive two-phase diffusion over single phase one for the considered algorithms on the LM dataset (WC model).
The performances of myopic and their corresponding farsighted algorithms were observed to be almost same (the maximum difference in the expected spread was observed to be 0.2 on the scale of 77 nodes),
so they share a common column for the expected spread.
These results, in conjunction with other results for $k_1\neq k_2$ and $d<D$ (which are not presented here), show that the myopic algorithms perform at par with the farsighted ones, while running orders of magnitude faster.
A possible reason for the excellent performance of myopic algorithms is that, the first set of $k_1$ nodes selected by most influence maximization algorithms, are generally the ones which would give a large enough observation and a well refined search space for the second phase seed set. 
Also, as mentioned earlier, there is no distinction between myopic and farsighted algorithms for heuristics such as PMIA, GDD, WD, SD, that do not consider the actual objective function for seed selection, so their running times also are the same.
The results obtained using TV model were qualitatively similar with a very slight dip in the \% gain with respect to the expected spread; the running times were significantly lower for most algorithms owing to lower edge probabilities in TV model as compared to WC model (in the case of LM dataset) and so the diffusion/simulation would terminate faster.

The results for NetHEPT are presented in Table~\ref{tab:impwithk}. For the purpose of this section, we need to only look at the first rows ($k_1=k_2$) of both WC and TV models.

\begin{observation}
Though it is clear that two-phase diffusion strictly performs better than single phase diffusion, the amount of improvement depends on the value of $k$ as well as the diffusion model under consideration (see Table~\ref{tab:impwithk}). 
\end{observation}

Note that the amount of improvement is significant, especially when the company is concerned with monetary profits or a long-term customer base.
We now attempt to further improve what we can get by using the two-phase diffusion.

\begin{table}[t!]
\centering
\begin{tabular}{c|c|c|c|c|c}
\hline \hline
\T \B
Model	&	$k \rightarrow$				&	50		&	100		&	200		&	300
\\ \hline \T \B
\multirow{3}{*}{WC}	&	\% Improvement ($k_1=k_2$) &	3.5		&	1.8		&	3.5		&	4.4
\\ \T \B
&	Opt. \% improvement &	4.5		&	2.0		&	4.0		&	4.5
\\ \T \B
	&	Optimizing $k_1$				&	15		&	35		&	70		&	105
\\ \hline \T \B
\multirow{3}{*}{TV}	&	\% improvement ($k_1=k_2$)	&	5.0		&	5.4		&	5.4		&	4.8
\\ \T \B
&	Opt. \% improvement	&	6.0		&	6.0		&	6.0		&	5.0
\\ \T \B
	&	Optimizing $k_1$						&	18		&	35		&	70	&	105
\\ \hline \hline
\end{tabular}
\caption{
\% Improvement of two-phase diffusion over single phase depending on $k$}
\label{tab:impwithk}
\end{table}

\section{Getting the Best out of Two-Phase Diffusion}
\label{sec:practical}

Till now, we assumed $k_1=k_2$ and $d=D$, and we needed to determine the best seed sets (a) of size $k_1$ for first phase and (b) of size $k_2$ for second phase based on the observed diffusion after a delay of $d$ time steps. However, in practical situations, there is also a need to determine (c) an appropriate split of the total budget $k$ into $k_1$ and $k_2$ as well as (d) an appropriate delay $d$ (we have proved that $d=D$ is optimal in absence of temporal constraints, but this may not be the case in their presence). In this section, we address these issues.
Henceforth, we use only farsighted algorithms, as they take the values of $k_2$ and $d$ into account while computing the objective function value.

\subsection{Budget Splitting}

Here we address the problem of splitting the total budget $k$ between the two phases, that is, determining an optimal $k_1$ and hence $k_2$.
Note that when $k_2$ is not fixed, the objective function $\mathbb{F}(S_1,d,k_2)$ is no longer monotone with respect to the first phase seed set $S_1$.
For instance, $\mathbb{F}(\{\},d,k-|\{\}|) = \mathbb{F}(\{\},d,k) = \sigma(S^O)$ where $S^O$ is the optimal seed set for single phase,
while for any $|S^\#|=k$, $\mathbb{F}(S^\#,d,k-|S^\#|) = \mathbb{F}(S^\#,d,0) = \sigma(S^\#)$. Unless $S^\#$ is an optimal seed set for single phase, we will have $ \sigma(S^O)> \sigma(S^\#)$ and hence $\mathbb{F}(\{\},d,k-|\{\}|)>\mathbb{F}(S^\#,d,k-|S^\#|)$, even though $\{\} \subset S^\#$.


In FACE algorithm, there is an implicit way to optimize over $k_1$ and $S_1$ (such that $|S_1|=k_1$) simultaneously by allowing each data sample to consist of a value of $k_1$ sampled from $\{1,\ldots,k\}$, as well as a sampled set $S_1$ of size $k_1$. 
\begin{remark}
For faster convergence, in the first iteration, instead of choosing each node $i$ in the set with probability $\frac{k_1}{n}$, we choose it with probability $q_i = \frac{k_1 w_i}{\sum_i w_i}$, where $w_i$ is as in Equation~(\ref{eqn:gdd_value}).
In cases wherein the value of $q_i$ exceeds 1, we distribute the surplus to other nodes with values less than 1, in proportion of their current values. We repeat this until all nodes have values at most 1. This process of distributing the surplus value is to ensure that 
the expected size of the sampled set does not drop below $k_1$.
%
The rest of the iterations follow as per the standard FACE algorithm. 
\end{remark}
%
In RMax method, for every sample, $k_1$ is chosen u.a.r. (uniformly at random) from $\{1,\ldots,k\}$ and hence a set $S_1$ of cardinality $k_1$ is sampled. The output set is one that maximizes the objective function among the sampled sets.
As there is no implicit way to optimize over $k_1$ in rest of the algorithms,
we do the following: as we add nodes one by one to construct the set $S_1$, we keep track of the maximum value attained so far, to determine a value maximizing set $S_1$ of size $k_1 \leq k$. 
%

  \begin{figure}[t!]
 \begin{minipage}{.5\textwidth}
 \centering
 \iftoggle{clr}{
    \includegraphics[scale=.6]{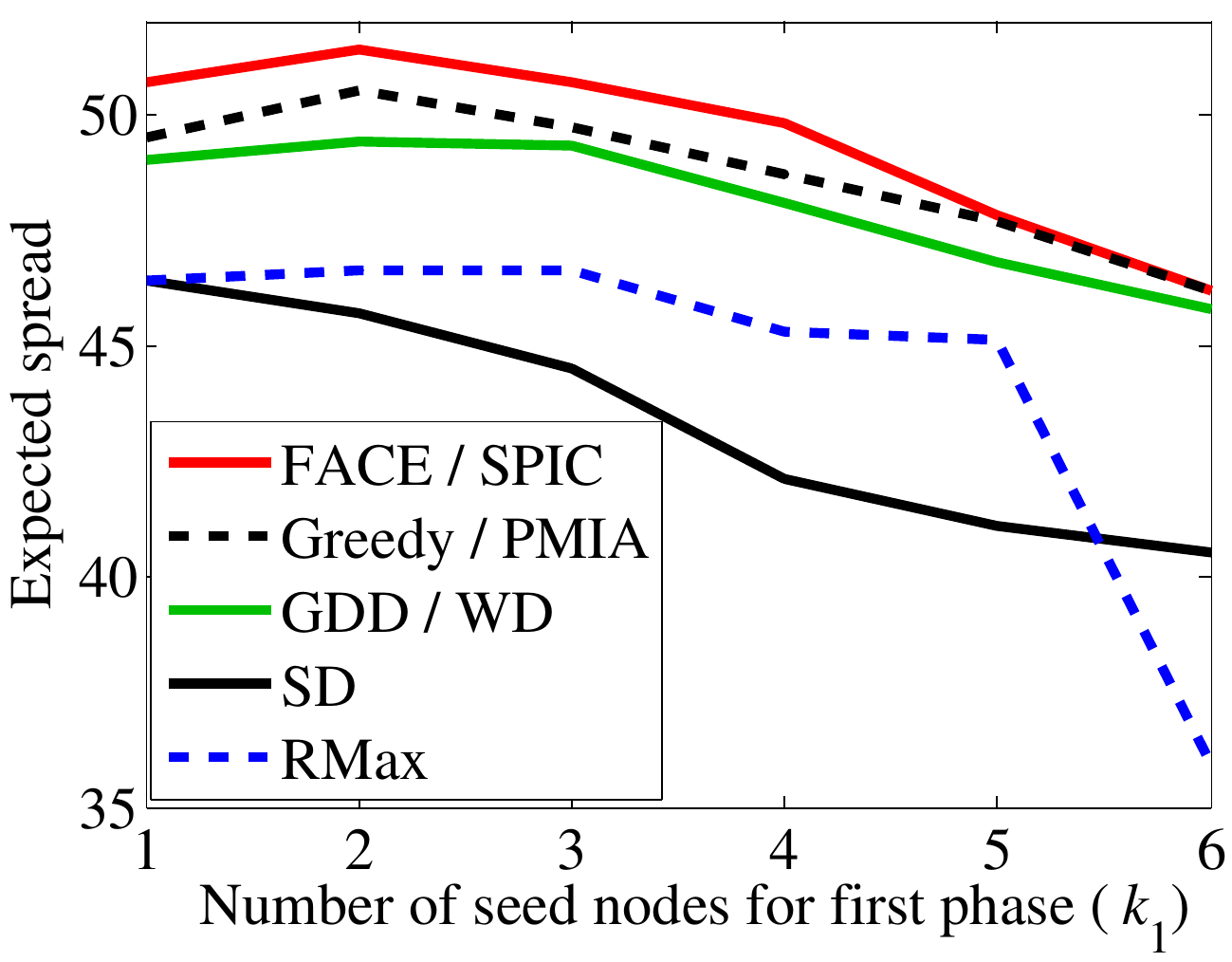} 
 }{
   \includegraphics[scale=.6]{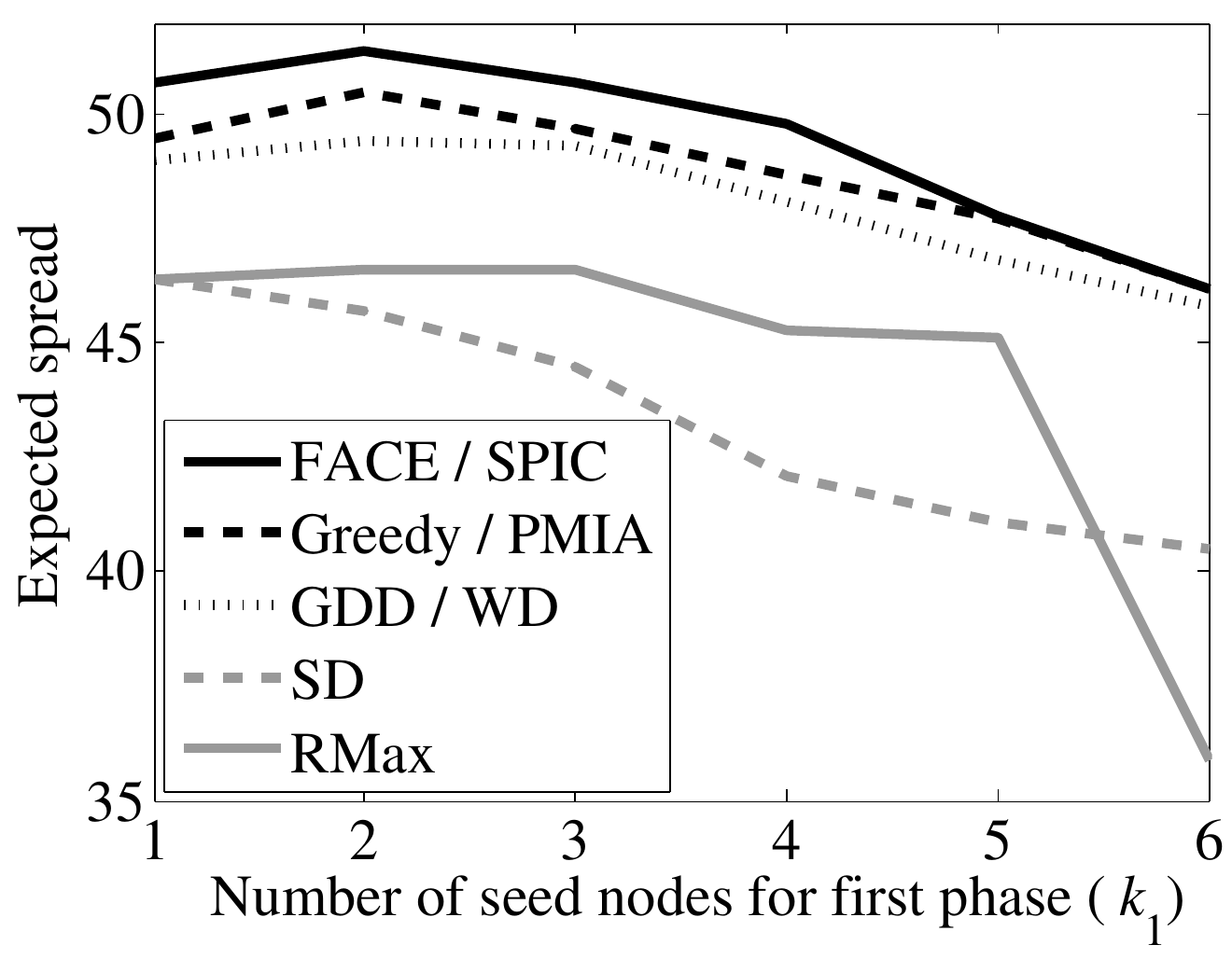} 
   }
 \\  (a)
\end{minipage}
\begin{minipage}{.5\textwidth}
\centering
\iftoggle{clr}{
\includegraphics[scale=.6]{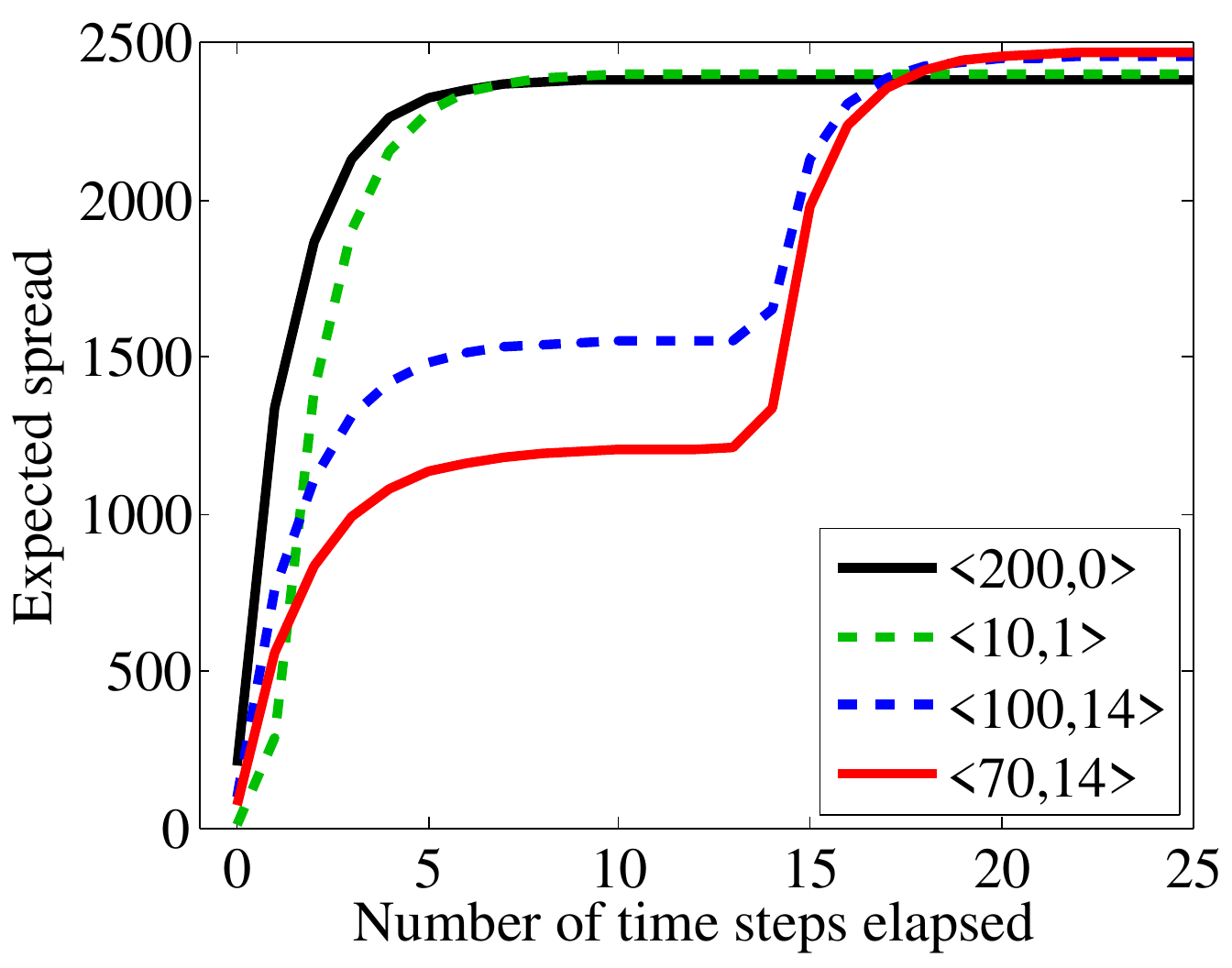} 
}{
\includegraphics[scale=.6]{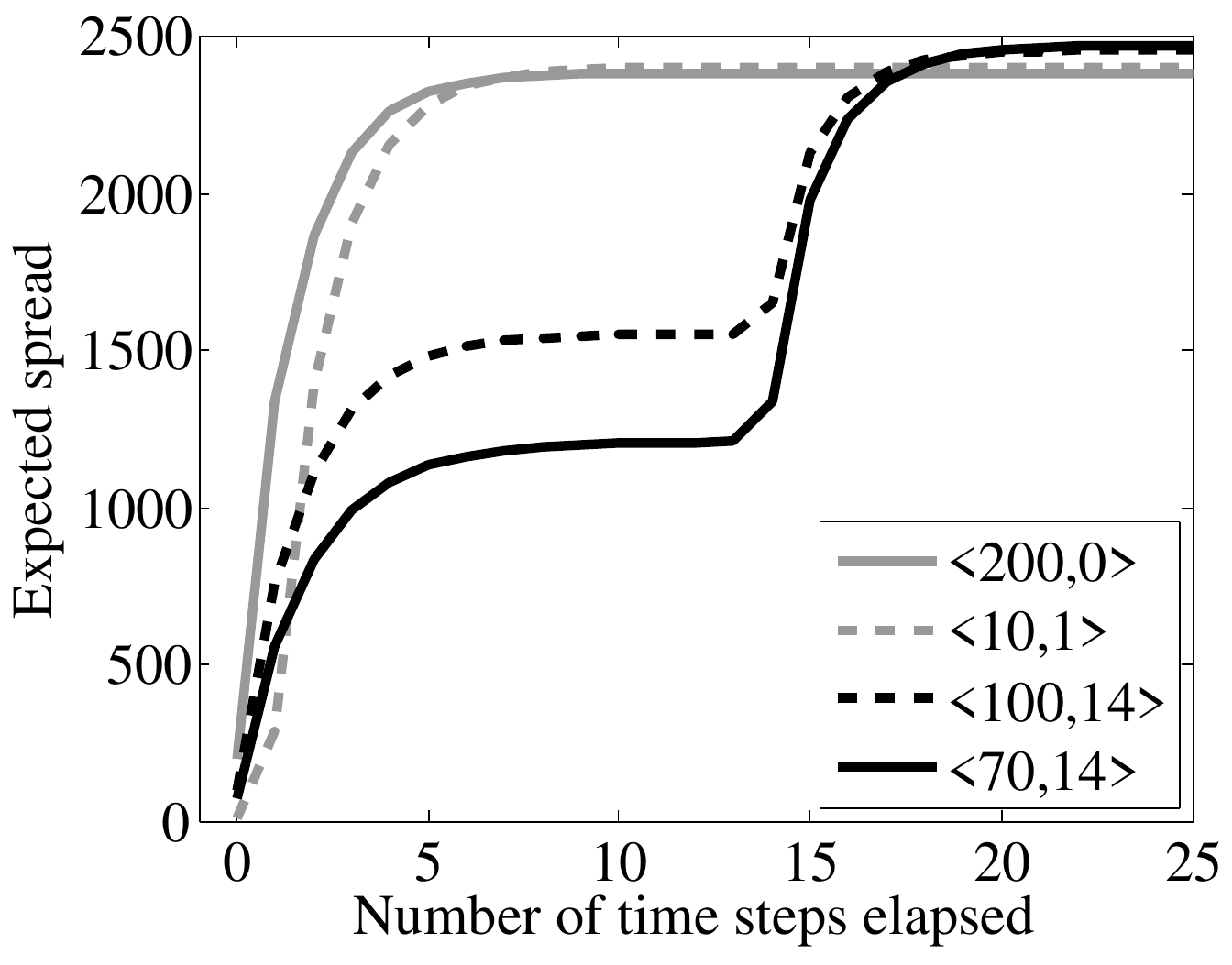} 
}
  \\ \centering (b)
\end{minipage}
 \begin{minipage}{\textwidth}
 \centering
 \vspace{5mm}
 \includegraphics[scale=.6]{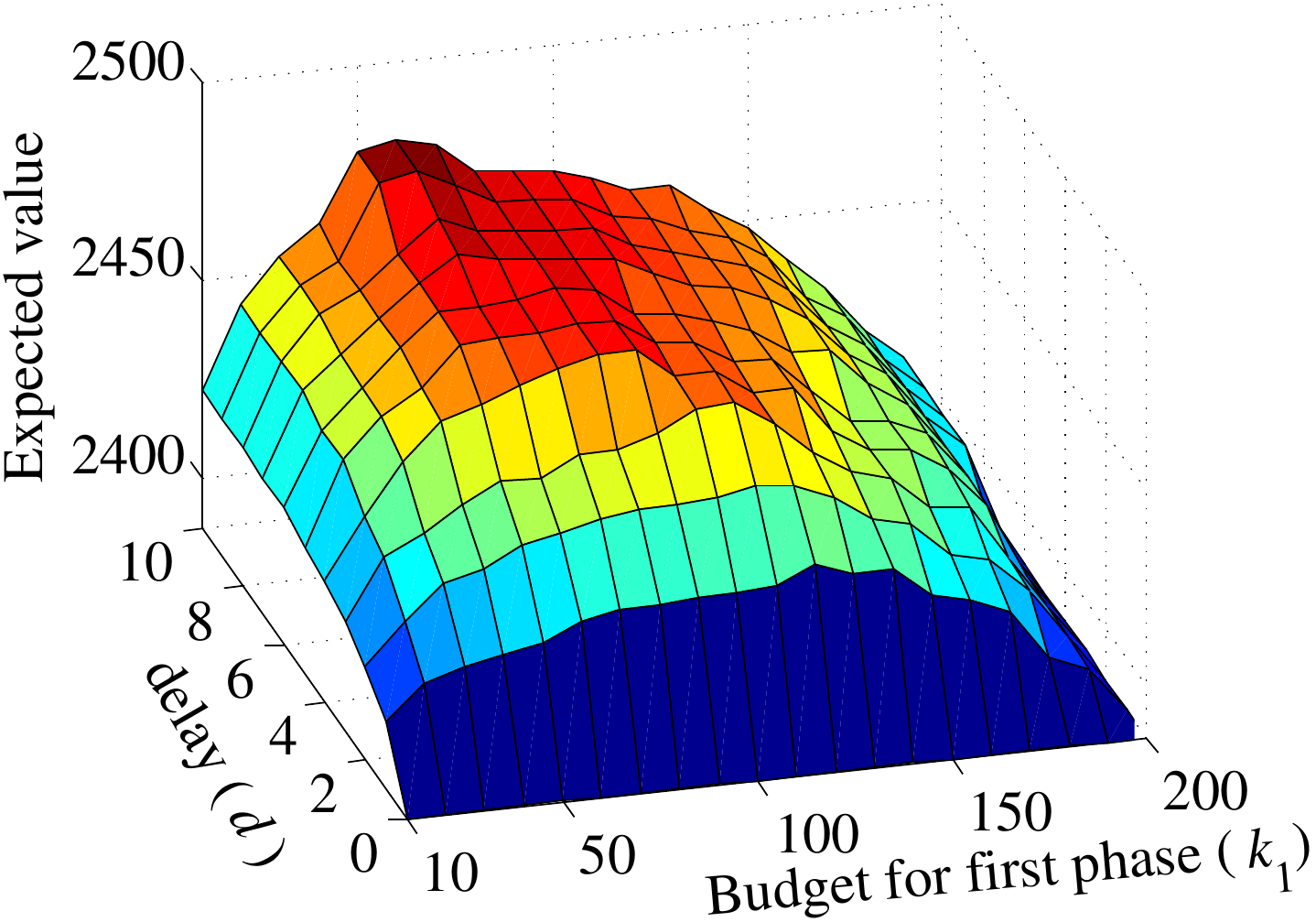} 
     \\ \centering (c)
\end{minipage}
   \caption{ 
   (a) Performance of algorithms for different values of ${k_1}$ on LM dataset under WC model (${k=6}$, ${d=D}$),
(b) Typical progression of diffusion with time for different $<k_1,d>$ pairs on NetHEPT  under WC model ($k=200$),
(c) 3D plot considering a range of $<k_1,d>$ pairs for $\delta=1$ on NetHEPT  under WC model ($k=200$)
   }
   \label{fig:plots_mpid}
  \end{figure}

Figure~\ref{fig:plots_mpid}(a) presents the results of different budget splits for the considered algorithms on LM dataset (results are for WC model, results for TV model were qualitatively similar).
We also studied various budget splits for NetHEPT dataset using both WC and TV models, the results of which are provided in
Figures~\ref{fig:plots_mpid}(c) and \ref{fig:3Dplots_TV}(c)
(see $d=10$; we have limited $d$ to $10$ for the purpose of clarity; the observations for $d>10$ were almost same as that for $d=10$)
and also Figure~\ref{fig:plot_all_deltas_hep_tv} ($\delta=1.00$).
The results for different values of $k$ are provided in
Table~\ref{tab:impwithk}.
These results show that our na\"ive first guess of splitting the budget equally was a good one, even though other splits give marginally higher values
(considering $d=D$).

\begin{observation}
For the datasets considered, under all settings (different diffusion models and values of $k$), a split of $k_1:k_2 \approx \frac{1}{3}:\frac{2}{3}$ is observed to be optimal. 
\end{observation}

  \begin{figure}[t!]
 \begin{minipage}{.5\textwidth}
 \centering
   \includegraphics[scale=.6]{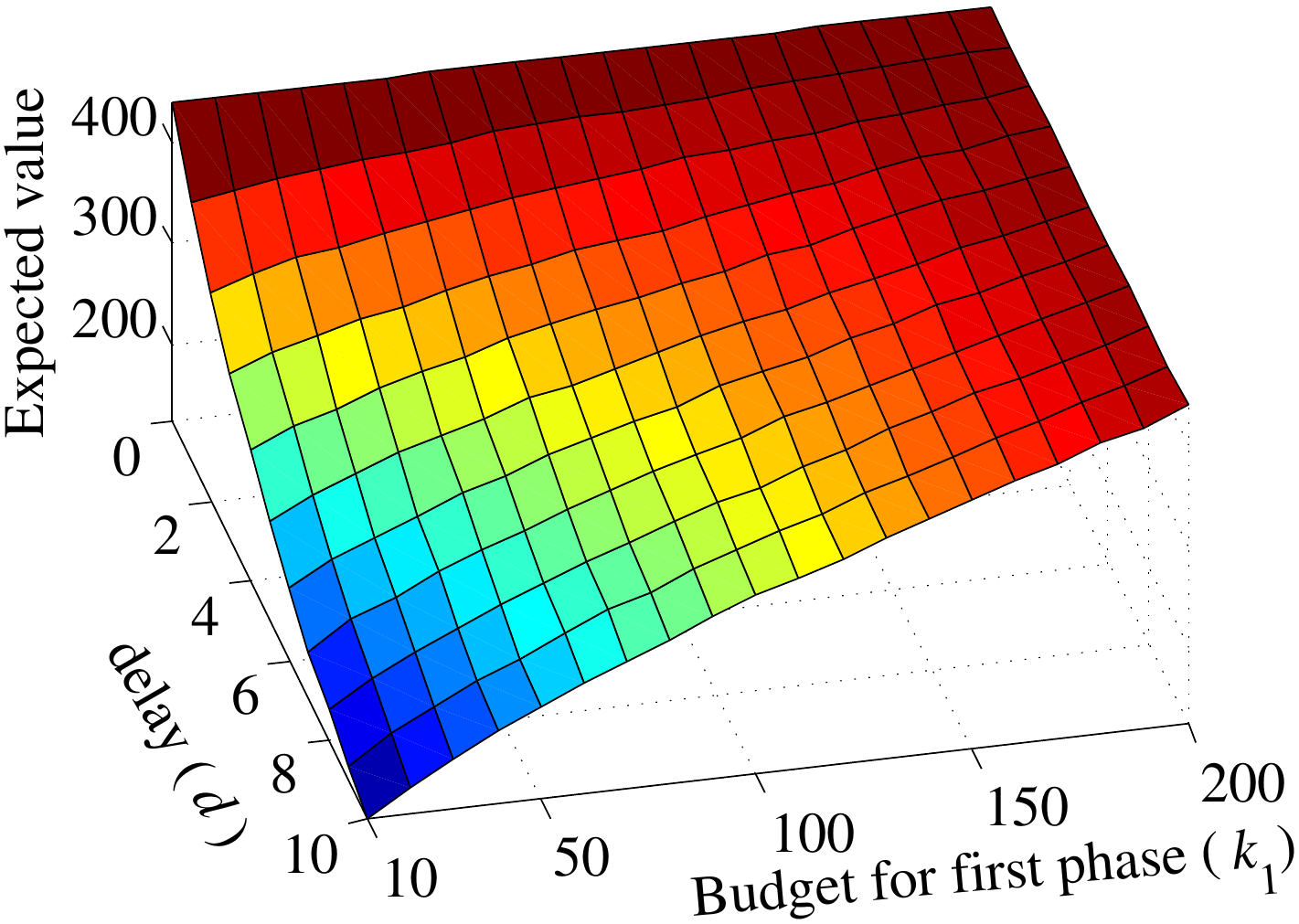} 
 \\ \centering (a) $\delta=0.85$
\end{minipage}
\vspace{5mm}
\begin{minipage}{.5\textwidth}
\centering
\includegraphics[scale=.6]{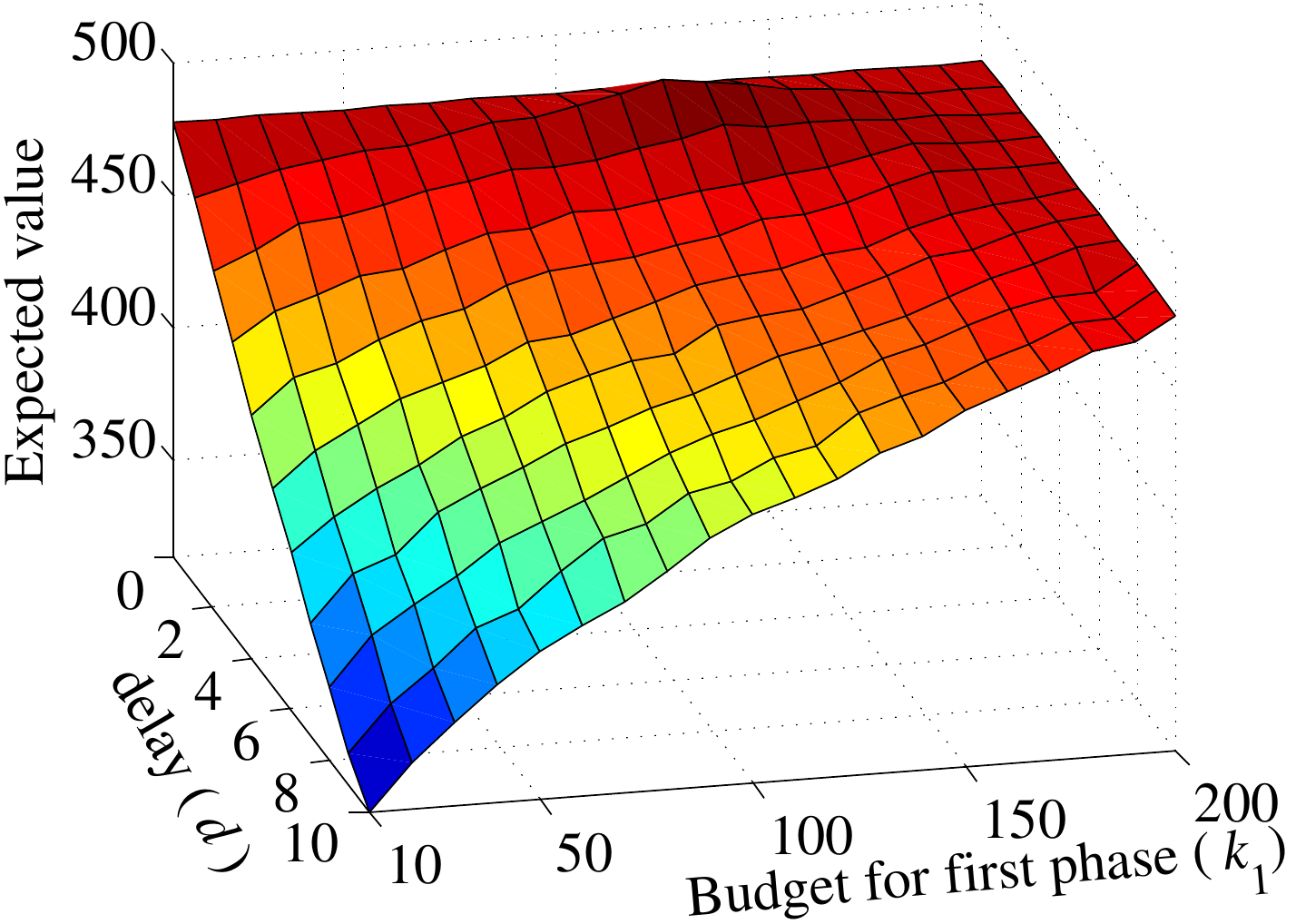}
  \\ \centering (b) $\delta=0.95$
\end{minipage}
 \begin{minipage}{\textwidth}
 \centering
 \vspace{5mm}
\includegraphics[scale=.6]{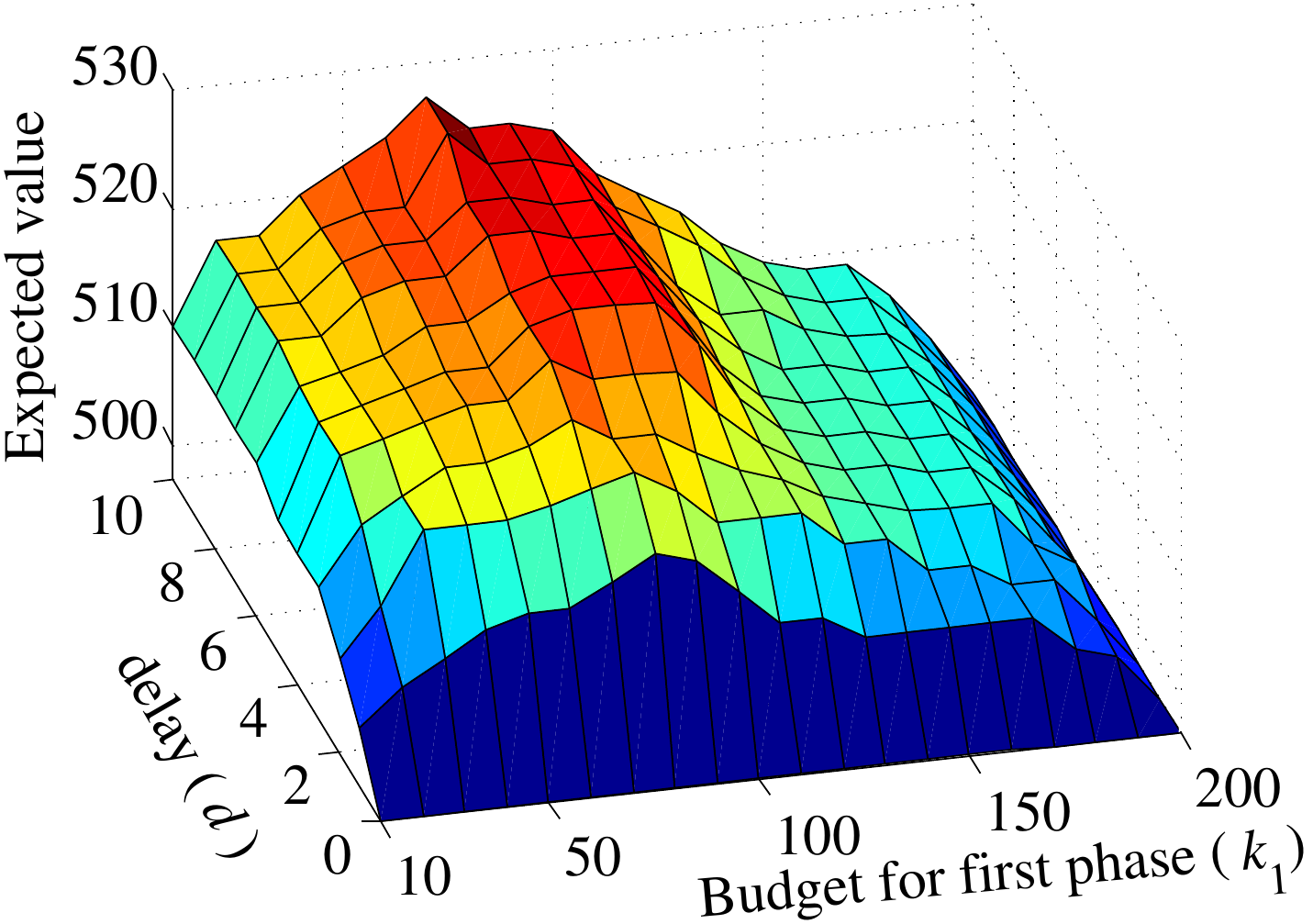}
     \\ \centering (c) $\delta=1$
\end{minipage}
   \caption{ 
   3D plots considering a range of $<k_1,d>$ pairs for different values of $\delta$ on NetHEPT dataset under TV model ($k=200$)
   (note the reversed delay axis in (c) as compared to (a-b))
   }
   \label{fig:3Dplots_TV}
  \end{figure}

  \begin{figure}
  \centering
   \iftoggle{clr}{
     \includegraphics[scale=0.62]{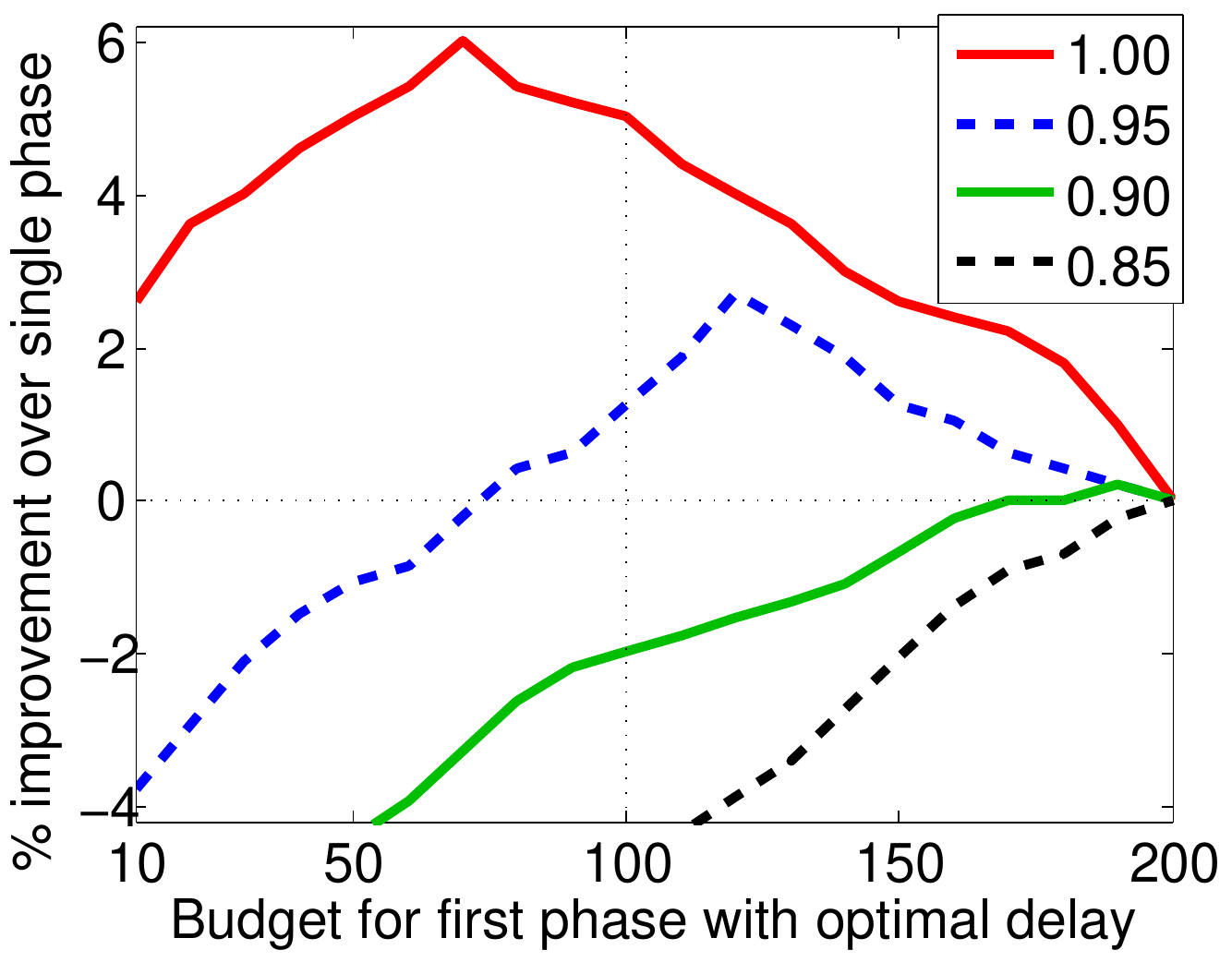}
   }{
  \includegraphics[scale=0.62]{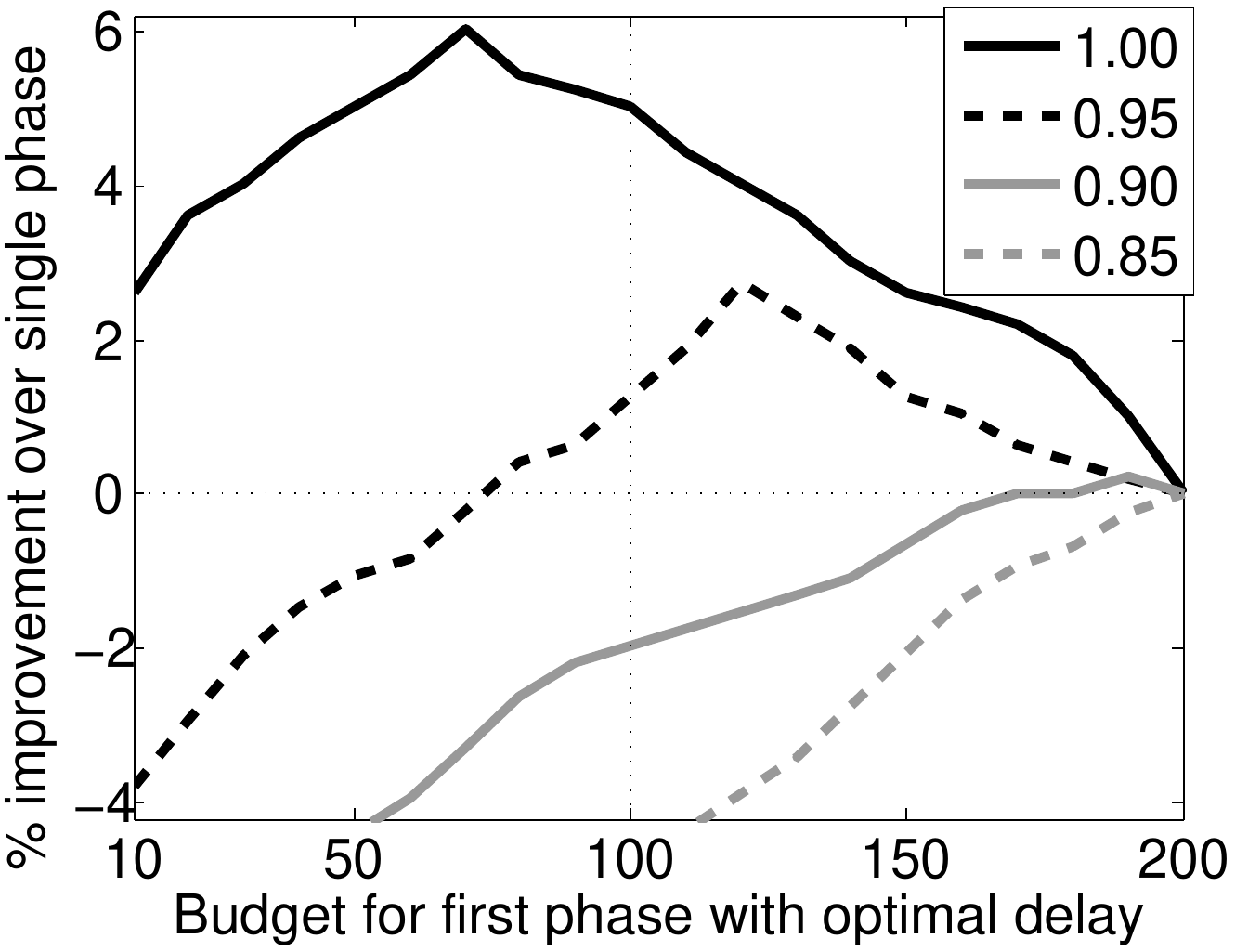}
  }
  \caption{Typical results of splitting budget ${k=200}$ with optimal delay ($d\geq 1$) for different ${\delta}$'s ($d=D$ for $\delta=1.00$, $d=1$ for other $\delta$'s) on NetHEPT  (TV model) (2D views of plots in Figure~\ref{fig:3Dplots_TV} corresponding to the optimal $<k_1,d>$ pairs)}
  \label{fig:plot_all_deltas_hep_tv}
  \vspace{-3mm}
  \end{figure}

A possible reason for $k_1 \approx k_2$ being a good guess is
a trade-off between (a) the size of the observed diffusion and (b) the exploitation based on the observed diffusion. If the value of $k_1$ is too low, not many nodes may be influenced and so we may not be able to observe the diffusion to a considerable extent, leaving us with little information on the basis of which we need to select the seed nodes for the second phase. On the other hand, if the value of $k_2$ is too low, we may not be able to select enough number of seed nodes for the second phase to exploit the information obtained from the observed diffusion.

The optimal split $k_1:k_2 \approx \frac{1}{3}:\frac{2}{3}$ (a skew towards lower values of $k_1$) can perhaps be attributed to the fact that the first set of seed nodes selected by most algorithms, are very influential, and it is not necessary to allocate half of the budget to first phase in order to obtain a large enough observable diffusion.

\subsection{Scheduling the Second Phase}
\label{sec:temporal}

It is clear that a two-phase diffusion would result in a higher influence spread than the single-phase one. This brings us to address the following questions: (a) why not use two-phase diffusion all the time? and (b) why not wait for the first phase to complete its diffusion process before starting the second phase? It is to be noted that the standard IC model fails to capture the effects of time taken for the diffusion process. 
A more realistic objective function would capture not only the influence spread, but also the rate of the diffusion process.
One such objective function could be 
$
\nu(S) = \sum_{t=0}^\infty \Gamma(t) \sigma_{(t)}(S)
$,
where $\Gamma(\cdot)$ is a non-increasing function such that $\Gamma(t) \leq 1$ for all values of $t$, and $\sigma_{(t)}(S)$ is the expected number of newly activated nodes at time step $t$.

Alternatively, let $t_j^{X,S}$ be the minimum number of time steps in which node $j$ can be reached from set $S$ in live graph $X$.
Then $\Gamma(t_j^{X,S})$ is the value obtained for influencing node $j$ in live graph $X$, and $\sum_X p(X) \Gamma(t_j^{X,S})$ is the expected value obtained for influencing node $j$.
So the expected influence value obtained starting from a seed set $S$ is $\nu(S) = \sum_j  \sum_X p(X) \Gamma(t_j^{X,S})$.
Note that if $\Gamma(t) = 1$ for all $t$, then $\nu(\cdot)$ reduces to $\sigma(\cdot)$.
Thus we modify our two-phase objective function 
by incorporating $\Gamma(t)$.
%
%
\begin{theorem}
$\nu(\cdot)$  is non-negative, monotone increasing, and submodular, for any non-increasing function $\Gamma(\cdot)$ where $0 \leq \Gamma(t) \leq 1, \; \forall t$. 
\end{theorem}
\begin{proof}
The non-negativity of $\nu(\cdot)$ is direct from the non-negativity of $\sigma_{(t)}(\cdot)$.
Now, it is clear that $t_j^{X,S} \geq t_j^{X,T}$ for any $S \subset T$, and owing to $\Gamma(\cdot)$ being a non-increasing function, we have $\Gamma(t_j^{X,S}) \leq \Gamma(t_j^{X,T})$.
Since this is true for any live graph $X$, we have $\sum_X p(X) \Gamma(t_j^{X,S}) \leq \sum_X p(X) \Gamma(t_j^{X,T})$. Also, since this is true for any node $j$, we have $\sum_j  \sum_X p(X) \Gamma(t_j^{X,S}) \leq \sum_j  \sum_X p(X) \Gamma(t_j^{X,T})$ or equivalently, $\nu(S) \leq \nu(T)$. 
This proves the monotone increasing property of $\nu(\cdot)$.

For the purpose of proving its submodularity, let us define another function $\psi_j^X(S) = \Gamma(t_j^{X,S})$. So $\nu(S) = \sum_j  \sum_X p(X) \psi_j^X(S)$, that is, $\nu(\cdot)$ is a non-negative linear combination of the functions $\psi_j^X(\cdot)$.
Consider arbitrary sets $S$ and $T$ and an arbitrary node $i$ such that $S \subset T$ and $i \in N\setminus T$.
We first prove the submodularity of $\psi_j^X(\cdot)$ for an arbitrary node $j$ and a live graph $X$ using two possible cases. 
In the first case, if addition of $i$ to the set $T$ does not reduce the number of time steps required to reach node $j$, then $\psi_j^X({T\cup\{i\}}) = \psi_j^X(T)$.
In the second case,  if addition of $i$ to the set $T$ reduces the number of time steps required to reach the node, then $\psi_j^X({T\cup\{i\}}) = \psi_j^X({\{i\}}) = \psi_j^X({S\cup\{i\}})$.
In both the cases, $\psi_j^X({S\cup\{i\}}) - \psi_j^X(S) \geq \psi_j^X({T\cup\{i\}}) - \psi_j^X(T)$. 
This proves the submodularity of $\psi_j^X(\cdot)$ and hence of their non-negative linear combination $\nu(\cdot)$.
\end{proof}
Thus following argument similar to that in Section~\ref{sec:compute}, the two-phase objective function (taking time into consideration) can be well approximated using greedy algorithm 
for seed selection in the second phase; and GDD heuristic can be used as an effective proxy for greedy.
%
Note that GDD heuristic is expected to perform very well for the temporal objective function $\nu(\cdot)$ because it maximizes the number of nodes influenced in the immediately following time step. In particular, GDD would be an excellent algorithm when $\Gamma(1)$ is significantly larger than $\Gamma(t)$ for $t \geq 2$.
In our simulations, we consider $\Gamma(t) = \delta^t$ where $\delta \in [0,1]$ (this is generally the first guess for a decay function in several problems).

Now our objective is to not only find an optimal $k_1$, but also an optimal delay $d$.
%
%
We have seen that FACE algorithm implicitly computes influential seed nodes while simultaneously optimizing over $k_1$. Now in addition to a sampled value of $k_1$ and a sampled set of cardinality $k_1$, we allow each data sample to also contain a value of $d$, sampled from $\{1,\ldots,D\}$. Table~\ref{tab:face_temporal} shows that the differences between (a) the spread achieved using this implicit optimization method and (b) that achieved using exhaustive search over $k_1$ and $d$, for different $\delta$'s on LM dataset, are low. 
The time taken for implicit optimization was observed to be approximately $\frac{1}{k D}$ of that taken for exhaustive search.
This shows the effectiveness of FACE algorithm for getting the best out of two-phase diffusion by addressing the combined optimization problem.

\begin{table}[t!]
\centering
%
\begin{tabular}{c||c||c|c|c||c|c|c}
\hline \hline
\T \B
\multirow{3}{*}{$\delta$} & Single & \multicolumn{6}{c}{Two-phase with optimal $<k_1,d>$} 
\\
\cline{3-8}
& \multicolumn{1}{c||}{phase} &  \multicolumn{3}{c||}{\T \B Implicit opt.} & \multicolumn{3}{c}{\T \B Exhaustive}
\\ \cline{3-8}
\T \B
& value & $k_1$ & $d$ & value & $k_1$ & $d$ & value 
\\ \hline
\T \B
0.75 & 33.2	&	6	&	0	& 33.2	& 6 & 0  	& 33.2	
\\ \T \B
0.80 & 36.0 &	6	&	0	& 36.0	& 5 & 1  	& 36.7	
\\ \T \B
0.85 & 38.0	& 	5	&	2	&  39.1 & 5 & 1  	& 39.7	
\\ \T \B
0.90 & 40.6 & 	5	&	2	&  41.8	& 4 & 1 	& 42.4	
\\ \T \B
0.95 & 42.9	& 	4	&	1	&  46.1	& 4 & 2 	& 46.5 	
\\ \T \B
1.00 & 46.2	& 	2	&	$D$	&  51.4 & 2 & $D$ 	& 51.4	
\\ \hline \hline
\end{tabular}
%
\caption{Performance of  FACE with implicit optimization versus that with exhaustive search}
\label{tab:face_temporal}
\end{table}

As mentioned earlier, for NetHEPT dataset also,
we observed that for $\delta=1$, it is optimal to allocate one-third of the budget to first phase and delay $d=D$. 
For $\delta \leq 0.85$, it is optimal to use single-phase diffusion. 
For intermediate values of $\delta$, it is optimal to allocate most of the budget to the first phase with a delay of one time step; the necessity of allocating most of the budget to the first phase increases as $\delta$ decreases. Figures~\ref{fig:3Dplots_TV}(a-c) and \ref{fig:plot_all_deltas_hep_tv} show this in an elaborate way.

Figure~\ref{fig:plots_mpid}(b) shows how the expected spread progresses with time for different $<k_1,d>$ pairs on NetHEPT dataset under WC model, given $k=200$.
$<200,0>$ corresponds to single phase diffusion, $<10,1>$ corresponds to two-phase diffusions with a random $<k_1,d>$ pair, $<100,14>$ corresponds to equal budget split $k_1=k_2=\frac{k}{2}$ with $d=D$, and $<70,14>$ corresponds to the optimal $<k_1,d>$ pair.
(We have $D=14$ in the plots as the first phase diffusion stagnated after 14 time steps for $k_1=70$ and $100$.) 
These types of plots showing the progression of diffusion with time may help a company to decide the ideal values of $k_1$ and $d$ based on its desired progression.

\subsection{An Efficient Method for the Combined Optimization Problem}
\label{sec:gss}

We have seen that the performance of FACE algorithm is excellent in terms of influence maximization. However, it is computationally intensive and hence impractical for use on large networks. With this in view, we propose another algorithm that is based on empirical observations in Figures~\ref{fig:plots_mpid}(c) and \ref{fig:3Dplots_TV}(a-c).

We note that the plots are unimodal in nature for the considered representative algorithms and datasets with respect to either $k_1$ (with a good enough interval between consecutive $k_1$'s) or $d$ as variable. We could exploit this nature for maximizing the objective function by using the {\em  golden section search} technique with $k_1$ as the variable, where the objective function itself is computed with an optimal $d$ for that particular $k_1$ (which can be found using golden section search). In the special case of the considered exponential decay function, since the optimal values of $d$ would be very small for almost any $\delta<1$, we find an optimal $d$ for a particular $k_1$ using sequential search starting from $d=0$.
It is to be noted that, as long as the function does not change its value drastically within small intervals (which would be true in general for the considered problem), the golden section search technique will give an optimal or near-optimal solution even when the function is not perfectly unimodal, but unimodal when the interval between consecutive $k_1$'s is good enough.

We also explored whether the plots are unimodal with respect to $k_1$ and $d$ simultaneously, so that faster methods such as multidimensional direct search, can be used. 
However, though the plots are observed to be unimodal with respect to $k_1$ and $d$ individually, they are not unimodal with respect to them simultaneously.

\section{Discussion}
\label{sec:conclusion_mpid}
We proposed and motivated the two-phase diffusion process, formulated an appropriate objective function, proposed an alternative objective function, developed suitable algorithms for selecting seed nodes in both the phases, and observed their performances using simulations. 
We observed that myopic algorithms perform closely to the corresponding farsighted algorithms, while taking orders of magnitude less time.

In order to make the best out of two-phase diffusion, we also studied the problems of budget splitting and scheduling second phase. We proposed the usage of FACE algorithm for the combined optimization problem. Further, we studied the nature of the plots and owing to their unimodal nature with respect to either $k_1$ or $d$ as variable, we proposed the usage of golden section search technique to find an optimal $<k_1,d>$ pair.
We concluded that: (a) under strict temporal constraints, use single-phase diffusion, (b) under moderate temporal constraints, use two-phase diffusion with a short delay while allocating most of the budget to the first phase, and (c) when there are no temporal constraints, use two-phase diffusion with a long enough delay with a budget split of $\frac{1}{3}:\frac{2}{3}$ between the two phases (one-third budget for the first phase).

We presented results for a few representative settings; these results are very general in nature.
We now provide notes on the decay function, satisfaction of the subadditivity property by the objective function, and how this work can be extended to the linear threshold model.


\subsection{A Note on the Decay Function}
\label{sec:decay_fn}

We considered a very strict decay function with time (exponential), which resulted in humbling two-phase diffusion for most range of $\delta$. In practice, the decay function would be more lenient, where the value would remain high for first few time steps and then decay at a slow rate. Such a decay function would be more suitable for two-phase diffusion. 

In practice, the value of a product follows a thresholding behavior. A company would value its initial sales to be high owing to factors like immediate profits, urgency to break even, higher selling prices, etc. But after some time period, it would value its sales to be considerably lower owing to concentration on newer products, more competing products, less number of potential customers for older products, etc. This behavior can be described by a reverse S-shaped curve, where the value is high for initial periods of time, and then starts decreasing until it reaches a stable lower limit. 
Note, however, that our choice of a simple exponential decay function allowed us to draw firm conclusions, which would not have been the case while working with a sophisticated function as described above.

One could also account for time by studying the problem in presence of competing diffusions, where a delay in diffusion may help competitors reach the potential customers first.

\subsection[{A Note on Subadditivity of Objective Function $f(\cdot)$}]{A Note on Subadditivity of Objective Function $\boldsymbol f(\cdot)$}
\label{sec:subadditive}

Though we have shown that the function $f(\cdot)$ given by Equation~(\ref{eqn:f}) is not submodular, it can be shown to be subadditive, that is, $f(S_1 \cup T_1) \leq f(S_1)+f(T_1), \;\forall S_1,T_1 \subseteq N$.
\begin{property}
\label{prop:subadditive}
\mbox{$f(\cdot)$ is subadditive.}
\end{property}
\begin{proof}
Let $W_1=S_1\cup T_1$
 and $W_2^{O(X,W_1,d,k_2)}$ be an optimal set of $k_2$ nodes given the observation corresponding to $X,d$ starting with seed set $W_1$.
\begin{align*}
 f(S_1)+f(T_1)
&=
\sum_{X} p(X) \sigma^X (S_1 \cup S_2 ^{O(X,S_1,d,k_2)}) 
+ \sum_{X} p(X) \sigma^X (T_1 \cup T_2 ^{O(X,T_1,d,k_2)})
\\&\geq
\sum_{X} p(X) \sigma^X (S_1 \cup W_2 ^{O(X,W_1,d,k_2)}) 
+ \sum_{X} p(X) \sigma^X (T_1 \cup W_2 ^{O(X,W_1,d,k_2)})
\\&\geq
\sum_{X} p(X) \sigma^X (S_1 \cup T_1 \cup W_2 ^{O(X,W_1,d,k_2)})
\\&=
\sum_{X} p(X) \sigma^X (W_1 \cup W_2 ^{O(X,W_1,d,k_2)})
\\&=
f(W_1) = f(S_1 \cup T_1)
\end{align*}
The first inequality is from optimality of $S_2 ^{O(X,S_1,d,k_2)}$ and $T_2 ^{O(X,T_1,d,k_2)}$, and the second one from subadditivity of $\sigma ^X (\cdot)$ (since submodularity and non-negativity  implies subadditivity).
\end{proof}

There exists an algorithm that provides an approximation guarantee of $\frac{1}{2}$ for maximizing a subadditive function
\cite{feige2009maximizing}. However, owing to its relatively high running time, we leave it out of our study. It would be interesting though to develop more efficient algorithms for exploiting the subadditivity of $f(\cdot)$.

\subsection{A Note on the Linear Threshold (LT) Model}
\label{sec:mpid_lt}

Throughout this chapter, we discussed multi-phase diffusion using IC model, primarily because it is a natural setting for such a diffusion. One can as well study multi-phase diffusion using the other most popular model, the LT model.
We now discuss this in brief.

In LT model, an influence degree $b_{u,v}$ is associated with every directed edge $(v,u)$, where $b_{u,v} \geq 0$ is the degree of influence that node $v$ has on node $u$, and an influence threshold $\chi_u$ is associated with every node $u$. The weights $b_{u,v}$ are such that $\sum_v b_{u,v} \leq 1$. Owing to lack of knowledge about the thresholds, which are held privately by the nodes, it is assumed that the thresholds are uniformly distributed in $[0,1]$. The diffusion starts at time step 0 and proceeds in discrete time steps, one at a time. In each time step, a node $u$ is influenced or activated if and only if the sum of influence degrees of the edges incoming from its activated neighbors (irrespective of the time of their activation) crosses its own influence threshold, that is, when
$
\sum_{v\in \mathcal{A}} b_{u,v} \geq \chi_u
$,
where $\mathcal{A}$ is the set of activated nodes.
Nodes, once activated, remain activated for the rest of the diffusion process. 
 In any given time step, all activated nodes (not just the recently activated ones) contribute to the diffusion. The diffusion stops when no further nodes can be activated.

At the beginning of the first phase, the thresholds of nodes are assumed to be uniformly distributed in $[0,1]$. When the second phase is scheduled to start, we have the information regarding the status of nodes thus far, that is, whether they are active or inactive. 
In addition, we also have the updated information regarding the thresholds of inactive nodes which are out-neighbors of active nodes. That is, such a node was not activated in the first phase even after receiving a total influence (sum of influence degrees of the edges incoming from activated neighbors) of $\sum_{v\in \mathcal{A}} b_{u,v}$. Thus we now have the information that this node has a threshold greater than $\sum_{v\in \mathcal{A}} b_{u,v}$, and so while determining the seed set for second phase, we can exploit this information by assuming its threshold to be uniformly distributed in $\left(\sum_{v\in \mathcal{A}} b_{u,v},1\right]$, instead of a wider (and hence more uncertain) range of $[0,1]$.

\vspace{10mm}
This chapter dealt with the problem of multi-phase information diffusion, while the preceding one dealt with the problem of network formation. Both network structure and the information flowing through it impact the opinions and preferences of individuals, depending on which other individuals they frequently interact with and what information they are more exposed to. The next chapter studies how the social network information can be harnessed to determine ideal representatives of the population, so as to aggregate the preferences of individuals in the population in an efficient and effective way in practice.

\begin{subappendices}


\chapter*{Appendix for Chapter~\ref{chap:mpid}}
\addcontentsline{toc}{chapter}{Appendix for Chapter~\ref{chap:mpid}}
\label{chap:appendix_mpid}

\section{Post-processing Methods for Cooperative Game Theoretic Solution Concepts}
\label{app:ppmethods}

The problem of finding the $k$ most critical nodes, referred to as the $top\text{-}k$ problem, is a very important one in several contexts such as influence maximization, influence limitation, virus inoculation, 
etc. It has been observed in the literature that the value allotted to a node by most of the popular cooperative game theoretic solution concepts, acts as a good measure of appropriateness of that node (or a data point) to be included in the $top\text{-}k$ set, by itself~\cite{narayanam2010shapley,suri2008determining}.  However, in general, nodes having the highest $k$ values are not the desirable $top\text{-}k$ nodes, because the appropriateness of a node to be a part of the $top\text{-}k$ set depends on other nodes in the set. As this is not explicitly captured by cooperative game theoretic solution concepts, it is necessary to post-process the obtained values in order to output the suitable $top\text{-}k$ nodes. Here, we propose a number of such post-processing methods and provide justification for each of them.

We now present a number of post-processing methods, primarily in the context of information diffusion; they can be extended to other contexts. However, some methods would be more suitable for a given application than others.
One can derive variants of the proposed methods or use multiple methods in conjunction.

Given an input that is a graph (or that can be converted to a graph), let $\beta_{xy}$ denote the weight of edge $xy$.
%
%
Starting from the null set, the $top\text{-}k$ set builds up as nodes get added to it, until it reaches cardinality $k$.
The most direct and na\"ive method of obtaining the $top\text{-}k$ set is to sort nodes in descending order of their values, say \textit{ordered\_list}, and then choose the first $k$ nodes from the list. However, as explained earlier, nodes having the highest $k$ values are generally not the desirable $top\text{-}k$ nodes, and so there is a need to post-process the obtained values in order to build an effective $top\text{-}k$ set.
We propose a number of post-processing methods that can be broadly classified into two types, namely,
(a) eliminating neighbors of chosen nodes
and
(b) discounting values of neighbors of chosen nodes.

In the following methods, neighbors of a node account for both its in-neighbors and out-neighbors; however, the accounting of neighbors can be altered based on the application.
In the methods that follow, in order to obtain an ordering over multiple nodes that are allotted equal values, the ties are broken randomly.

\subsection{Eliminating Neighbors of Chosen Nodes}

As choosing nodes na\"ively may result in the $top\text{-}k$ nodes to be clustered in one part of the network, these methods try to choose the nodes such that they are appropriately spread in the network, so that they influence  as many distinct nodes as possible.
In all the methods of this type, if it is not possible to choose any more nodes using the elimination approach, the methods reiterate over \textit{ordered\_list} and na\"ively choose the unselected nodes in order.
It is to be noted that in the methods of this type, 
if a node $x$ is both in-neighbor and out-neighbor of a node $y$, we consider the mutual edge weight to be $\max \{ \beta_{xy} , \beta_{yx} \}$ whenever applicable.

\subsubsection{Eliminating Neighbors of Chosen Nodes Always}

This method is the one used in \cite{narayanam2010shapley} for influence maximization in a social network.
It keeps on choosing the nodes in order from \textit{ordered\_list} and skips a node if any of its neighbors is already chosen in the $top\text{-}k$ set. 
This method is observed to perform well in the context of information diffusion since once a node is chosen, it would likely influence its out-neighbors
either directly or indirectly (through several common neighbors owing to triadic closures in social networks); 
furthermore, if a node is chosen before its in-neighbors because of its high value, it is likely that its in-neighbors would have a high value only because they influence the former (a more influential node); so the method eliminates both in-neighbors and out-neighbors. 

\subsubsection{Eliminating Neighbors of Chosen Nodes Based on a Threshold}

This method is similar to the one used in \cite{garg2013novel} in the context of clustering.
The method that eliminates neighbors of chosen nodes always, suffers from the fact that multiple nodes that are highly influential maybe connected with low edge weights; in such cases, it is undesirable to eliminate the neighbors of such influential nodes.
So instead of eliminating neighbors of chosen nodes always, this method keeps on choosing the nodes in order from \textit{ordered\_list} and skips a node if any of its neighbors which is already chosen in the $top\text{-}k$ set, is such that the corresponding edge weight exceeds a certain threshold.
This method would work well in all the contexts, provided the threshold is chosen appropriately.
One can come up with several variants of this method; for instance, the threshold could be a fixed one for the entire network or dataset \cite{garg2013novel}, or it could be a function of the value of the chosen node itself. 

\subsubsection{Eliminating Neighbors of Chosen Nodes Based on their Local Networks}

This method determines whether a node should be selected based on its local neighborhood.
It keeps on choosing the nodes in order from \textit{ordered\_list\/}, and skips a node $x$ if there exists its neighbor $y$ which is already chosen in the $top\text{-}k$ set such that, when all the neighbors of $y$ are ordered in decreasing order of their edge weights with $y$, then $x$ lies in the {\em first  half}.
Note that this fraction {\em half} is just a natural first guess; as it acts like a threshold, it can also be a fixed one for the entire network or it can be a function of the value of the chosen node.
Intuitively, this method does not eliminate a node when it is a good candidate in the local neighborhood of the already chosen nodes, that is, it is less likely to be influenced by or to influence the nodes in the $top\text{-}k$ set.

\subsection{Discounting Values of Neighbors of Chosen Nodes}

The elimination methods are strict owing to their $0/1$ nature of eliminating a node. Moreover, it is highly likely that it would not be possible to choose any more nodes using these methods beyond a certain $k$, resulting in a na\"ive selection of the unchosen nodes. One way to overcome these problems is to discount the values of the neighbors of the chosen nodes based on certain criteria instead of eliminating them. 

These methods run in $k$ steps where the value of each node gets updated in each step $t$. Let $top\text{-}k^{(t)}$ be the $top\text{-}k$ set in step $t$, and let $\phi_x^{(t)}$ be the value of node $x$ in step $t$. The initializing $top\text{-}k$ set can be given by $top\text{-}k^{(0)} = \{\}$, and the initializing value of a node $\phi_x^{(0)} = \phi_x$ is the original value allotted to it. In the methods of this category, we do not update the values of the nodes which are already chosen in the $top\text{-}k$ set, that is,
\begin{equation}
\nonumber
\phi_z^{(t)} = \phi_z^{(t-1)} \;\;\; \forall 
z \in top\text{-}k^{(t-1)}
\end{equation}
It is important that the values of the chosen nodes do not change after they get added to the $top\text{-}k$ set, since the discounting of the values of unselected nodes critically depends on these values.
In order to explain the methods of this category, let $y$ be the node chosen in step $t$, that is, $top\text{-}k^{(t)} \setminus top\text{-}k^{(t-1)} = \{y\}$. Also let $\mathcal{N}_w$ be the set of neighbors of a node $w$.

\subsubsection{Discounting Values of Neighbors of Chosen Nodes - I}
\label{app:discount1}

\begin{equation}
\nonumber
\phi_x^{(t)} = (1-\beta_{yx}) \, \phi_x^{(t-1)} \;\;\; \forall x \in \mathcal{N}_y \setminus top\text{-}k^{(t-1)}
\end{equation}

\begin{equation}
\nonumber
\text{or} \;\;\; \phi_x^{(t)} = \left( \prod_{w \in \mathcal{N}_x \cap top\text{-}k^{(t)}} (1-\beta_{wx}) \right) \, \phi_x^{(0)} 
\end{equation}

This discounting is natural in the independent cascade model of information diffusion where $\beta_{yx}$ corresponds to the parameter $p_{yx}$ (in independent cascade model, when node $y$ first becomes active at time $\tau$, it is given a single chance to activate each of its currently inactive neighbors $x$ at time $\tau+1$ and it succeeds with probability $p_{yx}$).
Since node $x$ would get activated because of node $y$ with probability $p_{yx}$, the value of node $x$ should be discounted by the factor of $p_{yx}$ whenever any of its neighbor $y$ gets chosen in the $top\text{-}k$ set. 
Equivalently, since node $x$ would not get directly activated by any of its neighbors that are chosen in the $top\text{-}k$ set, with probability $\prod_{w \in \mathcal{N}_x \cap top\text{-}k^{(t)}} (1-p_{wx})$, the value of node $x$ is updated using this factor. 
Note that we ignore the possibility that the influence of node $y$ can reach node $x$ in multiple hops.

\subsubsection{Discounting Values of Neighbors of Chosen Nodes - II}
\label{app:discount2}

\begin{equation}
\nonumber
\phi_x^{(t)} = \phi_x^{(t-1)} - \beta_{yx} \, \phi_x^{(0)} \;\;\; \forall x \in \mathcal{N}_y \setminus top\text{-}k^{(t-1)}
\end{equation}

\begin{equation}
\nonumber
\text{or} \;\;\; \phi_x^{(t)} = \left(1 - \sum_{w \in \mathcal{N}_x \cap top\text{-}k^{(t)}} \beta_{wx} \right) \phi_x^{(0)}
\end{equation}

This discounting is natural in the linear threshold model of information diffusion where $\beta_{yx}$ corresponds to the parameter $b_{x,y}$ (in linear threshold model, $b_{x,y}$ is the influence weight of node $y$ on node $x$ such that the sum of the influence weights from all of its incoming neighbors is at most $1$; node $x$ gets activated if the sum of the influence weights from its active incoming neighbors exceeds a certain threshold $\chi_{x}$ that is drawn from a uniform distribution in $[0,1]$). 
Analogous to the argument in the previous method, since node $x$ would not get directly activated by any of its neighbors that are chosen in the $top\text{-}k$ set, with probability $\left(1 - \sum_{w \in \mathcal{N}_x \cap top\text{-}k^{(t)}} b_{w,x}\right)$, the value of node $x$ is updated using this factor. 
Note that in this method also, we ignore the possibility that the influence of node $y$ can reach node $x$ in multiple hops.

\subsubsection{Discounting Values of Neighbors of Chosen Nodes - III}
\label{app:discount3}

\begin{equation}
\nonumber
\phi_x^{(t)} = \phi_x^{(t-1)} - \beta_{xy} \, \phi_y^{(t-1)} \;\;\; \forall x \in \mathcal{N}_y \setminus top\text{-}k^{(t-1)}
\end{equation}

\begin{equation}
\nonumber
\text{or} \;\;\; \phi_x^{(t)} = \phi_x^{(0)} - \sum_{w \in \mathcal{N}_x \cap top\text{-}k^{(t)}} \left( \beta_{xw} \, \phi_w^{(t-1)} \right)
\end{equation}

This discounting is natural in both independent cascade and linear threshold models of information diffusion (note the swapping of $x$ and $y$ with respect to the previous methods).
In the independent cascade model, as node $x$ influences node $y$ directly with probability $p_{xy}$, it gets a fractional share of the value of node $y$ (since $x$ would be influencing other nodes indirectly, through $y$).
Now given that node $y$ is chosen in the $top\text{-}k$ set, the share of $y$'s value should be removed from the value of $x$. 
This method uses a simplified expression for this share, namely, $p_{xy} \, \phi_y^{(t-1)}$.
Similar argument leads this share to be $b_{y,x} \, \phi_y^{(t-1)}$ in the linear threshold model.
Note that we may be possibly removing more share than required since there may exist multiple neighbors of $x$, that are already chosen in the $top\text{-}k$ set, with shares of similar nature (for example, they may be likely to influence almost the same set of nodes).  Owing to this, it is possible for the value of a node to become negative.
\\


Furthermore, depending on the application, one may also update the values of the nodes using a suitable  combination of the aforementioned discounting methods.
The SPIC algorithm described in Section~\ref{sec:shapley_method} uses post-processing that is a combination of the methods presented in Sections~\ref{app:discount1} and \ref{app:discount3}.

\begin{note}
For any practical problem, the natural valuation function would lead to intractable computation for most solution concepts. Though approximate algorithms exist for several solution concepts, an alternative would to formulate a valuation function that closely resembles the problem and at the same time, facilitates efficient computation of solution concepts. 
\end{note}

\end{subappendices}

\blankpagewithnumber


\newcommand{\meanoverall}{0.35}
\newcommand{\meanpersonal}{0.40}
\newcommand{\meansocial}{0.30}
\newcommand{\stdoverall}{0.09}
\newcommand{\stdpersonal}{0.12}
\newcommand{\stdsocial}{0.08}

\chapter[Modeling Spread of Preferences in Social Networks for Scalable Preference Aggregation]{Modeling Spread of Preferences in Social Networks for Scalable Preference Aggregation
  \blfootnote{A part of this chapter is published as \cite{dhamal2013scalable}:
Swapnil Dhamal and Y. Narahari. Scalable preference aggregation in social networks. In {\em First AAAI Conference on Human Computation and Crowdsourcing (HCOMP)}, pages 42--50. AAAI, 2013.}
}

\label{chap:pasn}

\begin{quote}
%
Given a large population, aggregating individual preferences over a set of alternatives could be extremely computationally intensive. Moreover, it may not even be feasible to gather preferences from all the individuals. If the individuals are nodes in a social network, we show that information available about this social network can be exploited to efficiently compute the aggregate preference, with knowledge of preferences of a small subset of representative nodes. To drive the research in this work, we have developed a Facebook app to create a dataset consisting of preferences of nodes for a range of topics as well as the underlying social network. We develop models that capture the spread of preferences among nodes in a typical social network. We next propose an appropriate objective function for the problem of selecting representative nodes, and subsequently propose two natural objective functions as effective alternatives for computational purposes. We then devise two algorithms for selecting the best representatives. For the first algorithm, we provide performance guarantees that require the preference aggregation rule to satisfy a property, expected weak insensitivity. For the second algorithm, we study desirable properties from a viewpoint of cooperative game theory. Also, we empirically find that the degree centrality heuristic performs quite well, demonstrating the ability of high-degree nodes to serve as good representatives of the population. The key conclusion of this work is that harnessing social network information in a suitable way will achieve scalable preference aggregation of a large population and will certainly outperform random polling based methods.
%
\end{quote}




\section{Introduction}
\label{sec:intro_pasn}

There are several scenarios such as elections, opinion polls, public project initiatives, funding decisions, etc., where a society faces a number of alternatives. In scenarios where the society's collective opinion is of importance, there is a need to get the society's collective preference over these alternatives.
Ideally, one would want to obtain the preferences of all the individuals in the society and aggregate them so as to represent the society's preference. This process of computing an aggregate preference over a set of alternatives given individual preferences is termed {\em preference aggregation} and is a well-studied topic in the field of social choice theory.
It is generally assumed that the preferences of all the individuals are known.
However, more often than not, obtaining all the individual preferences is a difficult and an expensive process in itself owing to factors such as the individuals' lack of interest to provide a prompt, truthful, well-informed, and well-thought preference over the given set of alternatives.
In an effort to circumvent this problem, we turn towards harnessing any additional information regarding the society of individuals, in particular, the underlying social network.

Social network information has been harnessed for a variety of purposes, ranging from viral marketing to controlling epidemic spread, from determining the most powerful personalities in a society to determining the behaviors of people. Social networks serve to explain several phenomena which cannot be explained otherwise, primarily because such phenomena are caused by the social interactions which are captured in the social network itself. 
Many of these phenomena can be explained with an important feature of social networks - {\em homophily}~\cite{networkscrowdsmarkets}.
Homophily refers to a bias in friendships towards similar individuals - individuals with similar interests, behaviors, opinions, etc. 
The tendency of individuals to form friendships with others who are like them is termed {\em selection}. 
On the other hand, similarities may also be a result of friendships; people tend to change their behaviors to align themselves more closely with the behaviors of their friends; this process is termed {\em social influence}.
Hence selection and social influence can be viewed as complements of each other.
It is evident that social networks and homophily are inseparable. 
%

We exploit this important feature of homophily in addressing the problem of preference aggregation.
In particular, we study how the social network among the individuals of a society can be used to determine the best representatives among them, so that obtaining only the representatives' preferences would suffice to effectively (and efficiently) deduce the aggregate preference of the society.
In this chapter, we use the terms voters, individuals, agents, and nodes interchangeably, so also  neighbors and friends. 


\subsection{Preliminaries}
\label{sec:prelim_pasn}

Given a set of alternatives, individuals have certain preferences over them. These alternatives could be any entity, ranging from political candidates to food cuisines.
We assume that 
an individual's preference is represented as a complete ranked list of alternatives.
Throughout this chapter, we refer to a ranked list of alternatives 
as a {\em preference} and the multiset consisting of the individual preferences as {\em preference profile}.
For example, if the set of alternatives is $\{X,Y,Z\}$ and individual $i$ prefers $Y$ the most and $X$ the least, then $i$'s preference 
is written as
$(Y,Z,X)_i$.
Suppose individual $j$'s preference is $(X,Y,Z)_j$, then the preference profile of the population $\{i,j\}$ is $\{(Y,Z,X),(X,Y,Z)\}$.
A widely used measure of dissimilarity between two preferences is {\em Kendall-Tau distance} which counts the number of pairwise inversions with respect to the alternatives.
In this chapter, given that the number of alternatives is $r$, we normalize Kendall-Tau distance to be in $[0,1]$, by dividing the actual distance by \begin{scriptsize}$\dbinom{r}{2}$\end{scriptsize}, the maximum distance between any two preferences on $r$ alternatives.
%
%
For example, the Kendall-Tau distance between preferences $(X,Y,Z)$ and $(Y,Z,X)$ is 2, since two pairs $\{X,Y\}$ and $\{X,Z\}$ are inverted between them. The normalized Kendall-Tau distance is $2/$\begin{scriptsize}$\dbinom{3}{2}$\end{scriptsize} $=\frac{2}{3}$.

Preference aggregation is a well-studied topic in social choice theory. An {\em aggregation rule} takes a preference profile as input and outputs the {\em aggregate preference(s)}, which in some sense reflect(s) the collective opinion of all the individuals. 
We consider a wide range of aggregation rules for our study, which are extensions of voting rules such as Bucklin, Smith set, Borda, Veto, Minmax (pairwise opposition), Dictatorship, Random Dictatorship, Schulze, Plurality, Kemeny, and Copeland.
A survey of voting rules
and related topics 
can be found in \cite{brandt2012computational}.
A more concise table of voting rules and their properties can be found in \cite{wiki:voting}.
Of these rules, only Kemeny, Dictatorship, and Random Dictatorship output the entire aggregate preference; others either determine a winning alternative or give each alternative a score.
For the sake of consistency, for all rules except Kemeny, Dictatorship, and Random Dictatorship, we employ the following well accepted approach for converting a series of winning alternatives into an aggregate preference: rank a winning alternative as first, then vote over the remaining alternatives and rank a winning alternative in this iteration as second, and repeat until all alternatives have been ranked
\cite{brandt2012computational}.
As we are indifferent among alternatives, we do not assume any tie-breaking rule in order to avoid any bias towards any particular alternative while determining a winner.
So an aggregation rule may not output a unique aggregate preference, that is, it is a correspondance.

We now motivate the problem of determining {\em best} representatives for deducing the aggregate preference of a population or society in an effective and efficient way.

\section{Motivation}
\label{sec:motiv_pasn}

In real-world scenarios, it may not be feasible to gather the individual preferences of all the voters owing to  factors like time and interest of the voters. One such scenario is preference aggregation in a large online social network where gathering all the preferences is an expensive process in terms of time as well as cost. 
Another scenario is that of a company desiring to launch a future product based on the feedback received from its customers. However, very few customers might be willing to respond to such feedback queries in a prompt and honest way, and devote the needed effort to provide a useful feedback. 
In such cases, if information on the underlying social network is available, one could harness the homophily property, which states that most friendship relations imply similar preferences. 
In order to estimate the aggregate preference of the entire population, an attractive approach would be to select a subset of individuals based on such properties and incentivize those individuals to report their preferences.
We refer to these individuals as {\em representatives}.

Another motivation for this work comes from the field of computational voting theory. Most aggregation rules are extremely computationally intensive. Hence, it is important to have efficient ways of estimating the aggregate preferences. In almost all aggregation rules (apart from those similar to dictatorship), as the number of voters decreases, computation of the aggregate preference becomes faster.
The problem tackled in this chapter is potentially one such approach where we use a subset of preferences to arrive at an acceptable aggregate preference. 
As we will see, the manner in which we aggregate the preferences reduces the number of voters, and weighs their preferences appropriately. This helps speed up the computation of aggregate preference for most aggregation rules, since we do not have to parse the preferences of all the voters and also owing to the reduced number of distinct preferences.

\section{Relevant Work}
\label{sec:relevant_pasn}

There have been studies in the literature that deal with the influence of social networks on voting in elections.
The pioneering Columbia and Michigan political voting research is discussed in~\cite{sheingold1973social} with an emphasis on importance of the underlying social network. 
It has been observed that the social network has higher impact on one's political party choice than background attributes like class or ethnicity~\cite{burstein1976social}.
A similar study finds that social influence through ties plays a more important role than similarities in attributes such as religion, education, and social status~\cite{Nieuwbeerta2000313}. 
It has also been observed that conversations with partisan discussants and family members do act as a statistically significant influence on voting~\cite{pattie1999british}.
Ryan~\cite{ryan2011shortcut} suggests that social communication is a useful information shortcut for uninformed independents, but not uninformed partisans, while informed individuals incorporate biased social messages into their candidate evaluations.
%
%
McClurg~\cite{Mcclurg01122003,McClurg2006electoral} argues that interactions in social networks have a strong, though often overlooked, influence on voting, since interactions allow an individual to gather information beyond personal resource constraints.
%
%
%
It has also been argued that, though it is obvious and well-known that social networks play a vital role in voting decisions, it has been overlooked by political scientists in their analysis~\cite{zuckerman2005social}.

Williams and Gulati~\cite{williams2007social} investigate the extent of Facebook profile use in the 2006 U.S. election, and analyze which candidates were more likely to use them, with what impact on their vote shares. Facebook users could register their support for specific candidates and also receive notifications when their Facebook friends registered support for a candidate. 
It was observed that the candidates' Facebook support had a significant effect on their final vote shares.
%
The impact of social networks has also been compared with that of mass media communication with respect to voting choices, where it is observed that social discussions outweigh the media effect \cite{beck2002calculus}, and that both the effects should be studied together \cite{campus2008social}.
Boldi et al.~\cite{Boldi2009votingnetworks} study the properties of a voting system suited for electronically mediated social networks, where an individual can either vote or appoint another individual who can vote on his/her behalf.
On the other hand, it has also been argued via a maximum likelihood approach to political voting, that it is optimal to ignore the network structure~\cite{conitzer2012should}.

There have been works on modeling homophily in social networks \cite{mcpherson2001birds,zhang2005learning,krivitsky2009representing,xiang2010modeling}.
%
Results of certain behavioral experiments suggest that agents compromise their individual preferences to achieve unanimity in a situation where agents gain some utility if and only if the entire population reaches a unanimous decision~\cite{kearns2009behavioral}. The scenario in a real group is similar, where, members who do not comply with group norms, either eventually compromise or leave the group to evade the tension between the preferences.


There have been efforts to detect the most critical nodes in social networks in other contexts \cite{jacksonbook,networkscrowdsmarkets} such as influence maximization~\cite{kempe2003maximizing}, limiting the spread of misinformation~\cite{budak2011limiting}, virus inoculation~\cite{abbassi2011toward}, etc.
There is extant literature on modeling individual preferences using {\em general random utility models} which consider the attributes of alternatives and agents. The most relevant work here considers the problem of node selection by exploiting these attributes~\cite{grum2013uai}; however, the underlying network structure is not taken into consideration.

To the best of our knowledge, there do not exist any models that capture the way preferences are spread among nodes in a social network and furthermore,
there do not exist any attempts to determine critical nodes that represent the aggregate preference of nodes in a social network. 

\section{Contributions of this Chapter}
\label{sec:contrib_pasn}
We are given a social graph with a set of individuals $N$ and an aggregation rule $f$.
Our objective is to choose a subset 
of nodes of a given (small) cardinality such that the aggregate preference of the nodes in this subset closely approximates the aggregate preference of the entire set of nodes. To solve this problem,
we proceed as follows in this chapter.

\begin{itemize}
\item
To drive the research in this work, we first develop a Facebook app for eliciting the preferences of individuals for a range of topics, while also obtaining the social network among them. The Facebook app has helped us obtain a real-world dataset consisting of over 1000 nodes.
We propose a number of simple yet faithful models with the aim of capturing how preferences are spread in a social network; some of these models provide an acceptable fit to the data collected above. (Section~\ref{sec:modeling})

\item
We formulate an appropriate objective function for the problem of determining critical nodes that best represent a given social network in terms of preferences. 
We propose a property which we call {\em expected weak insensitivity}, which captures the robustness of an aggregation rule, and we show that several aggregation rules satisfy this property.
We then establish a fundamental relation between (a) the closeness of the chosen representative set to the population in terms of expected distance and (b) the error incurred in the aggregate preference if that set is chosen as the representative set.
Following this, we present two alternate objective functions that are appropriate and natural for this setting and show the existence of efficient approximation algorithms for optimization problems involving these two objective functions. (Section~\ref{sec:problem_pasn})

\item
We then propose algorithms for selecting the best representatives. Our algorithms include the popular greedy hill-climbing algorithm for optimal set selection. We also provide a guarantee on the performance of one of the algorithms (Greedy-min), subject to the aggregation rule satisfying the expected weak insensitivity property, and study desirable properties of one other algorithm (Greedy-sum) from the viewpoint of cooperative game theory. We also see that the degree centrality heuristic performs very well, thus demonstrating the ability of high-degree nodes to serve as good representatives of the population. We conclude that instead of using the common and popular method of random polling, it is better to harness social networks for scalable preference aggregation as a more effective and reliable approach. (Section~\ref{sec:results})

\item
We provide insights on the effectiveness of our approach for aggregating preferences with respect to personal and social topics, as well as a note on an alternate way of defining the distribution of preferences between nodes. 
(Section~\ref{sec:conclusion_pasn})

\end{itemize}

We believe the results in this chapter offer a rigorous model for capturing spread of preferences in a social network, leading to an efficient approach
for scalable preference aggregation for large-scale social networks.
%

\section{Modeling Spread of Preferences in a Social Network}
\label{sec:modeling}

In this section, we first introduce the idea of modeling the spread of preferences in a social network with the help of an analogy to modeling information diffusion. We then describe the dataset obtained through our Facebook app and hence develop a number of simple yet faithful models for deducing the spread of preferences in a social network.

\subsection{An Analogy to Information Diffusion Models}

Several models have been proposed for studying how a piece of information diffuses in a social network, the most popular and well-studied being the independent cascade (IC) model and the linear threshold (LT) model \cite{kempe2003maximizing}. Given a weighted and directed network, and a set of seed nodes where the information starts propagating, these models provide ways of computing the individual probabilities of each node receiving the information (or getting influenced) given a particular seed set, and hence the expected spread achieved at the end of a diffusion process. 

The goal of this section is to develop models on similar lines, which will deduce the preference of each node and how similar it would be to the nodes in a given representative set, and hence compute the expected error in the aggregate preference obtained by considering only the preferences of the given set of representative nodes.
These models can also be seen as extending the modeling of similarities between connected nodes to that between unconnected nodes.

\begin{sidewaystable}
\begin{small}
\begin{tabular}{|c|c|c|c||c|c|c|c|}
\hline
\multicolumn{4}{|c||}{\T \B Personal}	&	\multicolumn{4}{c|}{\T \B Social}	
\\ \hline \T \B
Hangout	&	Chatting 			&	Facebook 	&	Lifestyle	&	Website		&	Government 		&	Serious 	&	Leader
\\ \T \B
Place			&	App					& Activity		&						&	visited		&	Investment			& Crime	&
\\ \hline \T \B
Friend's place 						&	WhatsApp			&	Viewing Posts				 & 		Intellectual		&	Google		&	Education					&	Rape		&		N. Modi (India) 
\\ \T \B
Adventure Park	&	Facebook chat	&	Chatting 					 & Exercising  &	Facebook	&		Agriculture		&	Terrorism 			& B. Obama (USA)
\\ \T \B
Trekking				&	Hangouts 					&	Posting 	& Social activist						&	Youtube		&	Infrastructure 				&	Murder				&	D. Cameron (UK) 
\\ \T \B
Mall	&	SMS		&	Games / Apps		& Lavish					&	Wikipedia	&	Military						&	Corruption	&	V. Putin (Russia)
\\ \T \B
Historical Place		&	Skype					&	Marketing			& Smoking			&	Amazon	&	Space exploration	&	Extortion		&	X. Jinping (China)
\\ \hline
\end{tabular}
\end{small}
\caption{Topics and alternatives in the Facebook app}
\label{tab:questions}
\end{sidewaystable}

\begin{figure}[t!]
\centering
\includegraphics[scale=.45]{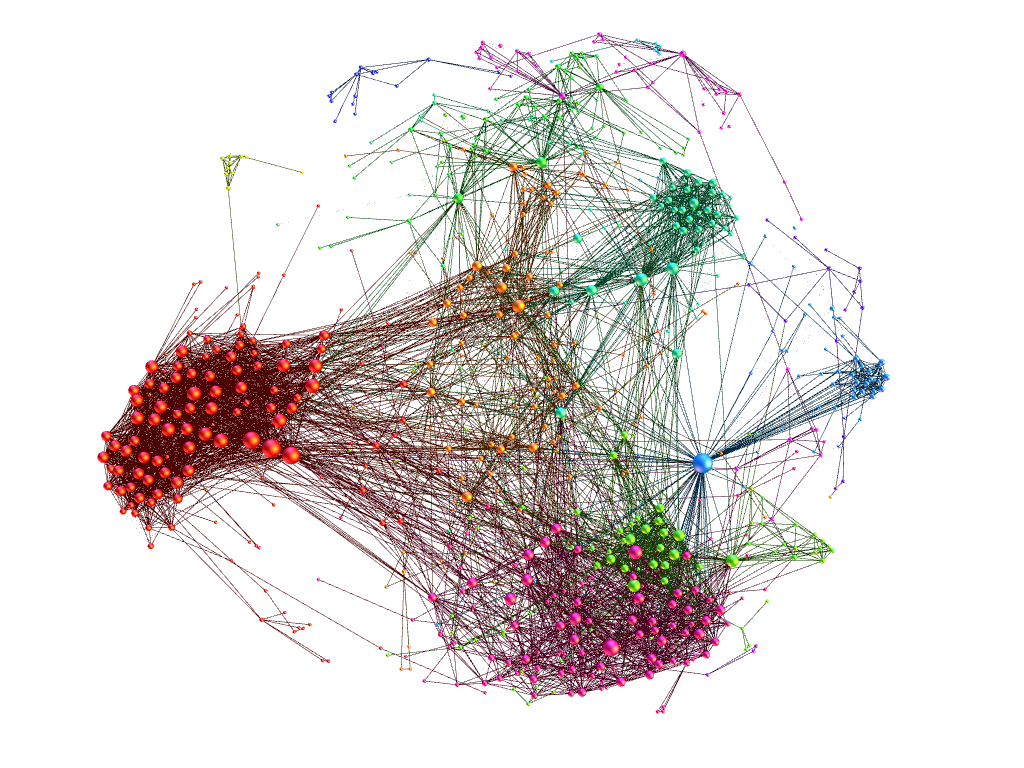}
\caption{Network of the app users}
\label{fig:network_app_users}
\end{figure}

\subsection{The Facebook App}

In order to develop such models, there was a need of a dataset that consists of (a) preferences of nodes for a range of topics and (b) the underlying social network.
Keeping this goal in mind, we developed a Facebook app titled {\em The Perfect Representer}, which asked the app users to give their preferences for 8 topics, over 5 alternatives per topic. 
The topics, which were broadly classified into personal and social types, and their alternatives are listed in Table~\ref{tab:questions}; the ordering of alternatives from top to bottom is based on the aggregate preference of the population as per the Borda count aggregation rule.
This enabled us to not only model the spread of preferences in a social network, but also validate our algorithms for selecting the best set of representative nodes in the network.
The app got over 1000 users and 500 likes. Figure~\ref{fig:network_app_users} shows the network of app users. More details about the app are provided in
Appendix~\ref{app:facebook_app}.


The obtained data consisted of 1006 nodes and 7112 edges, however it was necessary to preprocess the data before using it for model-fitting. For instance, we eliminated nodes which had given response to fewer than 6 topics so as to consider only those nodes which have responded to sufficient number of topics.
Furthermore, in order to observe the network effect, it was necessary that all the nodes belonged to the same component; so we considered only the giant component which consisted of 844 nodes and 6129 edges.

Our focus throughout this chapter will be on aggregating preferences across all topics or issues, without making any distinction among them with respect to their type or nature, for instance, personal or social. However, in practice, it maybe advantageous to consider different types separately, and choose representatives based on the type of topic or issue. We provide a brief note on this in Section~\ref{sec:conclusion_pasn}.

\subsection{Modeling Distance between Preferences of Pairs of Nodes}

In order to study the distance between preferences of two nodes, we used the measure of normalized Kendall-Tau distance.
The distribution of distance between preferences of most pairs of nodes (connected as well as unconnected), over the considered topics, was observed to follow a bell curve (most distances were clubbed together, with few of them spread apart). As a preliminary fit to the data, we considered Gaussian distribution since it is a natural and most commonly observed distribution in real-world applications. Since the range of values taken by the distance is between 0 and 1, we considered Truncated Gaussian distribution. Furthermore, as the range of values is discrete, we considered a discrete version of truncated Gaussian distribution; denote it by $\mathcal{D}$.
The discretization can be done in the following way.
We know that when the number of alternatives is $r$, 
the distance between consecutive discrete values is $1/$\begin{scriptsize}$\dbinom{r}{2}$\end{scriptsize} $=\frac{2}{r(r-1)}$.
Let $F$ be the cumulative distribution function of the continuous truncated Gaussian distribution. So the value of the probability mass function of $\mathcal{D}$ at $x$ can be shown to be
\begin{displaymath}
F\left(\min\left\{x+\frac{1}{r(r-1)},1\right\}\right) - F\left(\max\left\{x-\frac{1}{r(r-1)},0\right\}\right).
\end{displaymath}

So for any pair of nodes, the distance between their preferences was observed to follow  distribution $\mathcal{D}$. Let 
the expected distance between nodes $i$ and $j$ with respect to their preferences, be denoted by $d(i,j)$.
Let {\em distance matrix} be a matrix whose cell $(i,j)$ is $d(i,j)$ and {\em similarity matrix} be a matrix whose cell $(i,j)$ is $c(i,j) = 1-d(i,j)$.
%
%
Following are certain statistics about $d(i,j)$'s over all pairs of nodes in the obtained Facebook app data. 
The mean ($=\avg_{\{i,j\}} d(i,j)$) for overall, personal, social type of topics was respectively \meanoverall, \meanpersonal, \meansocial, while the standard deviation was respectively \stdoverall, \stdpersonal, \stdsocial.

The parameters of the distribution ($\mu_{ij}$ and $\sigma_{ij}$ of the original Gaussian distribution from which $\mathcal{D}$ is derived) for every pair $\{i,j\}$ were obtained using maximum likelihood estimation (MLE).
%


\subsection{Modeling Spread of Preferences in a Social Network}

Recollect that the primary objective of modeling the spread of preferences in a social network is to identify a best set of representative nodes for the entire social network. In order to do this, we need not only the distribution of distances between connected nodes, but also that between unconnected nodes. As an analogy to modeling information diffusion, in order to identify a best set of seed nodes for maximizing diffusion, we need the influence of a candidate set not only on its neighbors, but also on distant nodes. 

Given the preferences of a set of nodes (call it {\em initializing set}), our models aim to deduce the possible preferences of all the nodes in the social network. If this model is run for several iterations, say $\mathbb{T}$, with a randomized initializing set in each iteration, we would have deduced the preferences of nodes for these $\mathbb{T}$ generated (or simulated) topics (and hence $\mathbb{T}$ preference profiles). This would then enable us to deduce the distribution of distances between unconnected nodes as well.
Since the deduced preferences and hence the distances are randomized, in the remainder of this chapter, we will address each of our models as a {\em Random Preferences Model (RPM)}.
%


In the models that we propose, we partition the nodes into two sets at all time steps, namely, (1) {\em assigned nodes} which are assigned a preference, and (2) {\em unassigned nodes} which are not assigned a preference as yet.
Let {\em potentially next nodes} be the subset of unassigned nodes, which have at least one neighbor in the set of assigned nodes.
Starting with the nodes in the initializing set as the only assigned nodes,
a node is chosen uniformly at random from the set of potentially next nodes at each time step, and is assigned a preference based on the preferences of its {\em assigned neighbors} (neighbors belonging to the set of assigned nodes).
Algorithm~\ref{alg:generic_model} presents a generic model on these lines.
We now present a number of models as its special cases.

\begin{algorithm}[t!]
\KwIn{Connected graph $G$ with parameters $\mu,\sigma$ on its edges, Number of generated (or simulated) topics $\mathbb{T}$}
\KwOut{Preference profiles for $\mathbb{T}$ generated topics}
\For{$t \gets 1$ \textbf{to} $\mathbb{T}$}{
Randomly choose an initializing set of certain size $s$\;
Assign preferences to nodes in this initializing set\;
\For{$i \gets 1$ \textbf{to} $n-s$}{
Choose an unassigned node $u$ uniformly at random from the set of potentially next nodes\;
Assign a preference to $u$ based on:\\
 (i) the model under consideration and \\
 (ii) either \\~~~~~(a) the preferences of its assigned neighbors or \\~~~~~(b) the preference of one of its assigned neighbors which is chosen based on a \text{~~~~~~~~~~}certain criterion\;
}
}
\caption{{A generic model for spread of preferences in a social network}}
\label{alg:generic_model}
\end{algorithm}

\subsubsection{{Independent Conditioning (RPM-IC)}}
Let $P_j$ be the random preference to be assigned to a node $j$ and $A_j$ be the set of {\em assigned neighbors} of node $j$. So given the preferences of its assigned neighbors, the probability of node $j$ being assigned a preference $p_j$ is 
\begin{align*}
& \mathbb{P}\left(P_j = p_j | (P_{i} = p_{i})_{i \in A_j}\right)
\\&= \frac{\mathbb{P}\left((P_{i} = p_{i})_{i \in A_j} | P_j = p_j\right) \; \mathbb{P}(P_j = p_j)}{\sum_{p_j} \mathbb{P}\left((P_{i} = p_{i})_{i \in A_j} | P_j = p_j\right) \;  \mathbb{P}(P_j = p_j)}
\\&\propto \mathbb{P}\left((P_{i} = p_{i})_{i \in A_j} | P_j = p_j\right)
\end{align*}

The proportionality results since the denominator is common, and $\mathbb{P}(P_j = p_j) = \frac{1}{r!}$ for all $p_j$'s (assuming no prior bias).

Now we make a simplifying assumption of mutual independence among the preferences of assigned neighbors of node $j$, given its own preference. So the above proportionality results in
\begin{align}
\label{eqn:rpmic}
& \mathbb{P}\left(P_j = p_j | (P_{i} = p_{i})_{i \in A_j}\right)
\propto \prod_{i \in A_j} \mathbb{P}(P_{i} = p_{i} | P_j = p_j)
\end{align}
 We now see how to compute $\mathbb{P}\left(P_{i} = p_{i} | P_j = p_j\right)$.
 Let $D_{ij}$ be the random variable corresponding to the distance between nodes $i$ and $j$ (as described earlier, we assume that $D_{ij}$ has distribution $\mathcal{D}$ with the values of $\mu$ and $\sigma$ depending on the pair $\{i,j\}$), and $\tilde{d}(p_i,p_j)$ be the distance between preferences $p_i$ and $p_j$. So,
\begin{align}
\nonumber
&\mathbb{P}\left(P_{i} = p_{i} | P_j = p_j\right) 
\\
\nonumber
&= \mathbb{P}\left(P_{i} = p_{i} , D_{ij} = \tilde{d}(p_i,p_j) | P_j = p_j\right) 
\\
\nonumber
&\;\;\;\;\;\;\;\;\;\;\;\;\;\;\;\;\;\;\;\;\;\;\;\;\;\;\;\;\;\;\;\;\;\;\; (\because \text{given } P_j, \text{ we have } P_i \cap D_{ij} = P_i)
\\
\nonumber
&= \mathbb{P}\left(D_{ij} = \tilde{d}(p_i,p_j) | P_j = p_j\right) \;\mathbb{P}\,\Big(P_{i} = p_{i} | D_{ij} = \tilde{d}(p_i,p_j) , P_j = p_j\Big)
\\
\label{eqn:conditional}
&= \mathbb{P}\left(D_{ij} = \tilde{d}(p_i,p_j)\right) \;\mathbb{P}\left(P_{i} = p_{i} | D_{ij} = \tilde{d}(p_i,p_j) , P_j = p_j\right)
\\
\nonumber
&\;\;\;\;\;\;\;\;\;\;\;\;\;\;\;\;\;\;\;\;\;\;\;\;\;\;\;\;\;\;\;\;\;\;\;\;\;\;\;\;\;\;\;\;\;\;\;\;(\because D_{ij} \text{ is independent of } P_j)
\end{align}
Here, $\mathbb{P}\left(D_{ij} = \tilde{d}(p_i,p_j)\right)$ can be readily obtained.
Moreover, as we assume that no preference has higher priority than any other, $\mathbb{P}\left(P_{i} = p_{i} | D_{ij} = \tilde{d}(p_i,p_j) , P_j = p_j\right)$ is precisely the reciprocal of the number of preferences which are at distance $\tilde{d}(p_i,p_j)$ from a given preference. 
This value can be expressed in terms of distance $\tilde{d}(p_i,p_j)$ and the number of alternatives.
As an example for the case of 5 alternatives, the number of preferences which are at a normalized Kendall-Tau distance of $0.1$ (or Kendall-Tau distance of 1) from any given preference is 4; for example, if the given preference is $(A,B,C,D,E)$, the 4 preferences are $(B,A,C,D,E)$, $(A,C,B,D,E)$, $(A,B,D,C,E)$, $(A,B,C,E,D)$. It is clear that this count is independent of the given preference.

The initializing set for this model is a singleton set chosen uniformly at random, and is assigned a preference chosen uniformly at random from the set of all preferences.
For each unassigned node, this model computes probabilities for each of the $r!$ possible preferences, by looking at the preferences of its assigned neighbors, and hence chooses exactly one preference based on the computed probabilities (multinomial sampling). So the time complexity of this model for assigning preferences for $\mathbb{T}$ topics is $O(r!(\sum_{i\in N} {deg}(i))\mathbb{T}) = O(r!m\mathbb{T})$, where $deg(i)$ is the degree of node $i$.

\subsubsection{{Sampling (RPM-S)}}
In this model, an unassigned node $j$ is assigned a preference based on a distance sampled from the distribution with one of its assigned neighbors, say $i$, that is, from the distribution $\mathcal{D}$ having the parameters $(\mu_{ij},\sigma_{ij})$. This assigned neighbor could be selected in multiple ways; we enlist three simple ways which we have used in our experimentation:



\begin{enumerate}
\item[a)] \textit{Random:} A node is selected uniformly at random from $A_j$. 
This is the most natural way and is immune to overfitting.
\item[b)] \textit{$\mu$-based:} A node $i$ is selected randomly from $A_j$ with probability proportional to $1-\mu_{ij}$. 
This is consistent with the empirical belief that a node's preference should depend more on its similar friends.
\item[c)] \textit{$\sigma$-based:} A node $i$ is selected randomly from $A_j$ with probability proportional to $1/\sigma_{ij}$. 
This is statistically the best choice because, giving lower priority to distributions with low standard deviations may result in extremely large errors.
\end{enumerate}

The initializing set for this model is a singleton set chosen uniformly at random, and is assigned a preference chosen uniformly at random from the set of all preferences.
For each unassigned node, this model selects an assigned neighbor in one of the above ways, samples a distance value from the corresponding distribution, and chooses a preference uniformly at random from the set of preferences which are at that distance from the preference of selected assigned neighbor. So the time complexity of this model for assigning preferences for $\mathbb{T}$ topics is $O(r!(\sum_{i\in N} {deg}(i))\mathbb{T}) = O(r!m\mathbb{T})$.

\subsubsection{{Duplicating (RPM-D)}}
In this model, node $j$ is assigned a preference by duplicating the preference of its most similar assigned neighbor. 
This model pushes the similarity between a node and its most similar assigned neighbor to the extreme extent that, the preference to be assigned to the former is not just similar to the latter, but is exactly the same.

The initializing set for this model is a connected set of certain size $s$ which is obtained using the following iterative approach: start with a node chosen uniformly at random and then continue adding a new node to the set uniformly at random, from among the nodes that are connected to at least one node in the set. In our experiments, we choose $s$ itself to be uniformly at random from $\{1,\ldots,\lceil\sqrt{n}\rceil\}$. The nodes in this initializing set are assigned preferences based on RPM-IC.
The time complexity of this model for assigning preferences for $\mathbb{T}$ topics is $O((\sum_{i\in N} {deg}(i))\mathbb{T}) = O(m\mathbb{T})$.

\subsubsection{{Random (RPM-R)}}
In this model, preferences are assigned randomly to all the nodes without considering the distribution of distances from their neighbors, that is, without taking the social network effect into account. This model can be refined based on some known bias in preferences, for instance, if the topic is of social type, there would be a prior distribution on preferences owing to common external effects such as mass media.
The time complexity of this model for assigning the preferences for $\mathbb{T}$ topics is $O(n\mathbb{T})$.

\subsubsection{{Mean Similarity Model - Shortest Path Based (MSM-SP)}}

Unlike the models discussed so far, this model does not deduce the spread of preferences in a social network. Instead, it deduces the mean similarity between any pair of nodes, given the mean similarities of connected nodes.

Recall that cell $(i,j)$ of a {\em distance matrix} contains $d(i,j)$, the expected distance between preferences of nodes $i$ and $j$. 
We initialize all values in this matrix to $0$ for $i=j$ and to $1$ (the upper bound on the value of the distance) for any unconnected pair $\{i,j\}$.
In the case of a connected pair $\{i,j\}$, the value $d(i,j)$ is initialized to the actual observed expected distance (this value is known).
%
Following the initialization of the distance matrix, the next step is to update it.
%

Consider nodes $\{p,i,j\}$ where we know the expected distances $d(p,i)$ and $d(p,j)$ and we are interested in finding $d(i,j)$ via node $p$.
Given the preference of node $p$ and $d_x=d(p,i)$, let the preference of node $i$ be chosen uniformly at random from the set of preferences that are at a distance $\eta$ from the preference of node $p$, where $\eta$ is drawn from distribution $\mathcal{D}$  with mean $d_x$ (and some standard deviation).
Similarly, given $d_y=d(p,j)$, let the preference of node $j$ be obtained.
Using this procedure, the distance between the obtained preferences of nodes $i$ and $j$ via $p$ over several iterations and varying values of standard deviations, was observed to follow a bell curve; so we again approximate this distribution by $\mathcal{D}$.
Let the corresponding expected distance constitute the cell $(d_x,d_y)$ of a table, say $T_r$, where $r$ is the number of alternatives (for the purpose of forming a table, we consider only finite number of values of $d_x,d_y$).
It is clear that this distance is independent of the actual preference of node $p$.

We empirically observe that $T_r$ is different from $T_{r'}$ for $r \neq r'$.
Following are the general properties of $T_r$:
\begin{itemize}
\item $T_r(d_y,d_x) = T_r(d_x,d_y)$
\item $T_r(1-d_x,d_y) = T_r(d_x,1-d_y) = 1-T_r(d_x,d_y)$
\item $T_r(1-d_x,1-d_y) = T_r(d_x,d_y)$
\end{itemize}
We define an operator $\text{\textcircled{+}}_r$ as follows:
\begin{displaymath}
  d_x~\text{\textcircled{+}}_r~d_y=\begin{cases}
    T_r(d_x,d_y), & \text{if $d_x \leq 0.5$ and $d_y \leq 0.5$}\\
    \max\{d_x,d_y\}, & \text{if $d_x > 0.5$ or $d_y > 0.5$}
  \end{cases}
\end{displaymath}
The two different cases while defining $\text{\textcircled{+}}_r$ are based on the reasonable assumption that $d(i,j)$ via $p$ should be assigned a value which is at least $\max\{d(p,i),d(p,j)\}$ (but $T_r$ does not follow this rule when either $d(p,i)$ or $d(p,j)$ exceeds 0.5).

\begin{table}[h]
\begin{center}
\begin{tabular}{|>{\small}c|>{\small}c|>{\small}c|>{\small}c|>{\small}c|>{\small}c||>{\small}c|}
\hline 
0.00 &  0.10 &  0.20 &  0.30 &  0.40 &  0.50 & \slashbox{$d_x$}{$d_y$}	\\ \hline  \hline
     0.00   &  0.10  &  0.20  &  0.30  &  0.40  &  0.50 &  0.00 \\ \hline
     \multicolumn{1}{c|}{}  &  0.17  &  0.26  &  0.33  &  0.42  &  0.50 &  0.10  \\ \cline{2-7}
     \multicolumn{2}{c|}{}  &  0.32  &  0.37  &  0.43  &  0.50 &  0.20\\ \cline{3-7}
     \multicolumn{3}{c|}{}  &  0.40  &  0.45  &  0.50   &  0.30\\ \cline{4-7}
     \multicolumn{4}{c|}{}  &  0.47  &  0.50   &  0.40\\ \cline{5-7}
     \multicolumn{5}{c|}{}  &  0.50  &  0.50\\ \cline{6-7}
\end{tabular}
\end{center}
\caption{A partial view of table $T_5$ 
}
\label{tab:distance_table}
\end{table}

As the topics of our app had 5 alternatives, we obtain the table $T_5$ and hence $d_x~\text{\textcircled{+}}_5~d_y$ for any pair $\{d_x,d_y\}$.
In order to consider finite number of values of $d_x,d_y$ for forming the table, we only account for values that are multiples of 0.01 (and also round every entry in $T_r$ to the nearest multiple of 0.01).
Table~\ref{tab:distance_table} presents a partial view of $T_5$ which can be completed using the general properties of $T_r$ enlisted above; also it presents $d_x,d_y$ in multiples of 0.10 for ease of presentation.
Now the next question is to find $d(i,j)$ for any pair $\{i,j\}$.
In order to provide a fit to the distances obtained from the dataset, 
we initialize the distance matrix as explained in the beginning of this subsection (while rounding every value to the nearest multiple of 0.01) and update it
based on the {\em all pairs shortest path algorithm}~\cite{cormen2009introduction} with the following update rule:
\begin{center}
\textbf{if} $d(p,i)~\text{\textcircled{+}}_r~d(p,j) < d(i,j)$ \textbf{then} $d(i,j) = d(p,i)~\text{\textcircled{+}}_r~d(p,j)$,
\end{center}
where $r=5$ in our case.
The corresponding {similarity matrix} is obtained by assigning value $1-d(i,j)$ to its cell $(i,j)$.

The time complexity of deducing the mean distances between all pairs of nodes using MSM-SP is dominated by the all pairs shortest path algorithm, which is $O(n^2 \log n + nm)$ by Johnson's algorithm where the number of edges $m$ is generally small owing to sparsity of social networks.\\

We have already seen the time complexities of the other models for assigning preferences for $\mathbb{T}$ topics; the time complexity of deducing the mean distances between all pairs of nodes after that, is $O(\mathcal{G}_rn^2\mathbb{T})$, where $\mathcal{G}_r$ is the time complexity of computing the distance between two preferences ($\mathcal{G}_r = O(r^2)$ for Kendall-Tau distance).

\subsection{Validating the Models}

To validate a given model, we generated the preferences of all the nodes for $\mathbb{T}=10000$ simulated topics. Following this, we could get the distances between preferences of every pair of nodes in terms of normalized Kendall-Tau distance. 
In order to measure the error $err(\{i,j\})$ of this deduced model distribution against the distribution $\mathcal{D}$ with the actual values of $\mu_{ij}$ and $\sigma_{ij}$ for a particular pair of nodes $\{i,j\}$, we used the following two methods:
\begin{enumerate}
\item Kullback-Leibler (KL) divergence, a well-accepted way of measuring error between the actual and model distributions,
\item Absolute difference between the means of these two distributions, since some of our algorithms would be working with the distribution means.
\end{enumerate}

Using these two methods,
we measured the total error over all pairs of nodes as root mean square (RMS) error, that is, $\sqrt{\avg_{\{i,j\}} [err(\{i,j\})]^2}$.
Figure~\ref{fig:plot_models} provides a comparison among the models under study, with respect to these errors and running time. 
(Note that RMS KL divergence is not applicable for MSM-SP).
RPM-IC gave the least errors but at the cost of extremely high running time.
RPM-D and RPM-R ran fast but their errors were in a higher range.
RPM-S showed a good balance between the errors and running time; the way of choosing the assigned neighbor ($\mu$-based, $\sigma$-based, or random) did not show significant effect on its results. 
 MSM-SP was observed to be the best model when our objective was to deduce the mean distances between all pairs of nodes, and not the preferences themselves.

\begin{figure}[t!]
\centering
   \iftoggle{clr}{
\includegraphics[scale=.75]{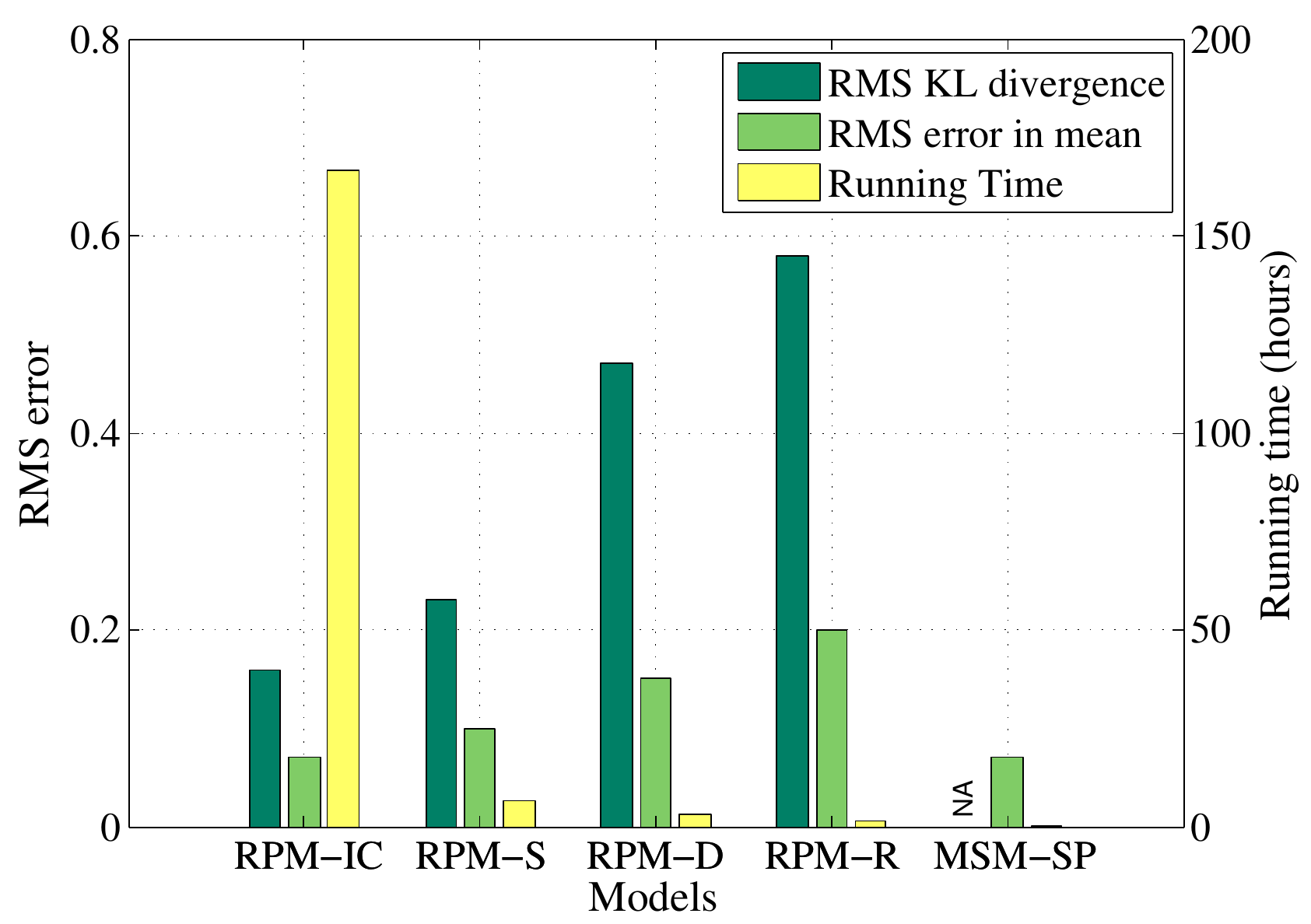}
   }{
\includegraphics[scale=.75]{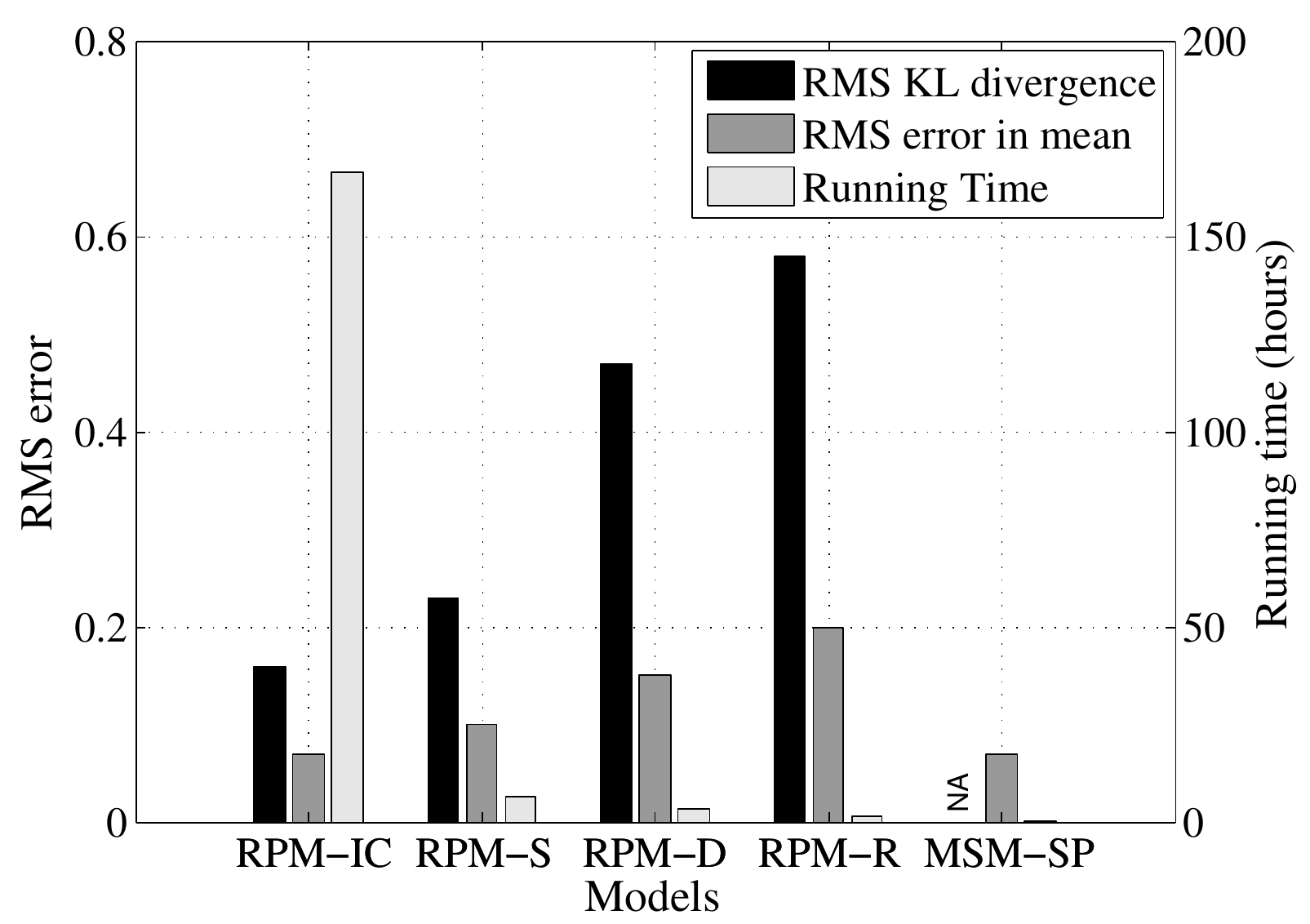}
}
\caption{Comparison among the considered models when run for 10000 iterations (or simulated topics)
}
\label{fig:plot_models}
\end{figure}


This concludes our study on modeling the spread of preferences in a social network. We now turn to the problem of determining the best set of representative nodes in a network.

\section{Scalable Preference Aggregation Problem}
\label{sec:problem_pasn}

Given a network with a set of nodes $N$ and an aggregation rule $f$, our objective is to choose a set of representative nodes $M \subseteq N$ of certain cardinality $k$, and aggregate their preferences to arrive at an aggregate preference that is `close enough' in expectation to the aggregate preference of $N$ using $f$. 
We now formalize this problem.
Table~\ref{tab:notation} presents the notation used in this chapter.

Let the expected distance between a set $S \subseteq N$ and a node $i \in N$ be
\begin{equation}
\label{eqn:dist}
d(S,i) = \min_{j \in S} d(j,i)
\end{equation}
We call $d(S,i)$ as the `expected' distance since $d(j,i)$ is the expected distance between nodes $j$ and $i$ with respect to their preferences.
Since $d(i,i)=0 \;\; \forall i \in N$, we have $d(S,j)=0 \;\; \forall j \in S$.
Let
\begin{equation}
\label{eqn:repr}
\Phi(S,i) \sim_\mathcal{U}\argmin_{j \in S} d(j,i)
\end{equation}
be a node chosen uniformly at random from the set of nodes in $S$ that are closest in expectation to node $i$ in terms of preferences. It can be said that $\Phi(S,i)$ represents node $i$ in set $S$. In other words, $\Phi(S,i)$ is the {\em representative} of $i$ in $S$.

The problem under consideration can be viewed as a setting where given certain individuals representing a population, every individual of the population is asked to choose one among them as its representative; now the representatives vote on behalf of the individuals who chose them.

\begin{table}[t]
\begin{center}
\begin{tabular}{l l}
\hline \hline
\T \B 
$N$ & set of nodes in the network \\ \hline
\T \B 
$E$ & set of edges in the network \\ \hline
\T \B 
$n$ & number of nodes in the network \\ \hline
\T \B 
$m$ & number of edges in the network \\ \hline
\T \B 
$i,j,p$ & typical nodes in the network \\ \hline
\T \B 
$r$ & number of alternatives \\ \hline
\T \B 
$P$ & preference profile of $N$ \\ \hline
\T \B 
$f$ & preference aggregation rule \\ \hline
\T \B 
$d(i,j)$ & expected distance between preferences of $i$ and $j$ \\ \hline
\T \B 
$c(i,j)$ & $1-d(i,j)$ \\ \hline
\T \B 
$\tilde{d}(x,y)$ & distance between preferences $x$ and $y$ \\ \hline
\T \B
$\Phi(S,i)$ & representative of node $i$ in set $S$ \\ \hline
\T \B 
$M$ & set of representatives who report their preferences \\ \hline
\T \B 
$k$ & $|M|$, cardinality of $M$ \\ \hline
\T \B 
$Q$ & profile containing unweighted preferences of $M$ \\ \hline
\T \B 
$Q'$ & profile containing weighted preferences of $M$ \\ \hline
\T \B 
$\Delta$ & error operator between aggregate preferences \\ \hline
\hline
\end{tabular}
\end{center}
\caption{Notation used in the chapter}
\label{tab:notation}
\end{table}

\subsection{Aggregating Preferences of Representative Nodes}

%
Recall that preference profile is a multiset containing preferences of the voters.
Let the preference profile of the population $N$ be $P$ and that of the selected set $M$ be $Q$. 
%
Suppose $M=\{i,j\}$ where
$j$ represents, say ten nodes
 including itself, 
while $i$ represents one
(only itself). 
If the preferences are aggregated by 
feeding $Q$ to aggregation rule $f$,
the aggregate preference $f(Q)$ so obtained may not reflect the preferences of the population, in general.
%
So
in order to capture this asymmetry in the importance of selected nodes, their preferences must be weighted. 
In our approach, the weight given to the preference of a node is precisely the number of nodes that it represents.
%

Let $Q'$ be the preference profile obtained by replacing every node's preference in $P$ by its uniquely chosen representative's preference. So, $k=|M|=|Q| \leq |Q'|=|P|=|N|=n$.
In our approach, the weight of a representative implies the number of times its preference appears in the new preference profile, that is, we use $R=Q'$.
So in the above example, 
the new  profile $R=Q'$ consists of ten preferences of $j$ and one of $i$.
Thus we aggregate the preferences of selected nodes using $f(Q')$.
%
%

\subsection{A Measure of `Close Enough'}

Now given $k$, our objective is to select a set of nodes $M$ such that $|M|=k$, who report their preferences such that, 
in expectation, the error incurred in using the aggregate preference $f(R)$ obtained by aggregating the preferences of the individuals in $M$ (in an unweighted manner if $R=Q$ or in a weighted manner if $R=Q'$) instead of $f(P)$ obtained by aggregating the preferences of the individuals in $N$, is minimized.
%
%
Note that an aggregation rule $f$ may not output a unique aggregate preference, that is, $f$ is a correspondence. 
So
the aggregation rule $f$ on the preferences of the entire population outputs $f(P)$ which is a set of preferences. 

Suppose $f(R)$ also is a set of several preferences, the question arises: which of these to choose as the output? As $f(P)$ is generally not known and all preferences in $f(R)$ are equivalent to us, we choose a preference from $f(R)$ uniformly at random and see how far we are from the actual aggregate preference, in expectation. 
In order to claim that a chosen preference in $f(R)$ is a good approximation, it suffices to show that it is close to at least one preference in $f(P)$.  Also, as any preference $y$ in $f(R)$ is chosen uniformly at random, we define the error incurred in using $f(R)$ instead of $f(P)$ as
\begin{equation}
\label{eqn:Delta}
f(P) \; \Delta \; f(R) = \mathbb{E}_{y \sim_\mathcal{U} f(R)} \left[ \min_{x \in f(P)} \tilde{d}(x,y) \right]
\end{equation}
where $\tilde{d}(x,y)$ is the distance between preferences $x$ and $y$ in terms of the same distance measure as $d(\cdot,\cdot)$. Notice that in general, $f(P) \; \Delta \; f(R) \neq f(R) \; \Delta \; f(P)$.
Also, $\Delta$ can be defined in several other ways depending on the application or the case we are interested in (worst, best, average, etc.). 
In this chapter, for the reasons explained above, we use the definition of $\Delta$ as given in Equation~(\ref{eqn:Delta}).

Recall that the distance between any pair of nodes is drawn from distribution $\mathcal{D}$, that is, the realized values for different topics are different in general. The value $f(P) \; \Delta \; f(R)$ can be obtained for every topic and hence the expected error $\mathbb{E} [ f(P) \; \Delta \; f(R) ]$ can be computed by averaging the values over all topics.
So now our objective is to find a set $M$ such that $\mathbb{E} [ f(P) \; \Delta \; f(R) ]$ is minimized.

\subsection{An Abstraction of the Problem}

If the aggregation rule is known, an objective function can be $\mathcal{F}(M) = 1- \mathbb{E} [ f(P)\;\Delta \;f(R) ]$ with the objective of finding a set $M$ that maximizes this value. 
However, even if the set $M$ is given, computing $\mathcal{F}(M)$ is computationally intensive for most aggregation rules and furthermore, hard for rules such as Kemeny. 
It can be seen that $\mathcal{F}(\cdot)$ is not monotone for non-dictatorial aggregation rules (the reader is referred to Figure~\ref{fig:plots_pasn} for the non-monotonic plots of Greedy-sum and Degree-cen algorithms since
in a run of these algorithms, a set of certain cardinality is a superset of any set having a smaller cardinality).
%
It can also be checked 
that $\mathcal{F}(\cdot)$ is neither submodular nor supermodular.
Owing to these properties of the objective function, even for simple non-dictatorial aggregation rules, it is not clear if one could efficiently find a set $M$ that maximizes $\mathcal{F}(\cdot)$, even within any approximation factor.
Moreover, the aggregation rule itself maybe unknown beforehand or maybe needed to be changed frequently in order to prevent strategic manipulation of preferences by the voters. 
%
This motivates us to propose an approach that finds set $M$ \textit{agnostic to the aggregation rule} being used.

To this end, we propose a property for preference aggregation rules, {\em weak insensitivity} which we define as follows.
\begin{definition}
\label{def:weakins}
A preference aggregation rule satisfies {\em weak insensitivity property\/} under a distance measure and a $\Delta$, if and only if for all $\epsilon_d$, a change of $\eta_i \leq \epsilon_d$ in the preferences of all $i$, results in a change of at most $\epsilon_d$ in the aggregate preference.
That is, $\forall \epsilon_d,$
\begin{equation}
\nonumber
\;\eta_i \leq \epsilon_d \;\;\; \forall i \in N \implies f(P) \; \Delta \; f(P')  \leq \epsilon_d
\end{equation}
where $P'$ is the preference profile of voters after deviations.
\end{definition}

We call it `weak' insensitivity property because it allows `limited' change in the aggregate preference (strong insensitivity can be thought of as a property that allows no change). 
This is potentially an important property that an aggregation rule should satisfy as it is a measure of its robustness in some sense.
It is clear that under normalized Kendall-Tau distance measure and $\Delta$ as defined in Equation~(\ref{eqn:Delta}), an aggregation rule that outputs a random preference does not satisfy weak insensitivity property as it fails the criterion for any $\epsilon_d < 1$, whereas dictatorship rule that outputs the preference of a single individual trivially satisfies the property.
For our purpose, we propose a weaker form of this property, which we call {\em expected weak insensitivity}.

\begin{definition}
\label{def:gaussweakins}
A preference aggregation rule satisfies {\em expected weak insensitivity property\/} under a distribution, a distance measure, and a $\Delta$, if and only if for all $\mu_d$, a change of $\eta_i$ in the preferences of all $i$, where $\eta_i$ is drawn from the distribution with mean $\delta_i \leq \mu_d$ and any permissible standard deviation $\sigma_d$, results in a change with an expected value of at most $\mu_d$ in the aggregate preference.
That is, $\forall \mu_d, \; \forall \text{ permissible } \sigma_d ,  $
\begin{equation}
\label{eqn:weak_ins}
\;\delta_i \leq \mu_d \;\;\; \forall i \in N \implies \mathbb{E} [ f(P) \; \Delta \; f(P') ] \leq \mu_d
\end{equation}
where $P'$ is the preference profile of voters after deviations.
\end{definition}
Note that in $\mathbb{E} [ f(P) \; \Delta \; f(P') ]$, the expectation is over the varying modified preferences of the agents (since $\eta_i$'s vary across instances and also, there are multiple preferences at a distance of $\eta_i$ from any given preference, in general). 
In this work, we study expected weak insensitivity property under  distribution $\mathcal{D}$, normalized Kendall-Tau distance, and $\Delta$ as defined in Equation~(\ref{eqn:Delta}). 
For distribution $\mathcal{D}$ with $\mu_d \in [0,1]$,
the permissible range of $\sigma_d$ depends on $\mu_d$. 
This range is wider for intermediate values of $\mu_d$ and shortens as we move towards the extremes.
In any case,
the permissible range for $\sigma_d$ cannot exceed $\frac{1}{\sqrt{12}} \approx 0.28$ (value at which the truncated Gaussian becomes a Uniform distribution), while for $\mu_d \in \{0,1\}$, the permissible $\sigma_d=0$.
%
\begin{observation}
We conducted extensive simulations for investigating empirical satisfaction of the expected weak insensitivity property under  distribution $\mathcal{D}$, normalized Kendall-Tau distance, and $\Delta$ as defined in Equation~(\ref{eqn:Delta}) by the considered aggregation rules. We observed that several aggregation rules satisfy this property, whereas most others violate it only by small margins; in fact, of the studied aggregation rules, only Veto was observed to violate this property by more than 20\%.

\end{observation}

\begin{lemma}
\label{lem:guarantee}
Given a distance measure and a $\Delta$, with a preference aggregation rule satisfying expected weak insensitivity property under  distribution $\mathcal{D}$,
if the expected distance between every individual and 
the set $M$ is at most $\epsilon_d \in [0,1]$,
then the expected error incurred in using $f(Q')$ instead of $f(P)$ is at most $\epsilon_d$. 
That is, for $\epsilon_d \in [0,1] , \; $
\begin{displaymath}
d(M,i) \leq \epsilon_d \;\; \forall i \in N \implies \mathbb{E} [ f(P) \; \Delta \; f(Q') ] \leq \epsilon_d
\end{displaymath}
\end{lemma}
\begin{proof}
In the preference profile $P$ of all voters, the preference of any node $i \in N$ is replaced by the preference of its representative node $p=\Phi(M,i)$
to obtain $Q'$. 
From Equations~(\ref{eqn:dist}), (\ref{eqn:repr}), and the hypothesis, we have $d(p,i) \leq \epsilon_d$.

Since in $P$, preference of every $i$ is replaced by that of the corresponding $p$ to obtain a new  profile $Q'$, and distance between $i$ and $p$ is distributed according to  distribution $\mathcal{D}$ with mean $d(p,i)$
and some standard deviation $\sigma_d$, 
the above is equivalent to node $i$ deviating its preference by some value which is drawn from distribution $\mathcal{D}$ with mean $d(p,i) = d(M,i)$. 
So we can map these variables to the corresponding variables in Equation~(\ref{eqn:weak_ins}) as follows: $\delta_i = d(M,i) \; \forall i$, $\mu_d = \epsilon_d$, and $P'=Q'$.
Also, recall that in $\mathbb{E} [ f(P) \; \Delta \; f(P') ]$, the expectation is over varying modified preferences of the agents, while in $\mathbb{E} [ f(P) \; \Delta \; f(Q') ]$, the expectation is over varying preferences of the agents' representatives in $M$ with respect to different topics (and hence preferences) of the agents.
These are equivalent given $P'=Q'$.
As this argument is valid for any permissible $\sigma_d$, the result follows.
\end{proof}

So under the proposed model and for aggregation rules satisfying the expected weak insensitivity property, this lemma establishes a 
 relation between (a) the closeness of the chosen representative set to the population in terms of expected distance and (b) the error incurred in the aggregate preference if that set is chosen as the representative set.
We now return to our goal of abstracting the problem of determining a representative set, by proposing an approach that is agnostic to the aggregation rule being used.

\subsection{Objective Functions in the Abstracted Problem}

Recall that $c(\cdot,\cdot) = 1 - d(\cdot,\cdot)$. 
Our objective is now to find a set of critical nodes $M$ that maximizes some objective function, with the hope of minimizing  $\mathbb{E} [ f(P) \; \Delta \; f(R) ]$ where $R=Q'$ in our case.
As the aggregation rule is anonymous, in order to ensure that the approach works well, even for rules such as random dictatorship, the worst-case objective function for the problem under consideration, representing least expected similarity, is
\begin{equation}
\label{eqn:influence_min}
\rho(S) = \min_{i \in N} c(S,i)
\end{equation}
The above is equivalent to $\max_{i \in N} d(S,i) = 1-\rho(S)$.
Thus $\epsilon_d = 1-\rho(S)$ in Lemma~\ref{lem:guarantee} and so this objective function offers a guarantee on $\mathbb{E} [ f(P) \; \Delta \; f(Q') ]$ irrespective of the aggregation rule, provided it satisfies the expected weak insensitivity property.
We will provide a detailed analysis for the performance guarantee of an algorithm that aims to maximize $\rho(S)$, in Section~\ref{sec:greedymin_guarantee}.

Now the above worst-case objective function ensures that our approach works well even for aggregation rules such as random dictatorship.
However, such extreme aggregation rules are seldom used in real-world scenarios; hence, an alternative objective function, representing average expected similarity, or equivalently sum of expected similarities, is
\begin{equation}
\label{eqn:influence_sum}
\psi(S) = \sum_{i \in N} c(S,i)
\end{equation}
        We will look into the desirable properties of an algorithm that aims to maximize $\psi(S)$, in Section~\ref{sec:cooperativeview}.

We now turn towards the problem of maximizing the above two objective functions.
\begin{proposition}
Given constants $\chi$ and $\omega$, \\
(a) it is NP-hard to determine whether there exists a set $M$ consisting of $k$ nodes such that $\rho(M) \geq \chi$, and \\
(b) it is NP-hard to determine whether there exists a set $M$ consisting of $k$ nodes such that $\psi(M) \geq \omega$.
\end{proposition}
\begin{proof}
We reduce an NP-hard Dominating Set problem instance to the problem under consideration.
Given a graph $G$ of $n$ vertices, the dominating set problem is to determine whether there exists a set $D$ of $k$ vertices such that every vertex not in $D$, is adjacent to at least one vertex in $D$.

Given a dominating set problem instance, we can construct a weighted undirected complete graph $H$ consisting of the same set of vertices as $G$ such that, the weight $c(i,j)$ of an edge $(i,j)$ in $H$ is some high value (say $0.9$) if there is edge $(i,j)$ in $G$, else it is some low value (say $0.6$). 

Now there exists a set $D$ of $k$ vertices in $G$ such that the distance between any vertex in $G$ and any vertex in $D$ is at most one, if and only if there exists a set $M$ of $k$ vertices in $H$ such that $\rho(M) \geq 0.9$ or $\psi(M) \geq k+0.9(n-k)$. Here $\chi=0.9$ and $\omega=k+0.9(n-k)$.
This shows that the NP-hard dominating set problem is a special case of the problems under consideration, hence the result.
\end{proof}

A function $h(\cdot)$ is said to be {\em submodular} if,
for all $v \in N \setminus T$ and for all $S,T$ such that $S \subset T \subset N$,
\begin{equation}
\label{eqn:submodular}
\nonumber
h(S \cup \{v\}) - h(S) \geq h(T \cup \{v\}) - h(T)
\end{equation}

\begin{proposition}
\label{prop:submodularity}
The objective functions $\rho(\cdot)$ and $\psi(\cdot)$ 
are non-negative, monotone increasing, and submodular.
\end{proposition}

We provide a proof of Proposition~\ref{prop:submodularity} in Appendix~\ref{app:submod_proof}.
%

For a non-negative, monotone increasing, submodular function, the greedy hill-climbing algorithm (selecting elements one at a time, each time choosing an element that provides the largest marginal increase in the function value), gives a $(1-\frac{1}{e}) \approx 0.63$-approximation to the optimal solution~\cite{nemhauser1978analysis}.
As the considered objective functions in Equations~(\ref{eqn:influence_min}) and (\ref{eqn:influence_sum}) satisfy these properties, we use the greedy hill-climbing algorithm to obtain a good approximation to the optimal solution.
Moreover, as desired, the functions are agnostic to the aggregation rule being used.

We next devise algorithms for determining a set of representative nodes and present their performances with the aid of extensive experimentation.

\section{Selection of the Representative Set: Algorithms and their Performances}
\label{sec:results}

Given the number of nodes to be selected $k$, our objective is to find a set $M$ of size $k$ such that 
$\mathbb{E} [ f(P) \; \Delta \; f(R) ]$ is minimized, where $R=Q'$ or $Q$ depending on the algorithm being used.
%

\subsection{Algorithms for Finding Representatives}
\label{sec:algos}

Recall that the preference profile of $N$ is $P$, that of $M$ is $Q$, and that obtained by replacing every node's preference in $P$ by that of its uniquely chosen representative in $M$, is $Q'$.

\begin{itemize}

\item \textbf{Greedy-orig} (Greedy hill-climbing for maximizing $1-\mathbb{E} [ f(P) \; \Delta \; f(Q') ]$):
Initialize $M$ to $\{\}$.
Until $|M|=k$, choose a node $j \in N \setminus M$ that 
minimizes the expected error or equivalently, maximizes $1-\mathbb{E} [ f(P) \; \Delta \; f( {Q}_M') ]$,
where $Q_M'$ is the preference profile obtained by replacing every node's preference in $P$ by the preference of its uniquely chosen representative in $M$.
Note that the optimal set would depend on the aggregation rule $f$.
Its time complexity for obtaining $M$ and hence $R$ is $O(kn\mathcal{T}_f)$, where $\mathcal{T}_f$ is the time complexity of obtaining an aggregate preference using the aggregation rule $f$. For instance, $\mathcal{T}_f$ is $O(rn)$ for plurality and $O(1)$ for dictatorship.

\item \textbf{Greedy-sum} (Greedy hill-climbing for maximizing $\psi(\cdot)$):
Initialize $M$ to $\{\}$.
Until $|M|=k$, choose a node $j \in N \setminus M$ that 
maximizes $\psi(M \cup \{j\}) - \psi(M)$.
Then obtain $f(R) = f(Q')$.
If the similarity matrix is known, its time complexity for obtaining $M$ and hence $R$ is $O(kn^2)$.
If the similarity matrix is unknown, the time complexity for deriving it is largely decided by the model used for deducing the mean distances between all pairs of nodes. 

\item \textbf{Greedy-min} (Greedy hill-climbing for maximizing $\rho(\cdot)$):
Initialize $M$ to $\{\}$.
Until $|M|=k$, choose a node $j \in N \setminus M$ that 
maximizes $\rho(M \cup \{j\}) - \rho(M)$.
Then obtain $f(R) = f(Q')$.
Its time complexity is the same as that of Greedy-sum.

\item \textbf{Degree-cen} (Degree centrality Heuristic): Choose $k$ nodes having the maximum degrees.
Then obtain $f(R) = f(Q)$.
Its time complexity for obtaining $M$ is $O(n(k+\log n))$.
\item \textbf{Random-poll} (Random selection without representation):
Choose a set of $k$ nodes uniformly at random.
Then obtain $f(R) = f(Q)$.
\end{itemize}
For all algorithms, the time complexity for arriving at an aggregate preference $f(R)$ depends on the aggregation rule $f$.
For dictatorship rule, in Random-poll, if the dictator is not in $M$, then the output is the preference of a node in $M$ chosen uniformly at random, else the output is the dictator's preference itself; in all other methods, the output is the preference of the dictator's unique representative in $M$.



Before proceeding to experimental observations, we provide an analytical guarantee on the performance of the Greedy-min algorithm, for aggregation rules satisfying the expected weak insensitivity property.

\subsection{A Guarantee on the Performance of Greedy-min Algorithm}
\label{sec:greedymin_guarantee}

The following result shows the performance guarantee of the Greedy-min algorithm.
\begin{theorem}
\label{thm:approx_guarantee}
For an aggregation rule satisfying expected weak insensitivity, the error incurred in using the aggregate preference given by the Greedy-min algorithm instead of the actual aggregate preference, is at most $\left(1 - \left(1-\frac{1}{e}\right) \rho^*\right)$, where $\rho^* = \max_{S \subseteq N,|S|\leq k} \rho(S)$.
\end{theorem}
\begin{proof}
Let $S^G$ be a set obtained using greedy hill-climbing algorithm for maximizing $\rho(\cdot)$. Since greedy hill-climbing provides a $\left(1-\frac{1}{e}\right)$-approximation to the optimal solution, we have
\begin{align*}
 &\;\;\;\;\;\;\;\;\;\;
 \rho(S^G) = \min_{i \in N} c(S^G,i) \geq \left(1-\frac{1}{e}\right) \rho^* \\
 &\implies
  1 - \max_{i \in N} d(S^G,i) \geq \left(1-\frac{1}{e}\right) \rho^* \\
&\implies
\max_{i \in N} d(S^G,i) \leq  1-\left(1-\frac{1}{e}\right) \rho^*\\
&\implies
d(S^G,i) \leq  1-\left(1-\frac{1}{e}\right) \rho^*, \;\;\; \forall i\in N
\end{align*}
For an aggregation rule satisfying expected weak insensitivity property, from Lemma~\ref{lem:guarantee}, when the representative set $M=S^G$, we have
\begin{equation}
\nonumber
\mathbb{E} [ f(P) \; \Delta \; f(Q') ] \leq 1 - \left(1-\frac{1}{e}\right) \rho^*
\end{equation}
\end{proof}

It is to be noted that though the approximation ratio given by the greedy algorithm is modest in theory, it has been observed in several domains that its performance is close to optimal in practice when it comes to optimizing non-negative, monotone increasing, submodular functions.

\subsection{Experimental Observations}

After obtaining the representative set using the above algorithms, we tested their performance on $\mathbb{T}=10000$ topics or preference profiles generated using the RPM-S model (with the assigned neighbor chosen in a random way) on the Facebook data that we collected. 
Owing to the nature of the Random-poll algorithm, we ran it sufficient number of times to get an independent representative set each time, and then defined the performance as the average over all the runs.
The values of $\mathbb{E} [ f(P) \; \Delta \; f(R) ]$ were computed using extensive simulations with the considered aggregation rules. 
It can also be noted that tie-breaking between alternatives was rarely required in our experiments, and so the error $\mathbb{E} [ f(P) \; \Delta \; f(R) ]$ can be viewed as Kendall-Tau distance between the actual aggregate preference and the obtained aggregate preference.

For random dictatorship, with increasing value of $k$, the error remained almost constant (with a very slight descent) using all algorithms; while for dictatorship, the error monotonically decreased with $k$ using algorithms except Random-poll (in which case it remained almost constant).
Apart from dictatorship and random dictatorship, the plots for all aggregation rules were similar (albeit with slightly different scaling) to the ones plotted in Figure~\ref{fig:plots_pasn}. 

\begin{figure}
\centering
   \iftoggle{clr}{
\includegraphics[scale=0.78]{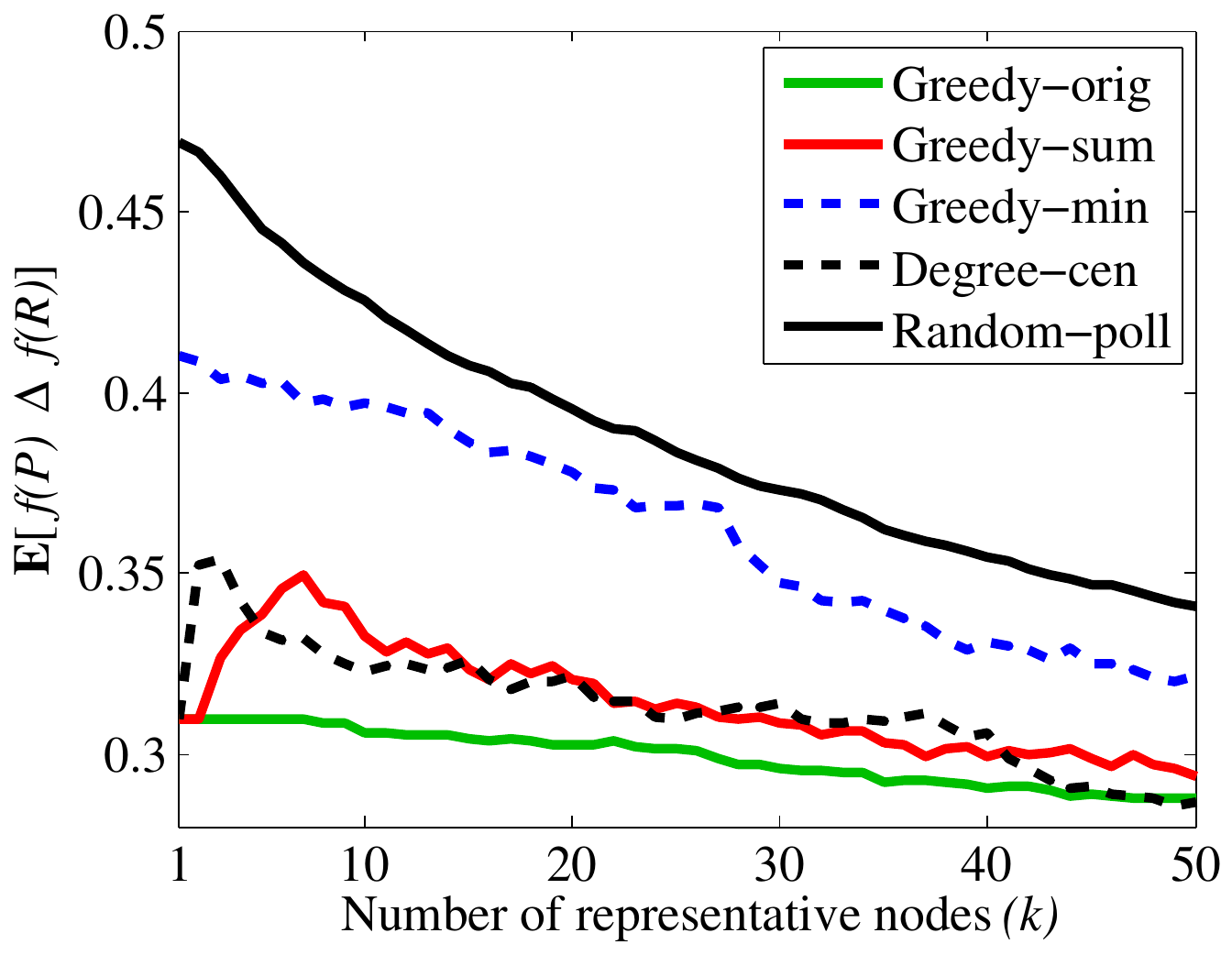} 
   }{
\includegraphics[scale=0.78]{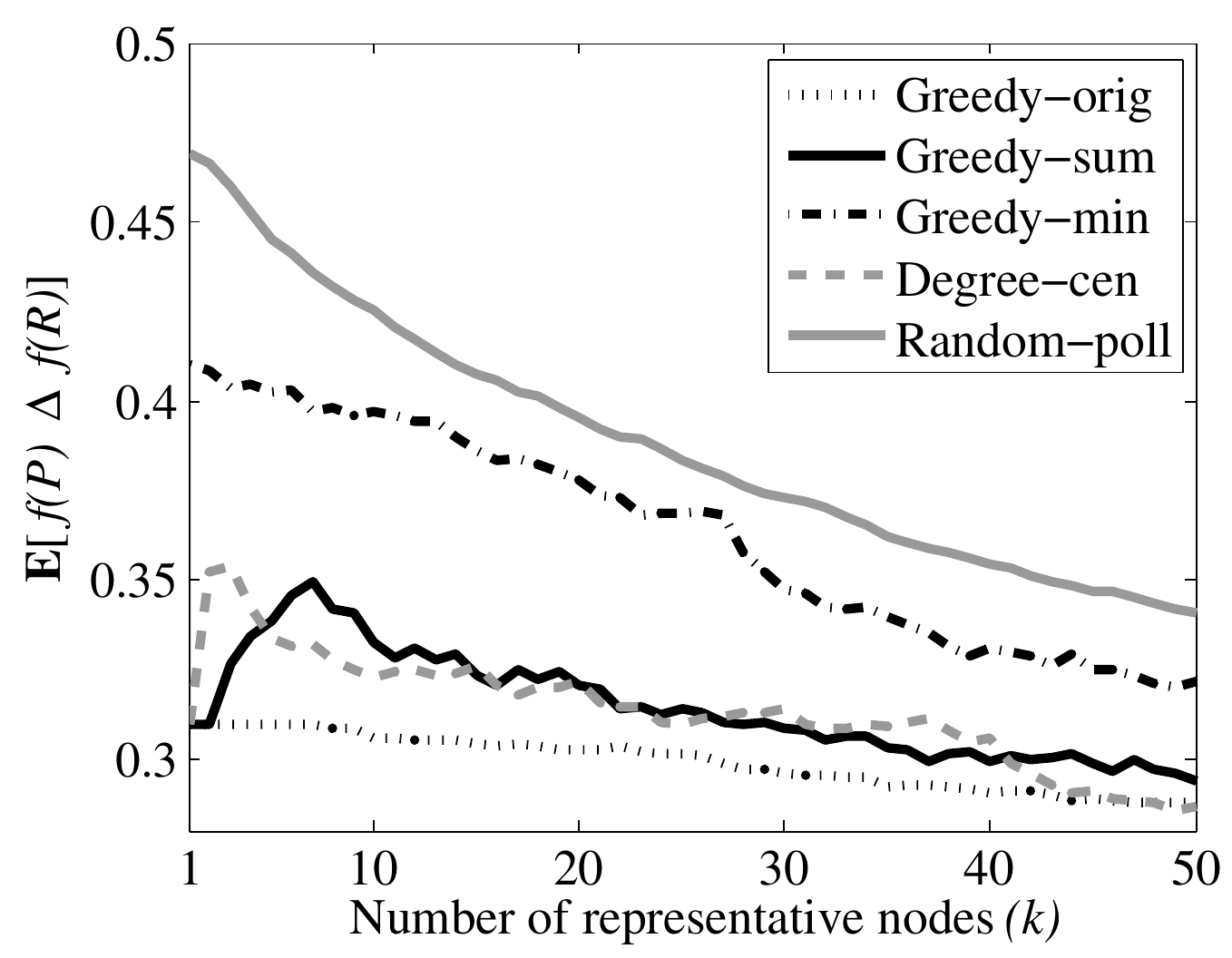} 
}
\caption{
Comparison among algorithms for plurality aggregation rule on Facebook dataset
}
\label{fig:plots_pasn}
\end{figure}

Our key observations are as follows:
\begin{itemize}
\item Greedy-orig performed the best throughout, however its execution took a few days even for computationally fast aggregation rules such as plurality; so it is practically infeasible to run this algorithm for computationally intensive rules, for example, Kemeny.
\item Greedy-sum performed close to Greedy-orig;
but its plots displayed non-monotonicity especially in the lower range of $k$, and so a higher $k$ might not always lead to a better result.
\item Greedy-min performed better than Random-poll for low values of $k$; this difference in performance decreased for higher values of $k$.
The effect of satisfaction or otherwise of expected weak insensitivity was not very prominent, because the property is not violated by an appreciable enough margin for any aggregation rule.
Nonetheless, the expected weak insensitivity property does provide a guarantee on the performance of Greedy-min for an aggregation rule.

\item As mentioned earlier, the performance of Random-poll is based on an average over several runs; the variance in performance was very high for low values of $k$, and the worst case performance was unacceptable. The variance was acceptable for higher values of $k$.

\item The performance of the simple Degree-cen heuristic showed a perfect balance between performance and running time. This demonstrates that high degree nodes indeed serve as good representatives of the population.
%

\item As mentioned earlier, the running time of Greedy-orig was unacceptably large. Greedy-sum and Greedy-min took about a minute to run, while Random-poll and Degree-cen took negligible time (less than a second).
\end{itemize}

It is clear that the Greedy-sum algorithm exhibited consistently excellent performance. In fact, the node which is chosen as the first node 
gives almost the same result as when the representative set is large. We study the selection of this first representative from a different viewpoint, that of cooperative game theory. 

\subsection{A Cooperative Game Theoretic Viewpoint of the Greedy-sum Algorithm}
\label{sec:cooperativeview}

We have seen that in order to maximize the objective function
$
\psi(S) = \sum_{i \in N} c(S,i)
$,
the greedy hill-climbing algorithm first chooses a node $j$ that maximizes
$\sum_{i \in N} c(i,j)$ or equivalently $\sum_{i \in N, i\neq j} c(i,j)$ (since $c(j,j)=1$ for all $j$'s).
It has been shown in \cite{garg2013novel} that the term $\sum_{i \in N, i\neq j} c(i,j) = \phi_j(\nu)$ is the Shapley value allocation of player $j$, in a convex Transferable Utility (TU) game $(N,\nu)$ with the characteristic function 
\begin{equation}
\label{eqn:game}
\nu(S) = \sum_{\substack{i,j \in S\\ i \neq j}} c(i,j) 
\end{equation}
In some sense, this characteristic function is an indication of how tightly knit a group is, or how similar the members of a set $S$ are to each other.


There are several solution concepts in the literature on cooperative game theory, of which we consider Nucleolus, Shapley value, Gately point, and $\tau$-value, which individually possess a number of desirable properties. The reader is referred to \cite{straffin1993game,tijs1987axiomatization} for further details.
In the game as defined in Equation~(\ref{eqn:game}), we show that the aforementioned solution concepts coincide.
The intuition is that the core is symmetric about a single point, which is a primary reason why the above solution concepts coincide with that very point. 
We now formally prove the coincidence of the above solution concepts.
%

Let $Nu(\nu),Gv(\nu),\tau(\nu)$ be Nucleolus, Gately point, $\tau$-value of the TU game $(N,\nu)$, and $\phi(\nu)=(\phi_j(\nu))_{j\in N}$ and so on.

\begin{theorem} 
\label{shapeqnu}
 For the TU game defined by Equation~(\ref{eqn:game}),
\begin{equation}
\nonumber
 \phi(\nu)=Nu(\nu)=Gv(\nu)=\tau(\nu).
\end{equation}
\end{theorem}
We provide a proof of Theorem~\ref{shapeqnu} in Appendix~\ref{app:coincidence}.
So one possible reason for the excellent performance of the Greedy-sum algorithm is that, it aims to maximize a term that is the solution suggested by several solution concepts for a TU game capturing the similarities within a group.

\section{Discussion}
\label{sec:conclusion_pasn}

This chapter focused on two problems on preference aggregation with respect to social networks, namely, (1) how preferences are spread in a social network and (2) how to determine the best set of representative nodes in the network.
We started by motivating both these problems. Based on the Facebook app that we had developed, we developed a number of simple and natural models, of which RPM-S showed a good balance between accuracy and running time; while MSM-SP was observed to be the best model when our objective was to deduce the mean distances between all pairs of nodes, and not the preferences themselves.
We then formulated an objective function for representative-set selection and followed it up with two alternative objective functions for practical usage. 
We then proposed algorithms for selecting best representatives, wherein we provided a guarantee on the performance of the Greedy-min algorithm, subject to the aggregation rule satisfying the expected weak insensitivity property; we also studied the desirable properties of the Greedy-sum algorithm from the viewpoint of cooperative game theory. 
We also observed that degree centrality heuristic performed very well, thus showing the ability of high-degree nodes to serve as good representatives of the population. 
Our main finding is that harnessing social network information in a suitable way will achieve scalable preference aggregation of a large population and will certainly outperform random polling based methods.

We now provide notes on the effectiveness of our approach for aggregating preferences with respect to personal and social topics, as well as on an alternative way of defining the distribution of preferences between nodes.

\subsection{A Note on Personal vs. Social Topics}

Throughout this chapter, we focused on aggregating preferences across all topics, without classifying them into different types (such as personal or social). 
We now provide a brief note on personal versus social type of topics.

As noted earlier, the performance of Random-poll is close to that of Greedy-min, a primary reason being the low standard deviation of the mean distances (\stdoverall). This also applies when the topics are restricted to social topics, which again has a low standard deviation of \stdsocial. The performance difference between the two is visible when it comes to aggregating preferences with respect to personal topics, which has a higher standard deviation of \stdpersonal. Also note that the variance in performance of Random-poll is unacceptably high for lower values of $k$, but is acceptable for higher values. 
This justifies its usage for social topics with a non-negligible sample size.

A high level of similarity between unconnected nodes with respect to social topics can perhaps be attributed to the impact of news and other common channels. It may also be justified by a theory of political communication~\cite{huckfeldt1995political} which stresses the importance of citizen discussion beyond the boundaries of cohesive groups for the dissemination of public opinion.

\subsection{An Alternative Way of Modeling the Distribution of Preferences between Nodes}
\label{sec:alt_dist}

Throughout this chapter, we conducted the study based on the distribution of (Kendall-Tau or Footrule) distances between nodes with respect to their preferences. 
Though modeling the distribution this way makes it a reasonably general approach, a primary drawback is the way distances between preferences themselves are distributed, that is, there are several preferences which are at a distance of, say 0.5, from a given preference as compared to those at a distance of, say 0.3, which leads to a bias in distances to be concentrated in the intermediate range. 

One could model the distribution of preferences between two nodes in a more explicit way, by modeling the correlation between the positions of an alternative in the preferences of the two nodes. 
For instance, the distribution between two nodes $\{i,j\}$ could answer the following question: what is the probability that an alternative is ranked at position $x$ by node $i$, given that it is ranked at position $y$ by node $j$?
We believe this way of modeling the distribution as a promising future direction to this work.


\vspace{10mm}
This concludes the technical contributions of the thesis. The next chapter will summarize the thesis by presenting conclusions of the addressed problems in brief, followed by some interesting future directions.

\begin{subappendices}

\chapter*{Appendix for Chapter~\ref{chap:pasn}}
\addcontentsline{toc}{chapter}{Appendix for Chapter~\ref{chap:pasn}}
\label{chap:appendix_pasn}


\section{Proof of Proposition \ref{prop:submodularity}}
\label{app:submod_proof}

\begin{customprop}{\ref{prop:submodularity}}
The objective functions $\rho(\cdot)$ and $\psi(\cdot)$ 
are non-negative, monotone increasing, and submodular.
\end{customprop}
\begin{proof}
We prove the properties in detail for $\psi(\cdot)$. The proof for $\rho(\cdot)$ is similar.

Consider sets $S,T$ such that $S \subset T \subset N$ and a node $v \in N \setminus T$. It is clear from Equation~(\ref{eqn:influence_sum}) that $\psi(\cdot)$ is non-negative.
Let 
$x_i = c(S,i),\;y_i = c(T,i),\;\bar{x_i} = c(S \cup \{v\} ,i),\;\bar{y_i} = c(T \cup \{v\} ,i)$.
%
For any $i \in N$,
\begin{displaymath}
c(S,i) = \max_{j \in S \subseteq T} c(j,i) \leq \max_{j \in T} c(j,i) = c(T,i)
\end{displaymath}
\begin{equation}
\label{eqn:x_y}
\implies x_i \leq y_i
\end{equation}
That is, $\psi(\cdot)$ is monotone. Similarly, it can be shown that
\begin{equation}
\label{eqn:three}
\bar{x_i} \leq \bar{y_i} \; ; \;\;\;
x_i \leq \bar{x_i} \; ; \;\;\;
y_i \leq \bar{y_i}
\end{equation}
%
\begin{eqnarray}
\hspace{-5mm}
\text{Now,} \;\;\; y_i < \bar{y_i} 
\label{eqn:maxT} & \implies & k \notin \argmax_{j \in T \cup \{v\}} c(j,i) \text{ } \forall k \in T \text{ }\\
\label{eqn:maxS} & \implies & k \notin \argmax_{j \in S \cup \{v\}} c(j,i) \text{ } \forall k \in S \subseteq T \text{ }\\
\label{eqn:y_ybar_x_xbar}
& \implies & x_i < \bar{x_i} 
\end{eqnarray}
The contrapositive of the above, from Inequalities~(\ref{eqn:three}) is
\begin{equation}
\label{eqn:contra}
x_i = \bar{x_i} \implies y_i = \bar{y_i} 
\end{equation}
Also, from Implications~(\ref{eqn:maxT}) and (\ref{eqn:maxS}),
\begin{eqnarray}
\nonumber 
y_i < \bar{y_i} 
& \implies & \{v\} = \argmax_{j \in T \cup \{v\}} c(j,i) = \argmax_{j \in S \cup \{v\}} c(j,i) \\
\label{eqn:xbar_equals_ybar} & \implies & \bar{x_i} = \bar{y_i}
\end{eqnarray}
%
%
Now from Inequalities~(\ref{eqn:three}), depending on node $i$, four cases arise that relate the values of $\bar{x_i} - x_i$ and $\bar{y_i} - y_i$.
\begin{description}
\item[Case 1:] $x_i = \bar{x_i}$ and $y_i = \bar{y_i}$:\\
In case of such an $i$, we have $\bar{x_i} - x_i = \bar{y_i} - y_i$
\item[Case 2:] $x_i = \bar{x_i}$ and $y_i < \bar{y_i}$:\\
By Implication~(\ref{eqn:contra}), there does not exist such an $i$.
\item[Case 3:] $x_i < \bar{x_i}$ and $y_i = \bar{y_i}$:\\
In case of such an $i$, we have $\bar{x_i} - x_i > \bar{y_i} - y_i$
\item[Case 4:] $x_i < \bar{x_i}$ and $y_i < \bar{y_i}$:
For such an $i$,
\begin{eqnarray}
\hspace{-4mm}
\nonumber
\bar{x_i} - x_i &=& \bar{y_i} - x_i \;\;\;\text{ (from Implications~(\ref{eqn:y_ybar_x_xbar}) and (\ref{eqn:xbar_equals_ybar}))} \\
\nonumber
&\geq& \bar{y_i} - y_i \;\;\;\text{ (from Inequality~(\ref{eqn:x_y}))} 
\end{eqnarray}
From the above cases, we have 
\begin{eqnarray}
\hspace{-4mm}
\nonumber
& & \bar{x_i} - x_i \geq \bar{y_i} - y_i, \;\;\;\forall i \in N \\
\nonumber
&\implies& \sum_{i \in N} (\bar{x_i} - x_i) \geq \sum_{i \in N} (\bar{y_i} - y_i) \\
\nonumber
&\implies& \sum_{i \in N} c(S\cup \{v\},i) - \sum_{i \in N} c(S,i) 
\geq \sum_{i \in N} c(T\cup \{v\},i) - \sum_{i \in N} c(T,i) \\
\nonumber
&\implies& \psi(S \cup \{v\}) - \psi(S) \geq \psi(T \cup \{v\}) - \psi(T)
\end{eqnarray}
\end{description} 
As the proof is valid for any $v \in N \setminus T$ and for any $S,T$ such that $S \subset T \subset N$, the result is proved.
\end{proof}

\section{Proof of Theorem~\ref{shapeqnu}}
\label{app:coincidence}

\begin{customthm}{\ref{shapeqnu}}
 For the TU game defined by Equation~(\ref{eqn:game}),
\begin{equation}
\nonumber
 \phi(\nu)=Nu(\nu)=Gv(\nu)=\tau(\nu).
\end{equation}
\end{customthm}

\begin{proof}
In Equation~(\ref{eqn:game}), when $|S|=2$ where $S=\{i,j\}$, 
\begin{equation}
\label{eqn:for2}
\nu(\{i,j\})=c(i,j)
\end{equation}
So we have
\begin{equation}
\label{eqn:pair}
\nu(S) = \sum_{\substack{T \subseteq N\\ |T|=2}} \nu(T)
\end{equation}
Now from Equation~(\ref{eqn:for2}),
 the Shapley value allocation for each $j \in N$, can be rewritten as
%
\begin{align}
\label{shapleyeqhalf}
\hspace{-2mm}
\phi_j(\nu) = \frac{1}{2} \sum_{\substack{i \in N\\i \neq j}}c(i,j)
= \frac{1}{2} \sum_{\substack{i \in N\\i \neq j}}\nu(\{i,j\})
= \frac{1}{2}\sum_{\substack{S \subseteq N\\ j \in S\\|S|=2}} \nu(S)
\end{align}

 The proof for $\phi(\nu)=Nu(\nu)$ follows from \cite{chun2007coincidence}.
 %
  Furthermore, it has been shown in \cite{chun2007coincidence} that, for the TU game satisfying Equation~(\ref{shapleyeqhalf}), for each $S \subseteq N$,
 \begin{equation}
  \nonumber
  \nu(S) - \sum_{i \in S}\phi_i(\nu) = \nu(N\backslash S) - \sum_{i \in N\backslash S}\phi_i(\nu)
 \end{equation}
 %
 When $S = \{i\}$, we have
 \begin{equation}
 \nonumber
  \nu(\{i\}) - \phi_i(\nu) = \nu(N\backslash\{i\}) - \sum_{\substack{j\in N\\j\neq i}}\phi_j(\nu)
 \end{equation}
 So, the propensity to disrupt for player $i$ \cite{straffin1993game} for the Shapley value allocation is
 \begin{equation}
 \nonumber
 d_i(\phi(\nu)) = \frac{\sum_{j\in N,j\neq i}\phi_j(\nu) - \nu(N\backslash \{i\})}{\phi_i(\nu) - \nu(\{i\})} = 1
 \end{equation}
 As the propensity to disrupt is equal for all the players ($=1$), 
 this allocation is the Gately point \cite{straffin1993game}, that is,
 $
  \phi(\nu)=Gv(\nu).
  $
Let $M(\nu)=(M_i(\nu))_{i\in N}$ and $m(\nu)=(m_i(\nu))_{i\in N}$, where $M_i(\nu)=\nu(N)-\nu(N\backslash \{i\})$ and $m_i(\nu)=\nu(\{i\})$.
For a convex game, 
$
\tau(\nu)= \lambda M(\nu) + (1-\lambda)m(\nu)
$,
where $\lambda\in [0,1]$ is chosen such that \cite{chun2007coincidence},
\begin{align}
\label{tausatis}
\sum_{i \in N} \left[ \lambda M_i(\nu) + (1-\lambda) m_i(\nu) \right] = \nu(N)
\end{align}
From 
Equation~(\ref{eqn:pair}),
we have
\begin{eqnarray*}
M_i(\nu)
= \sum_{\substack{S\subseteq N \\ |S| = 2}}\nu(S) - \sum_{\substack{S\subseteq N\backslash \{i\} \\ |S| = 2}}\nu(S) 
= \sum_{\substack{S\subseteq N \\ i \in S \\ |S| = 2}}\nu(S)
\end{eqnarray*}
This, with Equation~(\ref{tausatis}) and the fact that for our game, for all $i$, $m_i(\nu) = \nu(\{i\}) = 0$,
\begin{equation*}
 \nu(N) \;=\; \lambda\sum_{i\in N}M_i(\nu) 
\;=\; \lambda\sum_{i\in N} \sum_{\substack{S\subseteq N \\ i \in S \\ |S| = 2}}\nu(S)
\;=\; 2\lambda\sum_{\substack{S\subseteq N \\ |S| = 2}}\nu(S) \\ 
\end{equation*}
Using Equation~(\ref{eqn:pair}), we get $\lambda = \frac{1}{2}$.
So we have
\begin{equation}
 \nonumber
\tau_i(\nu) = \frac{1}{2}M_i(\nu) + \frac{1}{2}m_i(\nu) = \frac{1}{2}\sum_{\substack{S \subseteq N\\ i \in S\\|S| = 2}} \nu(S)
\end{equation}
This, with 
Equation~(\ref{shapleyeqhalf}), gives
$
\phi(\nu) = \tau(\nu)
$
\end{proof}

\section{Description of the Facebook App}
\label{app:facebook_app}

\subsection{Overview}
\label{sec:overview}

Online social networking sites such as Facebook, Twitter, and Google+ are highly popular in the current age; for instance, Facebook has over 1.5 billion users as of 2016. Using such online social networking sites for data collection has become a trend in several research domains. 
When given permission by a user, it is easy to obtain access to the user's friend list, birthday, public profile, and other relevant information using Facebook APIs.
%
Furthermore, Facebook provides a facility to its users to invite their friends to use any particular application, and hence propagate it. 
%

Owing to the above reasons, in order to obtain the data for our purpose, we developed a Facebook application titled {\em The Perfect Representer} for eliciting the preferences of users over a set of alternatives for a wide range of topics, as well as to obtain the underlying social network. 
Once a user logged into the app, the welcome page as shown in Figure~\ref{fig:intro_page} was presented, which described to the user what was to be expected from the app.

\begin{figure}[t!]
\centering
\includegraphics[scale=.37]{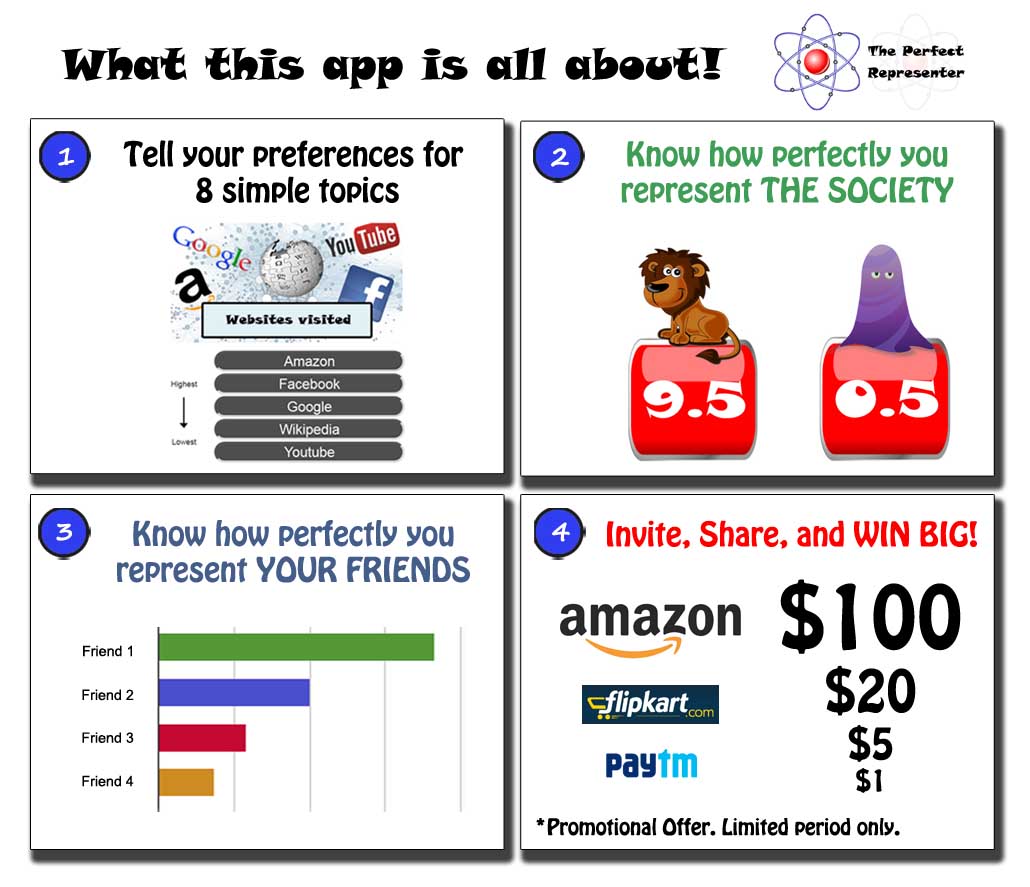}
\caption{Screenshot of the welcome page}
\label{fig:intro_page}
\end{figure}

\begin{figure}[t!]
\centering
\includegraphics[scale=.8]{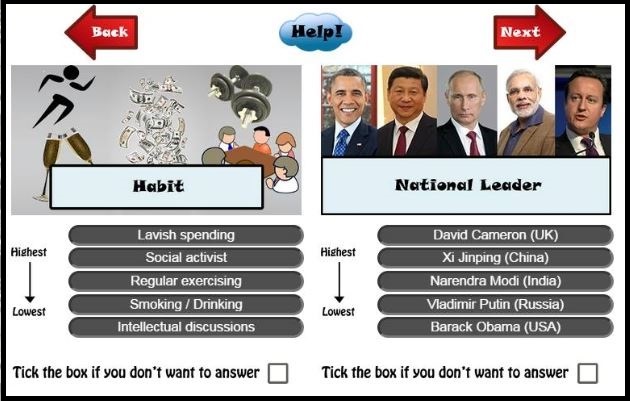}
\caption{Screenshot of a page with topics and their alternatives}
\label{fig:questions_page}
\end{figure}

First, the user would have to give his/her preferences over 5 alternatives for 8 topics, using a drag`n'drop interface as shown in Figure~\ref{fig:questions_page}.
The user was given the option of skipping any particular topic if he/she wished to.
%
From a user's viewpoint, the app gave the user a {\em social centrality score} out of 10, telling how well the user represents the society or how well the user's opinions are aligned with that of the society with respect to the provided preferences.
The score was dynamic and kept on updating as more users used the app (since the aggregate preference itself kept on updating); this score could be posted on the user's timeline.
The user also had an option of seeing how similar his/her preferences were to his/her selected friends. 
Explicit incentives were provided for users to propagate the app either by inviting their friends or sharing on their timelines as well as social networking and other websites (Figure~\ref{fig:sharing_page}). 

To host our application, we used Google App Engine, which provides a Cloud platform for facilitating large number of hits at any given time as well as large and fast storage.

\subsection{The Scores}
\label{sec:scores}

Let $\mathbb{A}$ be the set of alternatives and $r=|\mathbb{A}|$.
Let $p_{it}$ be the preference of user $i$ for topic $t$. Let $\tilde{c}(p,q)$ be the similarity between preferences $p$ and $q$. In our app, we consider $\tilde{c}(p,q)$ to be the normalized Footrule similarity, which we now explain.
Let $w_a^p$ denote the position of alternative $a$ in preference $p$. The Footrule distance between preferences $p$ and $q$ is given by 
$ 
\sum_{a \in \mathbb{A}} |w_a^p - w_a^q|
$. 
With $r$ being the number of alternatives, it can be shown that the maximum possible Footrule distance is
$ 
2 \lceil \frac{r}{2}\rceil \lfloor \frac{r}{2} \rfloor
$. 
So the normalized Footrule distance between preferences $p$ and $q$ can be given by
\begin{displaymath}
\tilde{d}(p,q) = 
\frac{\sum_{a \in \mathbb{A}} |w_a^p - w_a^q|}{2 \lceil \frac{r}{2}\rceil \lfloor \frac{r}{2} \rfloor}
\end{displaymath}
and normalized Footrule similarity by $\tilde{c}(p,q) = 1-\tilde{d}(p,q)$.

For example, the normalized Footrule similarity between preferences $p=(A,B,C,D,E)$ and $q=(B,E,C,A,D)$ is $\tilde{c}(p,q)=\left( 1-\frac{|1-4|+|2-1|+|3-3|+|4-5|+|5-2|}{2 \lceil \frac{5}{2}\rceil \lfloor \frac{5}{2} \rfloor} \right) = \frac{1}{3}$.

\subsubsection{Social Centrality Score - How Perfectly you Represent the Society?}


Let $p_{it}$ be the preference of node $i$ for topic $t$, and $p_{At}$ be the aggregate preference of the population for topic $t$. For the purpose of our app's implementation, we obtain the aggregate preference using the Borda count rule with tie-breaking based on a predefined ordering over alternatives.
For computing the social centrality, we give each topic $t$, a weight proportional to $n_t$ (the number of users who have given their responses for that topic).
So the fractional score of user $i$ is given by
$ 
\sum_t \left( \frac{n_t}{\sum_t n_t} \right) \tilde{c}(p_{it} , p_{At})
$. 

As also mentioned in 
Section
\ref{sec:alt_dist}, 
the distance between two preferences (here, user preference and the aggregate preference) for most topics would be concentrated in the intermediate range, which would forbid a user to get a very high or a very low score. A mediocre score would, in some sense, act as a hurdle in the way of user sharing the post on his/her timeline. So to promote posting their scores, we used a simple boosting rule (square root) and then enhanced it to the nearest multiple of 0.5, resulting in the final score (out of 10) of
\begin{displaymath}
\frac{1}{2}\;  \left\lceil 20 \;\sqrt{\sum_t \left( \frac{n_t}{\sum_t n_t} \right) \tilde{c}(p_{it} , p_{At})} \;\right\rceil
\end{displaymath}

\subsubsection{How Perfectly you Represent your Friends?}

Once a user selected a list of friends to see how similar they are to the user, the app would give the similarity for each friend in terms of percentage. This similarity was also a function of the number of common questions they responded to. So the similarity between nodes $i$ and $j$ in terms of percentage was given by
\begin{displaymath}
100 \cdot \sum_t \tilde{c}(p_{it} , p_{jt})
\end{displaymath}
where $\tilde{c}(p_{it} , p_{jt})=0$ if either $i$ or $j$ or both did not respond to topic $t$.

\subsection{The Incentive Scheme}
\label{sec:incentives}

\begin{figure}[t!]
\centering
\includegraphics[scale=.32]{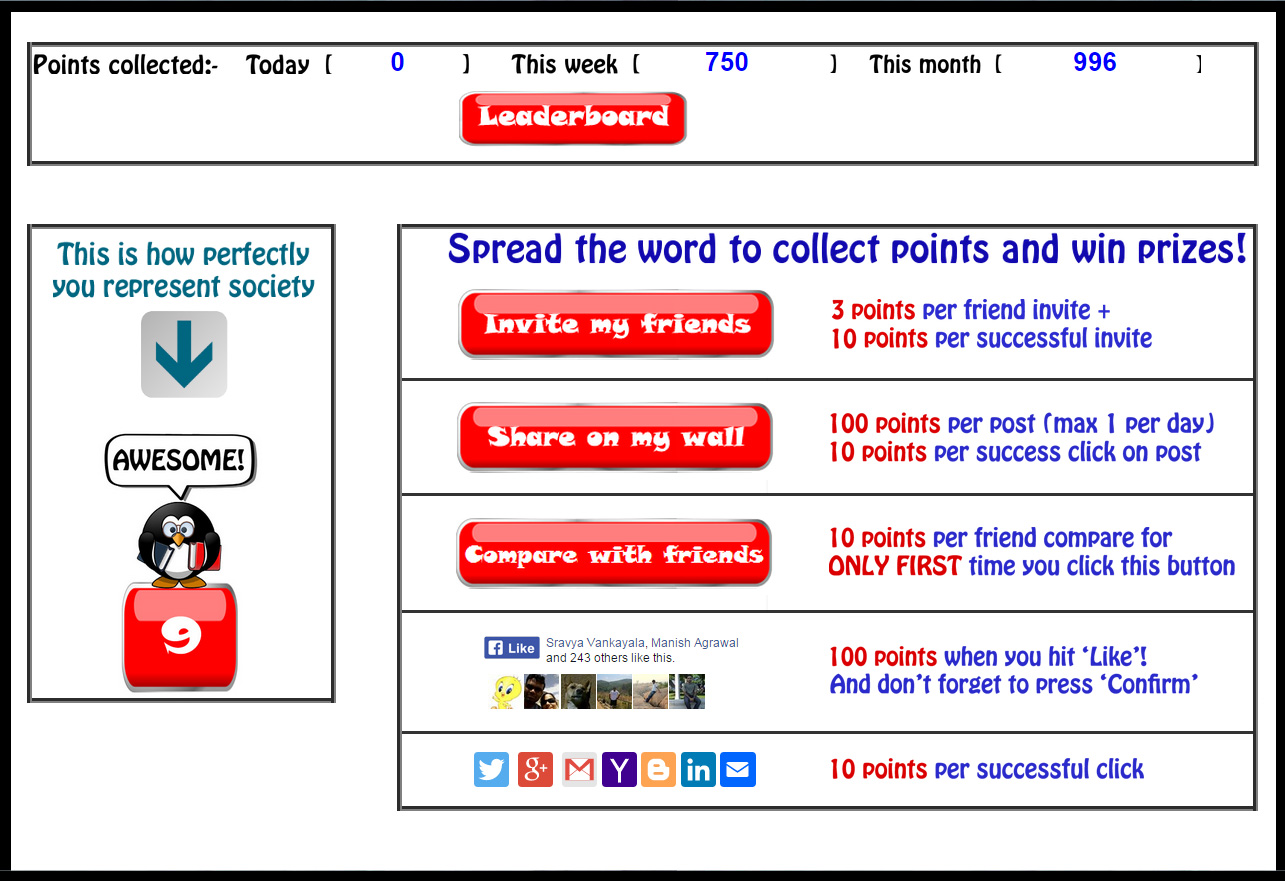}
\caption{Screenshot of the sharing and scoring page}
\label{fig:sharing_page}
\end{figure}

A typical active Facebook user uses several apps in a given span of time, and invites his/her friends to use it depending on the nature of the app. 
In order to ensure a larger reach, it was important to highlight the benefits of propagating our app. We achieved this by designing a simple yet effective incentive scheme for encouraging users to propagate the app by sharing it on their timelines and inviting their friends to use it. 

We incorporated a points system in our app, where suitable points were awarded on a daily, weekly, as well as on an overall basis, for spreading the word about the app through shares, invites, likes, etc. Bonus points were awarded when a referred friend used the app through the link shared by the user. 
 To ensure competitiveness in sharing and inviting, the daily and weekly `top 10' point collectors were updated in real-time and the winners were declared at 12 noon GMT (daily) and Mondays 12 noon GMT (weekly). 
%
A cutoff was set on the number of points to be eligible to get a prize. 
Users were also given a chance to win a big prize through daily, weekly, and bumper lucky draws, if they crossed a certain amount of points, so that users with less number of friends could also put their effort even though they did not have a chance to make it into the `top 10'.
Prizes were awarded in the form of gift coupons so that getting the prize was in itself, quick as well as hassle-free for the users.

The points structure as well as the links to invite, like, and share were given on the scoring page (Figure~\ref{fig:sharing_page}) giving the users a clear picture of how to earn points and win prizes.
%
The lists of daily and weekly winners were displayed on the welcome page as well as on the scoring page.

\end{subappendices}

\chapter{Summary and Future Work}
\label{chap:conclusion}

In this thesis, we proposed three novel problems in the context of social networks, and developed models and methods for solving them. 
First, we addressed the problem of deriving conditions under which autonomous link alteration decisions made by strategic agents lead to the formation of a desired network structure.
Second, we addressed the problem of effective viral marketing by determining how many and which individuals are to be selected in different phases, and what should be the delay between the phases.
Third, we addressed the problem of determining the best representatives using the underlying social network data, by modeling the spread of preferences among the individuals in a social network.
Below, we summarize our results for each problem chapter-wise.

\section{Summary of the Thesis}

\subsection*{Chapter~\ref{chap:nfsc}: \nameref{chap:nfsc}}

In this chapter, we proposed a model of recursive network formation where nodes enter a network sequentially, thus triggering evolution 
each time a new node enters.
We considered a sequential move game model with myopic nodes and pairwise stability as the equilibrium notion, but the proposed model 
is independent of the network evolution model, the equilibrium notion, as well as the utility model.
The recursive nature of our model enabled us to analyze the network formation process using an elegant induction-based technique.

For each of the relevant topologies, we derived sufficient conditions 
for the formation of that topology uniquely, by directing network evolution as desired.
 The derived conditions suggest that conditions on network entry impact degree distribution, while
  conditions on link costs impact density;
  also there arise constraints on intermediary rents owing to contrasting densities of connections in the desired topology.
We then analyzed the social welfare properties of the relevant topologies under the derived conditions,
and studied the effects of deviating from the derived conditions using simulations. 
In the process, we also provided efficient algorithms for computing graph edit distance from certain topologies.

\subsection*{Chapter~\ref{chap:mpid}: \nameref{chap:mpid}}

In this chapter, we proposed and motivated the two-phase diffusion process, formulated an appropriate objective function, proposed an alternative objective function, developed suitable algorithms for selecting seed nodes in both the phases, and observed their performances using simulations. 
We observed that myopic algorithms perform closely to the corresponding farsighted algorithms, while taking orders of magnitude less time.

In order to make the best out of two-phase diffusion, we also studied the problems of budget splitting and scheduling the second phase. We proposed the usage of FACE algorithm for the combined optimization problem. Further, we studied the nature of the plots and owing to their unimodal nature,
we proposed the usage of golden section search technique to find an optimal $<k_1,d>$ pair.
We concluded that: (a) under strict temporal constraints, use single-phase diffusion, (b) under moderate temporal constraints, use two-phase diffusion with a short delay while allocating most of the budget to the first phase, and (c) when there are no temporal constraints, use two-phase diffusion with a long enough delay with a budget split of $1:2$ between the two phases ($1/3$ of the budget for the first phase).
We presented results for a few representative settings; these results are very general in nature.
We then provided notes on the decay function accounting for the temporal constraints, the subadditivity property of the objective function, and how this work can be extended to the linear threshold model.

\subsection*{Chapter~\ref{chap:pasn}: \nameref{chap:pasn}}

This chapter focused on two problems on preference aggregation with respect to social networks, namely, (1) how preferences are spread out in a social network and (2) how to determine the best set of representative nodes in the network.
We started by motivating both these problems. Based on the Facebook app that we had developed, we developed a number of simple and natural models, of which RPM-S set a good balance between accuracy and running time; while MSM-SP was observed to be the best model when our objective was to deduce the mean distances between all pairs of nodes, and not the preferences themselves.

With a model for predicting spread of preferences in a social network at hand, we formulated an objective function for representative-set selection and followed it up with two alternative objective functions for practical usage. 
We then proposed algorithms for selecting best representatives, wherein we provided a guarantee on the performance of the Greedy-min algorithm, subject to the aggregation rule satisfying the expected weak insensitivity property; we also studied the desirable properties of the Greedy-avg algorithm from the viewpoint of cooperative game theory. 
We also observed that degree centrality heuristic performed very well, thus demonstrating the ability of high-degree nodes to serve as good representatives of the population. 
Our main finding is that harnessing social network information in a suitable way will achieve scalable preference aggregation of a large population and will certainly outperform random polling based methods.
We finally provided notes on the effectiveness of our approach for aggregating preferences with respect to personal and social topics, as well as on an alternative way of defining the distribution of preferences between nodes.

\section{Directions for Future Work}

Owing to the novel nature of the studied problems, there are many natural extensions to this thesis.
Below, we list some immediate possible directions for future work for each problem.

\subsection*{Chapter~\ref{chap:nfsc}: \nameref{chap:nfsc}}

\begin{itemize}
\item 
This work proposed a way to derive conditions on the network parameters (namely,
costs for maintaining link with an immediate neighbor as well as for entering a network), 
under which a desired network topology is uniquely obtained. 
It would be interesting to design incentives such that nodes in the network comply with the derived sufficient conditions. 

\item 
%
Based on the concept of dynamic conditions on a network that we introduced, 
it is an interesting future direction to study network formation when the conditions vary over time 
owing to certain internal and external factors.
%

\item 
The proposed model of recursive network formation can also be used in conjunction with other utility models to investigate the formation of interesting topologies under them.

\item 

We ensure 
that irrespective of the 
order in which nodes make their moves,
the network evolution is directed as desired;
a possible solution for simplifying analysis for more involved topologies would be 
to derive conditions so that a network has the desired topology with high probability.

\item 
From a practical viewpoint, one may study the problem of forming networks where the topology need not be exactly the desired one, for example, a near-$k$-star instead of a precise $k$-star.

\item 


%
%
We only derived sufficient conditions under which a desired network topology uniquely emerges. As these conditions serve the purpose of a network designer, the topic of necessary conditions is beyond the scope of this work. From a theoretical viewpoint, however, one may want to obtain a tight characterization by studying conditions that are both necessary and sufficient, that is, the union of all the conditions under which a topology uniquely emerges.
One na\"{i}ve way to obtain such conditions would be to search exhaustively by varying the parameters $c,c_0$ and $\gamma$ with respect to $b_i$'s and each other. As a tractable step towards searching for such conditions, 
one may search in the neighborhood of the values derived in this chapter. 
Although the objective of Section~\ref{sec:deviation} 
was to study the effects on resulting networks owing to temporary deviations of parameters from the derived sufficient conditions, the section can be viewed as a step in this direction.
\end{itemize}

\subsection*{Chapter~\ref{chap:mpid}: \nameref{chap:mpid}}

\begin{itemize}
\item 
This work can be extended to study diffusion 
in more than two phases
 and compare their performances against two-phase diffusion, with respect to the influence spread and time taken. 
%

 \item 
 We focused on the well-studied IC model and provided a note regarding the LT model; studying multi-phase diffusion under other diffusion models is another direction to look at.
 %
 
\item 
It would be useful to study how multi-phase diffusion can be harnessed to get a desired expected spread with a reduced budget.
%

\item 
It would be of theoretical interest to prove or disprove if there exists an algorithm that gives a constant factor approximation for the problem of two-phase influence maximization.

\item 
It would also be interesting to study equilibria in a game theoretic setting where multiple campaigns consider the possibility of multi-phase diffusion.
\end{itemize}


%
%

\subsection*{Chapter~\ref{chap:pasn}: \nameref{chap:pasn}}

\begin{itemize}
\item 
A primary objective of this work was to select $k$ nodes so as to minimize $\mathbb{E} [ f(P) \; \Delta \; f(R) ]$. This work can be extended to select minimum number of nodes such that this value is bounded with high confidence.

\item 
We empirically observed expected weak insensitivity for various aggregation rules.
This property maybe of prime importance in social choice theory and so, it will be interesting to analytically determine the aggregation rules that satisfy it.
%

\item 
We used a particular form of modified profile $R=Q'$ (profile consisting of  preferences of representatives, multiplied by the number of nodes they represent). It will be interesting to study the `best' form of $R$.
%
%

\item 
We assumed that the voters are not strategic and so report their preferences truthfully. From a game theoretic viewpoint, it would be interesting to look at the strategic aspect of the problem.

\item
We conducted the study based on the distribution of distances between nodes with respect to their preferences. 
It would be interesting to study alternative ways of modeling the distribution of preferences between nodes.

\item 
General random utility models are complementary to our model, exploiting attributes of nodes and alternatives instead of the underlying social network. It will be interesting to consider the attributes as well as the underlying social network for determining the best representatives.

\end{itemize}


%
%
\blankpagewithnumber
\bibliographystyle{plainnat}
\bibliography{references}
\blankpagewithnumber
\newpage
~ \vspace{5cm}

\noindent
This thesis was all about social networks, and online social networks (OSNs) are an effective way of implementing the methods suggested in this thesis. In today's age, OSNs are something which simply cannot be ignored. Facebook in particular, has now become an integral part of most people's lives. It has significantly impacted the way an individual's image is formed in his or her social circle with respect to reputation, viewpoints, lifestyle, etc. I conclude this thesis with a few lines on Facebook, because without OSNs like it, this thesis would not have existed.\\

\begin{center}
\begin{tabular}{l}
\textit{You must have many Facebook friends to prove you are social,} \\
\textit{You must add pics on Facebook to prove you are essential,} \\
\textit{You must share posts on Facebook to prove you are an activist,} \\
\textit{And you must exist on Facebook to prove that you exist.} \\
\end{tabular}
\end{center}

\vspace{5mm}
\begin{center}
-- THE END --
\end{center}

\end{document}